\documentclass{dis}
\usepackage{amsmath}
\usepackage{makeidx}

\makeindex

\usepackage{graphicx,amssymb} 
\usepackage{epsfig}
\usepackage{mathrsfs}

\usepackage{pstricks,pst-eucl}
\usepackage{graphicx}
\usepackage{bbm}

\DeclareMathAlphabet      {\mathbfit}{OML}{cmm}{b}{it}
\let\bm=\mathbfit

\newtheorem{theorem}{Theorem}[chapter]
\newtheorem{lemma}[theorem]{Lemma}
\newtheorem{prop}[theorem]{Proposition}
\newtheorem{cor}[theorem]{Corollary}

\theoremstyle{definition}
\newtheorem{defn}[theorem]{Definition}
\newtheorem{example}[theorem]{Example}

\theoremstyle{remark}
\newtheorem{remark}[theorem]{Remark}

\numberwithin{section}{chapter}
\numberwithin{equation}{chapter}

\numberwithin{equation}{section}

\newcommand{\re}{\operatorname{Re}}
\newcommand{\im}{\operatorname{Im}}

\def\supp{\mathop{\rm supp}}

\let\text=\mbox

\catcode`\@=11
\let\ced=\c
\def\a{\alpha}
\def\b{\beta}
\def\c{\gamma}
\def\d{\delta}
\def\g{\lambda}
\def\o{\omega}
\def\q{\quad}

\def\s{\sigma}

\def\H{{\cal H}}

\def\M{{\cal M}}

\def\P{{\mathcal P}}

\def\eR{{\bf R}}
\def\eN{{\bf N}}
\def\Ze{{\bf Z}}

\def\Ce{{\bf C}}

\def\ty{\infty}
\def\e{\varepsilon}
\def\f{\varphi}
\def\:{{\penalty10000\hbox{\kern1mm\rm:\kern1mm}\penalty10000}}
\def\ov#1{\overline{#1}}

\def\di{\mathrm{d}}
\def\O{\Omega}
\def\pa{\partial}

\def\st{\subset}
\def\stq{\subseteq}

\newcount\br@j
\def\q{\quad}

\def\bg{\begin}

\def\endeqnn{\end{eqnarray*}}
\def\bgeqn{\bg{eqnarray}}
\def\endeqn{\end{eqnarray}}

\def\bgeqq#1#2{\bgeqn\label{#1} #2\left\{\begin{array}{ll}}
\def\endeqq{\end{array}\right.\endeqn}

\def\abstract{\bgroup\leftskip=2\parindent\rightskip=2\parindent
        \noindent{\bf Abstract.\enspace}}
\def\endabstract{\par\egroup}

\def\udesno#1{\unskip\nobreak\hfil\penalty50\hskip1em\hbox{}
             \nobreak\hfil{#1\unskip\ignorespaces}
                 \parfillskip=\z@ \finalhyphendemerits=\z@\par
                 \parfillskip=0pt plus 1fil}
\catcode`\@=11

\def\cal{\mathcal}
\def\eR{\mathbb{R}}
\def\eN{\mathbb{N}}
\def\Ze{\mathbb{Z}}
\def\Qu{\mathbb{Q}}
\def\Ce{\mathbb{C}}

\def\res{\operatorname{res}}

\def\po{\mathcal{P}}
\def\I{\mathbbm{i}}
\def\E{\mathrm{e}}
\def\D{\mathrm{d}}
\def\qs{\q}

\def\upGamma{\Gamma}

\chaptersnewpage

\refnewpage

\begin{document}


\mathclass{Primary 11M41, 28A12, 28A75, 28A80, 28B15,  30D10, 42B20, 44A05; 
Secondary 11M06, 30D30, 37C30, 37C45, 40A10, 44A10, 45Q05.}

\thanks{The work of Michel L.\ Lapidus was partially supported by the US National Science Foundation (NSF) under the research grants DMS-0707524 and DMS-1107750, as well as by the Institut des Hautes Etudes Scientifiques (IHES) in Paris/Bures-sur-Yvette, France, where the first author was a visiting professor in the Spring of 2012 while part of this work was completed. The research of Goran Radunovi\'c and Darko \v Zubrini\'c was supported by the Croatian Science Foundation under the project IP-2014-09-2285 and by the Franco-Croatian 
PHC-COGITO~project.}

\abbrevauthors{M.\ M.\ Lapidus, G.\ Radunovi\'c and D.\ \v Zubrini\'c}
\abbrevtitle{Zeta Functions and Complex Dimensions of RFDs}

\title{Zeta Functions and Complex Dimensions\\ of Relative Fractal Drums:\\ Theory, Examples and Applications}

\author{Michel L.\ Lapidus}
\address{University of California\\
Department of Mathematics\\
900 University Avenue\\
231 Surge Building\\
Riverside, CA 92521-0135\\
USA\\
E-mail: lapidus@math.ucr.edu}

\author{Goran Radunovi\'c}
\address{University of Zagreb\\
Faculty of Electrical Engineering and Computing\\
Department of Applied Mathematics\\
Unska 3\\ 
10000 Zagreb\\
Croatia\\
E-mail: goran.radunovic@fer.hr}

\author{Darko \v Zubrini\'c}
\address{University of Zagreb\\
Faculty of Electrical Engineering and Computing\\
Department of Applied Mathematics\\
Unska 3\\ 
10000 Zagreb\\
Croatia\\
E-mail: darko.zubrinic@fer.hr}

\maketitledis

\begin{abstract}
In 2009, the first author introduced a new class of zeta functions, called `distance zeta functions', associated with arbitrary
compact fractal subsets of Euclidean spaces of arbitrary dimension. It represents a natural, but nontrivial extension of 
the existing theory of `geometric zeta functions' of bounded fractal strings.
In this work, we introduce the class of `relative fractal drums' (or RFDs), which contains the classes of bounded fractal strings and of compact fractal subsets
of Euclidean spaces as special cases.
Furthermore, the associated (relative) distance zeta functions of RFDs, extend 
the aforementioned classes of fractal zeta functions.
The abscissa of (absolute) convergence of any relative fractal drum is equal to the relative box dimension of the RFD. We pay particular attention to the question of constructing meromorphic
extensions of the distance zeta functions of RFDs, as well as to the construction of transcendentally $\infty$-quasiperiodic RFDs. 
We also describe a class of RFDs, 
called {\em maximal hyperfractals}, such that the critical line of convergence consists solely of nonremovable singularities of the associated relative distance zeta functions. Finally, we also describe a class of Minkowski measurable RFDs which possess an infinite sequence of complex dimensions of arbitrary multiplicity $m\ge1$, and even an infinite sequence of essential singularities along the critical line.
\end{abstract}
\makeabstract

\tableofcontents

\chapter*{Glossary}

{\parindent=0pt \advance\parskip by 6pt 


$a_k\sim b_k$ as $k\to\ty$, asymptotically equivalent sequences of complex numbers\dotfill\pageref{sim}



$f\sim g$, equivalence of the DTI $f$ and the meromorphic function $g$\dotfill\pageref{equ}




$A_{\mathcal L}:=\{a_j:=\sum_{k=j}^\ty\ell_k:j\in\eN\}$, canonical geometric representation\newline 
\hbox{}\hskip1cm of a bounded fractal string $\mathcal L=(\ell_k)_{k\ge1}$\dotfill\pageref{A_L}

$A_t:=\{x\in\eR^N:d(x,A)<t\}$, $t$-neighborhood of a subset $A$ of $\eR^N$ ($t>0$)\dotfill\pageref{Atn}


$A_N$, the inhomogeneous Sierpi\'nski $N$-gasket\dotfill\pageref{inh-gasket}

$(A_N,\O_N)$, the inhomogeneous Sierpi\'nski $N$-gasket RFD\dotfill\pageref{inh-gasketRFD}


$(A,\Omega)$, relative fractal drum\dotfill\pageref{rfd}

$(A,\O)_M:=(A_M,\O\times(-1,1)^{M})$, embedding of the relative fractal\newline 
\hbox{}\hskip1cm   drum $(A,\O)$ of $\eR^N$ into $\eR^{N+M}$ ($M\in\eN$)\dotfill\pageref{invar_rel}




$\mathrm{B}(a,b):=\int_0^{1}t^{a-1}(1-t)^{b-1}\di  t$, the Euler beta function\index{beta function, $B(a,b)$|textbf}\dotfill\pageref{beta_func}

$\mathrm{B}_x(a,b):=\int_0^{x}t^{a-1}(1-t)^{b-1}\di  t$, the incomplete beta function\index{beta function, $B(a,b)$!incomplete beta function, $\mathrm{B}_x(a,b)$|textbf}\dotfill\pageref{incbeta}

$B_r(a)$, open ball in $\eR^N$ (or in $\Ce$) of radius $r$ and with center at $a\in\eR^N$ (or in $\Ce$)\dotfill\pageref{spherer}



$C^{(m,a)}$, the generalized Cantor set with two parameters $m$ and $a$\dotfill\pageref{Cma}











$\underline\dim_BA$, $\ov\dim_BA$, $\dim_BA$, the box dimensions (Minkowski dimensions) of $A\subset\eR^N$\dotfill\pageref{dim}

$\underline\dim_B(A,\Omega)$, $\ov\dim_B(A,\Omega)$, $\dim_B(A,\Omega)$, the relative box dimensions of $(A,\O)$\dotfill\pageref{minkrel}


$\dim_{PC} A$, the set of principal complex dimensions of a bounded subset $A$ of $\eR^N$\dotfill\pageref{dimc} 

$\dim_{PC}(A,\O)$, the set of principal complex dimensions of the RFD $(A,\O)$\dotfill\pageref{dimcr}

$\dim_{PC}\mathcal {L}$, the set of principal complex dimensions\newline 
\hbox{}\hskip1cm of a fractal string $\mathcal {L}=(\ell_j)_{j\ge1}$\dotfill\pageref{dimcs}




$d(x,A):=\inf\{|x-a|:a\in A\}$, Euclidean distance from $x$ to $A$ in $\eR^N$\dotfill\pageref{d(x,A)}


$D(f)$, the abscissa of (absolute) convergence of the Dirichlet series or integral\dotfill\pageref{Dabs}


$D_{\rm hol}(f)$, the abscissa of holomorphic continuation of $f$\dotfill\pageref{Dhol}

$D_{\rm mer}(f)$, the abscissa of meromorphic continuation of $f$\dotfill\pageref{Dmer}


$\pa\O$, the boundary of a subset $\O$ of $\eR^N$\dotfill\pageref{boundary_of_a_set}

$\mathcal{D}_{\rm qp}$, the family of quasiperiodic relative fractal drums\dotfill\pageref{transcendental_drum}

$\mathcal{D}_{\rm aqp}$, the family of algebraically quasiperiodic relative fractal drums\dotfill\pageref{transcendental_drum}

$\mathcal{D}_{\rm tqp}$, the family of transcendentally quasiperiodic relative fractal drums\dotfill\pageref{transcendental_drum}

DTI, Dirichlet-type integral\dotfill\pageref{DTI3}




$|E|=|E|_N$, the $N$-dimensional Lebesgue measure of a measurable set $E\st\eR^N$\dotfill\pageref{Nmeasure}


$\mathbf{e}(m)$, the  exponent sequence of an integer $m\ge2$\dotfill\pageref{es}





$\operatorname{fl}(A,\O)$, the flatness of a relative fractal drum $(A,\O)$\dotfill\pageref{flatness}






$\upGamma(t):=\int_0^{+\ty} x^{t-1}\E^{-x}\,\D x$, the gamma function\index{gamma function, $\upGamma(t)$|textbf}\dotfill\pageref{gammaf}

$h$, gauge function\dotfill\pageref{gauge}

$H^D(A)$, the $D$-dimensional Hausdorff measure of the set $A$\dotfill\pageref{MHD}






$\I=\sqrt{-1}$, the imaginary unit\dotfill\pageref{sqrt-1}

IFS, iterated function system\dotfill\pageref{IFS_label}












$\mathcal{L}=(\ell_j)_{j=1}^\ty$, a fractal string with lengths $\ell_j$\dotfill\pageref{ellj}

$(\ell_j)_{j=1}^\ty$, sequence of lengths of a fractal string $\mathcal{L}$ written in nonincreasing\newline
\hbox{}\hskip1cm order\dotfill\pageref{ellj}


$\log_a x$, the logarithm of $x>0$ with base $a>0$; $y=\log_a x$\, $\Leftrightarrow$\, $x=a^y$\dotfill\pageref{logax}

$\log x:=\log_{\E}x$, the natural logarithm of $x$; $y=\log x$\, $\Leftrightarrow$\, $x=\E^y$\dotfill\pageref{logax}



$\ell^\ty(\eR)$, the Banach space of bounded sequences of real numbers $(\tau_j)_{j\ge1}$\dotfill\pageref{banach_bdd}

$\mathcal {L}_1\otimes\mathcal {L}_2$, the tensor product of two bounded fractal strings\dotfill\pageref{otimes}






$\M_*^r(A)$ and $\M^{*r}(A)$, the lower and upper $r$-dimensional Minkowski contents\newline
\hbox{}\hskip1cm of a bounded set $A\subset\eR^N$, where $r\ge0$\dotfill\pageref{mink}

$\M^D(A):=\lim_{t\to0^+}\frac{|A_t|}{t^{N-D}}$, the $D$-dimensional Minkowski content of a Minkowski\newline
\hbox{}\hskip1cm measurable bounded set $A\subseteq\eR^N$\dotfill\pageref{minkc}

$\M_*^r(A,\Omega)$ and $\M^{*r}(A,\Omega)$, the lower and upper relative $r$-dimensional Minkowski\newline
\hbox{}\hskip1cm contents of the relative fractal drum $(A,\O)$, where $r\in\eR$\dotfill\pageref{minkrel}

$\M_*^D(A,\Omega,h)$, $\M^{*D}(A,\Omega,h)$, the gauge relative lower and upper Minkowski 
\newline
\hbox{}\hskip1cm contents of $(A,\O)$ (with respect to the gauge function $h$)\dotfill\pageref{gauge_content}

$\M^D(A,\O):=\lim_{t\to0^+}\frac{|A_t\cap\O|}{t^{N-D}}$, the $D$-dimensional Minkowski content of a Minkowski\newline
\hbox{}\hskip1cm measurable relative fractal drum $(A,\O)$ in $\eR^N$\dotfill\pageref{minkc}



$\operatorname{Mer}\,(f):=\{\re s>D_{\rm mer}(f)\}$, the half-plane of meromorphic continuation of $f$\dotfill\pageref{Dmer}






$\eN:=\{1,\ 2,\ 3,\ldots\}$, the set of positive integers\dotfill\pageref{eN}

$\eN_0:=\eN\cup\{0\}=\{0,\ 1,\ 2,\ 3,\ldots\}$, the set of nonnegative integers\dotfill\pageref{eN}

$(\eN_0)_c^{\ty}$, the set of all sequences $\mathbf{e}$ with components in $\eN_0=\eN\cup\{0\}$, such that\newline
\hbox{}\hskip1cm all but at most finitely many components are equal to zero\dotfill\pageref{supp}


$\binom{|\a|}{\a_1,\a_2,\ldots,\a_n}:=\frac{|\a|!}{\a_1!\,\a_2!\cdots\a_n!}$, with $|\a|:=\a_1+\cdots+\a_n$, multinomial coefficient\dotfill\pageref{multinom_l}





$\omega_N:=|B_1(0)|_N$, the $N$-dimensional Lebesgue measure of the unit ball in $\eR^N$\dotfill\pageref{omega_N}

$\O_{\rm can}=\O_{{\rm can},\mathcal L}$, canonical geometric realization of a fractal string $\mathcal L$\dotfill\pageref{Oak}



$\otimes$, tensor product of a base relative fractal drum and a fractal string\dotfill\pageref{tensor1}


$\mathbf{p}$, the oscillatory period of a lattice self-similar set (or string or spray  \newline
\hbox{}\hskip1cm or RFD)\dotfill\pageref{operiod}


$\mathcal {P}(f)$, the set of poles of a meromorphic function $f$\dotfill\pageref{pof}

$\mathcal {P}(f,W)$, the set of poles of a meromorphic function $f$ contained in the \newline
\hbox{}\hskip1cm interior of the set $W\subseteq\Ce$\dotfill\pageref{pof}

$\mathcal {P}_c(f)$, the set of poles of a (tamed) Dirichlet-type integral (i.e., DTI) $f$ on \newline
\hbox{}\hskip1cm the critical line\dotfill\pageref{pof}






RFD, relative fractal drum $(A,\O)$ in $\eR^N$\dotfill\pageref{drum}


$S_N$, the classic $N$-dimensional Sierpi\'nski gasket\dotfill\pageref{classicNgasket}

$\bm S$, screen\dotfill\pageref{screen}







$\operatorname{Spray}(\Omega_0,(\lambda_j)_{j\ge1},(a_j)_{\ge1})$, a relative fractal spray in $\eR^N$\dotfill\pageref{spray}


$\supp\mathbf{e}$, the support of a sequence $\mathbf{e}\in(\eN_0)_c^\ty$\dotfill\pageref{supp}

$\supp m$, the support of an integer $m\ge2$\dotfill\pageref{supp}






$\bigsqcup_{j=1}^\ty\mathcal {L}_j$, the union of a countable family of fractal strings\dotfill\pageref{union_s}

$\bigsqcup_{j\in J}(A_j,\O_j)$, the disjoint union of a countable family of relative fractal \newline
\hbox{}\hskip1cm drums\dotfill\pageref{sqcup}








$\bm W$, window\dotfill\pageref{window}







$\zeta_{\mathcal {L}}(s):=\sum_{j=1}^\ty(\ell_j)^s$, the geometric zeta function of a fractal string $\mathcal {L}$\dotfill\pageref{string}

$\zeta_A(s):=\int_{A_\d}d(x,A)^{s-N}\D x$, the distance zeta function of a bounded\newline
\hbox{}\hskip1cm subset $A$ of $\eR^N$\dotfill\pageref{z}

$\tilde\zeta_A(s):=\int_0^\d t^{s-N-1}|A_t|\,\D t$, the tube zeta function of a bounded subset  \newline
\hbox{}\hskip1cm  $A$ of $\eR^N$\dotfill\pageref{tildz}

$\zeta_A(s,\Omega)=\zeta_{A,\O}(s):=\int_{\O}d(x,A)^{s-N}\,\D x$, the relative distance zeta function of \newline 
\hbox{}\hskip1cm a relative fractal drum $(A,\O)$\dotfill\pageref{rel_dist_zeta}

$\tilde\zeta_A(s,\Omega)=\tilde\zeta_{A,\O}(s):=\int_0^\d t^{s-N-1}|A_t\cap\O|\,\D t$, the relative tube zeta function of\newline 
\hbox{}\hskip1cm a relative fractal drum $(A,\O)$\dotfill\pageref{rel_tube_zeta}



$\zeta_{A,\O}^{\mathfrak{M}}(s):=\int_0^{+\ty}t^{s-N-1}|A_t\cap\O|\,\di t$, the Mellin zeta function of a relative\newline 
\hbox{}\hskip1cm  fractal drum $(A,\O)$\dotfill\pageref{mellinz}








$\zeta_{\mathfrak{S}}(s)$, the scaling zeta function (of a fractal spray or of a self-similar RFD)\dotfill\pageref{scalingzf}


}

%

\chapter*{Preface}

The purpose of this work is to develop the theory of complex dimensions for arbitrary compact subsets $A$ of Euclidean spaces $\eR^N$, of arbitrary dimension $N\ge1$. To this end, in 2009, the first author has introduced a new class of zeta functions, called {\em distance zeta functions} $\zeta_A$ of fractal sets $A$, the poles of which (after $\zeta_A$ has been suitably meromorphically extended) are defined as the {\em complex dimensions} of $A$. This notion establishes an important bridge between the geometry of fractal sets, Number Theory and Complex Analysis. 
\medskip

The development of the higher-dimensional theory of complex dimensions of fractal sets has led us to the discovery of the {\em tube zeta functions} $\tilde\zeta_A$ of fractal sets, which are not only a valuable technical tool, but a natural companion of the distance zeta functions $\zeta_A$. These two fractal zeta functions are connected by a simple functional equation, which shows that, in this generality, the theory of complex dimensions can be developed indifferently from the point of view of the distance or else of the tube zeta functions. Both the distance and tube zeta functions enable us to extend in a nontrivial way the existing theory of geometric zeta functions $\zeta_{\mathcal L}$ of bounded fractal strings ${\mathcal L}$. An even broader perspective is achieved by introducing the so-called {\em relative fractal drums} (RFDs) $(A,\O)$ in Euclidean spaces, which extend the notions of bounded fractal sets in $\eR^N$, as well as of bounded fractal strings. The associated {\em relative fractal zeta functions} $\zeta_{A,\O}$ enable us to consider the theory of fractal zeta functions from a unified perspective. An unexpected novelty is that a relative fractal drum $(A,\O)$ can have a (naturally defined) {\em Minkowski $($or box$)$ dimension} $\dim_B(A,\O)$ of {\em negative} value (and even of value $-\ty$), or more generally, that its principal complex dimensions (i.e., the poles of $\zeta_{A,\O}$ on the critical line $\{\re s=D\}$, where $D=\ov{\dim}_B(A,\O)$ is the upper Minkowski dimension of $(A,\O)$) can have negative real parts. 
\medskip

The residue of a fractal zeta function, computed at the value $D$ of the abscissa of (absolute) convergence of the zeta function (i.e., at the Minkowski dimension), is very closely related to the {\em Minkowski content} of the corresponding bounded set or RFD.
Furthermore, we also study the quasiperiodicity of relative fractal drums, by using a classical result from (transcendental) analytic number theory, due to Alan Baker. Roughly, for any given positive integer $n$, it is possible to construct a fractal set with $n$ algebraically independent quasiperiods; as a result, we obtain a {\em transcendentally $n$-quasiperiodic set}. Moreover, we can even construct transcendentally $\ty$-quasiperiodic sets, i.e., fractal sets with infinitely many algebraically independent quasiperiods.
\smallskip

Towards the end of this article, special emphasis is given to the construction of fractal sets $A$ which have principal complex dimensions (i.e., the poles of the distance zeta function $\zeta_A$ with real part equal to $D=\ov{\dim}_BA$) of any given multiplicity~$m\ge2$ and even, with `infinite multiplicity' $m=\ty$; i.e., in this case, the principal complex dimensions of $A$ are, in fact, essential singularities of its distance zeta function $\zeta_A$.
\smallskip
\newpage

Finally, we also construct fractal sets $A$ in $\eR^N$, which we call {\em maximal hyperfractals}, such that the corresponding distance zeta function has the entire critical line of (absolute) convergence $\{\re s=D\}$ as the set of its nonremovable singularities.

We conclude this paper by a discussion of the notion of ``fractality'', formulated in terms of the present higher-dimensional theory of complex dimensions.
Furthermore, we illustrate this discussion by means of an RFD suitably associated with the Cantor graph (or the ``devil's staircase'').   
\bigskip

\hbox to\hsize{\hfill August 18, 2016}
\medskip

\hbox to\hsize{Riverside, California, USA and Paris, France\hfill {\em Michel L.\ Lapidus}}

\hbox to \hsize{Zagreb, Croatia\hfill {\em Goran Radunovi\'c} and {\em Darko \v Zubrini\'c}}


\mainmatter

\index{RFD|see{relative fractal drum}}
\index{DTI|see{Dirichlet-type integral (DTI)}}
\index{negative box dimension|see{flatness of a relative fractal drum}}
\index{zeta function|see{fractal zeta function}}
\index{dimension|seealso{fractal dimension}}
\index{Hausdorff dimension|seealso{fractal dimension}}
\index{nondegenerate RFD|see{Minkowski nondegenerate RFD}}
\index{devil's staircase|see{Cantor graph (full)}}



\overfullrule=0pt

\chapter{Introduction}\label{intro}

\section{The development of the idea of dimension: From integers to complex numbers}

The development of the mathematical ideas behind the concept of {\em dimension} started in the 19th century, with the need to precisely define some basic notions like the `line' and `surface'. Its history can be very roughly subdivided into the following three parts, all of them deeply interlaced: the history of integer dimensions, fractal dimensions, and complex dimensions.

\subsection{Integer dimensions}\index{dimension!history of|(} Until the beginning of the 20th century, the notion of `dimension' has been in use exclusively as a {\em nonnegative integer}. It was rigourously defined in the 19th century, first for linear objects and then, for manifolds; i.e., in the area of linear algebra (where it was defined as the number of elements of any base of a given linear space), as well as in differential and algebraic geometry. Soon, several other integer dimensional quantities have been introduced, in order to study arbitrary subsets of Euclidean spaces (and, more generally, of topological spaces). These basic dimensional quantities are known as the small inductive dimension (Menger--Urysohn), the large inductive dimension (Brouwer--\v Cech) and the covering dimension (\v Cech--Lebesgue). A history of the extremely complex subject of integer dimensions appearing in Topology is given in the survey article \cite{cr}.

\subsection{Fractal dimensions}\index{fractal dimension!history of|textbf} The foundations of the theory of {\em fractal dimensions} have already been laid out in the 1920s, in the works of Minkowski, Hausdorff, Besicovich and Bouligand, by introducing (suitably defined) dimensions which can assume {\em noninteger} (more specifically, nonnegative real) values, in order to better understand the geometric properties of very general subsets of Euclidean spaces. These developments resulted in the {\em Hausdorff dimension} and the {\em Minkowski dimension} or the {\em Minkowski--Bouligand dimension} (also called the {\em box dimension}, the notion that we adopt in this paper), which have become essential tools of modern {\em Fractal Geometry}. 

Many distinguished scholars contributed in various ways to popularizing and developing these ideas, 
and thereby, in particular, to the introduction of the seemingly counterintuitive
concept of fractal dimension; there are too many of them to name them all here. (See, for example, \cite[Chapter XI]{Man} or [Fal1].) In addition, the methods of Fractal Geometry are today extremely developed and frequently used in various specialized scientific fields, both from the theoretical and applied points of view. 
An overview of the early history of Fractal Geometry and the development of its main ideas can be found in \cite{Man}. (See also \cite{lap7}.)

\subsection{Complex dimensions} The idea of introducing {\em complex dimensions} (more specifically, of {\em complex numbers} as dimensions) as a quantification of the inner (oscillatory) geometric properties of objects called {\em bounded fractal strings} $\mathcal L$, has been proposed in the beginning of the 1990s by the first author of this paper, based in part on earlier work in [{Lap1--3}, {LapPo1--2}, {LapMa1--2}].
Very roughly, bounded fractal strings can be identified with certain bounded subsets of the real line. In order to define the complex dimensions of a given bounded fractal string $\mathcal L$, one has to assign to it the corresponding (geometric) {\em zeta function} $\zeta_{\mathcal L}$. The `complex dimensions' of bounded fractal strings are then defined as the poles of a suitable meromorphic extension of the geometric zeta function in question. The development of the main ideas and results behind the mathematical theory of complex dimensions of fractal strings can be found in \cite{lapidusfrank12}. 

It is natural to ask the following question: Is it possible to define the `complex dimensions' for {\em any} (nonempty) bounded subset $A$ of Euclidean space?
In other words, is there a natural zeta function $\zeta_A$, such that its poles can be considered as the `complex dimensions' of a given set $A$ (assuming that a suitable meromorphic extension of $\zeta_A$ is possible)?
The answer to this question has been obtained by the first author in 2009, by introducing a class of {\em distance zeta functions} $\zeta_A$, as we call them in this work.

As the result of a collaboration between the authors of this paper, initiated by the first author in 2009, it soon became clear that the notion of `complex dimensions' can be introduced not only for bounded subsets of Euclidean spaces, but even for much more general geometric objects, denoted by $(A,\O)$, which we call {\em relative fractal drums} (RFDs). (An example of relative fractal drum is given by $(\pa\O,\O$), where $\O$ is a bounded open subset of some Euclidean space $\eR^N$, with $N\ge1$, while $\pa\O$ is its topological boundary. The special case when $N=1$ precisely corresponds to a bounded fractal string.) The study of the {\em complex dimensions of relative fractal drums} is the main goal of the present paper. The flexibility of this notion, as well as of the corresponding notion of relative distance zeta function $\zeta_{A,\O}$, has enabled us to view the existing theory of complex dimensions of fractal strings (and their generalizations to fractal sprays) from a unified perspective and to extend it beyond recognition. An unexpected novelty was the possibility 
for the relative Minkowski dimensions of some classes of relative fractal drums to take negative values (including the value $-\ty$).

We should mention that the set of complex dimensions of a given RFD (and, in particular, of a given bounded set of $\eR^N$) is always a discrete subset of $\Ce$, and hence, consists of a (finite or countable) sequence of complex numbers, with finite multiplicities (as poles of the corresponding fractal zeta functions).
In the future, we may extend this notion to also include the possible essential singularities of the corresponding fractal zeta functions.
Indeed, in this paper, we construct fractal sets and RFDs whose fractal zeta functions have infinitely many essential singularities.

The theory of complex dimensions of relative fractal drums, developed in this work, provides a useful bridge between Fractal Geometry, Number Theory and Complex Analysis.
It has brought to light numerous interesting questions and new challenging problems for further research; see, especially, \cite[Chapter 6]{fzf} and \cite[\S8]{brezish}.\index{dimension!history of|)}

\section{Relative fractal drums and their distance zeta functions}

In 2009, the first author has introduced a new class of zeta functions $\zeta_A$, called `distance zeta functions',\index{distance zeta function!of a set $A$, $\zeta_A$|textbf}\index{fractal zeta function!distance zeta function, $\zeta_A$|textbf} associated with arbitrary compact subsets $A$
of a given Euclidean space $\eR^N$ of arbitrary dimension $N$.
More specifically, the {\em distance zeta function} $\zeta_A$ of a bounded set $A\subset\eR^N$ is defined by
\begin{equation}\label{z}
\zeta_A(s):=\int_{A_\d}d(x,A)^{s-N} \D x,
\end{equation}
for all $s\in\Ce$ with $\re s$ sufficiently large, where $\d$ is a fixed positive real number
and $A_\d$ is the Euclidean $\d$-neighborhood\label{Atn} of $A$. (Here, $d(x,A):=\inf\{|x-a|:a\in A\}$ denotes the Euclidean distance\label{d(x,A)} from $x$
 to $A$ and the integral is understood in the sense of Lebesgue and is therefore absolutely convergent.) These new fractal zeta functions have been studied in [LapRa\v Zu2,3], as well as in the research monograph \cite[Chapter 2]{fzf}.

We extend the class of distance zeta functions from the family of compact subset of $\eR^N$ to a new class of objects that we call `relative fractal drums'
(RFDs) in $\eR^N$ (still for any $N\ge1$);
see Definition \ref{zeta_r} below. 
This enables us to provide a unified approach to the study of fractal zeta functions. An unexpected novelty is 
that RFDs may have an upper box (or Minkowski) dimension (defined by \eqref{dimrelu} below), which is {\em negative}, or is even equal to $-\ty$; see Proposition~\ref{ndim} and Corollary~\ref{flat} below.

\begin{defn}\label{rfd}
Let $\O$ be an open subset of $\eR^N$, possibly unbounded, but of finite $N$-dimensional Lebesgue measure.
Assume that $A$ is a subset of $\eR^N$ and
\begin{equation}\label{delta}
\mbox{
there exists $\d>0$ such that $\O\stq A_\d$.}
\end{equation}
 We then say that the ordered pair $(A,\O)$ is a {\em relative fractal drum}\label{drum}\index{relative fractal drum, RFD|textbf} (or an RFD, in short) in $\eR^N$.
\end{defn}

We stress that when working with an RFD $(A,\O)$, we always assume that both $A$ and $\O$ are nonempty.

Relative fractal drums represent a natural extension of the following classes of objects, simultaneously: 
\medskip

($a$) The class $\mathtt{STR}_b$ of (nonempty) {\em bounded fractal strings}\label{bounded fractal string}\index{fractal string, $\mathcal L$|textbf} $\mathcal {L}:=(\ell_j)_{j\in\eN}$; indeed, for any given bounded fractal string\footnote{A bounded fractal string $\mathcal {L}:=(\ell_j)_{j\in\eN}$ is defined as a nonincreasing sequence of positive real numbers $(\ell_j)_{j\in\eN}$ such that $\sum_{j=1}^\ty\ell_j<\ty$; see \cite{lapidusfrank12}.}\label{ellj} $\mathcal {L}:=(\ell_j)_{j\in\eN}$, we can define a disjoint union 
\begin{equation}
\O:=\bigcup_{j=1}^\ty I_j
\end{equation} 
of open intervals $I_j$ such that $|I_j|=\ell_j$ for each $j\ge1$, and $A:=\pa\O$; then $\mathcal {L}$ can be identified with any such RFD $(\pa\O,\O)$; the set $\O$ is referred to as a
{\em geometric realization\index{fractal string, $\mathcal L$!geometric realization, $\O$|textbf}\index{geometric realization $\O$ of a bounded fractal string $\mathcal L$|textbf} of the bounded fractal string $\mathcal {L}$}; using this identification, we can write $\mathcal {L}=(\pa\O,\O)$; 
it is sometimes convenient to deal with the {\em canonical geometric representation of $\mathcal L$},\index{canonical geometric representation $A_{\mathcal L}$ of a bounded fractal string $\mathcal L$|textbf}\index{fractal string, $\mathcal L$!canonical geometric representation, $A_{\mathcal L}$|textbf} 
defined by
\begin{equation}\label{A_L}
A_{\mathcal L}:=\Big\{a_j:=\sum_{k\ge j}\ell_k:j\in\eN\Big\},
\end{equation}
and in this case, the set 
\begin{equation}
\O_{\rm can,\mathcal L}:=(0,\ell_1)\setminus A_{\mathcal L}=\bigcup_{j=1}^\ty(a_k,a_{k+1})
\end{equation}
is obviously a geometric realization of $\mathcal L$, which we call the {\em canonical geometric realization\label{Oak} of $\mathcal L$};

\medskip

($b$) The class $\mathtt{COM}(\eR^N)$ of {\em 
compact fractal sets} $A$ in Euclidean space $\eR^N$, by identifying $A$ with the corresponding RFD $(A,A_\d)$, for any fixed $\d>0$.


\medskip
Moreover, denoting by $\mathtt{RFD}(\eR^N)$ the family of all RFDs in $\eR^N$, we have the following  natural inclusions, for any $N\ge1$:
\begin{equation}
\mathtt{STR}_b\st\mathtt{COM}(\eR^N)\st\mathtt{RFD}(\eR^N).
\end{equation}
Here, any bounded fractal string $\mathcal {L}:=(\ell_j)_{j\in\eN}$ can be identified with the compact set $\ov A\st\eR$, where $A:=\{a_j=\sum_{k=j}^{\ty}\ell_k:j\in\eN\}\subset\eR$.

\medskip

We now introduce the main definition of this paper.

\begin{defn}\label{zeta_r} Let $(A,\O)$ be a relative fractal drum (or an RFD) in $\eR^N$.
The {\em distance zeta function $\zeta_{A,\O}$ of the relative fractal drum $(A,\Omega)$}\index{relative distance zeta function, $\zeta_{A,\O}$|textbf}\index{fractal zeta function!relative distance zeta function, $\zeta_{A,\O}$|textbf}\index{distance zeta function!of a relative fractal drum $(A,\O)$, $\zeta_{A,\O}$|textbf} (or the {\em relative distance 
zeta function} of the RFD $(A,\O)$)
 is defined by
\begin{equation}\label{rel_dist_zeta}
\zeta_{A,\O}(s):=\int_\Omega d(x,A)^{s-N}\D x,
\end{equation}
for all $s\in\Ce$ with $\re s$ sufficiently large. 
\end{defn}

The family of relative distance zeta functions represents a natural extension of the following classes of fractal zeta functions:
\medskip 

($a$) The class of {\em geometric zeta functions}\label{string}\index{geometric zeta function, $\zeta_{\mathcal L}$|textbf}\index{fractal zeta function!geometric zeta function, $\zeta_{\mathcal L}$|textbf} $\zeta_{\mathcal L}$, associated with bounded fractal strings $\mathcal {L}:=(\ell_j)_{j\in\eN}$ 
and defined (for all $s\in\Ce$ with $\re s$ sufficiently large) by 
\begin{equation}\label{geomz}
\zeta_{\mathcal L}(s):=\sum_{j=1}^\ty\ell_j^s
\end{equation} 
(it has been extensively studied in the research monograph \cite{lapidusfrank12},
by the first author and van Frankenhuijsen, as well as in the relevant references therein); more precisely, we show that
\begin{equation}\label{geomr}
\zeta_{\mathcal L}(s)=s2^{s-1}\zeta_{\pa\O,\O}(s) 
\end{equation}
for all $s\in\Ce$ with $\re s$ sufficiently large, where $\O$ is any geometric realization of $\mathcal L$,
described in $(a)$, appearing immediately after Definition \ref{rfd}.

\medskip

($b$) The class of {\em distance zeta functions} $\zeta_A$ associated with compact fractal subsets $A$ of Euclidean spaces, defined by \eqref{z}.
\smallskip

\bigskip


We point out that some of the results of this article (especially in \S\ref{hols}) can be viewed in the context of very general convolution-type integrals
of the form 
\begin{equation}\label{H}
H(s)=\int_E f(s,t)\,\D\mu(t),
\end{equation}
about which we cite the following well-known result. We shall need it in the proofs of Theorem \ref{mero_ext_N+1} and Proposition \ref{F_mero} below.

\begin{theorem}\label{Hh}
Let $V$ be an open set in $\Ce$ $($or even in $\Ce^n$$)$.
Furthermore, let $(E,{\mathcal B}(E),\mu)$ be a measure space, where $E$ is a locally compact metrizable space, 
${\mathcal B}(E)$ is the Borel $\s$-algebra of $E$, and $\mu$ is a positive or complex $($local, i.e., locally bounded$)$ measure, with associated total variation measure denoted by $|\mu|$.
Assume that a function $f:V\times E\to\Ce$ is given, satisfying the following three conditions$:$  
\medskip

$(1)$ $f(\,\cdot\,,t)$ is holomorphic for $|\mu|$-a.e.\ $t\in E$,
\medskip

$(2)$ $f(s,\,\cdot\,)$ is $\mu$-measurable for all $s\in V$, and
\medskip

$(3)$ a suitable growth property is fulfilled by $f$$:$ for every compact subset $K$ of $V$, there exists $g_K\in L^1(|\mu|)$
such that $|f(s,t)|\le g_K(t)$, for all  $s\in V$ and $|\mu|$-a.e. $t\in K$.
\medskip

\noindent Then, the function $H$ defined by $($\ref{H}$)$ is holomorphic on $V$.
Moreover, one can interchange the derivative and the integral. $($The problem of complex differentiating under the integral sign is discussed, for example, in {\rm \cite{Mattn}}.$)$ More precisely, for every $s\in V$ and every $k\in\eN$, we have
\begin{equation}
F^{(k)}(s) =\int_E\frac{\pa^k}{\pa s^k} f(s,t)\,\D\mu(t).
\end{equation}
\end{theorem}

\begin{remark}\label{condition3}
According to \cite{Mattn} and as is well known, if conditions $(1)$ and $(2)$ from Theorem~\ref{Hh} are satisfied, then condition $(3)$ is equivalent to the following condition, which is generally slightly easier to verify in practice:

\medskip

$(3')$ $\int_E|f(\,\cdot\,,t)|\di|\mu|(t)$ is locally bounded; that is, for each fixed $s_0\in V$, there exists $\delta>0$ such that
\begin{equation}
\sup_{s\in V,|s-s_0|<\delta}\int_E|f(s,t)|\di|\mu|(t)<\ty.
\end{equation}
In other words, we can replace condition $(3)$ with condition $(3')$ in the statement of Theorem~\ref{Hh}.
(This is the case because the notion of holomorphicity is {\em local}.)
\end{remark}

\section{Overview of the main results}

We note that the notion of {\em complex dimensions of a relative fractal drum} (RFD), necessary for a clearer understanding of this overview, is introduced in Definition \ref{cdim} below. The definitions of the relative Minkowski content and of the relative box (or, more accurately, Minkowski) dimension can be found in Equations \eqref{minkrel} and \eqref{dimrell} below, respectively. 
\medskip

{\bf Overview of Chapter \ref{rfds}.} 
The main result of \S\ref{hols} is contained in parts $(a)$ and $(b)$ of Theorem \ref{an_rel}, according to which the abscissa of (absolute) convergence $D(\zeta_{A,\O})$ of the distance zeta function $\zeta_{A,\O}$ of any RFD $(A,\O)$ is equal to the upper box (i.e., Minkowski) dimension $\ov\dim_B(A,\O)$ of the RFD. Part $(c)$ of the same theorem provides some mild conditions under which the value of $D:=\dim_B(A,\O)$ (assuming that $\dim_B(A,\O)$ exists) is a singularity of the relative distance zeta function $\zeta_{A,\O}$, and therefore also coincides with the abscissa of holomorphic continuation of the distance zeta function of the RFD. 

Theorem \ref{pole1rel} shows that for a nondegenerate RFD $(A,\O)$ and provided $\zeta_{A,\O}$ possesses a meromophic extension to an open connected neighborhood of $D:=\dim_BA<N$, the residue of $\zeta_{A,\O}$ evaluated at $D$ always lies (up to the multiplicative constant $(N-D)$) between the lower and the upper $D$-dimensional Minkowski contents of the RFD. In particular, if the RFD is Minkowski measurable, then the residue $\res(\zeta_{A,\O},D)$ is, up to the same multiplicative constant, equal to the $D$-dimensional Minkowski content of the RFD.

In Proposition \ref{conem} of \S\ref{cone}, we show that if there is at least one point $a\in\ov A\cap\ov\O$ at which the RFD $(A,\O)$ satisfies a suitable {\em cone property} with respect to $\O$ (see Definition \ref{conep}), then $D(\zeta_{A,\O})\ge0$. In general, however, the value of $D(\zeta_{A,\O})$ (i.e., of the upper relative box dimension $\ov\dim_B(A,\O)$) can take on any prescribed negative value (see Proposition \ref{ndim}), and even the value $-\ty$ (see Corollary \ref{flat}). 
Here we stress that the phenomenon of negative box dimensions (including $-\ty$)  for relative fractal drums of the form $(\pa\O,\O)$ and $(\{a\},\O)$, with $a\in\pa\O$, have been studied independently by Tricot in \cite{Tri2}, where the notion of {\em inner box dimension of the boundary $\pa\O$} and of the point $a$ with respect to $\O$ is used.

In \S\ref{scaling_property}, dealing with the {\em scaling property of the relative distance zeta functions}, the main result is stated in Theorem \ref{scaling}. It has important applications to the study of self-similar sprays (or tilings).
Corollary \ref{10.151/2} states an interesting scaling property of the residues of relative distance zeta functions, evaluated at their simple poles; see Equation \eqref{10.91/4}.
The countable additivity of the relative distance zeta function with respect to a disjoint union of RFDs (a notion introduced in Definition \ref{union}) is established in Theorem \ref{unionz}.

In \S\ref{rtzf}, we introduce the notion of the {\em relative tube zeta function} (see Equation \eqref{rel_tube_zeta}), which is
 closely related to the relative distance zeta function (see the functional equation \eqref{rel_equality} connecting these two fractal zeta functions).
Equation \eqref{resdt} in Proposition \ref{resrfd} connects the residues
(evaluated at any visible complex dimension) of the relative tube and distance functions.
In Example \ref{torus}, we calculate (via a direct computation) the complex dimensions of the torus RFD.
Much more generally,
in Proposition \ref{preach}, by using Federer's tube formula\index{Federer's tube formula} established in [Fed1], we calculate the distance zeta function (and the complex dimensions) of the boundary of any compact set $C$ of {\em positive reach} (and, in particular, of any compact convex subset of $\eR^N$); as a special case, we obtain a similar result for a smooth compact submanifold of $\eR^N$ (thereby significantly extending the results of the aforementioned Example \ref{torus}).

The important problem of the existence and construction of meromorphic extensions of some classes of relative (tube and distance) zeta functions is studied in \S\ref{rzfe}.
It is treated in Theorem \ref{rel_measurable} for a class of Minkowski measurable RFDs,
and in Theorem \ref{rel_nonmeasurable} for a class of Minkowski nonmeasurable (but Minkowski nondegenerate) RFDs. Naturally, even though the two classes of examples dealt with here are of interest in their own right and in the applications, additional results should be obtained along these lines in the future developments of the theory.

The main result of \S\ref{ty_qp_rfd} is stated in Theorem \ref{qp} and deals with the construction of $\ty$-{\em quasiperiodic relative fractal drums}, a notion introduced in Definition \ref{quasi_periodic_t}.
Its proof makes an essential use of suitable families of {\em generalized Cantor sets} $C^{(m,a)}$ with two parameters $m$ and $a$, introduced in Definition \ref{Cma}; some of the properties of these Cantor sets are listed in Proposition \ref{Cmap}. 
 Theorem \ref{qp} can be considered as a fractal set-theoretic interpretation of Baker's theorem from transcendental number theory (see Theorem~\ref{baker0}). It provides an explicit construction of a transcendentally $\ty$-quasiperiodic relative fractal drum. In particular, this RFD possesses 
infinitely many algebraically incommensurable quasiperiods.  
In Definition \ref{hyperfractal}, we also introduce the new notions of {\em hyperfractal RFDs}, as well as of {\em strong hyperfractals} and of {\em maximal hyperfractals}. It turns out that the relative fractal drums constructed in Theorem \ref{qp} are not only $\ty$-quasiperiodic, but are also maximally hyperfractal. Accordingly, the critical line $\{\re s=D\}$, where $D:=\ov\dim_B(A,\O)$, consists solely of nonremovable singularities of the associated fractal zeta function; a fortiori, the distance and tube zeta functions of the RFD cannot be meromorphically extended beyond this vertical line.

The {\em scaling property of relative tube zeta functions} is provided in
Proposition \ref{scalingtr}. This result is analogous to the one obtained for relative distance zeta functions in Theorem~\ref{scaling}.
\medskip

{\bf Overview of Chapter \ref{sec_embed}.} This chapter deals with the problem of {\em embeddings
of RFDs} into higher-dimensional Euclidean spaces.
Theorem \ref{c_dim_inv} of \S\ref{emb_b}
shows that the notion
of complex dimensions of fractal sets does not depend on the dimension of the ambient space.
In Theorem \ref{c_dim_inv_rel} of
\S\ref{emb_rfd}, an analogous result is obtained for general RFDs.
An important role in the accompanying computations is played by the gamma function, the Euler beta function, as well as by the Mellin zeta function of an RFD, introduced in Equation \eqref{mellinz}. In Example \ref{c_dust}, we apply these results in order to calculate the complex dimensions of the Cantor dust.
\medskip

{\bf Overview of Chapter \ref{zfsprays}.} In this chapter, we study {\em relative fractal sprays} in $\eR^N$, introduced in
 Definition \ref{dsprays} of \S\ref{dsprays}. 
The main result is given in
Theorem \ref{sprayz}, which deals with the distance zeta function of relative fractal sprays.

In \S\ref{relative_other}, we study the relative Sierpi\'nski sprays and their complex dimensions.
Example \ref{6.15} deals with the relative Sierpi\'nski gasket, 
while Example \ref{Ngasket} deals with the inhomogeneous Sierpi\'nski $N$-gasket RFDs, for any $N\ge2$.
Furthermore, Example \ref{sierpinski_carpetr} deals with
the relative Sierpi\'nski carpet, while
Example \ref{carpetN} deals with the Sierpi\'nski $N$-carpet, for any $N\ge2$. Interesting new phenomena occur in this context, which are discussed throughout \S\ref{relative_other}.

In Definition \ref{ss_spray} of \S\ref{goran}, we recall (and extend to RFDs) the notion of
{\em self-similar sprays} (or {\em tilings}), defined by a
suitable {\em ratio list} of finitely many real numbers in $(0,1)$.

Theorem \ref{ss_spray_zeta} provides an explicit form for the distance zeta function of a self-similar spray, which can be found in Equation \eqref{ss_spray_form}. The results obtained here are illustrated by the new examples of the $1/2$-square fractal and of
the $1/3$-square fractal, discussed in Example \ref{kvadrat0.5} and in Example \ref{kvadrat_0.33}, respectively.

In \S\ref{multiplicity}, we describe a constructive method for generating principal complex dimensions\footnote{The {\em principal complex dimensions} of an RFD are the poles of the associated fractal (i.e., distance or tube) zeta function with maximal real part $D$, where $D$ is both the abscissa of (absolute) convergence of the zeta function and the (relative) Minkowski dimension of the RFD; see Definition \ref{cdim} and part $(b)$ of Theorem \ref{an_rel} below.} of relative fractal drums of any prescribed multiplicity $m\ge2$, including infinite multiplicity.
(The latter case when $m=\ty$ corresponds to essential singularities of the associated fractal zeta function.)
In Example \ref{Lmloga}, we provide the construction of the
{\em $m$-th order $a$-string}, while we define the {\em $m$-th order Cantor string} in Equation \eqref{Lm}.
In Examples \ref{tensorr} and \ref{tensorm}, we construct Minkowski measurable RFDs which possess infinitely many principal complex dimensions of arbitrary multiplicity $m$, with $m\ge2$ and even with $m=\ty$ (i.e., corresponding to essential singularities). 
\medskip

{\bf Overview of Chapter \ref{fractality}.} This chapter is dedicated to the discussion of the notion of {\em fractality} (of RFDs), and its intimate relationship with the notion of the complex dimensions of RFDs.
In \S\ref{sub_rfd}, {\em fractal and subcritically fractal RFDs} are discussed.
These notions are illustrated in 
\S\ref{cg_rfd} in the case of the {\em Cantor graph RFD}.

\section{Notation}\label{notation} 

In the sequel, an important role is played by the definition of the upper and lower Minkowski contents of RFDs and of the
upper and lower box (or Minkowski) dimensions of RFDs. We shall follow the definitions introduced by the third author in \cite{rae},
but with an essential difference: the parameter $r$ appearing below can be {\em any} real number, and not just a nonnegative real number. (See Remark \ref{grfd} below.)
Hence, for a given parameter $r\in\eR$, we define the {\em $r$-dimensional upper and lower\index{Minkowski content of an RFD $(A,\O)$|textbf}\index{relative fractal drum, RFD!Minkowski content of|textbf} Minkowski contents} of an RFD $(A,\O)$ in $\eR^N$ as follows:\footnote{For a given measurable set $E\st\eR^N$, its $N$-dimensional Lebesgue measure is denoted by $|E|=|E|_N$.}\label{Nmeasure}
\begin{equation}\label{minkrel}
\mathcal{M}^{*r}(A,\Omega):=\limsup_{t\to0^+}\frac{|A_t\cap\Omega|}{t^{N-r}}, \q
\mathcal{M}_*^r(A,\Omega):=\liminf_{t\to0^+}\frac{|A_t\cap\Omega|}{t^{N-r}}.
\end{equation}
We also call them the {\em relative upper Minkowski content {\rm and}\index{relative Minkowski!content|textbf}\index{relative fractal drum, RFD!Minkowski content of|textbf} lower Minkowski content} of $(A,\O)$, respectively. They represent a natural extension of the corresponding notions of upper and lower Minkowski contents of bounded sets in $\eR^N$, introduced by 
Bouligand [Bou] and Hadwiger [Had], as well as used by
Federer in \cite{federer} and by many other researchers in a variety of works, including [Sta, Tri1, BroCar, Lap1, LapPo2, Fal1--2, HeLap, Lap-vFr3, \v Zu2, Wi, KeKom, Kom, RatWi, LapRa\v Zu1, HerLap1] and the relevant references therein.

As usual, we then define the {\em upper box dimension} of the RFD $(A,\O)$ by
\begin{equation}\label{dimrell}
\begin{aligned}
\ov\dim_B(A,\Omega)&:=\inf\{r\in\eR:\mathcal{M}^{*r}(A,\Omega)=0\}\\
&\phantom{:}=\sup\{r\in\eR:\mathcal{M}^{*r}(A,\Omega)=+\infty\},
\end{aligned}
\end{equation} 
as well as the {\em lower box dimension} of $(A,\O)$ by
\begin{equation}\label{dimrelu}
\begin{aligned}
\underline\dim_B(A,\Omega)&:=\inf\{r\in\eR:\mathcal{M}_*^r(A,\Omega)=0\}\\
&\phantom{:}=\sup\{r\in\eR:\mathcal{M}_*^r(A,\Omega)=+\infty\}.
\end{aligned}
\end{equation} 
We refer to $\ov\dim_B(A,\O)$ and $\underline\dim_B(A,\O)$ as the {\em relative upper and lower box $($or Minkowski$)$ dimension}\index{Minkowski (or box) dimension, $\dim_B(A,\O)$|textbf}\index{relative Minkowski!dimension, $\dim_B(A,\O)$|textbf}\index{fractal dimension!relative box (or Minkowski) dimension of an RFD, $\dim_B(A,\O)$|textbf}\index{relative Minkowski (or box)!dimension, $\dim_B(A,\O)$|textbf}\index{relative fractal drum, RFD!Minkowski (or box) dimension of|textbf} of $(A,\O)$, respectively.
The novelty here is that, contrary to the usual upper and lower box dimensions, the {\em relative} upper and lower Minkowski dimensions can attain negative values as well, and even the value $-\ty$. More specifically, it is easy to see that 
$$
-\ty\le\underline\dim_B(A,\Omega)\le \ov\dim_B(A,\Omega)\le N.
$$
If $\ov\dim_B(A,\Omega)=\underline\dim_B(A,\Omega)$, then the common value is denoted by $\dim_B(A,\Omega)$
and we call it the {\em box} $($or {\em Minkowski$)$ dimension of the} RFD $(A,\O)$, or just the {\em relative box $($i.e., Minkowski$)$ dimension}.\index{fractal dimension!relative box (or Minkowski) dimension of an RFD, $\dim_B(A,\O)$|textbf}\footnote{We caution the reader,  however, that unlike in the standard case of bounded subsets of $\eR^N$, the notion of {\em relative Minkowski dimension} of an RFD has not yet been given a suitable geometric interpretation in terms of ``box counting''.}

If there exists $D\in\eR$ such that $0<\mathcal{M}_*^D(A,\Omega)\le
\mathcal{M}^{*D}(A,\Omega)<\ty$, then we say that the RFD $(A,\O)$ is {\em Minkowski nondegenerate}.\index{Minkowski nondegenerate RFD|textbf}\index{relative fractal drum, RFD!Minkowski nondegenerate|textbf} Clearly, in this case we have that $D=\dim_B(A,\O)$.

If for some $r\in\eR$, we have $\mathcal{M}^{*r}(A,\Omega)=
\mathcal{M}_*^r(A,\Omega)$, the common value is denoted by $\mathcal{M}^r(A,\Omega)$. If for some $D\in\eR$, $\mathcal{M}^D(A,\Omega)$ exists and $\mathcal{M}^D(A,\Omega)\in(0,+\ty)$, then we say that the RFD $(A,\O)$ is {\em Minkowski measurable}.\index{Minkowski measurable RFD|textbf}\index{relative fractal drum, RFD!Minkowski measurable|textbf} Clearly, in this case, the dimension of $(A,\O)$ exists and we have $D=\dim_B(A,\O)$.

For example, if the sets $A$ and $\Omega$ are a positive distance apart $($i.e., $\inf\{|x-y|:x\in A,\,y\in\O\}>0$$)$, then 
 it is easy to see that $\dim_B(A,\Omega)=-\ty$. 
Indeed, since $|A_t\cap\Omega|=0$ for all sufficiently small $t>0$, we have that $\mathcal{M}^r(A,\Omega)=0$ for all $r\in\eR$.
 A class of nontrivial examples for which $-\ty<\dim_B(A,\O)<0$ can be found in Proposition \ref{ndim}.

In the case when $\O:=A_\d$, where $A$ is a {\em bounded} subset of $\eR^N$ and $\d$ is a fixed positive real number, we obtain the usual (nonrelative) values of {\em box $($or Minkowski$)$ dimensions},\label{dim}\index{fractal dimension!box (or Minkowski) dimension of a set, $\dim_BA$|textbf} i.e., $\ov\dim_BA:=\ov\dim_B(A,A_\d)$, $\underline\dim_BA:=\dim_B(A,A_\d)$, $\dim_BA:=\dim_B(A,A_\d)$, which are all nonnegative in this case, as well as the values of the usual Minkowski contents of $A$; that is, $\M^{r*}(A):=\M^{r*}(A,A_\d)$, $\M_*^r(A):=\M_*^r(A,A_\d)$,\label{mink}  $\M^r(A):=\M^r(A,A_\d)$,\label{minkc} for any $r\ge0$. (It is easy to see that these values do not depend on the choice of $\d>0$.) Consequently, as was stated in \S\ref{hols}, bounded subsets of $\eR^N$ are special cases of RFDs in $\eR^N$. More specifically, if $A$ is a bounded subset of $\eR^N$, then the associated RFD in $\eR^N$ is $(A,A_\d)$, for any given $\d>0$. This comment extends to the theory of complex dimensions of RFDs developed in this paper, which therefore includes the theory of complex dimensions developed in [LapRa\v Zu2,3].

\begin{remark}\label{grfd}
These definitions extend to a general RFD the definitions used in \cite{Lap1} for an ordinary fractal drum (i.e., a drum with fractal boundary) in the case of Dirichlet boundary conditions; see also [{Lap-vFr2}, \S12.5] and the relevant references therein, including [{Brocar}, {Lap2--3}]). We then have $(A,\O)=(\pa\O,\O)$, where $\O$ is a (nonempty) bounded open subset of $\eR^N$; it follows at once that $D:=D(\zeta_{\pa\O,\O})\ge0$. In fact, we always have that $D\in[N-1,N]$; see \cite{Lap1}. The special case when $N=1$ corresponds to bounded fractal strings, for which we must have that $D\in[0,1]$; see, for example, [{Lap1--3}, {LapPo1--2}, {LapMa1--2}, {HeLap}, {LapLu-vFr1--3}]. Other references related to fractal strings include [DubSep, ElLapMacRo, Fal2, Fr, HamLap, HerLap1--2, Kom, LalLap1--2, Lap4--6, LapLu, LapRa\v Zu1--8, LapRo, LapL\'eRo, LapRo\v Zu, L\'eMen, MorSep, MorSepVi1--2, Ol1--2, Ra1--2, RatWi, Tep1--2].
\end{remark}

\medskip

In the sequel, we use the following notation. Given $\a\in\eR\cup\pm\ty$, we denote, for example, the open right half-plane $\{s\in\Ce:\re s>\a\}$ by $\{\re s>\a\}$, with the obvious convention if $\a=\pm\ty$; namely, for $\a=+\ty$, we obtain the empty set and for $\a=+\ty$, we have all of $\Ce$. Moreover, if $\a\in\eR$, we denote the vertical line $\{s\in\Ce:\re s=\a\}$ by $\{\re s=\a\}$.
\smallskip

We also let $\eN:=\{1,2,3\dots\}$ and $\eN_0:=\eN\cup\{0\}=\{0,1,2,\dots\}$\label{eN}. 
The logarithm of a positive real number $x$ with base $a>0$ is denoted by $\log_a x$; i.e., $y=\log_a x$\, $\Leftrightarrow$\, $x=a^y$.\label{logax}
Furthermore,
$\log x:=\log_{\E}x$ is the natural logarithm of $x$; i.e., $y=\log x$\, $\Leftrightarrow$\, $x=\E^y$.

\medskip

Let $f(s):=\int_E\f(x)^s\D\mu(x)$ be a {\em tamed}\index{Dirichlet-type integral (DTI)!tamed|textbf}\index{tamed Dirichlet-type integral (tamed DTI)|textbf}\index{Dirichlet-type integral (DTI)!generalized|textbf} generalized Dirichlet-type integral (DTI),\label{DTI3} in the sense of \cite[Definition 2.12]{dtzf}; that is, $\f\ge0$ $|\mu|$-a.e.\ on $E$ and $\f$ is essentially $|\mu|$-bounded on $E$, where $|\mu|$ is the total variation of the local complex or positive measure on the locally compact space $E$. Then, we denote by $D(f)$ the {\em abscissa of $($absolute$)$ convergence}\label{Dabs}\index{abscissa of!absolute convergence of a DTI, $D(f)$|textbf} of $f$; i.e., $D(f)\in[-\ty,\ty])$ is the infimum of all $\a\in\eR$ such that $\f^\a$ (or, equivalently, $\f(x)^{\re s}$, with $\alpha:=\re s$) is $|\mu|$-integrable.
If $D(f)\in\eR$, the corresponding vertical line $\{\re s=D(f)\}$ in the complex plane is called the {\em critical line}\index{critical line of a Dirichlet-type integral|textbf}\index{Dirichlet-type integral (DTI)!} of $f$.
 Furthermore, we denote by $D_{\rm mer}(f)$ the {\em abscissa of meromorphic continuation}\label{Dmer}\index{abscissa of!meromorphic continuation, $D_{\rm mer}(f)$|textbf} of $f$ (i.e., $D_{\rm mer}(f)\in[-\ty,\ty]$ is the infimum of all $\a\in\eR$ such that $f$ possesses a meromorphic extension to the open right half-plane $\{\re s>\a\}$).
We define $D_{\rm hol}(f)$,\label{Dhol} the {\em abscissa of holomorphic continuation}\index{abscissa of!holomorphic continuation, $D_{\rm hol}(f)$|textbf} of $f$, in exactly the same way as for $D_{\rm mer}(f)$, except for ``meromorphic'' replaced by ``holomorphic''.\footnote{Note that $D_{\rm mer}(f)$ and $D_{\rm hol}(f)$ can be defined for any given meromorphic function $f$ on a domain $U$ of $\Ce$, whereas $D(f)$ is only well defined if $f$ is a tamed DTI. (See [{LapRa\v Zu2}] or \cite[Appendix A]{fzf}.)} In general, for any tamed DTI $f$, we have that
\begin{equation}\label{merhol}
-\ty\le D_{\rm mer}(f)\le D_{\rm hol}(f)\le D(f)\le+\ty.
\end{equation}
(See \cite[Theorem A.2]{fzf} for the next to last inequality, and \cite[Appendix A]{fzf} for the general theory of tamed DTIs.)
\medskip

In order to be able to define the key notions of complex dimensions and of principal complex dimensions (see Equation \eqref{dimc0} in Chapter \ref{rfds} below),
we assume that the function $f$ has the property that it can be extended to a meromorphic function defined on $G\stq\Ce$, 
where $G$ is an open and connected neighborhood of the {\em window}\label{window}\index{window $\bm W$|textbf} $\bm W$ defined by
\begin{equation}
\bm{W}=\{s\in\Ce: \re s\ge S(\im s)\}.
\end{equation}
Here, the function $S:\eR\to(-\ty,D(\zeta_A)]$, called the {\em screen},\label{screen}\index{screen $\bm S$|textbf} is assumed to be Lipschitz continuous.
Note that if $f:=\zeta_{A,\O}$, then the closed set $\bm W$ contains the {\em
critical line}\index{critical line} $\{\re s=D(\zeta_{A,\O})\}$;\label{cr_line_w}\index{critical line!of (absolute) convergence, $\{\re s=D(f)\}$, of a Dirichlet-type integral $f$ (or, in short, the `critical line')}
in fact, it also contains the closed half-plane $\{\re s\ge D(\zeta_{A,\O})\}$. The boundary $\pa{\bm W}$ of the window is also called the {\em screen} and is denoted by~$\bm S$; it is the graph of the function $S$, with the horizontal and vertical axes interchanged. More specifically, we have that
\begin{equation}
{\bm S}=\{S(\tau)+\I \tau:\tau\in\eR\}.
\end{equation}
\medskip

\begin{defn}[Complex dimensions of an RFD]\label{cdim} The set of poles of $f$ located in a window $\bm W$ containing the critical line $\{\re s=D(f)\}$ is denoted by $\po(f,W)$. When the window $\bm W$ is known, or when $\bm{W}:=\Ce$, we often use the shorter notation $\po(f)$\label{pof} instead. 
If $f:=\zeta_{A,\O}$ and $\zeta_{A,\O}$ can be meromorphically extended to a connected open subset containing the closed right half-plane $\{\re s\ge D(\zeta_{A,\O})\}$, the multiset of poles (i.e., we also take the multiplicities of the poles into account) is called the multiset of ({\em visible}) {\em complex dimensions}\index{complex dimensions of an RFD $(A,\O)$|textbf}\index{fractal dimension!complex dimensions of an RFD $(A,\O)$|textbf} of $(A,\O)$. The multiset of complex dimensions on the critical line of $\zeta_{A,\O}$ is called the multiset of {\em principal complex dimensions}\label{dimcr}\index{complex dimensions of an RFD $(A,\O)$!principal complex dimensions, $\dim_{PC}(A,\O)$|textbf}\index{principal complex dimensions, $\dim_{PC}(A,\O)$|textbf}\index{fractal dimension!principal complex dimensions of an RFD, $\dim_{PC}(A,\O)$|textbf} of $(A,\O)$. This multiset is independent of the choice of $\d$, as well as of the meromorphic extension of $\zeta$.
We note that analogous definitions will be used in the case of the relative tube zeta function $\tilde\zeta_{A,\O}$, introduced in Equation \eqref{rel_tube_zeta} below, instead of the relative distance zeta function $\zeta_{A,\O}$. As we shall see, provided $\ov\dim_B(A,\O)<N$, the resulting multiset of principal complex dimensions (resp., of complex dimensions) will be the same for either $\zeta_{A,\O}$ or $\tilde\zeta_{A,\O}$. This will follow from the functional equation connecting $\zeta_{A,\O}$ and $\tilde\zeta_{A,\O}$.
\end{defn}
\medskip

If $\O$ is a given subset of $\eR^N$, its closure and boundary\label{boundary_of_a_set} are denoted by $\ov\O$ and $\pa\O$, respectively. 

For two given sequences of positive numbers $(a_k)_{k\ge1}$ and $(b_k)_{k\ge1}$, we write $a_k\sim b_k$\label{sim}
as $k\to \ty$ if $\lim_{k\to\ty}a_k/b_k=1$. Analogously, if $a(\,\cdot\,)$ and  $b(\,\cdot\,)$ are two real-valued functions defined on an open interval $(0,t_0)$, we write $a(t)\sim b(t)$
as $t\to0^{+}$ if $\lim_{t\to0^{+}}a(t)/b(t)=1$.

We shall also need the relation $\sim$ between Dirichlet-type integral functions (DTIs) and meromorphic functions (that is, $f\sim g$, where $f$ is a DTI and $g$ is a meromorphic function) in the sense of \cite[Definition 2.22]{dtzf}, which we now briefly recall.

\begin{defn}\label{equ}
Let $f$ and $g$ be tamed Dirichlet-type integrals, both admitting a (necessarily unique) meromorphic extension to an open connected subset $U$ of $\Ce$ which contains
the closed right half-plane $\{\re s\ge D(f)\}$. Then, the function $f$ is said to be {\em equivalent}\index{equivalence $f\sim g$|textbf} to $g$, and we write $f\sim g$,
if $D(f)=D(g)$ (and this common value is a real number) and furthermore, the sets of poles of $f$ and $g$, located on the common critical line $\{\re s=D(f)\}$, coincide. Here, the multiplicities of the poles should be taken into account. In other words, we view the set of principal poles ${\mathcal{P}}_c(f)$ of $f$ as a multiset.
More succinctly,
\begin{equation}\label{equ2}
f\sim g\quad\overset{\mbox{\tiny def.}}\Longleftrightarrow\quad D(f)=D(g)\,\,(\in\eR)\q \mathrm{and}\q \po_c(f)=\po_c(g).
\end{equation}
\end{defn}

\chapter{Basic properties of relative distance and tube zeta functions}\label{rfds}

\section{Holomorphicity of relative distance zeta functions}\label{hols}

In the sequel, we denote by $D(\zeta_{A,\O})$ the abscissa of (absolute) convergence of a given relative distance zeta function $\zeta_{A,\O}$. It is clear that $D(\zeta_{A,\O})\in[-\ty,N]$.
Recall from the discussion in Chapter \ref{intro} that an analogous definition can be introduced for much more general, tamed Dirichlet-type integrals, introduced in \cite{dtzf}, as well as in the monograph \cite[especially in Appendix A]{fzf}.

Some of the basic properties of distance zeta functions of RFDs are listed in the following theorem.

\begin{theorem}\label{an_rel}
Let $(A,\O)$ be a relative fractal drum in $\eR^N$. Then the following properties hold$:$
\bigskip 

$($a$)$ The relative distance zeta function $\zeta_{A,\O}$ is holomorphic in the half-plane 
$$
\{\re s>\overline{\dim}_B(A,\Omega)\},
$$ and for those same values of $s$, we have
$$
\zeta'_{A,\O}(s)=\int_{\Omega}d(x,A)^{s-N}\log d(x,A)\, \D x.
$$
\bigskip

$(b)$ The lower bound on the $($absolute$)$ convergence region $\{\re s>\overline{\dim}_B(A,\Omega)\}$ of the relative distance zeta function $\zeta_{A,\O}$ is optimal. In other words, 
\begin{equation}\label{rel_absc}
D(\zeta_{A,\O})=\overline{\dim}_B(A,\Omega).
\end{equation}
\bigskip

$(c)$ If $D:=\dim_B(A,\Omega)$ exists, $D<N$ and $\mathcal{M}_*^D(A,\Omega)>0$, then $\zeta_{A,\O}(s)\to+\infty$ as $s\in\eR$ converges to $D$ from the right. Hence, under these assumptions, we have that\footnote{The abscissa of holomorphic continuation, denoted by $D_{\rm hol}(\zeta_{A,\O})$, is defined in the discussion preceding Equation \eqref{merhol} above.}
\begin{equation}
D(\zeta_{A,\O})=D_{\rm hol}(\zeta_{A,\O})=\dim_B(A,\O).
\end{equation}
\end{theorem}

We omit the proof since it 
follows the same steps as
 in the case when $\O:=A_\d$ (that is, as in the case of a bounded set $A$); see \cite[Theorem 2.5]{dtzf}.
In the proof of part ($a$) of Theorem \ref{an_rel}, we need the following result. {\em For any relative fractal drum $(A,\O)$ in $\eR^N$, we have that}
\begin{equation}
\c<N-\ov\dim_B(A,\O)\q\implies\q\int_{\O}d(x,A)^{-\c}\D x<\ty.
\end{equation}
In the case when $\O=A_\d$, where $\d$ is a positive real number, this implication reduces to the Harvey--Polking result (see \cite{acta}) since in that case, $\ov\dim_B(A,A_\d)=\ov\dim_BA$.
We note that the technical condition \eqref{delta} on the RFD $(A,\O)$ from Definition \ref{rfd} is needed in order that the integrals appearing during the computation of $\zeta_{A,\O}$ are well defined for $\re s$ large enough.

It is clear that the function $\zeta_{A,\O}$ is a tamed generalized Dirichlet-type integral in the sense of \cite[Definition 2.12]{dtzf}. If the RFD $(A,\O)$ is such that the corresponding relative zeta function $\zeta_{A,\O}$ can be meromorphically extended to an open, connected window $\bm W$ containing the critical line $\{\re s=D(\zeta_{A,\O})\}$, then the poles of $\zeta_{A,\O}$ located on the critical line are called the {\em principal complex dimensions}\index{principal complex dimensions, $\dim_{PC}(A,\O)$} of the RFD $(A,\O)$. The corresponding multiset of complex dimensions of the RFD $(A,\O)$ is denoted by $\dim_{PC}(A,\O)$. In other words,
\begin{equation}\label{dimc0}
\dim_{PC}(A,\O):=\po(\zeta_{A,\O}, W)\cap\{\re s=D(\zeta_{A,\O})\}.
\end{equation}
It is easy to see that the multiset $\dim_{PC}(A,\O)$ does not depend on the choice of window $\bm W$.
If $A$ is a {\em bounded} subset of $\eR^N$ and $\d$ a fixed positive real number, the multiset of {\em principal complex dimensions of $A$}\index{fractal dimension!principal complex dimensions of a set, $\dim_{PC}A$|textbf} is defined by 
\begin{equation}\label{dimc}
\dim_{PC}A:=\dim_{PC}(A,A_\d), 
\end{equation}
and this multiset does not depend on the choice of $\d>0$. 

Analogously, for any bounded fractal string $\mathcal L$ for which $D:=\ov\dim_B\mathcal L>0$, we define the multiset of principal complex dimensions of $\mathcal L$ by
\begin{equation}\label{dimcs}
\dim_{PC}{\mathcal L}:=\dim_{PC}(\pa\O,\O), 
\end{equation}
where $\O$ is any geometric realization of the fractal string $\mathcal L$; see Equation \eqref{geomr} which connects the standard geometric zeta function $\zeta_{\mathcal L}$ with the relative distance zeta function $\zeta_{\pa\O,\O}$. It is clear that the multiset $\dim_{PC}(\pa\O,\O)$ does not depend on the choice of the geometric realization $\O$, because the same is true for the relative distance zeta function $\zeta_{\pa\O,\O}$.

\medskip

In light of \cite[Theorem~3.3]{dtzf}, we have the following result.

\begin{theorem}\label{pole1rel}
Assume that $(A,\O)$ is a Minkowski nondegenerate RFD in $\eR^N$,
that is, $0<\M_*^D(A,\O)\le\M^{*D}(A,\O)<\ty$ $($in particular, $\dim_B(A,\O)=D)$,
and $D<N$.
If $\zeta_{A,\O}$ can be meromorphically continued to a connected open neighborhood of $s=D$,
then $D$ is necessarily a simple pole of $\zeta_{A,\O}$, and
\begin{equation}\label{res_M}
(N-D)\M_*^D(A,\O)\le\res(\zeta_{A,\O},D)\le(N-D)\M^{*D}(A,\O).
\end{equation}
Furthermore, if $(A,\O)$ is Minkowski measurable, then
\begin{equation}
\label{pole1minkg1=1}
\res(\zeta_{A,\O}, D)=(N-D)\M^D(A,\O).
\end{equation}
\end{theorem}

\medskip

In the following example, we compute the relative distance zeta function of an open ball in $\eR^N$ with respect to its boundary.

\begin{example}\label{spherer}
Let $\O:=B_R(0)$ be the open ball in $\eR^N$ of radius $R$ and let $A=\pa\O$ be the boundary of $\O$, i.e., the $(N-1)$-dimensional sphere of radius $R$. Then, introducing the new variable $\rho=R-r$ and letting $\o_N:=|B_1(0)|_N$,\label{omega_N} the $N$-dimensional Lebesgue measure of the unit ball in $\eR^N$, we have that
$$
\begin{aligned}
\zeta_{A,\O}(s)&=N\o_N\int_0^R(R-r)^{s-N}r^{N-1}\D r =N\o_N\int_0^R\rho^{s-N}(R-\rho)^{N-1}\D \rho\\
&=N\o_N\int_0^R\rho^{s-N}\sum_{k=0}^{N-1}(-1)^k\binom{N-1}{k}R^{N-1-k}\rho^k\D \rho\\
&=N\o_N R^s\sum_{k=0}^{N-1}\binom{N-1}{k}\frac{(-1)^k}{s-(N-k-1)}\\
&=N\o_N R^s\sum_{j=0}^{N-1}\binom{N-1}{j}\frac{(-1)^{N-j-1}}{s-j},
\end{aligned}
$$
for all $s\in\Ce$ with $\re s>N-1$.
It follows that $\zeta_{A,\O}$ can be meromorphically extended to the whole complex plane and is given by
\begin{equation}\label{AON}
\zeta_{A,\O}(s)=N\o_N R^s\sum_{j=0}^{N-1}\binom{N-1}{j}\frac{(-1)^{N-j-1}}{s-j},
\end{equation}
for all $s\in\Ce$.

Therefore, we have that
\begin{equation}\label{dimBAON-1}
\begin{gathered}
\dim_B(A,\O)=D(\zeta_{A,\O})=N-1,\\
\po(\zeta_{A,\O})=\{0,1,\dots,N-1\}\q\mbox{and}\q\dim_{PC} (A,\O)=\{N-1\}.
\end{gathered}
\end{equation}
Furthermore,
\begin{equation}\label{resBN-1}
\res(\zeta_{A,\O},j)=(-1)^{N-j-1}N\o_N\binom {N-1}j R^j,
\end{equation}
for $j=0,1,\dots,N-1$.
 As a special case of \eqref{resBN-1}, for $j=D:=N-1$  we obtain that
\begin{equation}\label{resAN-1}
\res(\zeta_{A,\O},D)=N\o_NR^{N-1}=\M^D(A,\O)=
(N-D)\M^D(A,\O).
\end{equation}
(This is a very special case of Equation \eqref{pole1minkg1=1} appearing in Theorem \ref{pole1rel} above.) The second to last equality in \eqref{resAN-1}
follows from the following direct computation (with $D:=N-1$):
\begin{equation}
\begin{aligned}
\M^D(A,\O)&=\lim_{t\to0^+}\frac{|A_t\cap\O|}{t^{N-D}}\\
&=\lim_{t\to0^+}\frac{\o_NR^N-\o_N(R-t)^N}t=N\o_N R^{N-1}.
\end{aligned}
\end{equation}
Furthermore, recall that $\H^D(A)=\H^{N-1}(\pa B_R(0))=N\o_N R^{N-1}$, where $\H^{N-1}$ is the $(N-1)$-dimensional Hausdorff measure.\label{MHD}\index{Hausdorff measure} Hence, $\M^D(A,\O)=\H^D(A)$.
\end{example}

\begin{remark}\label{recon}
We note that the usual notions of distance and tube zeta functions, $\zeta_A$ and $\tilde\zeta_A$, associated with a bounded subset $A$ of $\eR^N$, 
can be recovered by considering the RFD $(A,A_\d)$, for some $\d>0$:
\begin{equation}\label{tildz}
\begin{gathered}
\zeta_A(s)=\zeta_{A,A_\d}(s)=\int_{A_\d}d(x,A)^{s-N}\,\D x,\\
\tilde\zeta_A(s)=\tilde\zeta_{A,A_\d}(s)=
\int_0^\d t^{s-N-1}|A_t|\,\D t.
\end{gathered}
\end{equation}
Here, $\tilde{\zeta}_{A,A_\d}$ is the relative tube zeta function of the RFD $(A,A_\d)$, as defined in Equation \eqref{rel_tube_zeta} of \S\ref{rtzf} below.
\end{remark}

\section{Cone property of relative fractal drums}\label{cone}
We introduce the cone property of a relative fractal drum $(A,\Omega)$ at a prescribed point, in order to show that the abscissa of convergence $D(\zeta_{A,\O})$
of the associated relative zeta function $\zeta_{A,\O}$ be nonnegative. The main result of this section is stated in Proposition~\ref{conem}. We also construct a class of nontrivial RFDs for which the relative box dimension is 
an arbitrary negative number (see Proposition~\ref{ndim}) or even equal to $-\ty$ (see Corollary \ref{flat} and Remark \ref{flatA}, along with part $(a)$ of Proposition \ref{conem}). 

\begin{defn}
Let $B_r(a)$ be a given ball in $\eR^N$, of radius $r$ and center $a$. Let $\pa B$ be the boundary of the ball, which is an $(N-1)$-dimensional sphere,
and assume that $G$ is a closed connected subset contained in a hemisphere of $\pa B$. Intuitively, $G$ is a disk-like subset (`calotte') of a hemisphere
contained in the sphere $\pa B$. We assume that $G$ is open with respect to the relative topology of $\pa B$. The {\em cone\index{cone in $\eR^N$|textbf} $K=K_r(a,G)$ with vertex at $a$}, and of radius $r$, is defined as the interior of the convex hull of the union of $\{a\}$ and $G$.
\end{defn}

\begin{defn}\label{conep}
Let $(A,\Omega)$ be a relative fractal drum in $\eR^N$ such that $\ov A\cap\ov\Omega\ne\emptyset$. We say that the {\em relative fractal drum $(A,\Omega)$ has the cone 
property\index{cone property of an RFD $(A,\O)$|textbf}
at a point} $a\in\ov A\cap\ov\Omega$ if there exists $r>0$ such that $\Omega$ contains a cone $K_r(a,G)$ with vertex at $a$ (and of radius $r$). 
\end{defn}

\begin{remark}\label{3.1.21}
If $a\in\ov A\cap\Omega$ (hence, $a$ is an inner point of $\Omega$), then
the cone property of the relative fractal drum $(A,\Omega)$ is obviously satisfied at this point. So, the cone property is actually interesting
only on the boundary of $\Omega$, that is, at $a\in\ov A\cap\pa\Omega$.
\end{remark}
\smallskip


\begin{example}\label{3.1.23}
Given $\alpha>0$, let $(A,\Omega_\alpha)$ be the relative fractal drum in $\eR^2$ defined by $A=\{(0,0)\}$ and 
$\Omega_\alpha=\{(x,y)\in\eR^2:0<y<x^\alpha,\,\,x\in(0,1)\}$.
If $0<\alpha\le1$,  then the cone property of $(A,\Omega)$ is fulfilled at $a=(0,0)$, whereas for $\alpha>1$ it is not satisfied (at $a=(0,0)$).
Using these domains, we can construct a one-parameter family of RFDs with negative relative box dimension; see Proposition~\ref{ndim} below.
\end{example}

In order to prove Proposition~\ref{conem} below, 
we first need an auxiliary result.

\begin{lemma}\label{conel}
Assume that $K=K_r(a,G)$ is an open cone in $\eR^N$ with vertex at $a$ $($and of radius $r>0)$, and $f\in L^1(0,r)$ is a nonnegative function. Then there exists a positive integer $m$, depending only on $N$ and on the opening angle of the cone,  such that
\begin{equation}\label{constm}
\int_{B_r(a)}f(|x-a|)\,\D x\le m\int_K f(|x-a|)\,\D x.
\end{equation}
\end{lemma}

\begin{proof}
Since the sphere $\pa B$ is compact, there exist finitely many calottes $G_1,\dots,G_m$ contained in the sphere, which are all congruent to $G$ (that is, each $G_i$ can be obtained from $G$
by a rigid motion, for $i=1,\dots,m$), and which cover $\pa B$. Let $K_i=K_r(a,G_i)$, with $i=1,\dots,m$, be the corresponding cones with vertex at $a$.
It is clear that the value of
\begin{equation}
\int_{K_i}f(|x-a|)\,\D x
\end{equation}
does not depend on $i$. Since $B_r(a)=\cup_{i=1}^m K_i$, we then have
\begin{equation}
\int_{B_r(a)}f(|x-a|)\,\D x\le\sum_{i=1}^m\int_{K_i}f(|x-a|)\,\D x=m\int_{K}f(|x-a|)\,\D x,
\end{equation}
as desired.
\end{proof}

\begin{prop}\label{conem}
Let $(A,\Omega)$ be a relative fractal drum in $\eR^N$. Then$:$


\medskip

$(a)$ If the sets $A$ and $\Omega$ are a positive distance apart $($i.e., if $d(A,\O)>0)$, 
then $D(\zeta_{A,\O})=-\ty$; that is, $\zeta_{A,\O}$ is an entire function. Furthermore, $\dim_B(A,\Omega)=-\ty$. 

\medskip

$(b)$ Assume that there exists at least one point $a\in\ov A\cap\ov\Omega$ at which the  relative fractal drum $(A,\Omega)$ satisfies the cone property. 
Then $D(\zeta_{A,\O})\ge0$.
\end{prop}

\begin{proof}
$(a)$ For $r>0$ small enough such that $r< d(A,\Omega)$, where $d(A,\Omega)$ is the distance between $A$ and $\Omega$, we have $A_r\cap\Omega=\emptyset$; so that $\zeta_{A,A_r\cap\Omega}(s)= 0$ for all $s\in\Ce$. Therefore, $D(\zeta_{A,A_r\cap\Omega})=-\ty$. Since $\zeta_{A,\O}(s)-\zeta_{A,A_r\cap\O}(s)$ is an entire function, we conclude that we also have that $D(\zeta_{A,\O})=-\ty$. 
Since $|A_\e\cap\Omega|=0$ for all sufficiently small $\e>0$, we have $\mathcal{M}^r(A,\Omega)=0$ for all $r\in\eR$, and therefore, $\dim_B(A,\Omega)=-\ty$.

\bigskip

$(b)$ Let us reason by contradiction and therefore assume that $D(\zeta_{A,\O})<0$.  In particular, $\zeta_{A,\O}(s)$ is continuous at $s=0$ (because it must then be holomorphic at $s=0$, according to part $(a)$ of Theorem \ref{an_rel} above). By hypothesis, there exists an open cone
$K=K_r(a,G)$, such that $K\subseteq\Omega$.  Using the inequality $d(x,A)\le|x-a|$ (valid for all $x\in\eR^N$ since $a\in\O$) and Lemma~\ref{conel}, we deduce that for any real number $s\in (0,N)$,
$$
\begin{aligned}
\zeta_{A,\O}(s)&\ge\zeta_{A,K}(s)=\int_K d(x,A)^{s-N}\D x\ge\int_K |x-a|^{s-N}\D x\\
&\ge\frac1m\int_{B_r(a)} |x-a|^{s-N}\D x=\frac{N\omega_N}mr^ss^{-1},
\end{aligned}
$$
where $m$ is the positive constant appearing in Equation \eqref{constm} of Lemma~\ref{conel}.
This implies that $\zeta_{A,\O}(s)\to+\ty$ as $s\to0^+$, $s\in\eR$, which contradicts the holomorphicity (or simply, the continuity) of $\zeta_{A,\O}(s)$ at $s=0$.
\end{proof}

The cone condition can be replaced by a much weaker condition, as we will now explain in the following proposition. 

\begin{prop}
Let $(r_k)_{k\ge0}$ be a decreasing sequence of positive real numbers,
converging to zero. We define a subset of the cone $K_r(a,G)$, as follows$:$
\begin{equation}
K_r(a,G,(r_k)_{k\ge0})=\big\{x\in K_r(a,G): |x-a|\in\bigcup_{k=0}^\ty(r_{2k},r_{2k+1})\big\}.
\end{equation}
If we assume that the sequence $(r_k)_{k\ge1}$ is such that
\begin{equation}\label{coner}
\sum_{k=0}^\ty(-1)^kr_k^s\to L>0\quad\mbox{{\rm as}\qs$s\to0^+$, $s\in\eR$,}
\end{equation} 
then the conclusion of Proposition~\ref{conem}$($b$)$ still holds, with the cone condition involving $K:=K(a,G)$ replaced by the above modified cone 
condition, involving the set 
$K':=K_r(a,G,(r_k)_{k\ge0})$ contained in $K$. 
\end{prop}

\begin{proof}
In order to establish this claim, it suffices to use a procedure analogous to the one used in the proof of Proposition~\ref{conem}:
$$
\begin{aligned}
\zeta_{A,\O}(s)&\ge \int_{K'} |x-a|^{s-N}\D x\ge
\frac1m\sum_{k=0}^\ty\int_{B_{r_{2k}}(a)\setminus B_{r_{2k+1}}(s)} |x-a|^{s-N}\D x\\
&=  \frac{N\omega_N}ms^{-1}\sum_{k=0}^\ty (r_{2k}^s-r_{2k+1}^s)=\frac{N\omega_N}ms^{-1}\sum_{k=0}^\ty(-1)^kr_k^s.
\end{aligned} 
$$
For example, if $r_k=2^{-k}$, then condition (\ref{coner}) is fulfilled since
$$
\sum_{k=0}^\ty(-1)^kr_k^s
=\sum_{k=0}^\ty(-1)^k2^{-ks}=
\frac1{1+2^{-s}}\to\frac12\quad\mbox{as\qs$s\to0^+$, $s\in\eR$.}
$$
This concludes the proof of the proposition.
\end{proof}


\setlength{\unitlength}{0.8cm}
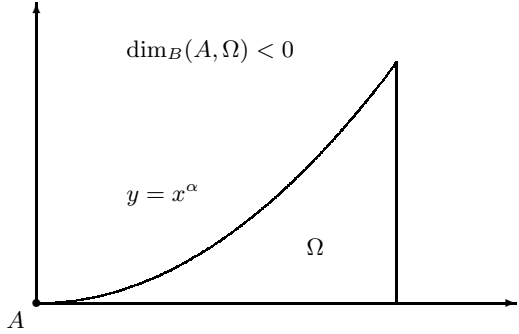
\begin{figure}[t]
\begin{center}
\begin{picture}(8,5)
\put(0,0){\vector(1,0){8}}
\put(0,0){\vector(0,1){5}}
\qbezier(0,0)(3,0)(6,4)
\put(6,0){\line(0,1){4}}
\put(0,0){\circle*{0.13}}
\put(-0.5,-0.4){\small$A$}
\put(1.5,1.7){\small$y=x^{\a}$}
\put(4.5,0.8){\small$\O$}
\put(1.5,4.1){\small$\dim_B(A,\O)<0$}
\end{picture}%
\end{center}
\vskip2mm
\caption{\small A relative fractal drum $(A,\O)$ with negative box dimension $\dim_B(A,\O)=1-\a<0$ (here $\a>1$), due to the `flatness' of the open set $\O$ at $A$; see Proposition \ref{ndim}. This provides a further illustration of {\em the drop in dimension phenomenon} (for relative box dimensions).}
\label{ndim_figure}\index{drop of dimension phenomenon}
\end{figure}

The following proposition (building on Example~\ref{3.1.23} above) shows that the box dimension of a relative fractal drum can be negative, and even take on any prescribed negative value; see Figure \ref{ndim_figure}.

\begin{prop}\label{ndim}
Let $A=\{(0,0)\}$ and 
\begin{equation}
\Omega=\{(x,y)\in\eR^2:0<y<x^\alpha,\,\,x\in(0,1)\},
\end{equation} 
where $\alpha>1$.
$($See Figure \ref{ndim_figure}.$)$
Then the relative fractal drum $(A,\O)$ has a negative box dimension. More specifically, $\dim_B(A,\O)$ exists, the relative fractal drum $(A,\O)$ is Minkowski measurable and
\begin{equation}
\begin{gathered}
\dim_B(A,\Omega)=D(\zeta_{A,\O})=1-\alpha<0,\\
\mathcal{M}^{1-\alpha}(A,\Omega)=\frac1{1+\alpha},\\
D_{\rm mer}(\zeta_{A,\O})\le3(1-\alpha).
\end{gathered}
\end{equation}
Furthermore, $s=1-\alpha$ is a simple pole of $\zeta_{A,\O}$.
\end{prop}

\begin{proof}
First note that $A_\e=B_\e((0,0))$. Therefore, for every $\e>0$, we have
$$
|A_\e\cap\Omega|\le\int_0^\e x^\alpha \D x=\frac{\e^{\alpha+1}}{\alpha+1}.
$$
If we choose a point $(x(\e),y(\e))$ such that
$$(x(\e),y(\e))\in\pa (A_\e)\cap\{(x,y):y=x^\alpha,\,\,x\in(0,1)\},$$
then the following equation holds:
\begin{equation}\label{ae}
x(\e)^2+x(\e)^{2\alpha}=\e^2.
\end{equation}
It is clear that
$$
|A_\e\cap\Omega|\ge\int_0^{x(\e)} x^\alpha \D x=\frac{x(\e)^{\alpha+1}}{\alpha+1}.
$$
Letting $D:=1-\alpha$, we conclude that
\begin{equation}\label{Malpha}
\frac1{\alpha+1}\Big(\frac{x(\e)}\e\Big)^{\alpha+1}\le\frac{|A_\e\cap\Omega|}{\e^{2-D}}\le\frac1{\alpha+1},\q\mbox{for all\qs$\e>0$}.
\end{equation}
We deduce from (\ref{ae}) that $x(\e)\sim\e$ as $\e\to0^+$, since
\begin{equation}\label{xee1}
\frac{x(\e)}{\e}=(1+x(\e)^{2(\alpha-1)})^{-1/2}\to1\quad\mbox{as\qs$\e\to0^+$;}
\end{equation}
therefore, (\ref{Malpha}) implies that $\dim_B(A,\Omega)=D$
and $\mathcal{M}^D(A,\Omega)=1/(\alpha+1)$. 

Using (\ref{Malpha}) again, we have that
\begin{equation}\label{ae2}
0\le f(\e):=\frac1{\alpha+1}-\frac{|A_\e\cap\Omega|}{\e^{2-D}}\le \frac1{\alpha+1}\Big(1-\Big(\frac{x(\e)}{\e}\Big)^{\alpha+1}\Big).
\end{equation}
Using (\ref{xee1})
and the binomial expansion,  we conclude that
$$
\Big(\frac{x(\e)}{\e}\Big)^{\alpha+1}=1-\frac{\alpha+1}2x(\e)^{2\alpha-2}+o(x(\e)^{2\alpha-2})\quad\mbox{as\qs$\e\to0^+$.}
$$
Hence, we deduce from (\ref{ae2}) that
$$
f(\e)=O(x(\e^{2\alpha-2}))=O(\e^{2\alpha-2})\quad\mbox{as\qs$\e\to0^+$.}
$$
Since $|A_\e\cap\Omega|=\e^{2-D}((\alpha+1)^{-1}+f(\e))$, 
we conclude that 
$$
D_{\rm mer}(\zeta_{A,\O})\le D-(2\alpha-2)=3(1-\alpha).
$$ 
Furthermore, $s=D$ is a simple pole.

Finally, we note that the equality $D(\zeta_{A,\O})=D$ follows from (\ref{rel_absc}).
\end{proof}



In the following lemma, we show that for any $\d>0$, the sets of principal complex dimensions of the RFDs $(A,\O)$ and $(A,A_\d\cap\O)$ coincide. 

\begin{lemma}\label{dimd} 
Assume that $(A,\O)$ is a relative fractal drum in $\eR^N$. Then, for any $\d>0$, we have 
\begin{equation}\label{excision}
\zeta_{A,\O}\sim\zeta_{A,A_\d\cap\O},
\end{equation}
where the equivalence relation $\sim$ is given in Definition \ref{equ}.
In particular, 
\begin{equation}
\dim_{PC}(A,\O)=\dim_{PC}(A,A_\d\cap\O)
\end{equation}
 and therefore, 
\begin{equation}
\ov\dim_B(A,\O)=\ov\dim_B(A,A_\d\cap\O).
\end{equation}
Here, $A_\d$, the $\d$-neighborhood of $A$, can be taken with respect to any norm on $\eR^N$.
$($This extra freedom will be used in Corollary \ref{flat} just below.$)$
\end{lemma}

\begin{proof} Recall that according to the definition of a relative fractal drum $(A,\O)$, there exists $\d_1>0$ such that $d(x,A)<\d_1$ for all $x\in\O$; see Definition \ref{rfd}.
On the other hand, we have that $d(x,A)>\d$ for all $x\in \O\setminus A_\d$.\label{setminus} Therefore, 
we conclude that the difference
$$
\zeta_{A,\O}(s)-\zeta_{A,A_\d\cap\O}(s)=\int_{\O\setminus A_\d}d(x,A)^{s-N}\D x
$$
defines an entire function. This proves the desired equivalence in \eqref{excision}. The remaining claims of the lemma follow immediately from this equivalence. Finally, the fact that any norm on $\eR^N$ can be chosen to define $A_\d$ follows from the equivalence of all the norms on $\eR^N$.
\end{proof}

The following result provides an example of a nontrivial relative fractal drum $(A,\Omega)$ such that $\dim_B(A,\Omega)=-\ty$.
It suffices to construct a domain $\Omega$ of $\eR^2$ which is {\em flat}\index{flatness of a relative fractal drum} in a neighborhood of one of its boundary points.

\begin{cor}[{A maximally flat RFD}]\label{flat}
Let $A=\{(0,0)\}$ and 
\begin{equation}\label{OmegaA'}
\Omega'=\{(x,y)\in\eR^2:0<y<{\E}^{-1/x},\,\,0<x<1\}.
\end{equation}
Then $\dim_B(A,\Omega')$ exists and
\begin{equation}
\dim_B(A,\Omega')=D(\zeta_{A,\O'})=-\ty.
\end{equation}
\end{cor}

\begin{proof}
Let us fix $\alpha>1$. Then, by l'Hospital's rule, we have that
$$
\lim_{x\to0^+}\frac{\E^{-1/x}}{x^\a}=\lim_{t\to+\ty}\frac{t^\a}{\E^t}=0. 
$$
 Hence, there exists $\delta=\delta(\alpha)>0$
such that $0<\E^{-1/x}<x^\a$ for all $x\in(0,\d)$; that is,
$$
\Omega_{\d(\a)}'\st\Omega_{\d(\alpha)},
$$
where
$$
\Omega'_{\delta(\alpha)}:=\{(x,y)\in\eR^2:0<y<{\E}^{-1/x},\,\,0<x<\delta(\alpha)\,\}
$$
and
$$
\Omega_{\d(\a)}:=\{(x,y)\in\eR^2:0<y<x^\alpha,\,\,0<x<\delta(\alpha)\,\}.
$$
Using Lemma~\ref{dimd}, with $\O'$ instead of $\O$ and with the $\ell^\ty$-norm on $\eR^2$ instead of the usual Euclidean norm (note that $\Omega'_{\delta(\alpha)}=\O'\cap B_{\d(\a)}(0)$,
where $B_\d(0):=\{(x,y)\in\eR^2:|(x,y)|_{\ty}<\d\}$ and $|(x,y)|_\ty:= \max\{|x|,|y|\}$), along with Proposition~\ref{ndim}, we see that 
$$
\ov\dim_B(A,\Omega')=\ov\dim_B(A,\Omega'_{\delta(\alpha)})\le\dim_B(A,\Omega_{\d(\alpha)})=1-\alpha.
$$
The claim follows by letting $\alpha\to+\ty$, since then, we have that
$$
-\ty\le\underline\dim_B(A,\Omega')\le\ov\dim_B(A,\Omega')=-\ty. 
$$
We conclude, as desired, that $\dim_B(A,\O)$ exists and is equal to $-\ty$.
\end{proof}

\begin{figure}[t]
\begin{center}
\includegraphics[width=7cm]{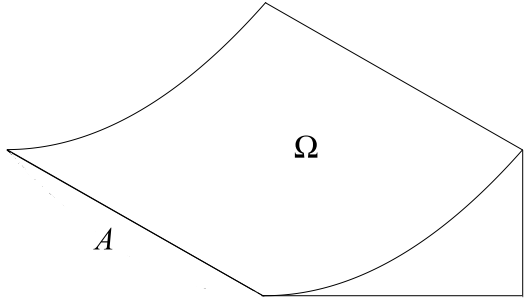}
\caption{\small A relative fractal drum $(A,\Omega)$ with infinite flatness, as described in Remark \ref{flatA}.
In other words, $\O$ has infinite flatness near $A$; equivalently, $\dim_B(A,\O)=-\ty$, which provides an even more dramatic illustration of {\em the drop in dimension phenomenon}\index{drop of dimension phenomenon|textbf} (for relative box dimensions).}
\label{flat_A}
\end{center}
\end{figure}

\begin{remark}\label{flatA} ({\em Flatness and `infinitely sharp blade'}).\index{flatness of a relative fractal drum|textbf}
It is easy to see that Corollary~\ref{flat} can be significantly generalized. For example, it suffices to assume that $a$ is a point
on the boundary of $\Omega$ such that the {\em flatness property of $A$ $($at $a$$)$ relative to} $\Omega$ holds. This can even be formulated in terms of subsets $A$ of the boundary of $\O$.
We can imagine a bounded open set $\Omega\st\eR^3$ with a Lipschitz boundary $\pa\O$, except on a subset $A\st\pa\O$, which may be a line segment, 
near which $\Omega$ is flat; see Figure \ref{flat_A}. A simple construction of such a set is $\Omega=\Omega'\times(0,1)$, where $\Omega'$ is given as in Corollary~\ref{flat}, and $A=\{(0,0)\}\times(0,1)$; see Equation (\ref{OmegaA'}).
Note that this domain is not Lipschitz near the points of $A$, and not even H\"olderian. The {\em flatness of a relative fractal drum}\label{flatness} $(A,\Omega)$ can be defined by
$$
\operatorname{fl}(A,\O)=\left(\ov\dim_B(A,\O)\right)^-,
$$
where $(r)^{-}:=\max\{0,-r\}$ is the negative part of a real number $r$. {\em We say that the flatness of $(A,\O)$ is nontrivial if $\operatorname{fl}(A,\O)>0$, that is, if $\ov\dim_B(A,\O)<0$}.
In the example mentioned just above, we have a relative fractal drum $(A,\Omega)$ with infinite flatness, i.e., with $\operatorname{fl}(A,\O)=+\ty$.
Intuitively, it can be viewed as an `ax' with an `infinitely sharp' blade.
\end{remark}

\section{Scaling property of relative distance zeta functions}\label{scaling_property}
We start this section with the following result, which shows that if $(A,\Omega)$ is a given relative fractal drum, then for any 
 $\lambda>0$, the zeta function $\zeta_{\g A,\g\O}(s)$ of the scaled relative fractal drum $\g(A,\O):=(\lambda A,\lambda\Omega)$ is equal to the zeta function $\zeta_{A,\O}(s)$
 of $(A,\Omega)$ multiplied by~$\lambda^s$. 

\begin{theorem}[{\rm Scaling property of relative distance zeta functions}]\label{scaling}\index{scaling property!of the relative distance zeta functions|textbf}
Let $\zeta_{A,\O}(s)$ be the relative distance zeta function of an RFD $(A,\O)$. Then, for any positive real number $\lambda$, we have that $D(\zeta_{\g A,\g \O)})=D(\zeta_{A,\O})=\overline{\dim}_B(A,\O)$ and
\begin{equation}\label{zeta_scaled}
\zeta_{\g A,\g\O}(s)=\lambda^s\zeta_{A,\O}(s),
\end{equation}
for all $s\in\Ce$ with $\re s>\overline{\dim}_B(A,\O)$ and any $\g>0$. $($See also Corollary~\ref{10.151/2} below for a more general statement.$)$
\end{theorem}

\begin{proof}
The claim is established by introducing a new variable $y=x/\lambda$, and by noting that $d(\lambda y,\lambda A)=\lambda\,d(y,A)$, for any $y\in\eR^N$ (which
is an easy consequence of the homogeneity of the Euclidean norm). Indeed, in light of part $(b)$ of Theorem~\ref{an_rel}, for any $s\in\Ce$ with $\re s>\overline{\dim}_B(A,\O)=D(\zeta_{A,\O})$, we have successively:
$$
\begin{aligned}
\zeta_{\g A,\g\O}(s)&=\int_{\lambda\Omega}d(x,\g A)^{s-N}\D x\\
&=\int_{\Omega}d(\lambda y,\lambda A)^{s-N}\lambda^N\D y\\
&=\lambda^s\int_{\Omega}d(y,A)^{s-N}\D y=\lambda^s\zeta_{A,\O}(s).
\end{aligned}
$$
It follows that \eqref{zeta_scaled} holds and $\zeta_{\lambda A,\g\O}$ is holomorphic for $\re s>\overline{\dim}_B(A,\O)$. Since $D(\zeta_{A,\O})=\overline{\dim}_B(A,\O)$ (by part $(b)$ of Theorem~\ref{an_rel}), we deduce that $D(\zeta_{\lambda A,\lambda\O})\leq D(\zeta_{A,\O})$, for every $\lambda>0$. But then, replacing $\lambda$ by its reciprocal $\lambda^{-1}$ in this last inequality, we obtain the reverse inequality (more specifically, we replace $(A,\O)$ by $(\lambda^{-1}A,\lambda^{-1}\O)$ to deduce that for every $\lambda>0$, $D(\zeta_{A,\O})\leq D(\zeta_{\lambda^{-1}A,\lambda^{-1}\O})$; we then substitute $\lambda^{-1}$ for $\lambda$ in this last inequality in order to obtain the desired reversed inequality: for every $\lambda>0$, $D(\zeta_{A,\O})\leq D(\zeta_{\lambda A,\lambda \O})$), and hence, we conclude that
$$
\ov{\dim}_B(A,\O)=D(\zeta_{A,\O})=D(\zeta_{\lambda A,\lambda\O}),
$$
for all $\lambda>0$, as desired.
\end{proof}

We note that if $\mathcal{L}=(\ell_j)_{j\ge1}$ is a fractal string, and $\lambda$ is a positive constant, then for the scaled string $\lambda\mathcal{L}:=(\lambda \ell_j)_{j\ge1}$,
the corresponding claim in Theorem~\ref{scaling} is trivial: $\zeta_{\lambda\mathcal{L}}(s)=\lambda^s\zeta_{\mathcal{L}}(s)$, for every $\lambda>0.$
Indeed, by definition of the geometric zeta function of a fractal string (see Equation \eqref{geomz}), we have\index{scaling property!of the geometric zeta functions|textbf}
$$
\zeta_{\lambda\mathcal{L}}(s)=\sum_{j=1}^{\ty}(\lambda \ell_j)^s=\lambda^s\sum_{j=1}^{\ty}\ell_{j}^s=\lambda^s\zeta_{\mathcal{L}}(s),
$$
for $\re s>D(\zeta_{\mathcal{L}})$.
(The same argument as above then shows that $D(\zeta_{\mathcal{L}})=D(\zeta_{\lambda\mathcal{L}})$.)
Then, by analytic (i.e., meromorphic) continuation, the same identity continues to hold in any domain to which $\zeta_{\mathcal{L}}$ can be meromorphically extended to the left of the 
critical line $\{\re s=D(\zeta_{\mathcal{L}})\}$. 

\bigskip

The following result supplements Theorem~\ref{scaling} in several different and significant ways.

\begin{cor}\label{10.151/2}
Fix $\lambda>0$. Assume that $\zeta_{A,\O}$ admits a meromorphic continuation to some open connected neighborhood $U$ of the open half-plane $\{\re s>\ov{\dim}_B(A,\O)\}$. Then, so does $\zeta_{\g A,\g\O}$ and the identity \eqref{zeta_scaled} continues to hold for every $s\in U$ which is not a pole of $\zeta_{A,\O}$ $($and hence, of $\zeta_{\g A,\g\O}$ as well$)$.

Moreover, if we assume, for simplicity, that $\omega$ is a simple pole of $\zeta_{A,\O}$ $($and hence also, of $\zeta_{\g A,\g\O}$$)$, then the following identity holds$:$
\begin{equation}\label{10.91/4}
\res(\zeta_{\g(A,\O)},\omega)=\lambda^\omega\res(\zeta_{A,\O},\omega).
\end{equation}
If $s$ is a multiple pole, then an analogous statement can be made about the principal parts $($instead of the residues$)$ of the zeta functions involved, as the reader can easily verify.
\end{cor}

\begin{proof}
The fact that $\zeta_{\g A,\g\O}$ is holomorphic at a given point $s\in U$ if $\zeta_{A,\O}$ is holomorphic at $s$ (for example, if $\re s>\ov{\dim}_B(A,\O)$), follows from \eqref{zeta_scaled} and the equality $D(\zeta_{\g A,\g\O})=D(\zeta_{A,\O})=\ov{\dim}_B(A,\O)$.
An analogous statement is true if ``holomorphic" is replaced with ``meromorphic".
More specifically, by analytic continuation 
of~\eqref{zeta_scaled}, $\zeta_{\g A,\g\O}$ is meromorphic in an open connected set $U$ (containing the critical line $\{\re s=\ov\dim_B(A,\O)\}$) if and only if $\zeta_{A,\O}$ is meromorphic in $U$, and then, clearly, identity~\eqref{zeta_scaled} continues to hold for every $s\in U$ which is not a pole of $\zeta_{A,\O}$ (and hence also, of $\zeta_{\g A,\g\O}$).
Therefore, the first part of the corollary is established.

Next, assume that $\omega$ is a simple pole of $\zeta_{A,\O}$. Then, in light of \eqref{zeta_scaled} and the discussion in the previous paragraph, we have that for all $s$ in a punctured neighborhood of $\omega$ (contained in $U$ but not containing any other pole of $\zeta_{A,\O}$),
\begin{equation}\label{10.191/2}
(s-\omega)\zeta_{\g(A,\O)}(s)=\lambda^s\left((s-\omega)\zeta_{A,\O}(s)\right).
\end{equation}
The fact that \eqref{10.91/4} holds now follows by letting $s\to\omega$, $s\neq\omega$ in \eqref{10.191/2}.
Indeed, we then have
$$
\res(\zeta_{A,\O},\omega)=\lim_{s\to\omega}(s-\omega)\zeta_{A,\O}(s),
$$
and similarly for $\res(\zeta_{\g(A,\O)},\omega)$.
\end{proof}


\medskip

The scaling property of relative zeta functions (established in Theorem~\ref{scaling} and Corollary~\ref{10.151/2}) motivates us to introduce the notion of relative
fractal spray, which is very close
to (but not identical with) the usual notion of fractal spray introduced by the first author 
and Carl Pomerance in \cite{lapiduspom3} (see \cite{lapidusfrank12} and the references therein, including [LapPe2--3, LapPeWi1--2, Pe, PeWi, DemDenKo\"U, DemKo\"O\"U]).
First, we define the operation of union of (disjoint) families of RFDs (Definition~\ref{union}).

\begin{defn}\label{union}\index{disjoint union!of RFDs, $\sqcup_{j=1}^\ty(A_j,\Omega_j)$}
Let $(A_j,\Omega_j)_{j\ge1}$ be a countable family of relative fractal drums in $\eR^N$, such that the corresponding family of open sets $(\Omega_j)_{j\ge1}$
is disjoint (i.e., $\O_j\cap\O_k=\emptyset$ for $j\neq k$), $A_j\subseteq\Omega_j$ for each $j\in\eN$, and the set $\Omega:=\cup_{j=1}^\ty\Omega_j$
is of finite $N$-dimensional Lebesgue measure (but may be unbounded).
Then, the {\em union of the $(${\rm finite or countable}$)$ family of relative fractal drums} $(A_j,\O_j)$ $(j\geq 1)$ is the relative fractal drum $(A,\Omega)$, where $A:=\cup_{j=1}^\ty A_j$ and $\O:=\cup_{j=1}^{\ty}\O_j$. We write 
\begin{equation}\label{sqcup}
(A,\Omega)=\bigsqcup_{j=1}^\ty(A_j,\Omega_j).
\end{equation}
\end{defn}

It is easy to derive the following countable additivity property of the distance zeta functions.

\begin{theorem}\label{unionz}
Assume that $(A_j,\Omega_j)_{j\ge1}$ is a  finite or countable family of RFDs satisfying the conditions of Definition~\ref{union}, and
let $(A,\Omega)$ be its union $($in the sense of Equation \eqref{sqcup} appearing in Definition~\ref{union}$)$.
Furthermore, assume that the following condition is fulfilled$:$
\begin{equation}\label{dj}
\mathrm{For\ any\ }j\in\eN\mathrm{\ and\ }x\in\Omega_j,\mathrm{\ we\ have\ that\ }d(x,A)=d(x,A_j).
\end{equation}
Then, for $\re s>\ov{\dim}_B(A,\O)$, 
\begin{equation}\label{10.11}
\zeta_{A,\O}(s)=\sum_{j=1}^\ty\zeta_{A_j,\O_j}(s).
\end{equation}
Condition \eqref{dj} is satisfied, for example, if for every $j\in\eN$, $A_j$ is equal to the boundary of $\Omega_j$ in $\eR^N$ $($that is, $A_j:=\pa\O_j$$)$.
\end{theorem}

\begin{proof}
The claim follows from the following computation, which is valid for $\re s>\ov{\dim}_B(A,\O)$:
\begin{equation}\label{zeta_sum}
\begin{aligned}
\zeta_{A,\O}(s)&=\int_{\Omega}d(x,A)^{s-N}\D x=\sum_{j=1}^\ty \int_{\Omega_j}d(x,A)^{s-N}\D x\\
&=\sum_{j=1}^\ty \int_{\Omega_j}d(x,A_j)^{s-N}\D x=\sum_{j=1}^\ty \zeta_{A_j,\O_j}(s).
\end{aligned}
\end{equation}

More specifically, clearly, \eqref{zeta_sum} holds for every real number $s$ such that $s>\ov{\dim}_B(A,\O)\geq D(\zeta_{A,\O})$. Therefore, for such a value of $s$,
$$
\zeta_{A,\O_j}(s)=\int_{\O_j}d(x,A)^{s-N}\D x\leq\int_{\O}d(x,A)^{s-N}\D x=\zeta_{A,\O}(s)<\infty,
$$
for every $j\geq 1$. Hence,
\begin{equation}\label{10.121/2}
\sup_{j\geq 1}\,\{D(\zeta_{A,\O_j})\}\leq D(\zeta_{A,\O})\leq\ov{\dim}_B(A,\O),
\end{equation}
from which \eqref{zeta_sum} now follows for all $s\in\Ce$ with $\re s>\ov{\dim}_B(A,\O)$, in light of the countable additivity of the local complex Borel measure (and hence, locally bounded measure) on $\O$, given by $\D\gamma(x):=d(x,A)^{s-N}\D x$.
(Note that according to the hypothesis of Definition~\ref{union}, we have $|\O|<\ty$, so that $\D\gamma$ is indeed a local complex Borel measure; see, e.g.,~\cite{folland} or~\cite{rudin}, along with [{DolFri}], [{JohLap}], [{JohLapNi}] and \cite[Appendix A]{fzf} for the notion of a local measure.)
\end{proof}

\begin{remark}\label{10.171/2}
In the statement of Theorem~\ref{unionz}, the numerical series on the right-hand side of \eqref{10.11} converges absolutely (and hence, converges also in $\Ce$) for all $s\in\Ce$ such that $\re s>\ov{\dim}_B(A,\O)$. In particular, for every real number $s$ such that $s>\ov{\dim}_B(A,\O)$, it is a convergent series of positive terms (i.e., it has a finite sum). It remains to be investigated whether (and under which hypotheses) Equation~\eqref{10.11} continues to hold for all $s\in\Ce$ in a common domain of meromorphicity of the zeta functions $\zeta_{A,\O}$ and $\zeta_{A,\O_j}$ for $j\geq 1$ (away from the poles). At the poles, an analogous question could be raised for the corresponding residues (assuming, for simplicity, that the poles are simple). 
\end{remark}

\section{Relative tube zeta functions}\label{rtzf}

We begin this section by introducing the {\em relative 
tube\index{relative tube zeta function, $\tilde\zeta_{A,\O}$|textbf}\index{fractal zeta function!relative tube zeta function, $\tilde\zeta_{A,\O}$|textbf} zeta function}\index{fractal zeta function!relative tube zeta function, $\tilde\zeta_{A,\O}$|textbf} associated with the relative fractal drum $(A,\Omega)$ in $\eR^N$. It is defined by
\begin{equation}\label{rel_tube_zeta}
\tilde\zeta_{A,\Omega}(s):=\int_0^\delta t^{s-N-1}|A_t\cap\Omega|\,\D t,
\end{equation}
for all $s\in\Ce$ with $\re s$ sufficiently large, where $\delta>0$ is fixed. As we see, $\tilde{\zeta}_{A,\O}$ involves the 
{\em relative tube function}\index{relative tube function of $(A,\O)$, $t\mapsto\vert A_t\cap\O\vert$|textbf} 
$t\mapsto|A_t\cap\Omega|$. 
As was noted in Remark \ref{recon}, if $\O:=A_\d$ with $A\subset\eR^N$ bounded, we recover the tube zeta function of the set $A$; that is, $\tilde\zeta_A(s):=\int_0^\delta t^{s-N-1}|A_t|\,\D t$, for all $s\in\Ce$ with $\re s$ sufficiently large.

The abscissa of convergence of the relative tube zeta function $\tilde\zeta_{A,\Omega}$ is given by $D(\tilde\zeta_{A,\Omega})=\ov\dim_B(A,\Omega)$.
This follows from the following fundamental identity (or {\em functional equation}), which connects the relative tube zeta function $\tilde\zeta_{A,\O}$ and the relative distance zeta function $\zeta_{A,\Omega}$,
defined by (\ref{rel_dist_zeta}):
\begin{equation}\label{rel_equality}
\zeta_{A,A_\delta\cap\Omega}(s)= \delta^{s-N}|A_\delta\cap\Omega|+(N-s)\tilde\zeta_{A,\Omega}(s),
\end{equation} 
for all $s\in\Ce$ such that $\re s>\ov\dim_B(A,\O)$.
Its proof is based on the the known identity
\begin{equation}
\int_{A_\delta\cap\Omega}d(x,A)^{-\gamma}\,\D x=\delta^{-\gamma}|A_\delta\cap\Omega|+\gamma\int_0^\delta t^{-\gamma-1}|A_t\cap\Omega|\,\D t,
\end{equation}
where $\gamma>0$; see \cite[Theorem~2.9(a)]{mink}, or a more general form provided in \cite[Lemma~3.1]{rae}. 
As a special case, when $\O:=A_\d$ with $A\subset\eR^N$ bounded, Equation \eqref{rel_equality} reduces to 
\begin{equation}\label{equality}
\zeta_{A}(s)= \delta^{s-N}|A_\delta|+(N-s)\tilde\zeta_{A}(s),
\end{equation} 
for all $s\in\Ce$ such that $\re s>\ov\dim_BA$, which has been obtained in \cite{dtzf}.

\medskip

The following proposition connects the residues of the relative tube and distance functions. 

\begin{prop}\label{resrfd} Assume that $(A,\O)$ is an RFD in $\eR^N$.
Let $U$ be a connected open subset of $\Ce$ to which the relative distance zeta function $\zeta_{A,\Omega}$ can be meromorphically extended, and such that it contains the critical line $\{\re s=D(\zeta_{A,\O})\}$. 
Then the relative tube function $\tilde\zeta_{A,\Omega}$ can be meromorphically extended to $U$ as well. 
Furthermore, if $\o\in U$ is a simple pole of $\zeta_{A,\Omega}$, then it is also a simple pole of
$\tilde\zeta_{A,\Omega}$ and we have that
\begin{equation}\label{resdt}
\res(\zeta_{A,\Omega},\o)=(N-\o)\cdot\res(\tilde\zeta_{A,\Omega},\o).
\end{equation}
Moreover, the functional equation \eqref{rel_equality} continues to hold for all $s\in U$.

The proposition also holds if we interchange the relative distance function and relative tube function in the above statement.
\end{prop}

\begin{proof}
Since the difference $\zeta_{A,\Omega}(s)-\zeta_{A,A_\delta\cap\Omega}(s)=\zeta_{A,\O\setminus A_\delta\cap\Omega}(s)$ is an entire function (note that $\d<d(x,A)<c$, where $c:=\sup_{x\in A} d(x,A)<\ty$; see property \eqref{delta} in Definition \ref{rfd} of an RFD), it suffices to prove the proposition in the case
of the RFD $(A,A_\d\cap\O)$ instead of $(A,\O)$. The claim now follows from the functional equation \eqref{rel_equality}. This concludes the proof.
\end{proof}

\begin{example}[Torus relative fractal drum]\label{torus}
Let $\O$ be an open solid torus in $\eR^3$ defined by two radii $r$ and $R$, where $0<r<R<\ty$, and let $A:=\pa\O$ be its topological boundary. In order to compute the tube zeta function of the {\em torus RFD}\index{torus RFD|textbf} $(A,\O)$, we first compute its tube function. Let $\d\in(0,r)$ be fixed. Using Cavalieri's principle, we have that 
\begin{equation}\label{tubet}
|A_t\cap\O|_3=2\pi R\left(r^2-(r-t)^2\right)=2\pi R(2rt-t^2), 
\end{equation}
for all $t\in(0,\d)$, from which it follows that
\begin{equation}
\tilde\zeta_{A,\O}(s):=\int_0^\d t^{s-4}|A_t\cap\O|_3\,\D t=2\pi R\Big(2r\frac{\d^{s-2}}{s-2}-\frac{\d^{s-1}}{s-1}\Big),
\end{equation}
for all $s\in\Ce$ such that $\re s>2$. The right-hand side defines a meromorphic function on the entire complex plane, so that, using the principle of analytic continuation, $\tilde\zeta_{A,\O}$ can be (uniquely) meromorphically extended to the whole of $\Ce$. 
In particular, we see that the multiset of complex dimensions of the torus RFD $(A,\O)$ is given by $\po(A,\O)=\{1,2\}$. Each of the complex dimensions $1$ and $2$ is simple. In particular, we have that
\begin{equation}
\dim_{PC}(A,\O)=\{2\}\q\mbox{and}\q\res(\tilde\zeta_{A,\O},2)=4\pi Rr.
\end{equation}
Also, $\ov\dim_B(A,\O)=D(\tilde\zeta_A)=2$.
From Equation \eqref{resM} appearing in Theorem \ref{rel_measurable} below, we conclude that the $2$-dimensional Minkowski content of the torus RFD $(A,\O)$ is given by
\begin{equation}
\M^{2}(A,\O)=4\pi Rr.
\end{equation}
Since $|A_t|_3=2\pi R\left((r+t)^2-(r-t)^2\right)$, we can also easily compute the `ordinary' tube zeta function $\tilde\zeta_A$ of the torus surface $A$ in $\eR^3$:
\begin{equation}\label{torusz}
\tilde\zeta_A(s)=8\pi Rr\frac{\d^{s-2}}{s-2},
\end{equation}
for all $s\in\Ce$. In particular, 
$
\res(\tilde\zeta_A,2)=8\pi Rr
$.
Using Equations \eqref{rel_equality} and \eqref{equality}, we deduce from \eqref{torusz} the corresponding expressions for the distance zeta functions, which are valid for all $s\in\Ce$:
\begin{equation}
\zeta_{A,\O}(s)=2\pi R\Big(2r\frac{\d^{s-1}}{s-2}-\frac2{s-1}\Big),
\q
\zeta_A(s)=8\pi Rr\frac{\d^{s-2}}{s-2}.
\end{equation}
Also, 
$$
{\mathcal P}(\zeta_{A,\O})={\mathcal P}(\tilde\zeta_{A,\O})=\{1,2\}
$$
and
$$
{\mathcal P}_c(\zeta_{A,\O})={\mathcal P}_c(\tilde\zeta_{A,\O})=\{2\}
$$
(with each pole $1$ and $2$ being simple) and
$$
\ov\dim_B(A,\O)=D(\zeta_{A,\O})=D(\tilde\zeta_{A,\O})=2.
$$
Furthermore, we see that $\res(\zeta_{A,\O},2)=4\pi Rr$ and $\res(\zeta_A,2)=8\pi Rr$, in agreement with \eqref{resdt} in Proposition \ref{resrfd} above.
\end{example}

\medskip

One can easily extend the example of the $2$-torus to any (smooth) closed submanifold of $\eR^N$ 
(and, in particular, of course, to the $n$-torus, with $n\ge2$). This can be done by using Federer's tube formula\index{Federer's tube formula} \cite{Fed1} for sets of positive reach, which extends and unifies Weyl's tube formula \cite{Wey}
for (proper) smooth submanifolds of $\eR^N$ and Steiner's formula (obtained by Steiner \cite{Stein} and his successors) for compact convex subsets of $\eR^N$. The global form of Federer's tube formula\index{Federer's tube formula} expresses the volume of $t$-neighborhoods of a (compact) set of positive reach\footnote{A closed subset $C$ of $\eR^N$ is said to be of {\em positive reach} if there exists $\d_0>0$ such that every point $x\in \eR^N$ within a distance less than $\d_0$ from $C$ has a unique metric projection onto $C$; see \cite{Fed1}. 
The {\em reach} of $C$ is defined as the supremum of all such positive numbers $\d_0$.
Clearly, every closed convex subset of $\eR^N$ is of infinite (and hence, positive) reach.}\index{positive reach of a closed set|textbf}\index{reach of a closed set|textbf} $A\st\eR^N$ as a polynomial of degree at most $N$ in $t$, whose coefficients are (essentially) the so-called {\em Federer's curvatures}\index{Federer's curvatures|textbf} and which generalize Weyl's curvatures\index{Weyl's curvatures} in \cite{Wey} (see \cite{BergGos} for an exposition) and Steiner's curvatures in \cite{Stein}
(see \cite[Chapter 4]{Schn1} for a detailed exposition) in the case of submanifolds and compact convex sets, respectively.
We also draw the attention of the reader to the related notions of `fractal curvatures' and `fractal curvature measures' for fractal sets, introduced by Winter in [Wi] and Winter and Zahle in [WiZ\"a1--3]; for information on closely related topics in integral geometry and on tube formulas, see also, for example, [Stein], [Mink], [Wey],  [Bla], [Fed1], [KlRot], [Z\"a1--3], [Schn1--2], [HugLasWeil], [LapPe3], [LapPeWi1], [LapRa\v Zu1] and [LapRa\v Zu4--6], along with the many relevant references therein.

In the present context, for a compact set of positive reach $C\st\eR^N$, it is easy to deduce from the tube formula in \cite{Fed1} an explicit expression for $\tilde\zeta_A(s)$, as follows (with $A:=\pa C$).\footnote{Relative versions are also possible, for example for the RFD $(A,\overset {\circ}C)$, where $\overset {\circ}C$ is the interior of $C$, assumed to be nonempty. Under appropriate assumptions, the associated expression for $\tilde\zeta_A$ may take a slightly different form (because the corresponding tube formula would be of pluriphase type in the sense of [LapPe2--3,LapPeWi1], that is, piecewise polynomial of degree $\le N-1$) but Equation \eqref{Ptildez} would remain valid in this case.}

\begin{prop}\label{preach}
Let $A=\pa C$ be the boundary of a $($nonempty$)$ compact set $C$ of positive reach in $\eR^N$. Then, for any $\d>0$ sufficiently small $($and less than the reach of $A$, in particular$)$, we have that 
\begin{equation}\index{Federer's tube formula}
\tilde\zeta_A(s):=\tilde\zeta_A(s;\d)=\sum_{k=0}^{N-1}c_k\frac{\d^{s-k}}{s-k},
\end{equation}
where $|A_t|=\sum_{k=0}^{N-1}c_kt^{N-k}$ for all $t\in(0,\d)$ and the coefficients $c_k$ are the $($normalized$)$ 
Federer curvatures. $($From the functional equation \eqref{rel_equality} above, one then deduces at once a corresponding explicit expression for the distance zeta function $\zeta_A(s):=\zeta_A(s;\d)$.$)$

Hence, $\dim_BA$ exists and 
\begin{equation}
D:=D(\tilde\zeta_A)=D(\zeta_A)=\dim_BA=\max\{k\in\{0,1,\dots,N-1\}:c_k\ne0\}
\end{equation}
and $($since $D\le N-1<N$$)$,
\begin{equation}\label{Ptildez}
\P:=\P(\tilde\zeta_A)=\P(\zeta_A)\stq\{0,1,\dots,N-1\}.
\end{equation}
In fact,
\begin{equation}
\P=\big\{k\in\{0,1,\dots,N-1\}:c_k\ne0\big\}\stq\{k_0,\dots,D\},
\end{equation}
where $k_0:=\min\big\{k\in\{0,1,\dots,D\}:c_k\ne0\big\}$.
Furthermore, each of the complex dimensions of $A$ is simple. 

Finally, if $C$ is such that its affine hull is all of $\eR^N$ $($which is the case when the interior of $C$ is nonempty and, in particular, if $C$ is a convex body$)$, then $D=N-1$, while if $C$ is a $($smooth$)$ submanifold with boundary the closed $d$-dimensional smooth submanifold $A:=\pa C$ $($with $0\le d\le N-1$$)$, then $D=d$.\footnote{One could also work with a closed (i.e., boundaryless) $d$-dimensional submanifold of $\eR^N$, with $0\le d\le N-1$.}
\end{prop}

\medskip

For the $2$-torus $A$, we have $N=3$, $D=2$ (since the Euler characteristic of $A$ is equal to zero), $c_2\ne0$,\footnote{Note that $c_2$ is just proportional to the area of the $2$-torus, with the proportionality constant being a standard positive constant.} $c_1\ne0$, and hence, $c_0=0$, $k_0=1$ and $\P=\{1,2\}$, as was found in Example \ref{torus} via a direct computation.

We note that much more general tube formulas called 
``fractal tube formulas''\index{fractal tube formula} are obtained in \cite{cras2} (as well as in \cite[Chapter 5]{fzf}, see also \cite{cras1})
for arbitrary bounded sets (and even more generally, RFDs) in $\eR^N$, under mild growth assumptions on the associated fractal zeta functions.

\section{Meromorphic extensions of relative zeta functions}\label{rzfe}

We shall use the following assumption on the asymptotics of the relative tube function
$t\mapsto|A_t\cap\Omega|$:
\begin{equation}\label{h}
|A_t\cap \Omega|=t^{N-D}h(t)(\mathcal{M}+O(t^{\alpha}))\quad\mbox{as\qs$t\to0^+$,}
\end{equation}
where $\mathcal{M}>0$, $\alpha>0$ and $D\le N$ are given in advance. Here, we assume that the function $h(t)$ is positive and
has a sufficiently slow growth near the origin, in the sense that for any $c>0$, $h(t)=O(t^c)$ as $t\to0^+$.
Typical examples of such functions are $h(t)=(\log t^{-1})^m$, $m\ge1$, or more generally, 
$$
h(t)=\big(\underbrace{\log\ldots\log}_{n}(t^{-1})\big)^m
$$ 
(the $m$-th power of the $n$-th iterated logarithm of $t^{-1}$, for $n\ge1$), and in these cases we obviously have $\mathcal{M}^D(A,\Omega)=+\ty$.
For this and other examples, see \cite{lapidushe}.
The function $t\mapsto t^Dh(t)^{-1}$ is usually called a {\em gauge function},\label{gauge}\index{gauge function $h$ of an RFD $(A,\O)$} but for the sake of simplicity, we shall instead use this name only for the function $h(t)$.

Assuming that a relative fractal drum $(A,\Omega)$ in $\eR^N$ is such that $D=\dim_B(A,\Omega)$ exists, and $\mathcal{M}_*^D(A,\Omega)=0$ or $+\ty$ (or $\mathcal{M}^{*D}(A,\Omega)=0$ or $+\ty$), it makes sense to define as follows a new class of relative lower and upper Minkowski contents of $(A,\Omega)$, associated with a suitably chosen gauge function $h(t)$:
\begin{equation}\label{gauge_content}
\begin{aligned}
\mathcal{M}_*^D(A,\Omega,h)&=\liminf_{t\to0^+}\frac{|A_t\cap \Omega|}{t^{N-D}h(t)},\\
\mathcal{M}^{*D}(A,\Omega,h)&=\limsup_{t\to0^+}\frac{|A_t\cap \Omega|}{t^{N-D}h(t)}.
\end{aligned}
\end{equation}
The aim is to find an {\em explicit} gauge function so that these two contents are in $(0,+\ty)$, and the functions $r\mapsto\mathcal{M}_*^r(A,\Omega,h)$
and $r\mapsto \mathcal{M}^{*r}(A,\Omega,h)$, $r\in\eR$, defined exactly as in (\ref{gauge_content}), except for $D$ replaced with $r$, have a jump from $+\ty$ to $0$ when $r$ 
crosses the value of $D$.
In this generality, the above contents are called {\em gauge relative 
Minkowski contents}\label{gauge_relative}\index{Minkowski content of an RFD $(A,\O)$!relative|textbf}
(with respect to $h$).

If for some gauge function $h$, say, we have that $\mathcal{M}^D(A,\Omega,h)\in(0,+\ty)$ (which means, as usual, that $\mathcal{M}_*^D(A,\Omega,h)=\mathcal{M}^{*D}(A,\Omega,h)$ and that this common value, denoted by $\mathcal{M}^D(A,\Omega,h)$, lies in $(0,+\ty)$), we say (as in \cite{lapidushe}) that the fractal drum $(A,\Omega)$ is {\em $h$-Minkowski 
measurable}.\index{h-Minkowski@$h$-Minkowski measurable RFD $(A,\O)$}

In what follows, we denote the Laurent 
expansion\index{Laurent expansion} of a meromorphic extension (assumed to exist) 
of the relative tube zeta function $\tilde\zeta_{A,\Omega}$ to an open, connected neighborhood of $s=D$ (more specifically, an open punctured disk centered at $s=D$) by
\begin{equation}\label{laurent}
\tilde\zeta_{A,\O}(s)=\sum_{j=-\ty}^\ty c_j(s-D)^j,
\end{equation}
where, of course, $c_j=0$ for all $j\ll 0$ (that is, there exists $j_0\in\Ze$ such that $c_j=0$ for all $j<j_0$).

Let us first introduce some notation.
Given a $T$-periodic function $G:\eR\to\eR$, we denote by $G_0$ its truncation to $[0,T]$; that is
\begin{equation}
G_0(\tau)=
\begin{cases}
G(t)& \mbox{if $\tau\in[0,T]$}\\
0& \mbox{if $\tau\notin[0,T]$},
\end{cases}
\end{equation}
 while the Fourier transform\index{Fourier transform|textbf} of $G_0$
is denoted by $\hat G_0$:
\begin{equation}\label{fourier}
\hat G_0(t)=\int_{-\ty}^{+\ty}\E^{-2\pi \I \,t\tau}G_0(\tau)\,\D\tau=\int_0^T\E^{-2\pi \I \,t\tau}G(\tau)\,\D\tau,
\end{equation}
where $\I:=\sqrt{-1}$\label{sqrt-1} is the imaginary unit.

The following theorem shows that, in order to obtain a meromorphic extension of the zeta function to the left of the abscissa of convergence,
it is important to have some information about the second term in the asymptotic expansion of the relative tube function $t\mapsto|A_\delta\cap\Omega|$ near $t=0$. 
We stress that the presence (in Theorem \ref{rel_measurable}) of the gauge function $h(t):=(\log t^{-1})^m$ is closely related to the multiplicity of the principal complex dimension $D$, which is equal to $m+1$.
Theorem \ref{rel_measurable} extends \cite[Theorem 4.24]{mezf} to the general setting of RFDs.

Observe that since the case when $m=0$ is allowed in Theorem \ref{rel_measurable} just below, that theorem enables us to deal, in particular, with the usual class of Minkowski measurable RFDs (for which the gauge function $h$ is trivial, i.e., satisfies $h(t)\equiv 1$).

\begin{theorem}[Minkowski measurable RFDs]\label{rel_measurable}
Let $(A,\Omega)$ be a relative fractal drum in $\eR^N$ such that {\rm({\ref{h}})} holds for some $D\le N$, $\mathcal{M}>0$, $\alpha>0$ and with $h(t):=(\log t^{-1})^m$ for all $t\in(0,1)$, where $m$ is a nonnegative
integer. Then the relative fractal drum $(A,\Omega)$ is $h$-Minkowski measurable, $\dim_B(A,\Omega)=D$,
and $\mathcal{M}^D(A,\Omega,h)=\mathcal{M}$. Furthermore, the relative tube zeta function $\tilde\zeta_{A,\O}$  has for abscissa of convergence
$D(\tilde\zeta_{A,\Omega})=D$, and it possesses a $($necessarily unique$)$ meromorphic extension $($at least$)$ to the open right half-plane $\{\re s>D-\alpha\}$; that is, the abscissa of meromorphic continuation
$D_{\rm mer}(\tilde\zeta_{A,\Omega})$ of $\tilde\zeta_{A,\Omega}$ can be estimated as follows$:$
\begin{equation}
D_{\rm mer}(\tilde\zeta_{A,\Omega})\le D-\a.
\end{equation}
Moreover, $s=D$ is the unique pole in this half-plane, and it is of order $m+1$.
In addition, the coefficients of the Laurent series expansion\index{Laurent expansion} \eqref{laurent} corresponding to the principal part of $\tilde\zeta_{A,\O}$ at $s=D$ are given by 
\begin{equation}\label{m!}
\begin{aligned}
c_{-m-1}&=m!\mathcal{M},\\
c_{-m}&=\dots=c_{-1}=0\q\mbox{$($\rm provided $m\ge1$.$)$}
\end{aligned}
\end{equation}
If $m=0$, then $D$ is a simple pole of $\tilde\zeta_{A,\Omega}$ and we have that
\begin{equation}\label{resM}
\res(\tilde\zeta_{A,\O},D)=\mathcal{M}.
\end{equation}
\end{theorem}

\begin{proof}
Let us set
\begin{equation}\label{zeta12}
\begin{aligned}
\zeta_1(s)&=\mathcal{M} z_m(s),\quad z_m(s)=\int_0^\delta t^{s-D-1}(\log t^{-1})^m \D t,\\ 
\zeta_2(s)&=\int_0^\delta t^{s-N-1}(\log t^{-1})^m \, O(t^{N-D+\alpha})\,\D t.
\end{aligned}
\end{equation}
Since $\tilde\zeta_{A,\O}(s)=\zeta_1(s)+\zeta_2(s)$, we can proceed as follows.
It is easy to see that for each $\e>0$, we have $(\log t^{-1})^m=O(t^{-\e})$ as $t\to0^+$;
hence, 
$$
|\zeta_2(s)|\le\int_0^\delta O(t^{\re s-1-D+(\alpha-\e)})\,\D t.
$$
 Then, 
since the integral is well defined for all $s\in\Ce$ with $\re s>D-(\a-\e)$,
we deduce that $D(\zeta_2)\le D-(\alpha-\e)$. Letting $\e\to0^+$, we obtain
the desired inequality: $D(\zeta_2)\le D-\alpha$.

By means of the change of variable $\tau=\log t^{-1}$ (for $0<t\leq\delta$), it is easy to see that
\begin{equation}
z_m(s)=\int_{\log\delta^{-1}}^\ty {\E}^{-\tau(s-D)}\tau^m\D\tau.
\end{equation}
Integration by parts yields the following recursion relation, where we have to assume (at first) that $\re s>D$:
\begin{equation}\label{z_m}
z_m(s)=\frac1{s-D}\left((\log \delta^{-1})^m\delta^{s-D}+m z_{m-1}(s)\right),\quad \textrm{for}\ m\ge1,
\end{equation}
and $z_0(s)=(s-D)^{-1}\delta^{s-D}$.
Since $D(\zeta_2)\le D-\alpha$, it is clear that the coefficients $c_j$, $j<0$, of the Laurent\index{Laurent expansion} series expansion~\eqref{laurent} of $\tilde\zeta_{A,\O}(s)=\zeta_1(s)+\zeta_2(s)$ in a connected open neighborhood of $s=D$ do not depend on $\delta>0$. Indeed, changing the value of $\delta>0$ to $\delta_1>0$ in (\ref{rel_tube_zeta}) is equivalent to adding $\int_{\delta}^{\delta_1}t^{s-N-1}|A_t\cap\Omega|\,\D t$, which is an entire function of~$s$. Therefore,
without loss of generality, we may take $\delta=1$ in (\ref{z_m}):
\begin{equation}
z_m(s)=\frac{m}{s-D} z_{m-1}(s)=\dots=\frac{m!}{(s-D)^m}z_0(s)=\frac{m!}{(s-D)^{m+1}}.
\end{equation}
In this way, we obtain that
\begin{equation}
\zeta_1(s)=\frac{m!}{(s-D)^{m+1}}\mathcal{M},
\end{equation}
and we can meromorphically continue $\zeta_1$ from the half-plane $\{\re s>D\}$ to the entire complex plane.
The claim then follows from the equality $\tilde\zeta_{A,\O}(s)=\zeta_1(s)+\zeta_2(s)$.
\end{proof}

A large class of examples of RFDs satisfying condition \eqref{h}, involving power logarithmic gauge functions, can be found in Example \ref{Lmloga} of \S\ref{multiplicity}
below, based on \cite[Theorem 5.4]{mm}.
(In fact, \cite[Theorem 5.4]{mm} can be understood as a partial converse of the above result, Theorem \ref{rel_measurable}.) 
These RFDs are constructed by using consecutive tensor products of a suitable bounded fractal string $\mathcal L$, i.e., by an iterated spraying of $\mathcal L$;
see \cite{cras2} for details. A nontrivial class of examples is already obtained when
$\mathcal L$ is the ternary Cantor string. 
A similar comment can be made about the analogous condition \eqref{rel_G} appearing in the folowing theorem. 

\begin{theorem}[Minkowski nonmeasurable RFDs]\label{rel_nonmeasurable}
Let $(A,\Omega)$ be a relative fractal drum in $\eR^N$ such that there exist $D\le N$, a nonconstant periodic function $G:\eR\to\eR$ with minimal period $T>0$,
and a nonnegative integer $m$, satisfying
\begin{equation}\label{rel_G}
|A_t\cap\Omega|=t^{N-D}(\log t^{-1})^m\left(G(\log t^{-1})+O(t^\alpha)\right)\quad\mbox{\rm as\qs$t\to0^+$.}
\end{equation}
 Then $\dim_B(A,\Omega)$ exists and $\dim_B(A,\Omega)=D$, $G$ is continuous, and 
$$
\mathcal{M}_*^D(A,\Omega,h)=\min G,\quad \mathcal{M}^{*D}(A,\Omega,h)=\max G,
$$
where $h(t):=(\log t^{-1})^m$ for all $t\in(0,1)$.
Furthermore, the tube zeta function $\tilde\zeta_{A,\O}$ has for abscissa of convergence $D(\tilde\zeta_{A,\O})=D$, and it possesses a $($necessarily unique$)$
meromorphic extension
{\rm({\it at least})} to the open right half-plane $\{\re s>D-\alpha\}$; that is,
\begin{equation}
D_{\rm mer}(\tilde\zeta_{A,\O})\le D-\a.
\end{equation}
Moreover, all of its poles 
located in this half-plane are of order $m+1$, and the set of poles $\mathcal{P}(\tilde \zeta_{A,\O})$ is contained in the vertical line $\{\re s=D\}$.
More precisely,
\begin{equation}\label{rel_Dpoles}
\begin{aligned}
\mathcal{P}(\tilde \zeta_{A,\O}))&=\mathcal{P}_c(\tilde\zeta_{A,\O})\\
&=\left\{s_k=D+\frac{2\pi}Tk{\I}\in\Ce:\hat G_0\left(\frac kT\right)\ne0,\,\,k\in\Ze\right\},
\end{aligned}
\end{equation}
where $s_0=D\in \mathcal{P}(\tilde \zeta_{A,\O})$ and $\hat G_0$ is the Fourier transform\index{Fourier transform} of $G_0$ $($as given by $(\ref{fourier})$$)$. 
The nonreal poles come in complex conjugate pairs; that is, for each $k\geq 1$, if $s_k$ is a pole, then $s_{-k}$ is a pole as well.

In addition, for any given $k\in\Ze$, if $\tilde\zeta_{A,\O}(s)=\sum_{j=-\ty}^{\ty}c_j^{(k)}(s-s_k)^j$ is the Laurent 
expansion\index{Laurent expansion} of the tube zeta function in a connected open neighborhood of $s=s_k$, then
\begin{equation}
\begin{aligned}
c_j^{(k)}&=0\quad\mbox{\rm for\qs$j<0$,\q$j\ne -m-1$,}\\
c_{-m-1}^{(k)}&=\frac{m!}T\hat G_0\left(\frac kT\right),
\end{aligned}
\end{equation}
where $G_0$ is the restriction of $G$ to the interval $[0,T]$,
and $\hat G_0$ is given by \eqref{fourier}, as above. Also,
\begin{equation}
|c_{-m-1}^{(k)}|\le\frac{m!}T\int_0^T G(\tau)\,\D\tau,\quad\lim_{k\to\ty}c_{-m-1}^{(k)}=0.
\end{equation}
In particular, for $k=0$, that is, for $s_0=D$, we have
\begin{equation}
\begin{aligned}
c_{-m-1}^{(0)}&=\frac{m!}T\int_0^T G(\tau)\,\D\tau,\\
m!\,\mathcal{M}_*^D(A,\Omega,h)&<c_{-m-1}^{(0)}<m!\,\mathcal{M}^{*D}(A,\Omega,h).
\end{aligned}
\end{equation}
If $m=0$ $($i.e., $h(t)=1$ for all $t\in(0,1)$$)$, then $D$ is a simple pole of $\tilde\zeta_{A,\Omega}$ and we have that
\begin{equation}\label{av}
\res(\tilde\zeta_{A,\O},D)=\frac1T\int_0^T G(\tau)\,\D\tau=\widetilde{\mathcal M}
\end{equation}
and
\begin{equation}
\mathcal{M}_*^D(A,\Omega)<\res(\tilde\zeta_{A,\O},D) <\mathcal{M}^{*D}(A,\Omega),
\end{equation}
where $\widetilde{\mathcal M}=\widetilde{\mathcal M}^{D}(A,\O)$ denotes the average Minkowski content of $(A,\O)$.
$($See Remark \ref{average} below.$)$
\end{theorem}

\begin{proof}
For $m\in\eN_0$, let us define $z_m$ by
\begin{equation}
z_m(s)=\int_0^\delta t^{s-D-1}(\log t^{-1})^m G(\log t^{-1})\,\D t.\nonumber
\end{equation}
The function $z_0(s)$ is the exact counterpart of $\zeta_1(s)$ from the proof of \cite[Theorem 4.24]{mezf},
with $|A_t|$ changed to $|A_t\cap\Omega|$ and where, much as in that proof, $\tilde\zeta_{A,\O}=\zeta_1+\zeta_2$ and $\zeta_2$ is an entire function.
It is easy to see that $z_m(s)=(-1)^m z_0^{(m)}(s)$; therefore, the functions $z_m(s)$ and $z_0(s)$ have the same meromorphic extension, and the same sets of poles.
This proves that $\tilde\zeta_{A,\O}(s)$ can be meromorphically extended from $\{\re s>D\}$ to the half-plane $\{\re s>D-\alpha\}$. The set of poles (complex dimensions) of the relative
zeta function, belonging to this half-plane, is given by
\begin{equation}\nonumber
\begin{aligned}
\mathcal{P}(\tilde \zeta_{A,\O})&=\mathcal{P}(z_m)=\mathcal{P}(z_0)\\
&=\left\{s_k=D+\frac{2\pi}Tk{\I}\in\Ce:\hat G_0\left(\frac kT\right)\ne0,\,\,k\in\Ze\right\}.
\end{aligned}
\end{equation}
Each of these poles is simple. Furthermore, if
\begin{equation}
z_0(s)=\sum_{j=-1}^\ty a_j^{(k)}(s-s_k)^j\nonumber
\end{equation}
is the Laurent\index{Laurent expansion} series of $z_0(s)$ in a neighborhood of $s=s_k$, then
\begin{equation}
z_0^{(m)}(s)= (-1)^m m!\,a_{-1}^{(k)}(s-s_k)^{-m-1}+\sum_{j=0}^\ty \frac{(m+j)!}{j!}a_{m+j}^{(k)}(s-s_k)^j.\nonumber
\end{equation}
Hence,
\begin{equation}
c_{-m-1}^{(k)}=m!a_{-1}^{(k)}=m!\,\frac1T\hat G_0\left(\frac kT\right),\nonumber
\end{equation}
where, in the last equality, we have used \cite[Eq.\ (4.32)]{mezf}.
The remaining claims are proved much as the corresponding ones
made in \cite[Theorem 4.24]{mezf}.
\end{proof}

\begin{remark}\label{average}
In Equation \eqref{av},  $\widetilde{\mathcal M}=\widetilde{\mathcal M}^{D}(A,\O)$, the {\em average Minkowski content}\index{average Minkowski content of an RFD $(A,\O)$|textbf}\index{Minkowski content of an RFD $(A,\O)$!average|textbf} of $(A,\O)$, is defined as the multiplicative Ces\`aro average of 
$t^{-(N-D)}|A_t\cap\O|$:
\begin{equation}
\widetilde{\mathcal M}^{D}(A,\O):=\lim_{\tau\to+\ty}\frac1{\log\tau}\int_{1/\tau}^1 \frac{|A_t\cap\O|}{t^{N-D}}\,\frac{\D t}t,
\end{equation}
provided the limit exists in $[0,+\ty]$.
(See Equation \eqref{h} and compare with \cite[Definition 8.29, Eq.\ (8.55)]{lapidusfrank12}.)
\end{remark}

\begin{remark}\label{meas12}
In light of the functional equation \eqref{rel_equality} connecting $\zeta_{A,\O}$ and $\tilde\zeta_{A,\O}$, Theorems \ref{rel_measurable} and \ref{rel_nonmeasurable} also hold for relative {\em distance} zeta functions (instead of relative tube zeta functions), provided $D<N$, and in that case, all of the expressions for the residues and the Laurent coefficients must be multiplied by $N-D$
\end{remark}

\section{Construction of $\ty$-quasiperiodic relative fractal drums}\label{ty_qp_rfd}

Our construction of quasi\-pe\-ri\-o\-dic RFDs (see Definition \ref{quasi_periodic_t} below) given in this section is based on a certain two-parameter family of generalized Cantor sets, which we now describe.

\begin{defn}\label{Cma}
The generalized Cantor sets $C^{(m,a)}$ are determined by an integer $m\ge2$ and a positive real number $a$ such that $ma<1$.
In the first step of the analog of Cantor's construction of the standard ternary Cantor set, we start with $m$ equidistant, closed intervals in $[0,1]$ of length $a$, with $m-1$ `holes', each of length $(1-ma)/(m-1)$. In the second step, we continue by scaling by the factor $a$ each of the $m$ intervals of length $a$; and so on, ad infinitum.
The  $($two-parameter$)$ {\em generalized Cantor set}\index{generalized Cantor set, $C^{(m,a)}$|textbf} $C^{(m,a)}$ is defined as the intersection of the decreasing sequence of compact sets constructed in this way.
It is easy to check that $C^{(m,a)}$ is a perfect, uncountable compact subset of $\eR$.
(Recall that a {\em perfect set}\index{perfect set|textbf} is a closed set without any isolated points.) Furthermore, $C^{(m,a)}$ is also self-similar.

In order to avoid any possible confusion, we note that the generalized Cantor sets introduced here are different from the generalized Cantor strings introduced and studied in \cite[Chapter 10]{lapidusfrank12}.
With our present notation, the classic ternary Cantor set is obtained as $C^{(2,1/3)}$.
\end{defn}

We note that the box dimension of $C^{(m,a)}$ exists and is equal to its Hausdorff dimension, as well as its similarity dimension (here, $\log_{1/a}m$).
 The proof of this fact in the case of the classic Cantor set can be found in \cite{falc} and is due to Moran [Mora] (in the present case when $N=1$); see also \cite{hutchinson}.
For any pair $(m,a)$ as above, this follows from a general result in \cite{hutchinson} (described in [{Fal1}, Theorem 9.3]) because $C^{(m,a)}$ is a self-similar set satisfying the open set condition. (See also [{Mora}].)

It can be shown that the generalized Cantor sets $C^{(m,a)}$ have the following properties. 
Apart from the proof of \eqref{zetaCma}, the proof of the next proposition is similar to that for the standard Cantor set (see \cite[Eq.\ (1.11)]{lapidusfrank12}).

\begin{prop}\label{Cmap}
 If $C^{(m,a)}\subset\eR$ is the generalized Cantor set introduced in Definition~\ref{Cma}, 
where $m$ is an integer, $m\ge2$, and $a\in(0,1/m)$,
then
\begin{equation}\label{2.1.1}
D:=\dim_B C^{(m,a)}=D(\zeta_A)=\log_{1/a}m.
\end{equation}
Furthermore, the tube formula associated with $C^{(m,a)}$ is given by
\begin{equation}\label{Cmat}
|C^{(m,a)}_t|=t^{1-D}G(\log t^{-1})
\end{equation}
for all $t\in(0,\frac{1-ma}{2(m-1)})$, where $G=G(\tau)$ is a nonconstant periodic function, with minimal period $T=\log (1/a)$,  and is  defined by 
\begin{equation}\label{Gtau}
G(\tau)=c^{D-1}(ma)^{g\left(\frac{\tau-c}{T}\right)}+2\,c^Dm^{g\left(\frac{\tau-c}{T}\right)}.
\end{equation}
Here, $c=\frac{1-ma}{2(m-1)}$, and $g:\eR\to[0,+\ty)$ is the $1$-periodic function defined by $g(x)=1-x$ for $x\in(0,1]$.

Moreover,
\begin{equation}
\begin{aligned}
\mathcal{M}_*^D(C^{(m,a)})&=\min G=\frac1D\left(\frac{2D}{1-D}\right)^{1-D},\\
\mathcal{M}^{*D}(C^{(m,a)})&=\max G=\left(\frac{1-ma}{2(m-1)}\right)^{D-1}\frac{m(1-a)}{m-1}.
\end{aligned}
\end{equation}

Finally, if we assume that $\delta\ge\frac{1-ma}{2(m-1)}$, then the distance zeta function of $A:=C^{(m,a)}$ is given by
\begin{equation}\label{zetaCma}
\zeta_A(s):=\int_{-\delta}^{1+\delta}d(x,A)^{s-2}\D x=\left(\frac{1-ma}{2(m-1)}\right)^{s-1}\frac{1-ma}{s(1-ma^s)}+\frac{2\delta^s}s.
\end{equation}
As a result, $\zeta_A(s)$ admits a meromorphic continuation to all of $\Ce$, given by the last expression in $(\ref{zetaCma})$. In particular, the set of poles of $\zeta_A(s)$ $($in $\Ce)$ and the residue of $\zeta_A(s)$ at $s=D$ are respectively given by
\begin{equation}\label{2.1.6}
\begin{aligned}
\po(\zeta_A)&=(D+\mathbf p{\I}\Ze)\cup\{0\}\q\mathrm{and}\\ 
\res(\zeta_A,D)&=\frac{1-ma}{DT}\left(\frac{1-ma}{2(m-1)}\right)^{D-1},
\end{aligned}
\end{equation}
where $\mathbf p=2\pi/T$  is the oscillatory period\label{operiod}\index{oscillatory period!of the generalized Cantor set $C^{(m,a)}$} $($in the sense of {\rm\cite{lapidusfrank12}}$)$.
Furthermore, 
$$
D=\frac{\log m}{2\pi}\,\mathbf{p},
$$
and both $\mathbf p\to0^+$ and $D\to0^+$ as $a\to0^+$. In particular, $\po(\zeta_A)$ converges to the imaginary axis in the Hausdorff metric,\index{Hausdorff metric} as $a\to0^+$.
 Finally, each pole in $\po(\zeta_A)$ is simple.
\end{prop}

In the sequel, we shall need the following important theorem from transcendental number theory, due to Baker \cite[Theorem~2.1]{baker}.

\begin{theorem}[{\rm Baker, \cite[Theorem~2.1]{baker}}]\index{Baker's theorem|textbf}
\label{baker0}
Let $n\in\eN$ with $n\geq 2$.
If $m_1,\dots, m_n$ are positive algebraic numbers such that $\log m_1,\dots,\log m_n$ are linearly independent over the rationals,
then 
$$
1,\log m_1,\dots,\log m_n
$$ 
are linearly independent over the field of all algebraic numbers $($or algebraically independent, in short$)$. In particular,
the numbers $\log m_1,\dots,\log m_n$ are transcendental, as well as their pairwise quotients.
\end{theorem}
\medskip

Here, we describe a general construction of quasiperiodic fractal drums 
possessing infinitely many algebraically incommensurable periods.
It is based on properties of generalized Cantor sets, as well as on Baker's theorem (Theorem \ref{baker0} just above).

Let $m\ge2$ be a given integer and $D\in(0,1)$ a given real number. Then, for $a>0$ defined by $a=m^{-1/D}$,
we have $am=m^{1-1/D}<1$, and hence, the generalized Cantor set $A=C^{(m,a)}$ is well defined and $\dim_BA=\log_{1/a}m=D$.

\begin{defn} 
A finite set of real numbers is said to be {\em rationally $($resp., algebra\-ically$)$ linearly 
independent}\index{linearly independent set in $\eR$!rationally|textbf}\index{linearly independent set in $\eR$!algebraically|textbf} or simply, {\em rationally $($resp., algebra\-ically$)$ independent}, if
it is linearly independent over the field of rational (resp., algebraic) real numbers.
\end{defn}

\begin{defn} A sequence $(T_i)_{i\ge1}$ of real numbers 
is said to be {\em rationally $($resp., algebraically$)$ linearly 
independent} if any of its finite subsets
is rationally (resp., algebraically) independent.
We then say that $(T_i)_{i\ge1}$ is {\em rationally $($resp., algebraically$)$  
independent}, for short.
\end{defn}

\begin{defn} Let $m\ge2$ be a positive integer. Let $\mathbf{p}=(p_i)_{i\ge1}$ be the sequence of all prime numbers, arranged in increasing order; that is,
$$
\mathbf{p}=(2,3,5,7,11,\dots).
$$ 
We then define the {\em exponent sequence\label{es}\index{exponent sequence of $m\in\eN$, $\mathbf{e}(m)$|textbf} $\mathbf e=\mathbf{e}(m):=(\a_i)_{i\ge1}$ associated with $m$},
where $\a_i\ge0$ is the multiplicity of $p_i$ in the factorization of $m$. We also let
\begin{equation}
\mathbf{p}^{\,\mathbf{e}}:=\prod_{\{i\ge1:\a_i>0\}} p_i^{\a_i}.
\end{equation}
The set of all sequences $\mathbf{e}$ with components in $\eN_0=\eN\cup\{0\}$, such that all but at most finitely many components are equal to zero, is denoted by $(\eN_0)_c^{\ty}$.
\end{defn}

With this definition, for any integer $m\ge2$, we obviously have $m=\mathbf{p}^{\,\mathbf{e}(m)}$.
Conversely, any $\mathbf{e}\in (\eN_0)_c^\ty$ defines a unique integer $m\ge2$ such that $m=\mathbf{p}^{\,\mathbf{e}}$.

\begin{defn}\label{supp}
Given an exponent vector $\mathbf{e}=(\a_i)_{i\ge1}\in(\eN_0)_c^\ty$, we define the 
{\em support of $\mathbf{e}$}\index{support!of the exponent sequence|textbf} as the set of all indices $i\in\eN$ for which $\a_i>0$, and we write
\begin{equation}
S(\mathbf{e})=\supp(\mathbf{e}):=\{i\ge1:\a_i>0\}.
\end{equation}
The {\em support of an integer}\index{support!of an integer|textbf} $m\ge2$ is denoted by $\supp m$ and defined by $\supp m:=\supp\mathbf{e}(m)$.
\end{defn}

The following definition will be useful in the sequel.

\begin{defn}
We say that a set $\{\mathbf{e}_i:i\ge1\}$ of exponent vectors
is {\em rationally linearly independent}\index{rational linear independence|textbf} if any of its finite subsets
is linearly independent over $\Qu$. We then say for short that the exponent vectors are {\em rationally independent}.
\end{defn}
\medskip

The following two definitions, Definition \ref{quasiperiodic_fct} and Definition \ref{transcendental_drum}, refine and extend the definition of $n$-quasiperiodic functions and sets introduced in \cite{dtzf}.

\begin{defn}\label{quasiperiodic_fct}
We say that a function $G:\eR\to\eR$ is {\em $\ty$-quasiperiodic}\index{quasiperiodic!function|textbf}
if it is of the form
$$
G(\tau)=H(\tau,\tau,\dots),
$$
where $H:\ell^\ty(\eR)\to\eR$,\footnote{Here, $\ell^\ty(\eR)$ stands for the usual Banach space\label{banach_bdd} of bounded sequences $(\tau_j)_{j\ge1}$ of real numbers, endowed with
the norm $\|(\tau_j)_{j\ge1}\|_\ty:=\sup_{j\ge1}|\tau_j|$.}
 $H=H(\tau_1,\tau_2,\dots)$ is a function which is $T_j$-periodic in its $j$-th component, for each $j\in\eN$,
with $T_j>0$ as minimal periods, and such that the set of periods 
\begin{equation}\label{periods}
\{T_j:j\ge1\}
\end{equation}
is {\em rationally} independent.
We say that the {\em order of quasiperiodicity}\index{order of quasiperiodicity} of the function $G$ is equal to infinity
(or that the function $G$ is {\em $\ty$-quasiperiodic}).\index{quasiperiodic@$\ty$-quasiperiodic function|textbf}

In addition, we say that $G$ is 
\medskip

($a$) {\em transcendentally quasiperiodic of infinite order}\index{transcendentally $\ty$-quasiperiodic function|textbf} (or {\em transcendentally $\ty$-quasi\-pe\-riodic})
  if the periods in (\ref{periods}) are {\em algebraically} independent;
\smallskip

($b$) {\em algebraically quasiperiodic of infinite order} (or {\em algebraically $\ty$-quasiperi\-o\-dic})\index{algebraically $\ty$-quasiperiodic function}
of infinite order if the periods in (\ref{periods}) are {\em rationally} independent and algebraically dependent.
\end{defn}

We say that a sequence $(T_i)_{i\ge1}$ of real numbers is 
{\em algebraically dependent}\index{algebraically dependent real numbers|textbf} of infinite order if there
exists a finite subset $J$ of $\eN$ such that $(T_i)_{i\in J}$ is algebraically dependent (that is, linearly dependent over the field of algebraic numbers). Recall that a finite set of
real numbers $\{T_1,\dots,T_k\}$ is said to be {\em algebraically dependent} if there exist $k$ algebraic real numbers $\g_1,\dots,\g_k$, not all of which are equal to zero, such that $\g_1 T_1+\dots+\g_k T_k=0$.

The notion of quasiperiodic function provided in Definition~\ref{quasiperiodic_fct} above has been
motivated by [Vin]. However, while in [Vin], it is assumed that the reciprocals of the
quasiperiods $T_1,\dots,T_n$ are rationally independent, we assume in Definition~\ref{quasiperiodic_fct} that
the quasiperiods $T_1,\dots,T_n$ themselves are rationally independent. The distinction between
algebraically $n$-quasiperiodic and transcendentally $n$-quasiperiodic functions
seems to be new.

\begin{defn}\label{quasi_periodic_t}\label{transcendental_drum}
Let $(A,\O)$ be a relative fractal drum in $\eR^N$ having the following tube formula:
\begin{equation}\label{Gqp}
|A_t\cap\O|=t^{N-D}(G(\log t^{-1})+o(1))\q\mbox{as\qs$t\to0^+$},
\end{equation}
where $D\in(-\ty,N]$,
and $G$ is a nonnegative function such that 
$$
0<\liminf_{\tau\to+\ty} G(\tau)\le\limsup_{\tau\to+\ty} G(\tau)<\ty. 
$$
(Note that it then follows that $\dim_B(A,\O)$
exists and is equal to $D$. Moreover, 
$
\M_*^D(A,\O)=\liminf_{\tau\to+\ty} G(\tau)$ and $\M^{*D}(A,\O)=\limsup_{\tau\to+\ty} G(\tau).%
$)

We then say that the {\em relative fractal drum $(A,\O)$ in $\eR^N$ is 
 quasiperiodic}\label{transcendental_drumdef}\index{relative fractal drum, RFD!quasiperiodic@$\ty$-quasiperiodic}\index{quasiperiodic!RFD $(A,\O)$|textbf}  
and of {\em infinite order of quasiperiodicity} (or, in short, {\em $\ty$-quasiperiodic}) if the function $G=G(\tau)$ is $\ty$-quasiperiodic; see Definition~\ref{quasiperiodic_fct}. 

In addition, $(A,\O)$ is said to be
\medskip

($a$) a {\em transcendentally $\ty$-quasiperiodic relative fractal drum}\index{relative fractal drum, RFD!transcendentally $\ty$-quasiperiodic}\index{transcendentally $\ty$-quasiperiodic RFD|textbf}
if the corresponding function $G$ is transcendentally $\ty$-quasiperiodic;
\smallskip 

($b$) an {\em algebraically $\ty$-quasiperiodic relative fractal drum}\index{relative fractal drum, RFD!algebraically $\ty$-quasiperiodic}
if the corresponding function $G$ is algebraically $\ty$-quasiperiodic.
\end{defn}

\medskip

The following definition is closely related to the 
the notion of fractality (given in \cite{lapidusfrank12},  \S12.1.1 and \S12.1.2, including Figures 12.1--12.3, along with \S13.4.3).

\begin{defn}\label{analytically}\label{hyperfractal} 
Let $A$ be a bounded subset of $\eR^N$ and let $D:=\ov\dim_BA$. Then$:$
\medskip

$(i)$ The set $A$ is a {\em hyperfractal}\index{hyperfractal|textbf} (or is {\em hyperfractal}) if there is a screen $\bm S$ along which the associated tube (or equivalently, if $D<N$, distance) zeta function is a 
natural boundary. This means that the zeta function cannot be meromorphically continued to an open neighborhood of $\bm S$ (or, equivalently, of the associated window~$\bm W$).
\medskip

$(ii)$ The set $A$ is a {\em strong hyperfractal}\index{hyperfractal!strong|textbf}\index{strong hyperfractal|textbf} (or is {\em strongly hyperfractal}) if the 
critical line $\{\re s =D\}$ is a  
(meromorphic) natural boundary
of the associated zeta function; that is, if we can choose $\bm{S} = \{\re s = D\}$ in~$(i)$.
\medskip

$(iii)$ Finally, the set $A$ is {\em maximally hyperfractal}\index{maximal hyperfractal|textbf}\index{hyperfractal!maximal|textbf} if it is strongly  hyperfractal and every point of the critical line $\{\re s=D\}$ is a nonremovable singularity of the zeta function.
\medskip

An analogous definition can be provided (in the obvious manner) where instead of $A$, we have a fractal string $\mathcal{L}=(\ell_j)_{j\ge1}$ in $\eR$ or, more generally, a relative fractal drum $(A,\O)$ in~$\eR^N$.
\end{defn}
\medskip

\begin{remark}
Following \cite{lapidusfrank12}, but now using the higher-dimensional theory of complex dimensions developed here and in [{LapR\v Zu1--8}], we say that a bounded set $A\st\eR^N$ (or, more generally, an RFD $(A,\O)$ in $\eR^N$)
is ``fractal'' if it has at least one nonreal visible complex dimension\footnote{Then, it clearly has at least two complex conjugate nonreal complex dimensions.} (i.e., if the associated fractal zeta function has a nonreal visible pole) or if it is hyperfractal (in the sense of part $(i)$ of Definition \ref{analytically} just above). 
\end{remark}

The following result can be considered as a fractal set-theoretic interpretation of Baker's\label{bakers} theorem \cite[Theorem~2.1]{baker} (i.e., of Theorem~\ref{baker0} above)
from transcendental number theory. It provides a construction of a transcendentally $\ty$-quasiperiodic relative fractal drum. In particular, this drum possesses 
infinitely many algebraically incommensurable quasiperiods $T_i$. In our construction, we use the two-parameter family of generalized 
Cantor sets $C^{(m,a)}$ introduced in Definition~\ref{Cma} and whose basic properties are described in Proposition~\ref{Cmap}.

\begin{theorem}\label{qp}
Let $D\in(0,1)$ be a given real number, and let $(m_i)_{i\ge1}$ be a sequence of integers such that $m_i\ge 2$ for each $i\geq 1$. For any $i\geq 1$, define $a_i=m_i^{-1/D}$ and let
$C^{(m_i,a_i)}$ be the corresponding generalized Cantor set $($see Definition~\ref{Cma}$)$.
Assume that $(\O_i)_{i\ge1}$ is a family of disjoint open intervals on the real line such that $|\O_i|\le C_1 m_i^{1-1/D}c_i^{1/D}$ for each $i\geq 1$, where the sequence $(c_i)_{i\ge1}$ 
of positive real numbers is summable, and $C_1>0$.
Let 
$$
(A,\O):=\bigcup_{i\ge1}(A_i,\O_i),\q\textrm{{\rm where}}\q A_i:=|\O_i|\,C^{(m_i,a_i)}+\inf\O_i,\q\textrm{{\rm for all}}\q i\ge1.
$$ 
 Assume that the sequence of real numbers
\begin{equation}\label{mri}
\{\log m_1,\dots,\log m_n,\dots\}\ \mathrm{is\ rationally\ independent.}
\end{equation}
Then the sequence of real numbers
\begin{equation}
\Big\{\frac1D,T_1,T_2,\dots\Big\}
\end{equation}
is algebraically independent.
In other words, the relative fractal drum $(A,\O)$ is transcendentally quasiperiodic
with infinite order of quasiperiodicity and associated sequence of quasiperiods $(T_i)_{i\ge1}$, where $T_i:=log(1/a_i)=(\log m_i)/D$ for each $i\ge1$.
Furthermore,  
\begin{equation}\label{DzetaDzeta_mer}
D(\zeta_{A,\O})=D_{\rm mer}(\zeta_{A,\O}),
\end{equation}
and moreover, all the points on the critical line $\{\re s=D\}$ are nonremovable singularities of $\zeta_{A,\O}$; in other words, the relative fractal drum $(A,\O)$ is also maximally hyperfractal $($in the sense of Definition \ref{hyperfractal}$(iii)$ above$)$.

Finally, the relative fractal drum $(A,\O)$ is Minkowski nondegenerate, in the sense that
$$
0<\mathcal{M}_*^D(A,\O)\le\mathcal{M}^{*D}(A,\O)<\ty.
$$
\end{theorem}

Theorem \ref{qp} admits a partial extension. If instead of condition (\ref{mri}) we assume that $m_i\to\ty$ as $i\to\ty$,
then (\ref{DzetaDzeta_mer}) still holds, and, moreover, all of the points of the critical line are nonremovable singularities of $\zeta_A$, and hence, the RFD $(A,\O)$ is maximally hyperfractal. Furthermore, the fractal drum $(A,\O)$
is Minkowski nondegenerate.

\medskip

We shall need the following lemma, which states a simple scaling property of the tube functions and Minkowski contents of RFDs.
We note that the identity~\eqref{minkrel_scaling} below yields a partial extension of \cite[Proposition 4.4.]{rae}. Compare also with
the scaling property of the corresponding distance zeta function $\zeta_{A,\O}$, obtained in Theorem~\ref{scaling}.

\begin{lemma}\label{lAO}
$(a)$ Let $(A,\O)$ be a relative fractal drum in $\eR^N$. Then, for any fixed $\g>0$ and for all $t>0$, we have that
\begin{equation}\label{lAON}
(\g A)_t\cap\g \O=\g(A_{t/\g}\cap\,\O),\q |(\g A)_t\cap\,\g\O|=\g^N|A_{t/\g}\cap\,\O|.
\end{equation}
Furthermore, for any real parameter $r\in \eR$, we have the following scaling $($or homogeneity$)$ properties of the relative upper and lower Minkowski contents$:$\index{scaling property!of the relative Minkowski content|textbf}
\begin{equation}\label{minkrel_scaling}
\mathcal{M}^{*r}(\g A,\g \O)=\g^r\mathcal{M}^{*r}( A, \O),\q\mathcal{M}_*^r(\g A,\g \O)=\g^r\mathcal{M}_*^r( A, \O).
\end{equation}

\bigskip

$(b)$ If $A$ is a generalized 
Cantor set $C^{(m,a)}$
$($as in Definition~\ref{Cma}$)$, then
$$
|(\g C^{(m,a)})_t\cap(0,\g)|= t^{1-D}(G_\g(\log t^{-1})-2t^D),
$$
where
$$
G_\g(\tau):=\g^DG(\tau+\log\g)
$$
and $G$ is the $T$-periodic function defined in Equation~\eqref{Gtau} of Proposition~\ref{Cmap}.
\end{lemma}

\begin{proof} We shall establish parts $(a)$ and $(b)$ separately.
\medskip

$(a)$ Scaling the set $A_t\cap \O$ by the factor $\g$, we obtain $\g(A_t\cap \O)$. On the other hand, the same result is then obtained as the intersection of
the scaled sets $(\g A)_{\g t}$ and $\g \O$; that is,
$$
\g(A_t\cap \O)=(\g A)_{\g t}\cap\g \O.
$$
The first equality in (\ref{lAON}) now follows by replacing $t$ with $t/\g$. The second one is an immediate consequence of the first one.
We also have
$$
\begin{aligned}
\mathcal{M}^{*r}(\g A,\g\O)&=\limsup_{t\to0^+}\frac{|(\g A)_t\cap\g\O|}{t^{N-r}}=\g^N\limsup_{t\to0^+}\frac{|( A)_{t/\g}\cap\O|}{t^{N-r}}\\
&=\g^N\limsup_{\tau\to0^+}\frac{|( A)_{\tau}\cap\O|}{(\g \tau)^{N-r}}=\g^r\mathcal{M}^{*r}( A, \O).
\end{aligned}
$$
The second equality in (\ref{minkrel_scaling}) is proved in the same way, but by now using the lower limit instead of the upper limit.

\bigskip

$(b)$ In the case of the generalized Cantor set, we use (\ref{lAON}) with $N=1$ together with Proposition~\ref{Cmap}:
$$
\begin{aligned}
|(\g C^{(m,a)})_t\cap(0,\g)&=\g|C^{(m,a)}_{t/\g}\cap(0,1)|=\g\Big(\frac t\g\Big)^{1-D}\Big(G\Big(\log\frac{\g}{t}\Big)-2(t/\g)^D\Big)\\
&=t^{1-D}\Big(\g^DG(\log \g+\log t^{-1})-2t^D\Big).
\end{aligned}
$$
This completes the proof of the lemma.
\end{proof}

Relative tube zeta functions have a scaling property which is analogous to that obtained for the tube zeta functions of bounded sets; see [LapRa\v Zu1, Proposition 2.2.22]. We omit the corresponding simple direct proof.\footnote{An alternative proof of Proposition~\ref{scalingtr} would rely on the functional equation \eqref{rel_equality} combined with Theorem \ref{scaling}, the scaling property for distance zeta functions.}

\begin{prop}[{\rm Scaling property of relative tube zeta functions}]\label{scalingtr}\index{scaling property!of the relative tube zeta functions|textbf}
Let $(A,\O)$ be a relative fractal drum and let $\d>0$. Let us denote by $\tilde\zeta_{A,\O;\d}(s)$ the associated relative fractal zeta function defined by Equation \eqref{rel_tube_zeta}. Then, for any $\g>0$, we have $D(\tilde\zeta_{\g A,\g\O;\g\d})=D(\tilde\zeta_{A,\O;\d})=\ov\dim_B(A,\O)$ and
\begin{equation}
\tilde\zeta_{\g A,\g\O;\g\d}(s)=\g^s \tilde\zeta_{A,\O;\d}(s),
\end{equation}
for all $s\in\Ce$ such that $\re s>\ov\dim_B(A,\O)$. Furthermore, if $\o\in\Ce$ is a simple pole of $\tilde\zeta_{A,\O;\d}$, where $\tilde\zeta_{A,\O;\d}$ is meromophically extended to an open
connected neighborhood of the critical line $\{\re s=\ov\dim_B(A,\O)\}$ $($as usual, we keep the same notation for the extended function$)$, then
\begin{equation}
\res(\tilde\zeta_{\g A,\g\O;\g\d}\,,\,\o)=\g^\o\res(\tilde\zeta_{A,\O;\d}\,,\,\o).
\end{equation}
\end{prop}

In the proof of Theorem~\ref{qp}, we shall use the following simple fact.
If a function $G(\tau):=H(\tau,\tau,\dots)$ is transcendentally quasiperiodic with respect to a sequence of quasiperiods $(T_i)_{i\ge1}$,
it is clear that for any fixed sequence of real numbers $\mathbf d=(d_i)_{i\ge1}$, the corresponding function
$$
G_{\mathbf d}(\tau):=H(d_1+\tau,d_2+\tau,\dots)
$$ 
is quasiperiodic with respect to the same sequence of quasiperiods $(T_i)_{i\ge1}$.

\begin{proof}[Proof of Theorem \ref{qp}] The proof of the theorem is divided into three steps.
\medskip

{\em Step} 1: First of all, note that the generalized Cantor sets $C^{(m_i,a_i)}$ are well defined, since $m_ia_i=m_i^{1-1/D}<1$. Furthermore,
$$
|\O|=\sum_{i=1}^\ty|\O_i|\le C_1\sum_{i=1}^\ty m_i^{1-1/D}c_i^{1/D}\le C_1\sum_{i=1}^\ty c_i^{1/D}\le C_1\sum_{i=1}^\ty c_i<\ty,
$$
where we have assumed without loss of generality that $c_i\le1$ for all $i\ge1$.
Using Lemma~\ref{lAO}, we have
$$
\begin{aligned}
|A_t\cap\O|&=\sum_{i=1}^\ty|(A_i)_t\cap\O_i|=t^{1-D}\sum_{i=1}^\ty |\O_i|^D\left(G_i\Big(\log|\O_i|+\log\frac1t\Big)-2t^D\right)\\
&=t^{1-D}\Big(G\Big(\log\frac 1t\Big)-2|\O|\,t^D\Big),
\end{aligned}
$$
where
$$
G(\tau):=\sum_{i=1}^\ty |\O_i|^DG_i(\log|\O_i|+\tau),
$$
and the functions $G_i=G_i(\tau)$ are $T_i$-periodic, with $T_i:=\log(1/a_i)$, for all $i\ge1$.
This shows that $G(\tau)=H(\tau,\tau,\dots)$, where 
$$
H((\tau_i)_{i\ge1}):=\sum_{i=1}^\ty |\O_i|^D\,G_i(\log|\O_i|+\tau_i).
$$
Note that the last series is well defined, and that so is the series defining $G(\tau)$.
Indeed, letting $\mathcal{M}_i=\mathcal{M}^{*D}(C^{(m_i,a_i)})$ and using Proposition~\ref{Cmap}, we see that
\begin{equation}\label{Mi}
0<G_i(\tau)\le\mathcal{M}_i=\left(\frac{2(m_i-1)}{1-m_ia_i}\right)^{1-D}\frac{m_i}{m_i-1}(1-a_i)\le C\, m_i^{1-D},
\end{equation}
where $C$ is a positive constant independent of $i$, since $m_i\to\ty$ and $m_ia_i\to0$ as $i\to\ty$. Therefore,
$$
\sum_{i=1}^\ty |\O_i|^DG_i(\tau_i)\le \sum_{i=1}^\ty (C_1^Dm_i^{D-1}c_i)\,( C m_i^{1-D})=CC_1^D\sum_{i=1}^\ty c_i<\ty.
$$
In particular,
$$
\mathcal{M}^{*D}(A,\O)\le CC_1^D\sum_{i=1}^\ty c_i<\ty.
$$

\noindent On the other hand, since $(A_1,\O_1)\supset(A,\O)$, we can use Lemma~\ref{lAO}$(a)$ (with $r:=D$) and Proposition~\ref{Cmap} to obtain that 
$$
\begin{aligned}
\mathcal{M}_*^D(A,\O)&\ge\mathcal{M}_*^D(A_1,\O_1)=|\O_1|^D\mathcal{M}_*^D(`C^{(m_1,a_1)})\\
&=|\O_1|^D\frac1D\left(\frac{2D}{1-D}\right)^{1-D}>0.
\end{aligned}
$$

\bigskip

{\em Step} 2: Let $n$ be any fixed positive integer. Since the set of real numbers
$$
\{\log m_1,\dots,\log m_n\}
$$
is rationally independent, we conclude from Baker's theorem (see Theorem~\ref{baker0} above or \cite[Theorem~2.1]{baker}) that the set of numbers
$\{1,\log m_1,\dots,\log m_n\}$ is algebraically independent. Dividing all of these numbers by $D$, and using $D=(\log m_i)/T_i$,
where $T_i=\log(1/a_i)$ for all $i\ge1$ (see Proposition~\ref{Cmap}), we deduce that
$$
\Big\{\frac1D,\frac{\log m_1}D,\dots,\frac{\log m_n}D\Big\}=\Big\{\frac1D,T_1,\dots,T_n\Big\}
$$
is algebraically independent as well. Since $n$ is arbitrary, this proves that the relative fractal drum $(A,\O)$ is transcendentally
$\ty$-quasiperiodic, in the sense
of Definition~\ref{quasi_periodic_t}.

\bigskip

{\em Step} 3: To prove the last claim, note that the critical line
 $\{\re s=D\}$ contains the union of the set of poles $\mathcal{P}_i:=\mathcal{P}(\tilde\zeta_{A_i,\O_i},\Ce)=D+\mathbf{p}_i{\I}\Ze$ of the tube zeta functions $\tilde\zeta_{A_i,\O_i}$, $i\ge 1$.
Since the integers $m_i$ are all distinct, we have that $m_i\to\ty$ as $i\to\ty$, and therefore, $\mathbf{p}_i=2\pi/T_i=2\pi D/\log m_i\to0$. This proves that the union 
$\cup_{i\ge1}\mathcal{P}_i$, as a set of nonisolated singularities of $\tilde\zeta_{A,\O}=\sum_{i\ge1}\tilde\zeta_{A_i,\O_i}$,
is dense in the critical line $\{\re s=D\}$.
(Indeed, it is easy to deduce from the definitions that the subset of nonremovable singularities of $\zeta_{A,\O}$ along the critical line $L:=\{\re s=D\}$ is closed in $L$ and, hence, must coincide with $L$ since it is also dense in $L$; see the proof of [{LapR\v Zu2}, Theorem 5.3].)
It follows, in particular, that~(\ref{DzetaDzeta_mer}) holds, as desired.
\end{proof}

It is noteworthy that the sequence $\mathcal{M}^{*D}(C^{(m_i,a_i)},(0,1))$ appearing in Theorem~\ref{qp} is divergent. More precisely, it is easy to deduce
from the equality in (\ref{Mi}) that
$$
\mathcal{M}^{*D}(C^{(m_i,a_i)},(0,1))\sim (2m_i)^{1-D}\q\mbox{as\qs$i\to\ty$.}
$$

The conditions of Theorem~\ref{qp} are satisfied if, for example, $m_i:=p_i$ for all $i\ge1$ (that is, $(m_i)_{i\ge1}$ is the sequence of prime numbers $(p_i)_{i\ge1}$, written in increasing order), and if 
$C_1:=1$ and $c_i:=2^{-i}$ for every $i\ge1$. 

\medskip


\chapter{Embeddings into higher-dimensional spaces}\label{sec_embed}

In this chapter, we obtain useful results concerning relative fractal drums and bounded subsets of $\eR^N$ embedded into higher-dimensional spaces.
In particular, we show that the complex dimensions (and their multiplicities) of a bounded set (or, more generally, of a relative fractal drum) are independent of the dimension of the ambient space.
(See Theorem~\ref{c_dim_inv} and Theorem~\ref{c_dim_inv_rel}, respectively.)
In addition, we apply some of these results in order to calculate the complex dimensions of the Cantor dust.
(See Example \ref{c_dust}.)

\section{Embeddings of bounded sets}\label{emb_b}

We begin this section by stating a result which (along with the subsequent result, Theorem~\ref{mero_ext_N+1}) will be key to the developments in this chapter.

\begin{prop}\label{prop_N+1}
Let $A\subset\eR^N$ be a bounded set and let $\ov{D}:=\overline{\dim}_BA$.
Then, for the tube zeta functions of $A$ and $A\times\{0\}\st\eR^{N+1}$, the following equality holds$:$
\begin{equation}\label{eq_N+1}
\tilde{\zeta}_{A\times\{0\}}(s;\delta)=2\int_0^{\pi/2}\frac{\tilde{\zeta}_A(s;\delta\sin \tau)}{\sin^{s-N-1}\tau}\,\di \tau,
\end{equation}
for all $s\in\Ce$ such that $\re s>\overline{D}$.
\end{prop}

\begin{proof}
First of all, it is well known and easy to check that $\ov{\dim}_B(A\times\{0\})=\ov{\dim}_BA$, from which we conclude that the tube zeta functions of $A$ and $A\times\{0\}$ are both holomorphic in the right half-plane $\{\re s>\ov{D}\}$.
Furthermore, we use the fact (see~\cite[Proposition~6]{maja}) that for every $t>0$, we have
\begin{equation}
|(A\times\{0\})_t|_{N+1}=2\int_0^t|A_{\sqrt{t^2-u^2}}|_N\,\di  u,
\end{equation}
where, as before, $|\cdot|_N$ denotes the $N$-dimensional Lebesgue measure.
(See also the proof of Lemma \ref{rel_emb} in \S\ref{emb_rfd} below.)
After having made the change of variable $u:=t\cos v$, this yields
\begin{equation}
|(A\times\{0\})_t|_{N+1}=2t\int_0^{\pi/2}|A_{t\sin v}|_N\sin v\,\di  v.
\end{equation}
Finally, for the tube zeta function of $A\times\{0\}$, we can write successively:
\begin{equation}\nonumber
\begin{aligned}
\tilde{\zeta}_{A\times\{0\}}(s;\delta)&=\int_0^{\delta}t^{s-N-2}|(A\times\{0\})_t|_{N+1}\,\di t\\
&=2\int_0^{\delta}t^{s-N-1}\di  t\int_0^{\pi/2}|A_{t\sin v}|_N\sin v\,\di  v\\
&=2\int_0^{\pi/2}\sin v\,\di  v\int_0^{\delta}t^{s-N-1}|A_{t\sin v}|_N\,\di  t\\
&=2\int_0^{\pi/2}\sin^{N+1-s}v\,\di  v\int_0^{\delta\sin v}\tau^{s-N-1}|A_{\tau}|_N\,\di  \tau\\
&=2\int_0^{\pi/2}\frac{\tilde{\zeta}_A(s;\delta\sin v)}{\sin^{s-N-1}v}\,\di  v,
\end{aligned}
\end{equation}
where we have used the Fubini--Tonelli theorem in order to justify the interchange of integrals (in the third equality), as well as made another change of variable (in the fourth equality), namely, $\tau:=t\sin v$.
This completes the proof of the proposition.
\end{proof}

In the following theorem, $\upGamma(t):=\int_0^{+\ty} x^{t-1}\E^{-x}\,\D x$, initially defined by this integral for $t>0$, is the usual gamma function, meromorphically extended to all of $\Ce$.\label{gammaf}\index{gamma function, $\upGamma(t)$|textbf}

\begin{theorem}\label{mero_ext_N+1}
Let $A\st\eR^N$ be a bounded set and let $\ov{D}:=\ov{\dim}_BA$. 
Then, we have the following equality between $\tilde{\zeta}_A$, the tube zeta function of $A$, and $\tilde{\zeta}_{A_M}$, the tube zeta function  of $A_M:=A\times\{0\}\cdots\times\{0\}\times\st\eR^{N+M}$, with $M\in\eN$ arbitrary$:$
\begin{equation}\label{link_zzm}
\tilde{\zeta}_{A_M}(s;\delta)= \frac{\left(\sqrt{\pi}\right)^M\, \upGamma\left(\frac{N-s}{2}+1\right)}{\upGamma\left(\frac{N+M-s}{2}+1\right)}\tilde{\zeta}_A(s;\delta)+E(s;\delta),
\end{equation}
initially valid for all $s\in\Ce$ such that $\re s>\ov{D}$.
Here, the error function $E(s):=E(s;\d)$ $($initially defined in the case when $M=1$ by the integral on the right-hand side of Equation~\eqref{ERRor} below$)$ admits a meromorphic extension to all of $\Ce$.
The possible poles $($in $\Ce$$)$ of $E(s;\delta)$ are located at $s_k:=N+2+2k$ for every $k\in\eN_0$, and all of them are simple.
$($It follows that $\tilde{\zeta}_A$ is well defined at each $s_k$.$)$
Moreover, we have that for each $k\in\eN_0$, 
\begin{equation}\label{E_res}
\res(E(\,\cdot\,;\delta),s_k)=\frac{(-1)^{k+1}\left(\sqrt{\pi}\right)^{M}}{k!\,\upGamma\left(\frac{M}{2}-k\right)}\tilde{\zeta}_A(s_k;\delta).
\end{equation}
$($We refer to Theorem \ref{c_dim_inv} below for more precise information about the domain of validity of the approximate functional equation \eqref{link_zzm}, and to Corollary \ref{res_emb} for information about the relationship between the visible poles of $\tilde{\zeta}_A$ and $\tilde{\zeta}_{A_M}$.$)$
More specifically, if $M$ is even, then all of the poles $s_k$ of $E(s;\d)$ are canceled for $k\geq M/2$; i.e., the corresponding residues in \eqref{E_res} are equal to zero.
On the other hand, if $M$ is odd, then there are no such cancellations and all of the residues in \eqref{E_res} are nonzero; so that all the $s_k$'s are $($simple$)$ poles of $E(s;\d)$ in that case.

\end{theorem}

\begin{proof}
We will prove the theorem in the case when $M=1$.
The general case then follows immediately by induction.
From Proposition~\ref{prop_N+1}, we have that formula~\eqref{eq_N+1} holds for all $s\in\Ce$ such that $\re s>\overline{\dim}_BA$. 
In turn, this latter identity can be written as follows:
\begin{equation}\label{beta_func}
\begin{aligned}
\tilde{\zeta}_{A\times\{0\}}(s;\delta)&=2\tilde{\zeta}_A(s;\delta)\int_0^{\pi/2}\frac{\di  \tau}{\sin^{s-N-1}\tau}\\
&\phantom{=}-2\int_0^{\pi/2}\frac{\di  v}{\sin^{s-N-1}v}\int_{\delta\sin v}^{\delta}\tau^{s-N-1}|A_{\tau}|_N\,\di  \tau\\
&=\tilde{\zeta}_A(s;\delta)\cdot\mathrm{B}\left(\frac{N-s}{2}+1,\frac{1}{2}\right)+E(s;\delta),
\end{aligned}
\end{equation}
where $\mathrm{B}$ denotes the Euler beta function and
\begin{equation}\label{ERRor}
E(s;\delta):=-2\int_0^{\pi/2}\frac{\di  v}{\sin^{s-N-1}v}\int_{\delta\sin v}^{\delta}\tau^{s-N-1}|A_{\tau}|_N\,\di  \tau\\.
\end{equation}
By using the functional equation which links the beta function\index{beta function, $B(a,b)$} with the gamma function\index{gamma function, $\upGamma(t)$} (namely, $\mathrm{B}(x,y)={\upGamma(x)\upGamma(y)}/{\upGamma(x+y)}$ for all $x,y>0$ and hence, upon meromorphic continuation, for all $x,y\in\Ce$), we obtain that~\eqref{link_zzm} holds (with $M=1$) for all $s\in\Ce$ such that $\re s>\overline{\dim}_BA$.

By looking at $E(s;\delta)$, we see that the integrand is holomorphic for every $v\in(0,\pi/2)$ since the integral $\int_{\delta\sin v}^{\delta}\tau^{s-N-1}|A_{\tau}|_N\di  \tau$ is equal to $\tilde{\zeta}_A(s;\delta)-\tilde{\zeta}_A(s;\delta\sin v)$, which is an entire function.
Furthermore, if we assume that $\re s<N+1$, then since $\tau\mapsto\tau^{\re s-N-1}$ is decreasing, we have the following estimate:
\begin{equation}
\begin{aligned}\label{Esd}
|E(s;\delta)|&\leq 2\int_0^{\pi/2}\sin^{N+1-\re s}v\,\di  v\int_{\delta\sin v}^{\delta}\tau^{\re s-N-1}|A_{\tau}|_N\,\di  \tau\\
&\leq 2|A_{\delta}|_N\int_0^{\pi/2}\sin^{N+1-\re s}v\,\di  v\int_{\delta\sin v}^{\delta}\tau^{\re s-N-1}\,\di  \tau\\
&\leq{2\delta^{\re s-N-1}|A_{\delta}|_N}\int_0^{\pi/2}\sin^{N+1-\re s}v\,\sin^{\re s-N-1}v\int_{\delta\sin v}^{\delta}\,\di  \tau\\
&={2\delta^{\re s-N}|A_{\delta}|_N}\int_0^{\pi/2}(1-\sin v)\,\di  v\\
&={2\delta^{\re s-N}|A_{\delta}|_N}\left(\frac{\pi}{2}-1\right).
\end{aligned}
\end{equation}
From this we conclude that for $s_0\in\{\re s<N+1\}$, the condition $(3')$ of Remark~\ref{condition3} is satisfied, which implies, in light of Theorem~\ref{Hh}, that $E(s;\delta)$ is holomorphic on the open left half-plane $\{\re s<N+1\}$.

On the other hand, we know that both of the tube zeta functions $\tilde{\zeta}_{A}$ and $\tilde{\zeta}_{A_M}$ are holomorphic on $\{\re s>\ov{\dim}_BA\}\supseteq\{\re s> N\}$.
The fact that $E(s;\delta)$ is meromorphic on $\Ce$, as well as the statement about its poles, now follows from Equation~\eqref{link_zzm} (with $M=1$) and the fact that the gamma function is nowhere vanishing in $\Ce$.
(In fact, $1/\upGamma(s)$ is an entire function with zeros at the nonpositive integers.)
More specifically, the locations of the poles of $E(s;\delta)$ must coincide with the locations of the poles $s_k=N+2+2k$, for $k\in\eN_0$, of $\upGamma((N-s)/2+1)$ since the left-hand side of~\eqref{link_zzm} is holomorphic on $\{\re s>\ov{\dim}_BA\}$ and because $\tilde{\zeta}_A(s_k)>0$ (since it is defined as the integral of a positive function).
(Note that since $N\geq \ov{D}$, we have $s_k>\ov{D}$, and hence, $\tilde{\zeta}_A$ is well defined at $s_k$, for each $k\in\eN_0$.)

Finally, by multiplying~\eqref{link_zzm} by $(s-s_k)$, taking the limit as $s\to s_k$ and then using the fact that the residue of the gamma function at $-k$ is equal to $(-1)^k/k!$, we deduce that~\eqref{E_res} holds, as desired.
Furthermore, if $M$ is odd, there are no cancellations between the poles of the numerator and of the denominator in \eqref{link_zzm} since an integer cannot be both even and odd; i.e., the residues are nonzero for each $k\in\eN_0$.
On the other hand, if $M$ is even, then it is clear that all of the residues at $s_k$ for $k\geq M/2$ are equal to zero; i.e., the corresponding poles at $s_k$ cancel out with the poles of the denominator in \eqref{link_zzm}. This concludes the proof of the theorem.
\end{proof}



Theorem~\ref{mero_ext_N+1} has as an important consequence, namely, the fact that the notion of complex dimensions does not depend on the dimension of the ambient space.

\begin{theorem}\label{c_dim_inv}
Let $A\st\eR^N$ be a bounded set and $A_M$ be its embedding into $\eR^{N+M}$, with $M\in\eN$ arbitrary.
Then, the tube zeta function $\tilde{\zeta}_{A}$ of $A$ has a meromorphic extension to a given connected open neighborhood $U$ of the critical line $\{\re s=\ov{\dim}_BA\}$ if and only if the analogous statement is true for the tube zeta function $\tilde{\zeta}_{A_M}$ of $A_M$.
Furthermore, in that case, the approximate functional equation \eqref{link_zzm} remains valid for all $s\in U$. 
In addition, the multisets of the poles of $\tilde{\zeta}_A$ and $\tilde{\zeta}_{A_M}$ located in $U$ coincide; i.e., $\po(\tilde{\zeta_A},U)=\po(\tilde{\zeta}_{A_M},U)$.\footnote{Recall that the bounded sets $A$ and $A_M$ have the same upper Minkowski dimension, $\ov{\dim}_BA=\ov{\dim}_BA_M$, and hence, the same critical line $\{\re s=\ov{\dim}_BA\}$.}
Consequently, neither the values nor the multiplicities of the complex dimensions of $A$ depend on the dimension of the ambient space.
\end{theorem}

\begin{proof}
This is a direct consequence of Theorem~\ref{mero_ext_N+1} and the principle of analytic continuation.
More specifically, identity~\eqref{link_zzm} is valid for all $s\in\Ce$ such that $\re s>\ov{\dim}_BA$ and the function $E(s;\delta)$ is meromorphic on $\Ce$.
Furthermore, according to Theorem~\ref{mero_ext_N+1}, the poles of $E(s;\delta)$ belong to $\{\re s\geq N+2\}$, which implies that the function $s\mapsto E(s;\delta)$ is holomorphic on $\{\re s<N+2\}$.
Identity~\eqref{link_zzm} then remains valid if any of the two zeta functions involved (namely, $\tilde{\zeta}_A$ or $\tilde{\zeta}_{A_M}$) has a meromorphic continuation to some connected open neighborhood of the critical line $\{\re s=\ov{\dim}_BA\}$.
This completes the proof of the theorem. 
\end{proof}

\begin{cor}\label{res_emb}
Let $A\st\eR^N$ be a bounded set $($with $\ov{D}:=\ov{\dim}_BA$$)$ such that its tube zeta function $\tilde{\zeta}_A$ has a meromorphic continuation to a connected open neighborhood $U$ of the critical line $\{\re s=\overline{\dim}_BA\}$.
Furthermore, suppose that $s=\ov{D}$ is a simple pole of $\tilde{\zeta}_A$.
Let $A_M\st\eR^{N+M}$ be the embedding of $A$ into $\eR^{N+M}$, as in Theorem~\ref{mero_ext_N+1}.
Then
\begin{equation}
\res(\tilde{\zeta}_{A_M},\ov{D})=\frac{\left(\sqrt{\pi}\right)^M\, \upGamma\left(\frac{N-\ov{D}}{2}+1\right)}{\upGamma\left(\frac{N+M-\ov{D}}{2}+1\right)}\res(\tilde{\zeta}_{A},\ov{D}).
\end{equation}
\end{cor}

We point out here that the above corollary is compatible with the dimensional invariance of the normalized Minkowski content, established by M.\ Kneser in [Kne] and later recovered independently in~\cite{maja}.
More specifically, if in the above corollary, we assume, in addition, that $\ov{D}$ is the only pole of the tube zeta function of $A$ on the critical line $\{\re s=\overline{D}\}$ (i.e., $\ov{D}$ is the only complex dimension of $A$ with real part $\ov{D}$), then, according to \cite[Theorem~5.2]{cras2}, $A$ and $A\times\{0\}$ are Minkowski measurable, with Minkowski dimension $D:=\ov{D}$ and Minkowski contents satisfying the following identity:
\begin{equation}
\frac{\M^D(A)}{\pi^{\frac{D-N}{2}}\,\upGamma\left(\frac{N-{D}}{2}+1\right)}=\frac{\M^D(A\times\{0\})}{{\pi}^{\frac{D-N-1}{2}}\,\upGamma\left(\frac{N+1-{D}}{2}+1\right)}.
\end{equation}

\section{Embeddings of relative fractal drums}\label{emb_rfd}

The results obtained in the previous section in the context of bounded subsets of $\eR^N$ can also be obtained in the more general context of relative fractal drums (RFDs) in $\eR^N$.
More specifically, let $(A,\O)$ be a relative fractal drum in $\eR^N$ and let
$$
(A\times\{0\},\O\times(-1,1))
$$
be its natural embedding into $\eR^{N+1}$.
We want to connect the relative tube zeta functions of these two RFDs; the following lemma will be needed for this purpose.

\begin{lemma}\label{rel_emb}
Let $(A,\O)$ be a relative fractal drum in $\eR^N$ and fix $\delta\in(0,1)$.
Then we have
\begin{equation}
\big|(A\times\{0\})_\delta\cap(\O\times(-1,1))\big|_{N+1}=2\int_0^{\delta}|A_{\sqrt{\delta^2-u^2}}\cap\O|_N\,\di  u.
\end{equation}
\end{lemma}

\begin{proof}
We proceed much as in the proof of ~\cite[Proposition~6]{maja}.
Namely, we let $(x,y)\in\eR^{N}\times\eR\equiv\eR^{N+1}$ and define
\begin{equation}
V:=\left\{(x,y)\,:\,d_{N+1}((x,y),A\times\{0\})\leq\delta\}\cap\{(x,y)\,:\,x\in\O, |y|< 1\right\},
\end{equation}
where for any $k\in\eN$, $d_k$ denotes the Euclidean distance in $\eR^k$.
It is clear that the following equality holds: $d_{N+1}((x,y),A\times\{0\})=\sqrt{d_{N}(x,A)^2+y^2}$.
This implies that for a fixed $y\in[-\delta,\delta]$, we have
\begin{equation}
\begin{aligned}
V_y:&=\left\{x\in\eR^N\,:\,d_{N+1}((x,y),A\times\{0\})\leq\delta\right\}\\
&=\left\{x\,:\,d_{N}(x,A)\leq\sqrt{\delta^2-y^2}\right\}.
\end{aligned}
\end{equation}
(Note that if $|y|>\delta$, then $V_y$ is empty.)
Finally, Fubini's theorem implies that
$$
\begin{aligned}
\big|(A\times\{0\})_\delta\cap(\O\times(-1,1))\big|_{N+1}&=\int_{V}\di  x\,\di  y\\
&=\int_{-\delta}^{\delta}\di  y\int_{V_y\cap\{x\in\eR^N\,:\,x\in\O\}}\di  x\\
&=2\int_0^{\delta}|A_{\sqrt{\delta^2-y^2}}\cap\O|_N\,\di  y,
\end{aligned}
$$
which completes the proof of the lemma.
\end{proof}

The above lemma will eventually yield (in Theorem \ref{c_dim_inv_rel} below) an RFD analog of Proposition~\ref{prop_N+1} from \S\ref{emb_b} above.
First, however, we will show that the upper and lower relative box dimensions of an RFD are independent of the ambient space dimension.

\begin{prop}\label{invar_rel}
Let $(A,\O)$ be an RFD in $\eR^N$ and let
\begin{equation}
(A,\O)_M:=(A_M,\O\times(-1,1)^{M})
\end{equation}
be its embedding into $\eR^{N+M}$, for some arbitrary $M\in\eN$.
Then we have that
\begin{equation}
\ov{\dim}_B(A,\O)=\ov{\dim}_B(A,\O)_M
\end{equation}
and
\begin{equation}
\underline{\dim}_B(A,\O)=\underline{\dim}_B(A,\O)_M.
\end{equation}
\end{prop}

\begin{proof}
We only prove the proposition in the case when $M=1$, from which the general result then easily follows by induction.
It is clear that for $0<\delta<1$, we have
$$
\begin{aligned}
(A\times\{0\})_\delta\cap(\O\times(-1,1))&\subseteq(A\times\{0\})_\delta\cap(\O\times(-\delta,\delta))\\
&\subseteq (A_{\delta}\cap\O)\times(-\delta,\delta);
\end{aligned}
$$
so that
\begin{equation}
\left|(A\times\{0\})_\delta\cap(\O\times(-1,1))\right|_{N+1}\leq 2\delta|A_\delta\cap\O|_N.
\end{equation}
This observation, in turn, implies that for every $r\in\eR$, we have
\begin{equation}\label{jjbb}
\frac{\left|(A\times\{0\})_\delta\cap(\O\times(-1,1))\right|_{N+1}}{\delta^{N+1-r}}\leq\frac{2|A_\delta\cap\O|_N}{\delta^{N-r}}.
\end{equation}
Furthermore, by successively taking the upper and lower limits as $\delta\to0 ^+$ in Equation \eqref{jjbb} just above, we obtain the following inequalities, involving the $r$-dimensional upper and lower relative Minkowski contents of the RFDs $(A,\O)_1$ and $(A,\O)$, respectively:
\begin{equation}
{{\M}}^{*r}(A,\O)_1\leq 2{{\M}}^{*r}(A,\O)\quad\mathrm{and}\quad{{\M}}_*^{r}(A,\O)_1\leq 2{{\M}}_*^{r}(A,\O).
\end{equation}
In light of the definition of the relative upper and lower box (or Minkowski) dimensions (see Equations~\eqref{dimrell} and~\eqref{dimrelu} and the text surrounding them), we deduce that
\begin{equation}\label{ovdimineq}
\ov{\dim}_B(A,\O)_1\leq\ov{\dim}_B(A,\O)\quad\mathrm{and}\quad\underline{\dim}_B(A,\O)_1\leq\underline{\dim}_B(A,\O).
\end{equation}

On the other hand, for geometric reasons, we have that
$$
(A_{\delta/2}\cap\O)\times\left(-\frac{\delta\sqrt{3}}{2},\frac{\delta\sqrt{3}}{2}\right)\subseteq(A\times\{0\})_\delta\cap(\O\times(-1,1));
$$
so that
\begin{equation}
\delta\sqrt{3}|A_{\delta/2}\cap\O|_N\leq\left|(A\times\{0\})_\delta\cap(\O\times(-1,1))\right|_{N+1}.
\end{equation}
Much as before, this inequality implies that for every $r\in\eR$, we have
\begin{equation}
\frac{\sqrt{3}|A_{\delta/2}\cap\O|_N}{2^{N-r}(\delta/2)^{N-r}}\leq\frac{\left|(A\times\{0\})_\delta\cap(\O\times(-1,1))\right|_{N+1}}{\delta^{N+1-r}}
\end{equation}
and by successively taking the upper and lower limits as $\delta\to 0^+$, we obtain that
\begin{equation}\label{rev_ineq}
\frac{\sqrt{3}{{\M}}^{*r}(A,\O)}{2^{N-r}}\leq{{\M}}^{*r}(A,\O)_1\quad\mathrm{and}\quad \frac{\sqrt{3}{\M}_*^{r}(A,\O)}{2^{N-r}}\leq{\M}_*^{r}(A,\O)_1.
\end{equation}
Finally, this completes the proof because (again in light of Equations~\eqref{dimrell} and~\eqref{dimrelu} and the text surrounding them) \eqref{rev_ineq} implies the reverse inequalities for the upper and lower relative box dimensions in~\eqref{ovdimineq}.
\end{proof}

\begin{remark}\label{4.7.81/2}
Observe that it follows from Proposition \ref{invar_rel} (combined with part $(b)$ of Theorem \ref{an_rel}) that the RFDs $(A,\O)$ and $(A,\O)_M$ have the same upper Minkowski dimension, $\ov{\dim}_B(A,\O)=\ov{\dim}_B(A,\O)_M$, and hence, the same critical line $\{\re s=\ov{\dim}_B(A,\O)\}$.
This fact will be used implicitly in the statement of Proposition \ref{prop_N+1_rel} as well as in the statements of Theorems \ref{mero_ext_N+1_rel} and \ref{c_dim_inv_rel} just below.
\end{remark}

We can now state the desired results for embedded RFDs and their relative zeta functions.
In light of Lemma~\ref{rel_emb} and Proposition~\ref{invar_rel}, the proofs follow the same steps as in the corresponding results established in \S\ref{emb_b} about bounded subsets of $\eR^N$ (namely, Proposition \ref{prop_N+1} and Theorem \ref{mero_ext_N+1}, respectively), and for this reason, we will omit them.

\begin{prop}\label{prop_N+1_rel}
Fix $\d\in(0,1)$ and let $(A,\O)$ be an RFD in $\eR^N$, with $\overline{D}:=\overline{\dim}_B(A,\O)$.
Then, for the relative tube zeta functions of $(A,\O)$ and $(A,\O)_1:=(A\times\{0\},\O\times(-1,1))$, the following equality holds$:$
\begin{equation}\label{eq_N+1_rel}
\tilde{\zeta}_{A\times\{0\},\O\times(-1,1);\delta}(s)=2\int_0^{\pi/2}\frac{\tilde{\zeta}_{A,\O;\delta\sin \tau}(s)}{\sin^{s-N-1}\tau}\,\di\tau,
\end{equation}
for all $s\in\Ce$ such that $\re s>\overline{D}$.
\end{prop}

\begin{theorem}\label{mero_ext_N+1_rel}
Fix $\d\in(0,1)$ and let $(A,\O)$ be an RFD in $\eR^N$, with $\ov{D}:=\ov{\dim}_B(A,\O)$. 
Then, we have the following equality between $\tilde{\zeta}_{A,\O}$, the tube zeta function of $(A,\O)$, and $\tilde{\zeta}_{A_M,\O\times(-1,1)^M}$, the tube zeta function of the relative fractal drum $(A,\O)_M:=(A_M,\O\times(-1,1)^M)$ in $\eR^{N+M}$, for some arbitrary $M\in\eN$$:$
\begin{equation}\label{link_zzm_rel}
\tilde{\zeta}_{A_M,\O\times(-1,1)^M;\delta}(s)= \frac{\left(\sqrt{\pi}\right)^M\upGamma\left(\frac{N-s}{2}+1\right)}{\upGamma\left(\frac{N+M-s}{2}+1\right)}\tilde{\zeta}_{A,\O;\delta}(s)+E(s;\delta),
\end{equation}
initially valid for all $s\in\Ce$ such that $\re s>\ov{D}$. $($See Theorem \ref{c_dim_inv_rel} for more precise information about the domain of validity of the approximate functional equation \eqref{link_zzm_rel}.$)$
Here, the error function $E(s):=E(s;\delta)$ is meromorphic on all of $\Ce$.
Furthermore, the possible poles $($in $\Ce$$)$ of $E(s;\delta)$ are located at $s_k:=N+2+2k$ for every $k\in\eN_0$, and all of them are simple.
$($It follows that $\tilde{\zeta}_A$ is well defined at each $s_k$.$)$
Moreover, we have that for each $k\in\eN_0$,
\begin{equation}\label{E_res_rel}
\res(E(\,\cdot\,;\delta),s_k)=\frac{(-1)^{k+1}\left(\sqrt{\pi}\right)^{M}}{k!\,\upGamma\left(\frac{M}{2}-k\right)}\tilde{\zeta}_{A,\O;\delta}(s_k).
\end{equation}
More specifically, if $M$ is even, then all of the poles $s_k$ of $E(s;\d)$ are canceled for $k\geq M/2$; i.e., the corresponding residues in \eqref{E_res_rel} are equal to zero.
On the other hand, if $M$ is odd, then there are no such cancellations and all of the residues in \eqref{E_res_rel} are nonzero; so that all of the $s_k$'s are $($simple$)$ poles of $E(s;\d)$ in that case.
\end{theorem}

We deduce at once from Theorem \ref{mero_ext_N+1_rel} the following key result about the invariance of the complex dimensions of a relative fractal drum with respect to the dimension of the ambient space.
This result extends Theorem \ref{c_dim_inv} to general RFDs.

\begin{theorem}\label{c_dim_inv_rel}
Let $(A,\O)$ be an RFD in $\eR^N$ and let the RFD $(A,\O)_M:=(A_M,\O\times(-1,1)^M)$ be its embedding into $\eR^{N+M}$, for some arbitrary $M\in\eN$.
Then, the tube zeta function $\tilde{\zeta}_{A,\O}:=\tilde{\zeta}_{A,\O}$ of $(A,\O)$ has a meromorphic extension to a given open connected neighborhood $U$ of the critical line $\{\re s=\ov{\dim}_B(A,\O)\}$ if and only if the analogous statement is true for the tube zeta function $\tilde{\zeta}_{(A,\O)_M}:=\tilde{\zeta}_{A_M,\O\times(-1,1)^M}$ of $(A,\O)_M$.
$($See Remark \ref{4.7.81/2} just above.$)$
Furthermore, in that case, the approximate functional equation \eqref{link_zzm_rel} remains valid for all $s\in U$.
In addition, the multisets of the poles of $\tilde{\zeta}_{A,\O}$ and $\tilde{\zeta}_{(A,\O)_M}$ belonging to $U$ coincide; i.e., 
\begin{equation}\label{4.7.26.1/2}
\po(\tilde{\zeta}_{A,\O},U)=\po(\tilde{\zeta}_{(A,\O)_M},U).
\end{equation}
Consequently, neither the values nor the multiplicities of the complex dimensions of the RFD $(A,\O)$ depend on the dimension of the ambient space.
\end{theorem}

\begin{remark}\label{discussion}
In the above discussion about embedding RFDs into higher-dimensi\-onal spaces, we can also make similar observations if we embed $(A,\O)$ as a `one-sided' RFD, for example of the form $(A\times\{0\},\O\times(0,1))$, a fact which can be more useful when decomposing a relative fractal drum into a union of relative fractal subdrums in order to compute its distance (or tube) zeta function.
This observation follows immediately from the above results for `two-sided' embeddings of RFDs since, by symmetry, we have
\begin{equation}
\tilde{\zeta}_{A\times\{0\},\O\times(-1,1)}(s)=2\,\tilde{\zeta}_{A\times\{0\},\O\times(0,1)}(s).
\end{equation}
We note that when using the above formulas, one only has to be careful to take into account the factor $2$.
Furthermore, we can also embed $(A,\O)$ as $$(A\times\{0\},\O\times(-\alpha,\alpha))\quad \textrm{or}\quad (A\times\{0\},\O\times(0,\alpha)),$$ for some $\alpha>0$, but in that case, the corresponding formulas will only be valid for all $\delta\in(0,\alpha)$.
\end{remark}

We could now use the functional equation \eqref{rel_equality} connecting the tube and distance zeta functions, in order to translate the above results in terms of $\zeta_{A,\O}$, the (relative) distance zeta function of the RFD $(A,\O)$.
However, we will instead use another approach because it gives some additional information about the resulting error function.
More specifically, consider the {\em Mellin\index{Mellin zeta function of an RFD|textbf} zeta function of a relative fractal drum} 
$(A,\O)$ defined by
\begin{equation}\label{mellinz}
\zeta_{A,\O}^{\mathfrak{M}}(s):=\int_0^{+\ty}t^{s-N-1}|A_t\cap\O|\,\di t,
\end{equation}
for all $s\in\Ce$ located in a suitable vertical strip.
In fact, in light of \cite[Theorem 5.7]{cras2} (see also \cite[Theorem 5.4.7]{fzf}), the above Lebesgue integral is absolutely convergent (and hence, convergent) for all $s\in\Ce$ such that $\re s\in(\overline{\dim}_B(A,\O),N)$.
Moreover, the relative distance and Mellin zeta functions of $(A,\O)$ are connected by the functional equation
\begin{equation}\label{eqqq}
\zeta_{A,\O}(s)=(N-s)\zeta_{A,\O}^{\mathfrak{M}}(s),
\end{equation}
on every open connected set $U\subseteq\Ce$ to which any of the two zeta functions has a meromorphic continuation.
Observe that in \eqref{eqqq}, the parameter $\d$ is absent.
Indeed, this means implicitly that the functional equation \eqref{eqqq} is valid only for the parameters $\d>0$ for which the inclusion $\O\subseteq A_\d$ is satisfied; that is, when the equality $\zeta_{A,\O;\d}(s)=\int_{\O}d(x,A)^{s-N}\di x$ is satisfied.

We will now embed the relative fractal drum $(A,\O)$ of $\eR^N$ into $\eR^{N+1}$ as
$$
(A\times\{0\},\O\times\eR).
$$
Strictly speaking, this is not a relative fractal drum in $\eR^{N+1}$ since there does not exist a $\d>0$ such that $\O\times\eR\subseteq(A\times\{0\})_{\d}$.
On the other hand, observe that Lemma \ref{rel_emb} is now valid for every $\d>0$; that is,
\begin{equation}\label{foralld}
\big|(A\times\{0\})_\delta\cap(\O\times\eR)\big|_{N+1}=2\int_0^{\delta}|A_{\sqrt{\delta^2-u^2}}\cap\O|_N\di  u.
\end{equation}

\begin{prop}\label{F_mero}
Let $(A,\O)$ be an RFD in $\eR^{N}$ such that $\ov{\dim}_B(A,\O)<N$.
Then the function $F=F(s)$, defined by the integral
\begin{equation}\label{fsss}
F(s):=\int_0^{+\ty}t^{s-N-2}\big|(A\times\{0\})_t\cap(\O\times\eR)\big|_{N+1}\,\di t,
\end{equation}
is holomorphic inside the vertical strip $\{\ov{\dim}_B(A,\O)<\re s<N\}$.
\end{prop}

\begin{proof}
We split the integral into two integrals: $F(s)=\int_0^1+\int_1^{+\ty}$.
According to Proposition \ref{invar_rel}, the first integral
$$
\begin{aligned}
&\int_0^{1}t^{s-N-2}\big|(A\times\{0\})_t\cap(\O\times\eR)\big|_{N+1}\,\di t\\
=&\int_0^{1}t^{s-N-2}\big|(A\times\{0\})_t\cap(\O\times(-1,1))\big|_{N+1}\,\di t
\end{aligned}
$$
defines a holomorphic function on the right half-plane $\{\re s>\ov{\dim}_B(A,\O)\}$.

In order to deal with the second integral, we observe that
$$
\big|(A\times\{0\})_t\cap(\O\times\eR)\big|_{N+1}\leq 2t|\O|_N,
$$
and, consequently,
$$
\begin{aligned}
\left|\int_1^{+\ty}t^{s-N-2}\big|(A\times\{0\})_t\cap(\O\times\eR)\big|_{N+1}\di t\right|&\leq 2|\O|_N\int_1^{+\ty}t^{\re s-N-1}\,\di t\\
&=\frac{2|\O|_N}{N-\re s},
\end{aligned}
$$
for all $s\in\Ce$ such that $\re s<N$. 
In light of Theorem \ref{Hh} and Remark~\ref{condition3}, the latter inequality implies that the integral over $(1,+\ty)$ defines a holomorphic function on the left half-plane $\{\re s<N\}$.
Therefore, it follows that $F(s)$ is holomorphic in the vertical strip $\{\ov{\dim}_B(A,\O)<\re s<N\}$ and the proof of the proposition is complete.
\end{proof}

In light of the above proposition, we continue to use the convenient notation $\zeta_{A\times\{0\},,\O\times\eR}^{\mathfrak{M}}$ for the integral on the right-hand side of \eqref{fsss} although, as was noted earlier, $(A\times\{0\},\O\times\eR)$ is not technically a relative fractal drum in $\eR^{N+1}$; see Remark \ref{discussion} above.
The following result is the counterpart of Theorem \ref{mero_ext_N+1} in the present, more general context.

\begin{theorem}\label{thm_dist_emb}
Let $(A,\O)$ be a relative fractal drum in $\eR^N$ such that $\overline{D}:=\overline{\dim}_B(A,\O)<N$. Then, for every $a>0$, the following approximate functional equation holds$:$
\begin{equation}\label{pomocna_j}
\zeta_{A\times\{0\},\O\times(-a,a)}(s)=\frac{\sqrt{\pi}\upGamma\left(\frac{N-s}{2}\right)}{\upGamma\left(\frac{N+1-s}{2}\right)}\zeta_{A,\O}(s)+E(s;a),
\end{equation}
initially valid for all $s\in\Ce$ such that $\re s>\ov{D}$.
Here, the error function $E(s):=E(s;a)$ is initially given $($for all $s\in\Ce$ such that $\re s<N$$)$ by 
\begin{equation}
E(s;a):=(s\!-\! N\!-\! 1)\int_a^{+\ty}t^{s-N-2}|(A\times\{0\})_{t}\cap\O\times(\eR\setminus(-a,a))|_{N+1}\,\di t,
\end{equation}
and admits a meromorphic extension to all of $\Ce$,
with a set of simple poles equal to $\{N+2k:k\in\eN_0\}$.

Moreover, Equation \eqref{pomocna_j} remains valid on any connected open neighborhood of the critical line $\{\re s=\ov{D}\}$ to which $\zeta_{A,\O}$ $($or, equivalently, $\zeta_{A\times\{0\},\O\times(-a,a)}$$)$ can be meromorphically continued.
\end{theorem}

\begin{proof}
In a completely analogous way as in the proof of Theorem \ref{mero_ext_N+1}, we obtain that
\begin{equation}\label{link_zzm_rel_R}
\tilde{\zeta}_{A\times\{0\},\O\times\eR;\delta}(s)= \frac{\sqrt{\pi}\upGamma\left(\frac{N-s}{2}+1\right)}{\upGamma\left(\frac{N+1-s}{2}+1\right)}\tilde{\zeta}_{A,,\O;\delta}(s)+\tilde{E}(s;\delta),
\end{equation}
now valid for all $\d>0$ (see Equation \eqref{foralld} above and the discussion preceding it).
Furthermore, the error function $\tilde{E}(s):=\tilde{E}(s;\delta)$ is holomorphic on $\{\re s<N+1\}$ and
\begin{equation}\label{esst}
|\tilde{E}(s,\d)|\leq{2\delta^{\re s-N}|A_{\delta}\cap\O|_N}\left(\frac{\pi}{2}-1\right)
\end{equation}
for all $s\in\Ce$ such that $\re s<N+1$.
See the proof of Theorem \ref{mero_ext_N+1} and Equation \eqref{Esd} for the derivation of the above estimate.
The estimate \eqref{esst} now implies that the sequence of holomorphic functions $\tilde{E}(\,\cdot\,;n)$ tends to $0$ as $n\to\ty$, uniformly on every compact subset of $\{\re s<N\}$, since $|A_{n}\cap\O|=|\O|$ for all $n$ sufficiently large.
Furthermore, we also have that $\tilde{\zeta}_{A,\O;n}\to\zeta_{A,\O}^{\mathfrak{M}}$ and
\begin{equation}
\tilde{\zeta}_{A\times\{0\},\O\times\eR}(s;n)\to\zeta_{A\times\{0\},\O\times\eR}^{\mathfrak{M}}\quad\textrm{as}\quad n\to\ty,
\end{equation}
uniformly on every compact subset of $\{\overline{D}<\re s<N\}$.
This implies that by taking the limit in \eqref{link_zzm_rel_R} as $\d\to+\ty$, we obtain the following functional equality between holomorphic functions:
\begin{equation}\label{frak0}
\zeta_{A\times\{0\},\O\times\eR}^{\mathfrak{M}}(s)=\frac{\sqrt{\pi}\upGamma\left(\frac{N-s}{2}+1\right)}{\upGamma\left(\frac{N+1-s}{2}+1\right)}\zeta_{A,\O}^{\mathfrak{M}}(s),
\end{equation}
valid in the vertical strip $\{\overline{D}<\re s<N\}$.
(We can obtain this equality even more directly by applying Lebesgue's dominated convergence theorem to a counterpart of \eqref{eq_N+1_rel}.)

Moreover, according to \eqref{eqqq} and \eqref{frak0}, we have the functional equation
 \begin{equation}\label{mfrak}
\zeta_{A\times\{0\},\O\times\eR}^{\mathfrak{M}}(s)=\frac{2\sqrt{\pi}\upGamma\left(\frac{N-s}{2}\right)}{\upGamma\left(\frac{N+1-s}{2}+1\right)}\zeta_{A,\O}(s),
\end{equation}
from which we deduce that the right-hand side admits a meromorphic extension to the right half-plane $\{\re s>\overline{D}\}$, with simple poles located at the simple poles of $\upGamma((N-s)/2)$; that is, at $s_k:=N+2k$ for all $k\in\eN_0$.
(Observe that in the above ratio of gamma functions, there are no cancellations between the poles of the numerator and of the denominator; indeed, an integer cannot be both even and odd.)
From this we conclude that by the principle of analytic continuation, the same property also holds for the left-hand side of \eqref{mfrak} and, furthermore, the left-hand side has a meromorphic extension to any domain $U\subseteq\Ce$ to which the right-hand side can be meromorphically extended.

In order to complete the proof of the theorem, we now observe that for any $a>0$, since
\begin{equation}\nonumber
\begin{aligned}
\big|(A\times\{0\})_t\cap(\O\times\eR)\big|&=\big|(A\times\{0\})_t\cap(\O\times(-a,a))\big|\\
&\phantom{=}+\big|(A\times\{0\})_t\cap(\O\times(\eR\setminus(-a,a)))\big|,
\end{aligned}
\end{equation}
the left-hand side of \eqref{mfrak} can be split into two parts:
\begin{equation}\nonumber
\begin{aligned}
\zeta_{A\times\{0\},\O\times\eR}^{\mathfrak{M}}(s)&=\zeta_{A\times\{0\},\O\times(-a,a)}^{\mathfrak{M}}(s)\\
&\phantom{=}+\int_a^{+\ty}t^{s-N-2}\big|(A\times\{0\})_t\cap(\O\times(\eR\setminus[-a,a]))\big|\,\di t\\
&=\frac{\zeta_{A\times\{0\},\O\times(-a,a)}(s)}{N+1-s}-\frac{E(s;a)}{N+1-s}.
\end{aligned}
\end{equation}
We then combine this observation with \eqref{mfrak} to obtain \eqref{pomocna_j}.
In light of part $(a)$ of Theorem~\ref{an_rel}, we know that $\zeta_{A\times\{0\},\O\times(-a,a)}(s)$ is holomorphic on the open right half-plane $\{\re s>\overline{D}\}$.
Furthermore, much as in the proof of Proposition \ref{F_mero}, we can show that $E(s):=E(s;a)$ defines a holomorphic function on the open left half-plane $\{\re s<N\}$.
This fact, together with the functional equation \eqref{pomocna_j}, now ensures that $E(s;a)$ admits a meromorphic continuation to all of $\Ce$, with a set of simple poles equal to $\{N+2k:k\in\eN_0\}$.
(Note that $\zeta_{A,\O}(s)>0$ for all $s\in[N,+\ty)$, which implies that there are no zero-pole cancellations on the right-hand side of \eqref{pomocna_j}.)
This completes the proof of Theorem \ref{thm_dist_emb}.
\end{proof}

We note that in Example \ref{c_dust} below, we actually want to embed $(A,\O)$ into $\eR^{N+1}$, as $(A\times\{0\},\O\times(0,a))$ for some $a>0$.
By looking at the proof of the above theorem and using a suitable symmetry argument, we can obtain the following result, which deals with this type of embedding.

\begin{theorem}\label{ulaganje_jedan}
Let $(A,\O)$ be a relative fractal drum in $\eR^N$ such that $\overline{D}:=\overline{\dim}_B(A,\O)<N$. Then, the following approximate functional equation holds$:$
\begin{equation}\label{pomocna_j_2}
\zeta_{A\times\{0\},\O\times(0,a)}(s)=\frac{\sqrt{\pi}\upGamma\left(\frac{N-s}{2}\right)}{2\upGamma\left(\frac{N+1-s}{2}\right)}\zeta_{A,\O}(s)+E(s;a),
\end{equation}
initially valid for all $s\in\Ce$ such that $\re s>\ov{D}$.
Here, the error function $E(s):=E(s;a)$ is initially given $($for all $s\in\Ce$ such that $\re s<N)$ by
\begin{equation}
E(s;a):=(s\!-\! N\!-\! 1)\int_a^{+\ty}t^{s-N-2}|(A\times\{0\})_{t}\cap\O\times(\eR\setminus(0,a))|_{N+1}\di t,
\end{equation}
and admits a meromorphic continuation to all of $\Ce$,
with a set of simple poles equal to $\{N+2k:k\in\eN_0\}$.

Moreover, Equation \eqref{pomocna_j_2} remains valid on any connected open neighborhood of the critical line $\{\re s=\ov{D}\}$ to which $\zeta_{A,\O}$ $($or, equivalently, $\zeta_{A\times\{0\},\O\times(0,a)}$$)$ can be meromorphically continued.
\end{theorem} 

\begin{example}\label{c_dust} ({\em Complex dimensions of the Cantor dust RFD}).\index{Cantor dust|textbf}
In this example, we will consider the relative fractal drum consisting of the Cantor dust contained in $[0,1]^2$ and compute its distance zeta function.
More precisely, let $A:=C^{(1/3)}\times C^{(1/3)}$ be the Cantor dust (i.e., the Cartesian product of the ternary Cantor set $C:=C^{1/3}$ by itself) and let $\O:=(0,1)^2$.
We will not obtain for $\zeta_A$ an explicit formula in a closed form but we will instead use Theorem~\ref{ulaganje_jedan} in order to deduce that the distance zeta function of the Cantor dust has a meromorphic continuation to all of $\Ce$.

More interestingly, we will also show that the set of complex dimensions of the Cantor dust is the union of (a nontrivial subset of) a periodic set contained in the critical line $\{\re s=\log_34\}$ and the set of complex dimensions of the Cantor set (which is a periodic set contained in the critical line $\{\re s=\log_32\}$).
{\em This fact is significant because it shows that in this case, the distance $($or tube$)$ zeta function also detects the `lower-dimensional' fractal nature of the Cantor dust.}

Note that, as is well known, the Minkowski dimension of the RFD (or Cantor string) $(C,(0,1))$ is given by $\dim_B(C,(0,1))=\log_32$ (see, e.g., \cite[\S1.2.2]{lapidusfrank}).
Furthermore, it will follow from the discussion below that, as might be expected since $(A,\O)=(C,(0,1))\times(C,(0,1))$), 
\begin{equation}\label{4.7.391/2}
\dim_B(A,\O)=2\dim_B(C,(0,1))=\log_34.
\end{equation}
Consequently, the critical line of the RFD $(C,(0,1))$ in $\eR$ (the {\em Cantor string RFD})\index{Cantor string!relative fractal drum|textbf} is the vertical line $\{\re s=\log_32\}$, while the critical line of the 
Cantor dust, viewed as the RFD $(A,\O)$ in $\eR^2$,
 is the vertical line $\{\re s=\log_34\}$, as was stated in the previous paragraph.

The construction of the RFD $(A,\O)$ can be carried out by beginning with the unit square and removing the open middle-third `cross', and then iterating this procedure ad infinitum.
This procedure implies that we can subdivide the Cantor dust into a countable union of RFDs which are scaled down versions of two base (or generating) RFDs, $(A_1,\O_1)$ and $(A_2,\O_2)$.
The first one of these base RFDs, $(A_1,\O_1)$, is defined by $\O_1:=(0,1/3)^2$ and by $A_1$ being the union of the four vertices of the closure of $\O_1$ (namely, of the square $[0,1/3]^2$).
Furthermore, the second base RFD is defined by $\O_2:=(0,1/3)\times(0,1/6)$ and by $A_2$ being the ternary Cantor set contained in $[0,1/3]\times\{0\}$.

At the $n$-th step of the iteration, we have exactly $4^{n-1}$ RFDs of the type $(a_nA_1,a_n\O_1)$ and $8\cdot 4^{n-1}$ RFDs of the type $(a_nA_2,a_n\O_2)$, where $a_n:=3^{-n}$ for each $n\in\eN$.
This observation, together with the scaling property of the relative distance zeta function (see Theorem \ref{scaling}), yields successively (for all $s\in\Ce$ with $\re s$ sufficiently large)$:$
\begin{equation}
\begin{aligned}
\zeta_{A,\O}(s)&=\sum_{n=1}^{\ty}4^{n-1}\zeta_{a_nA_1,a_n\O_1}(s)+8\sum_{n=1}^{\ty}4^{n-1}\zeta_{a_nA_2,a_n\O_2}(s)\\
&=\left(\zeta_{A_1,\O_1}(s)+8\zeta_{A_2,\O_2}(s)\right)\sum_{n=1}^{\ty}4^{n-1}\cdot 3^{-ns}\\
&=\frac{1}{3^s-4}\left(\zeta_{A_1,\O_1}(s)+8\zeta_{A_2,\O_2}(s)\right).
\end{aligned}
\end{equation}
Moreover, for the relative distance zeta function of $(A_1,\O_1)$, we have
\begin{equation}
\begin{aligned}
\zeta_{A_1,\O_1}(s)&=8\int_0^{1/6}\di  x\int_0^{x}\left(\sqrt{x^2+y^2}\right)^{s-2}\,\di  y\\
&=8\int_0^{\pi/4}\di \varphi\int_0^{1/6\cos\varphi}r^{s-1}\,\di  r\\
&=\frac{8}{6^ss}\int_0^{\pi/4}\cos^{-s}\varphi\,\di \varphi=\frac{8I(s)}{6^ss},
\end{aligned}
\end{equation}
where $I(s):=\int_0^{\pi/4}\cos^{-s}\varphi\,\di \varphi$ and is easily seen to be an entire function. (In fact, $I(s)=2^{-1}\mathrm{B}_{{1}/{2}}\left({1}/{2},{(1-s)}/{2}\right)$, where $\mathrm{B}_x(a,b):=\int_0^{x}t^{a-1}(1-t)^{b-1}\di  t$ is the incomplete beta function.\label{incbeta})\index{beta function, $B(a,b)$!incomplete beta function, $\mathrm{B}_x(a,b)$}
Consequently, $\zeta_{A,\O}$ admits a meromorphic continuation to all of $\Ce$ and we have
\begin{equation}
\zeta_{A,\O}(s)=\frac{8}{3^s-4}\left(\frac{I(s)}{6^ss}+\zeta_{A_2,\O_2}(s)\right),
\end{equation}
for all $s\in\Ce$.
Furthermore, let $\zeta_{C,,(0,1)}$ be the relative distance zeta function of the Cantor middle-third set constructed inside $[0,1]$; see \cite[Example 6.3]{cras2}. Alternatively, use the relation
$\zeta_{C,(0,1)}(s)=\zeta_{{\mathcal L}_{CS}}(s)(2^{1-s}/s)$, where ${\mathcal L}_{CS}$ is the Cantor string and (by \cite[Eq.\ (1.29), p.\ 22]{lapidusfrank12}) $\zeta_{{\mathcal L}_{CS}}(s)=1/(3^s-2)$, for all $s\in\Ce$.
From Theorem~\ref{ulaganje_jedan} and the scaling property of the relative distance zeta function (Theorem \ref{scaling}), we now deduce that
\begin{equation}\label{deno}
\begin{aligned}
\zeta_{A_2,\O_2}(s)&=\frac{\sqrt{\pi}\mathrm{\Gamma}\left(\frac{1-s}{2}\right)}{2\mathrm{\Gamma}\left(\frac{2-s}{2}\right)}\zeta_{3^{-1}C,3^{-1}(0,1)}(s)+E(s;6^{-1})\\
&=\frac{\mathrm{\Gamma}\left(\frac{1-s}{2}\right)}{\mathrm{\Gamma}\left(\frac{2-s}{2}\right)}\frac{\sqrt{\pi}}{6^ss(3^s-2)}+E(s;6^{-1}),
\end{aligned}
\end{equation}
where $E(s;6^{-1})$ is meromorphic on all of $\Ce$ with a set of simple poles equal to $\{2k+1:k\in\eN_0\}$; so that for all $s\in\Ce$, we have
\begin{equation}\label{form}
\zeta_{A,\O}(s)=\frac{8}{s(3^s-4)}\left(\frac{I(s)}{6^s}+\frac{\upGamma\left(\frac{1-s}{2}\right)}{\upGamma\left(\frac{2-s}{2}\right)}\frac{\sqrt{\pi}}{6^ss(3^s-2)}+E(s;6^{-1})\right).
\end{equation}
Formula~\eqref{form} implies that $\po(\zeta_{A,\O}$, the set of all complex dimensions (in $\Ce$) of the `relative' Cantor dust, is a subset of
\begin{equation}
\left(\log_34+\frac{2\pi}{\log3}\I\Ze\right)\cup\left(\log_32+\frac{2\pi}{\log3}\I\Ze\right)\cup\{0\}
\end{equation}
and consists of simple poles of $\zeta_{A,\O}$.
Of course, we know that $\log_34\in\po(\zeta_{A,\O})$, but we can only conjecture that the other poles on the critical line $\{\re s=\log_34\}$ are in $\po(\zeta_{A,\O})$ since it may happen that there are some zero-pole cancellations in~\eqref{form}.
On the other hand, since it is known that the Cantor dust is not Minkowski measurable (see~\cite{FaZe}), we can deduce from 
 [LapRa\v Zu6, Theorem 4.2]
that there must exist at least two other (necessarily nonreal) poles $s_{\pm k_0}=\log_34\pm\frac{2k_0\pi\I}{\log 3}$ of $\zeta_{A,\O}$, for some $k_0\in\eN$. (Indeed, according to the Minkowski measurability criterion established in  [LapRa\v Zu6, Theo\-rem~4.2] (see also [LapRa\v Zu1, Theorem~5.4.20]), $D:=\log_34$ cannot be the only complex dimension of $(A,\O)$ on the critical line $\{\re s=D\}$ since otherwise, the Cantor dust would be Minkowski measurable, which is a contradiction.)
Based on~\eqref{form}, we cannot even claim that $0\in\po(\zeta_{A,\O})$ for sure, but we can see that all of the principal complex dimensions of the Cantor set are elements of $\po(\zeta_{A,\O})$; i.e., $\log_32+\frac{2\pi}{\log3}\I\Ze\subseteq\po(\zeta_{A,\O}))$.
We conjecture that we also have $\log_34+\frac{2\pi}{\log3}\I\Ze\subseteq\po(\zeta_{A,\O})$; that is, we conjecture that $\po_c(\zeta_{A,\O})=\log_34+\frac{2\pi}{\log3}\I\Ze$.
\end{example}

The above example can be easily generalized to the case of Cartesian products of any finite number of generalized Cantor sets (as given by Definition \ref{Cma}), in which case we conjecture that the set of complex dimensions of the product is contained in the union of the sets of complex dimensions of each of the factors, modulo any zero-pole cancellations which may occur.
In light of this and other similar examples, it would be interesting to obtain some results about zero-free regions for fractal zeta functions.
We leave this problem as a possible subject for future investigations.

\chapter{Relative fractal sprays and principal complex dimensions of any multiplicity}\label{zfsprays}


In this chapter, we consider a special type of RFDs, called {\em relative fractal sprays}, and study their distance zeta functions. We then illustrate our results by computing the corresponding complex dimensions of relative  Sierpi\'nski sprays. More specifically, we determine the complex dimensions (as well as the associated residues) of the relative Sierpi\'nski gasket (Example \ref{6.15}) and of the relative Sierpi\'nski carpet (Example \ref{sierpinski_carpetr}); we also consider higher-dimensional analogs of these examples, namely, the inhomogeneous Sierpi\'nski $N$-gasket RFD (Example \ref{Ngasket}) and the Sierpi\'nski $N$-carpet RFD (Example \ref{carpetN}).

\section{Relative fractal sprays in $\eR^N$}\label{dsprays}
We now introduce the definition of relative fractal spray, which is very similar to (but more general than) the notion of fractal spray (see \cite{lapiduspom3}, \cite[Definition 13.2]{lapidusfrank12},
[{LapPe1--2}]
and
[{LapPeWi1--2}]), itself a generalization of the notion of (ordinary) fractal string
[{LapPo1--2}, {Lap1--3}, {Lap-vFr3}].

\begin{defn}\label{spray}
Let $(\pa\Omega_0,\Omega_0)$ be a fixed relative fractal drum in $\eR^N$ (which we call the 
{\em base relative fractal drum}, or {\em generating
relative fractal drum}),\index{relative fractal drum, RFD!generating}\index{generating relative fractal drum}  $(\lambda_j)_{j\ge0}$ a decreasing 
sequence of positive numbers (scaling factors), converging to zero, and $(b_j)_{j\ge0}$ a given 
 sequence of positive integers (multiplicities). The associated {\em relative fractal 
 spray}\index{relative fractal spray|textbf} is 
a relative fractal drum $(A,\O)$ obtained as the disjoint union
of a sequence of RFDs, $\mathcal{F}:=\{(\pa\Omega_i,\Omega_i):i\in\eN_0\}$, where $\eN_0:=\eN\cup\{0\}$, such that each $\Omega_i$ can be obtained
from $\lambda_j\Omega_0$ by a rigid motion in $\eR^N$, and for each $j\in\eN_0$ there are precisely $b_j$ RFDs
 in the family $\mathcal{F}$ that can be obtained from $\lambda_j\Omega_0$ by a rigid motion.
Any relative fractal spray $(A,\Omega)$, generated by the base relative fractal drum (or `basic shape') $\Omega_0$ and the sequences of `scales'  $(\lambda_j)_{j\ge0}$ with associated `multiplicities' $(b_j)_{\ge0}$, is denoted by 
\begin{equation}
(A,\Omega):=\operatorname{Spray}(\Omega_0,(\lambda_j)_{j\ge0},(b_j)_{\ge0}).
\end{equation}
The family $\mathcal{F}$ is called the {\em skeleton of the spray}.
The distance zeta function $\zeta_{A,\O}$ of the relative fractal spray $(A,\O)$ is computed in Theorem~\ref{sprayz} below. 

If there exist $\lambda\in(0,1)$ 
and an integer $b\ge2$ such that $\lambda_j=\lambda^j$
and $b_j=b^j$, for all $j\in\eN_0$,  then we simply write
$$
(A,\Omega)=\operatorname{Spray}(\Omega_0,\lambda,b).
$$
\end{defn}

\medskip

\begin{defn}\label{tensor}
The relative fractal spray $(A,\Omega)=\operatorname{Spray}(\Omega_0,(\lambda_j)_{j\ge0},$ $(b_j)_{j\ge0})$ can be viewed as a relative fractal drum
generated by $(\pa\Omega_0,\Omega_0)$ and a fractal string $\mathcal{L}=(\ell_j)_{j\ge0}$, consisting of the decreasing sequence $(\lambda_j)_{j\geq 0}$ of positive real numbers, in which each $\lambda_j$
has multiplicity $b_j$ for every $j\ge0$. Thus, we can write $(A,\Omega)=\operatorname{Spray}(\Omega_0,\mathcal{L})$. It is also convenient to view the construction of $(A,\O)$
in Definition~\ref{spray} as the {\em tensor product}\label{tensor1}\index{tensor product!$(\pa\O_0,\O_0)\otimes\mathcal{L}$|textbf} of the base relative fractal drum $(A_0,\O_0)$ and the fractal string $\mathcal{L}$:
\begin{equation}
(A,\O)=(\pa\O_0,\O_0)\otimes\mathcal{L}.
\end{equation}
We can also define the {\em tensor product of two $($possibly unbounded$)$ fractal strings}\index{tensor product!of fractal strings, ${\mathcal L}_1\otimes{\mathcal L}_2$|textbf}
$\mathcal{L}_1=(\ell_{1j})_{j\ge1}$ and $\mathcal{L}_2=(\ell_{2k})_{k\ge1}$ as
the following fractal string (note that here, $\mathcal{L}_1$ and $\mathcal{L}_2$ are viewed as nonincreasing sequences of positive numbers tending to zero, but that we may have $\sum_{j=1}^\ty\ell_{1j}=+\ty$ or $\sum_{k=1}^\ty\ell_{2k}=+\ty$):
\begin{equation}\label{otimes}
\mathcal{L}_1\otimes\mathcal{L}_2:=(\ell_{1j}\ell_{2k})_{j,k\ge1}.
\end{equation}
By construction, the multiplicity of any $l\in \mathcal{L}_1\otimes\mathcal{L}_2$  is equal to the number of ordered pairs of $(\ell_{1j},\ell_{2k})$ in the Cartesian product $\mathcal{L}_1\times\mathcal{L}_2$ of multisets such that $l=\ell_{1j}\ell_{2k}$.
\end{defn}

\begin{example}\label{tensorE}
Let $\O_0:=B_1(0)$ be the open unit disk in the Euclidean plane $\eR^2$ and $A_0:=\pa\O_0$ the unit circle. Let $\mathcal L:=(\ell_j)_{j\ge0}$ be the Cantor string.\index{Cantor string|textbf} In other words, $\mathcal L$ is the multiset consisting of $l_k=3^{-k-1}$ with multiplicity $2^k$ for each $k\ge0$. As in Definition \ref{tensor}, we define the RFD $(A,\O)$ as the tensor product
\begin{equation}
(A,\O):=(A_0,\O_0)\otimes{\mathcal L}.
\end{equation} 
 Then,
\begin{equation}\label{last_ex}
\begin{aligned}
\zeta_{A,\O}(s)&=\sum_{k=0}^\ty2^{k}\zeta_{3^{-k-1}(A_0,\O_0)}(s)\\
&=3^{-s}\sum_{k=0}^\ty2^{k}3^{-sk}\zeta_{A_0,\O_0}(s)=\frac{3^{-s}\zeta_{A_0,\O_0}(s)}{1-2\cdot3^{-s}}\\
&=\frac{2\pi(s-2)3^{-s}}{s(s-1)(1-2\cdot3^{-s})},
\end{aligned}
\end{equation}
where in the last equality we have used Equation \eqref{AON} with $N=2$. 
It follows that $\zeta_{A,\O}$ has a meromorphic continuation to all of $\Ce$ given by the last expression of \eqref{last_ex}.
Therefore, we deduce that
\begin{equation}
\po(\zeta_{A,\O})=\{0\}\cup\Big(\log_32+\frac{2\pi}{\log3}\I\Ze\Big)\cup\{1\},\q\dim_{PC}(\tilde A,\tilde\O)=\{1\},
\end{equation}
and all the complex dimensions are simple.
The RFD $(A,\O)$ has a vertical sequence of equidistant complex dimensions (namely, $\{\log_32+\frac{2\pi}{\log3}\I k\}_{k\in\Ze}$), while there is only one principal complex dimension (namely, $s=1$), and it is simple. Using the Minkowski measurability criterion obtained in [LapRa\v Zu6, Theorem 4.2] (see also [LapRa\v Zu1, Theorem 5.4.20]), we conclude that the RFD $(A,\O)$ is Minkowski measurable.
\end{example}


\medskip

We point out, however, that one can also show 
(by using the results and techniques of [{LapRa\v Zu5}] and \cite[Chapter 5]{fzf}) 
that the RFD $(A,\O)$ is ({\em strictly}) {\em subcritically Minkowski nonmeasurable\index{subcritically Minkowski nonmeasurable RFD} in dimension $d:=\log_32$}, in a sense specified in the just mentioned references.
Heuristically, this means that it has {\em geometric oscillations of lower order} $d=\log_32$, but none of leading order $D=1$.
\medskip

We can easily modify the notion of relative fractal spray 
in Definition~\ref{spray} in order to deal with a finite collection of $K$ basic RFDs
(or generating RFDs) $(\pa\Omega_{01},\Omega_{01})$,\dots,$(\pa\Omega_{0K},\Omega_{0K})$, where $K$ is and integer $\geq 1$, similarly as in [{LapPo3}], \cite[Definition~13.2]{lapidusfrank12} (and
[{LapPe1--2}, {LapPeWi1--2}]). A slightly more general notion would consist in replacing $(\partial\O_0,\O_0)$ by any relative fractal drum $(A_0,\O_0)$; see \cite{fzf}.

\bigskip

It is important to stress that, from our point of view, the sets $\Omega_i$ in the definition of a relative fractal spray (Definition~\ref{spray})
do not have to be `densely packed'. In fact, in general, they cannot be `densely packed', as indicated by Example~\ref{spraye}$(c)$ below.
They can just be viewed as the union of the disjoint family $\{(\pa\Omega_i,\Omega_i)\}_{i\ge0}$ of RFDs in $\eR^N$. The corresponding disjoint 
family of open sets $\{\Omega_i\}_{i\ge0}$ can even be unbounded in $\eR^N$,
since its union does not have to be of finite $N$-dimensional Lebesgue measure. 

\medskip

The following simple lemma provides necessary and sufficient conditions for a relative fractal spray $(A,\O)$ to be such that $|\O|<\ty$.

\begin{lemma}\label{sprayv}
Assume that
$(A,\Omega):=\operatorname{Spray}(\Omega_0,(\lambda_j)_{j\ge0},(b_j)_{\ge0})$ in $\eR^N$ is a relative fractal spray. Then $|\O|<\ty$ if and only if $|\O_0|<\ty$ and
\begin{equation}\label{spraysum}
\sum_{j=0}^\ty b_j\g_j^N<\ty.
\end{equation}
In that case, we have
\begin{equation}\label{3.1.371/2}
|\O|=|\O_0|\sum_{j=0}^\ty b_j\g_j^N.
\end{equation}
In particular, the relative fractal drum $(A,\O)$ is well defined and
$\ov\dim_B(A,\O)\le N$.
\end{lemma}

\begin{proof} Let us prove the sufficiency part.
For $\O_j=\g_j\O_0$ we have $|\O_j|=|\g_j \O_0|=\g_j^N|\O_0|$, and therefore,
$$
|\O|=\sum_{j=0}^\ty|\O_j|=\sum_{j=0}^\ty b_j|\g_j \O_0|=|\O_0|\sum_{j=0}^\ty b_j\g_j^N.
$$
The proof of the necessity part is also easy and is therefore omitted.
\end{proof}


\medskip

\begin{example}\label{spraye}
Here, we provide a few simple examples of relative fractal sprays: 

\bigskip

$(a)$ The ternary Cantor set can be viewed as a relative fractal drum 
$$
(A,\Omega)=\operatorname{Spray}(\Omega_0,
1/3,2)
$$ 
(called the {\em Cantor relative fractal drum},\index{Cantor relative fractal drum|textbf}\index{relative fractal drum, RFD!Cantor|textbf} or the {\em relative Cantor fractal spray}),\index{relative Cantor fractal spray|textbf} generated by
$$
(\pa\Omega_0,\Omega_0)=(\{1/3,2/3\}\,,\,(1/3,2/3))
$$ 
as the  base  relative fractal drum, $\lambda=1/3$ and $b=2$. Its relative box dimension exists and is given by $D=\log_32$. 
Of course, this is just an example of ordinary fractal string, namely, the well-known Cantor string; see \cite[\S1.1.2]{lapidusfrank12}.
\bigskip

$(b)$ The Sierpi\'nski gasket can be viewed as a relative fractal drum 
(called the {\em Sierpi\'nski relative fractal drum},\label{sierpinski_drum}\index{Sierpi\'nski relative fractal drum (or spray)|textbf}\index{relative fractal drum, RFD!Sierpi\'nski|textbf}
 or {\em Sierpi\'nski relative fractal spray}), 
generated by $(\pa\Omega_0,\Omega_0)$
as the basic relative fractal drum,  where $\Omega_0$ is an open equilateral triangle of sides of length $1/2$,
$\lambda=1/2$ and $b=3$.
(See the left part of Figure \ref{Sierpinski_slika}.)
Its relative box dimension exists and is given by $D=\log_23$.
\bigskip

$(c)$ If $\Omega_0$ is any bounded open set in $\eR^2$ (say, an open disk), $\lambda=1/2$ and $b=3$, we obtain a fractal
spray $(A,\Omega)=\operatorname{Spray}(\Omega_0,1/2,3)$, in the sense of Definition~\ref{spray}. In Theorem~\ref{sprayz}, we shall see that if $\Omega_0$ has a Lipschitz boundary, then the set of poles of the relative zeta function of this fractal spray
(which is a relative fractal drum), as well as the multiplicities of the poles,
do not depend on the choice of $\Omega_0$. In this sense, examples $(b)$ and $(c)$ are equivalent. In particular, the box dimension
of the {\em generalized Sierpi\'nski relative fractal drum} is constant, and equal to $D=\log_23$.
\end{example}

\begin{figure}[t]
\begin{center}
\includegraphics[trim= 1cm 0cm 1cm 0cm,clip,width=5.8cm]{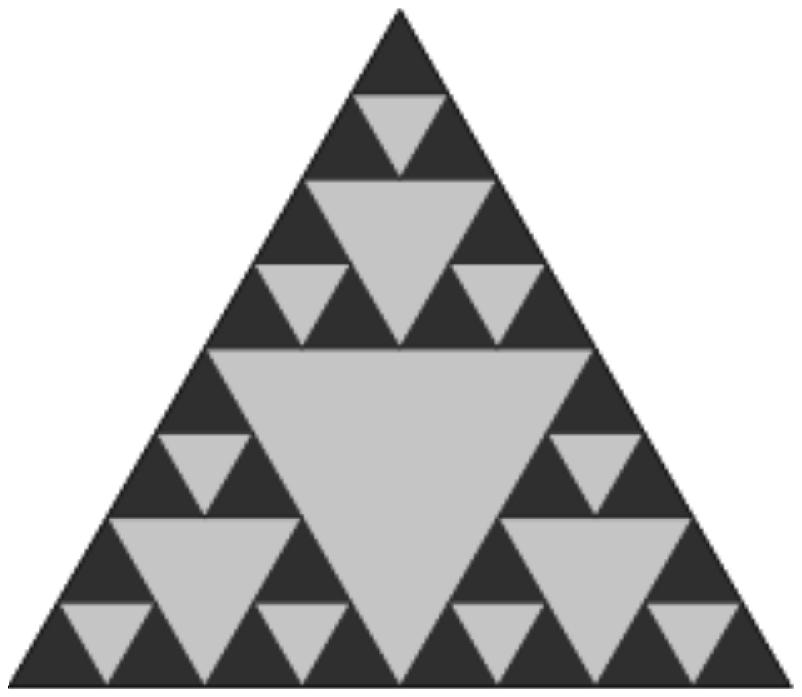}
\includegraphics[trim= 1cm 0cm 1cm 0cm,clip,width=5.8cm]{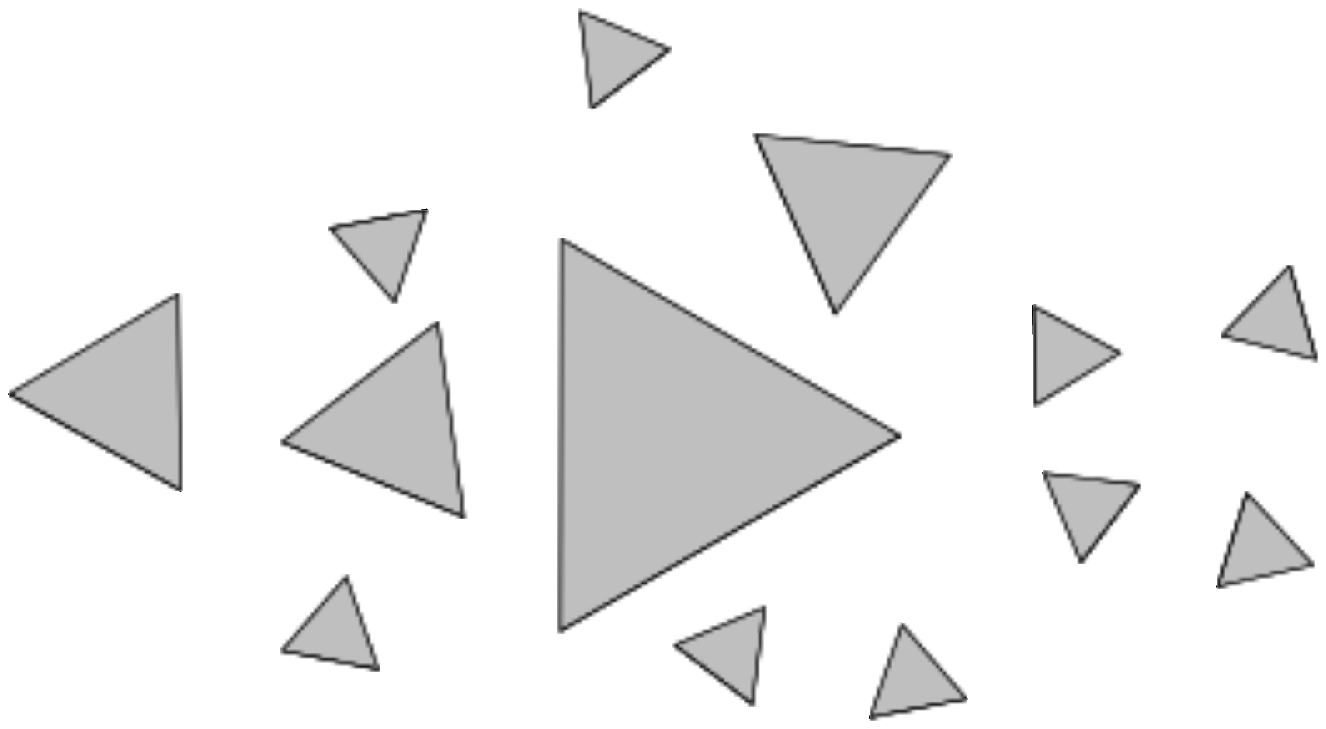}
\caption{\small Left: The Sierpi\'nski gasket $A$, viewed as a relative fractal drum $(A,\O)$, with $\O$ being the countable disjoint union of open triangles contained in the
unit triangle $\O_0$.
Right: An equivalent interpretation of the Sierpi\'nski gasket drum $(A,\O)$.
Here, $\O$ is a countable disjoint union of open equilateral triangles, and $A=\partial\O$.
(There are $3^j$ triangles with sides $2^{-j-1}$ in the union, $j\in\eN_0$.)
Both pictures depict the first three iterations of the construction. We can also view the standard Sierpi\'nski gasket $A$ as a relative fractal drum $(A,\O)$, in which $\O$ is just
the open unit triangle in the left picture.}
\label{Sierpinski_slika}
\end{center}
\end{figure}

In other words, the Sierpi\'nski gasket $(A,\Omega)=\operatorname{Spray}(\Omega_0,1/2,3)$, appearing in Example~\ref{spraye}$(b)$, can be viewed
 as {\em any} countable disjoint collection of open triangles in the plane (which can even be an unbounded collection) and their bounding triangles, of sizes $\lambda_j=2^{-j-1}$ and multiplicities $b_j=3^{j}$, $j\in\eN_0$, and not just as the standard disjoint collection of open triangles, densely packed inside the unit open triangle; see the right part of Figure~\ref{Sierpinski_slika}.

\medskip

By using the scaling property stated in Theorem~\ref{scaling}, it is easy to explicitly compute the distance zeta function of relative fractal sprays.
Note that the zeta function involves the Dirichlet series $f(s)=\sum_{j=0}^\ty b_j \lambda_j^s$. Theorem~\ref{sprayz} just below can be considered as an extension of Theorem~\ref{scaling}.

\begin{theorem}[Distance zeta function of relative fractal sprays]\label{sprayz}
Let 
$$
(A,\Omega)=\operatorname{Spray}(\Omega_0,(\lambda_j)_{j\ge0},(b_j)_{j\ge0})
$$ 
be a relative fractal spray in $\eR^N$, in the sense of Definition~\ref{spray}, and such that $|\O_0|<\ty$. Assume that condition \eqref{spraysum} of Lemma~\ref{sprayv} is satisfied; that is, $|\O|<\ty$. Let $\Omega$
be the $($countable, disjoint$)$ union of all the 
open sets appearing in the skeleton, corresponding to the fractal spray. $($In other words, $\O$ is the disjoint union of the open sets $\O_j$, each repeated with the multiplicity $b_j$ for $j\in\eN_0$.$)$ 
Let $f(s):=\sum_{j=0}^\ty b_j \lambda_j^{s}$. $($Note that according to~\eqref{spraysum}, this Dirichlet series\index{Dirichlet series} converges absolutely for $\re s\geq N$; hence, $D(f)\le N$.$)$ 
Then, for all $s\in\Ce$ with $\re s>\max\{\ov{\dim}_B(A,\O),D(f)\}$, the distance zeta function
of the relative fractal spray $(A,\Omega)$ is given by the {\rm factorization formula}
\begin{equation}\label{dirichlete}
\zeta_{A,\Omega}(s)=f(s)\cdot\zeta_{\pa\Omega_0,\Omega_0}(s),
\end{equation}
and
\begin{equation}\label{dirichletd}
\ov\dim_B(A,\Omega)=\max\{\ov\dim_B(\pa\Omega_0,\Omega_0),D(f)\}.
\end{equation}
\end{theorem}

\begin{proof} Clearly, it follows from~\eqref{spraysum} that $f(N)<\ty$. Hence, $D(f)\le N$; so that $\ov\dim_B(A,\O)\le N$.
Each open set of the skeleton of the relative fractal spray is obtained by a rigid motion of sets of the form $\lambda_j\Omega_0$, and for 
any fixed $j\in\eN_0$, there are precisely $b_j$
such sets.
Identity (\ref{dirichlete}) then follows immediately from Theorems~\ref{scaling} and \ref{unionz}. The remaining claims are easily derived by using this identity.  
\end{proof}

It follows from Definition~\ref{tensor} and relation (\ref{dirichlete})
that the distance zeta function of the tensor product is equal to the product of the zeta functions of its components:
\begin{equation}
\zeta_{(\pa\O_0,\O_0)\otimes\mathcal{L}}(s)=\zeta_{\mathcal{L}}(s)\cdot\zeta_{\pa\O_0,\O_0}(s).
\end{equation}
Equation (\ref{dirichletd}) can therefore be written as follows:
\begin{equation}
\ov\dim_B((\pa\O_0,\O_0)\otimes\mathcal{L})=\max\{\ov\dim (\pa\O_0,\O_0),\ov\dim_B\mathcal{L}\}.
\end{equation}

\medskip

\begin{defn}\label{zetas}
The Dirichlet series $f(s):=\sum_{j=1}^\ty b_j\g_j^s$ (or, more generally, its meromorphic extension to a connected open subset $U$ of $\Ce$, when it exists), is called the {\em scaling zeta function}\label{scalingzf}\index{scaling zeta function, $\zeta_{\mathfrak S}$}\index{fractal zeta function!scaling zeta function, $\zeta_{\mathfrak S}$|textbf} of the relative fractal spray $(A,\O)$ and is denoted by $\zeta_{\mathfrak S}(s)$; hence, the factorization formula \eqref{dirichlete} can also be rewritten as follows:
\begin{equation}
\zeta_{A,\O}(s)=\zeta_{\mathfrak S}(s)\cdot\zeta_{\pa\O_0,\O_0}(s).
\end{equation}
\end{defn}

\medskip

\begin{theorem}\label{sprayzc}
Assume that a relative fractal spray $(A,\O)=\operatorname{Spray}(\Omega_0,\lambda,b)$, as introduced at the end of Definition~\ref{spray}, is such that $|\O_0|<\ty$, $\lambda\in(0,1)$,  $b\ge2$ is an integer, and 
 $b\g^N<1$.
 Then, for $\re s>\max\{\ov{\dim}_B(\partial\O_0,\O_0),\log_{1/\lambda}b\})$, we have
\begin{equation}\label{10151/4c}
\zeta_{A,\Omega}(s)=\frac{\zeta_{\pa\Omega_0,\Omega_0}(s)}{1-b\lambda^s},
\end{equation}
and the lower bound for $\re s$ is optimal.
In particular, it is equal to $D(\zeta_{A,\O})$, and hence,
$$
\ov\dim_B(A,\O)=D(\zeta_{A,\O})=\max\{\ov{\dim}_B(\partial\O_0,\O_0),\log_{1/\lambda}b\}.
$$

If, in addition, $\O_0$ is bounded and has a Lipschitz boundary $\pa\O_0$ that can be described by finitely many Lipschitz charts, then $\dim_B(A,\Omega)$ exists and 
\begin{equation}\label{10151/2}
\dim_B(A,\Omega)=\max\{N-1,\log_{1/\lambda} b\}.
\end{equation}
Therefore, if we assume that $\log_{1/\lambda}b\in(N-1,N)$, then the set $\dim_{PC}(A,\O)=\po_c(\zeta_{A,\O})$ of principal complex dimensions of the relative fractal spray $(A,\O)$ is given by
\begin{equation}\label{rfs_pcd}
\dim_{PC}(A,\O)=\log_{1/\lambda}b+\frac{2\pi}{\log(1/\lambda)}{\I}\Ze.
\end{equation}
\end{theorem}

\begin{proof}  
If $\lambda_j:=\lambda$ and $b_j:=b^{j}$ for all $j\in\eN_0$, with $b\g^N<1$, then
$\sum_{j=0}^\ty b^{j}\g^{jN}=\frac{1}{1-b\lambda^N}<\ty$; so that $|\O|<\ty$, as desired. Identity (\ref{10151/4c}) follows immediately from (\ref{dirichlete}),
using the fact that for $\O_0$ with a Lipschitz boundary satisfying the stated assumption, we have $\dim_B(\pa\O_0,\O_0)=\dim_B\pa\O_0=N-1$ (this follows, for example, from \cite[Lemma~3]{zuzup2}; see also [Lap1]), together with the property of finite stability of the upper box dimension; see, e.g., \cite[p.\ 44]{falc}.
\end{proof}

\begin{example}
Here, we construct a relative fractal spray 
$$
(A,\Omega)=\operatorname{Spray}(\Omega_0,(\lambda_j)_{j\ge1},(b_j)_{\ge1})
$$ 
in $\eR^2$ such that $|\O_0|<\ty$, $b_j\equiv1$, $\sum_{j=1}^\ty \g_j^2<\ty$ 
(hence, $|\O|<\ty$ by Lemma~\ref{sprayv}), and such that the base set $\O_0$ is {\em unbounded}, as well as its boundary
$\pa\O_0$. Let $\O_0$ be any unbounded Borel set of finite $2$-dimensional Lebesgue
measure,
such that both $\O_0$ and $\pa\O_0$ are unbounded, and $\O_0$ is contained in the horizontal strip 
$$
V_1:=\{(x,y)\in\eR^2:0<y<1\}.
$$ 
We can explicitly construct such a set as follows: 
$$
\O_0:=\{(x,y)\in\eR^2:0<y<x^{-\a},\,\,x>1\},
$$ 
where $\a>1$, so that $|\O_0|<\ty$. 

Furthermore, let $(V_j)_{j\ge 1}$ be a countable, disjoint sequence of horizontal strips in the plane, defined for each $j\in\eN$ by
$V_j=V_1+(0,j)$, the Minkowski sum of $V_j$ and $(0,j)$. Let $(\g_j)_{j\ge1}$ be a sequence of real numbers in $(0,1)$ such that $\sum_{j=1}^\ty \g_j^2<\ty$.
It is clear that for any $\g_j$, $j\ge2$, the set $\g_j\O_0$ is congruent (up to a rigid motion) to the subset $\O_j:=\g_j\O_0+(0,j)$ of $V_j$.
 Then, the fractal spray 
$$
(A,\O)=\bigcup_{j=1}^\ty(\pa\O_j,\O_j)
$$ 
has the desired properties.
\end{example}

\section{Relative Sierpi\'nski sprays and their complex dimensions}\label{relative_other}
We provide two examples of relative fractal sprays, dealing with the {\em relative Sierpi\'nski gasket} and the {\em relative Sierpi\'nski carpet}, respectively. In the sequel, it will be useful to introduce the following definition.

\begin{defn}\label{congruentd} We say that {\em two given relative fractal drums $(A_1,\O_1)$ and $(A_2,\O_2)$ in $\eR^N$ are 
congruent}\index{congruent RFDs} if there exists an
isometry $f:\eR^N\to\eR^N$ such that $A_2=f(A_1)$ and $\O_2=f(\O_1)$. 
\end{defn}

It is easy to see that the congruence of RFDs is an equivalence relation.  

The following lemma states, in particular, that any two congruent RFDs have equal distance zeta functions. We leave its proof as a simple exercise for the interested reader.

\begin{lemma}\label{congruentl}
Let $(A_1,\O_1)$ and $(A_2,\O_2)$ be two congruent RFDs in $\eR^N$. Then, for any $r\in\eR$, we have
\begin{equation}
\M_*^r(A_1,\O_1)=\M_*^r(A_2,\O_2),\q\M^{*r}(A_1,\O_1)=\M^{*r}(A_2,\O_2)
\end{equation}
and
\begin{equation}\label{ovD12}
\underline\dim_B(A_1,\O_1)=\underline\dim_B(A_2,\O_2),\q\ov\dim_B(A_1,\O_1)=\ov\dim_B(A_2,\O_2)=:\ov D.
\end{equation}
Furthermore, 
\begin{equation}\label{zeta12c}
\zeta_{A_1,\O_1}(s)=\zeta_{A_2,\O_2}(s)
\end{equation}
for any $s\in\Ce$ with $\re s>\ov\dim_B(A_1,\O_1)$.
\end{lemma}

It follows from \eqref{zeta12c} that under the hypotheses of Lemma \ref{congruentl} and given a connected open set $U\stq\Ce$ (containing the critical line $\{\re s=\ov D\}$ of the RFDs $(A_1,\O_1)$ and $(A_2,\O_2)$; see Equation~\eqref{ovD12} above),
$\zeta_{A_1,\O_1}$ and $\zeta_{A_2,\O_2}$ have the exact same mermomorphic continuation to $U$, and therefore the same poles in $U$
and associated residues (or more generally, principal parts in the case of multiple poles). In particular, two congruent RFDs have the same (visible) complex dimensions.

\begin{figure}
\psset{unit=1.3}
\begin{pspicture}(-2.6,-1)(6,4)
\pstTriangle[PointName=none](-2,0){A}(2,0){B}(0,3.46){C}
\pstMiddleAB[PointName=none]{A}{B}{C'}
\pstMiddleAB[PointName=none]{C}{A}{B'}
\pstMiddleAB[PointName=none]{B}{C}{A'}
\pstLineAB[linestyle=dashed]{A}{A'}
\pstLineAB[linestyle=dashed]{B}{B'}
\pstLineAB[linestyle=dashed]{C}{C'}
\put(-2.1,-0.4){\small$0$}
\put(-0.2,-0.4){\small$1/4$}
\put(1.9,-0.4){\small$1/2$}
\put(1,2){\small$(\pa\O_0,\O_0)$}
\put(-0.5,2){\small$\O_0$}
\put(-1.9,1.3){\small$\pa\O_0$}

\pstGeonode[PointName=none](3,0){D}(5,0){E}(5,1.153){F}
\pstLineAB{D}{E}
\pstLineAB[linestyle=dashed]{E}{F}
\pstLineAB[linestyle=dashed]{D}{F}
\put(2.9,-0.4){\small$0$}
\put(4.9,-0.4){\small$1/4$}
\put(3.5,1.1){\small$y=x/\sqrt3$}
\put(5,0){\vector(1,0){0.5}}
\put(3,0){\vector(0,1){1.7}}
\put(5.5,-0.3){\small$x$}
\put(2.8,1.7){\small$y$}
\put(4,-0.4){\small$A'$}
\put(4.4,0.4){\small$\O'$}
\put(5.2,0.6){\small$(A',\O')$}
\end{pspicture}
\caption{\small On the left is depicted the base relative fractal drum $(\pa\O_0,\O_0)$ of the relative Sierpi\'nski gasket, where $\O_0$ is the associated (open) equilateral triangle with sides $1/2$. It can be viewed as the
(disjoint) union of six RFDs, all of which are congruent to the relative fractal drum $(A',\O')$ on the right.
This figure explains Equation \eqref{sierpinski6} appearing in Example \ref{6.15}; see Lemma \ref{congruentl}.}
\label{sierpinski_triangle}
\end{figure}
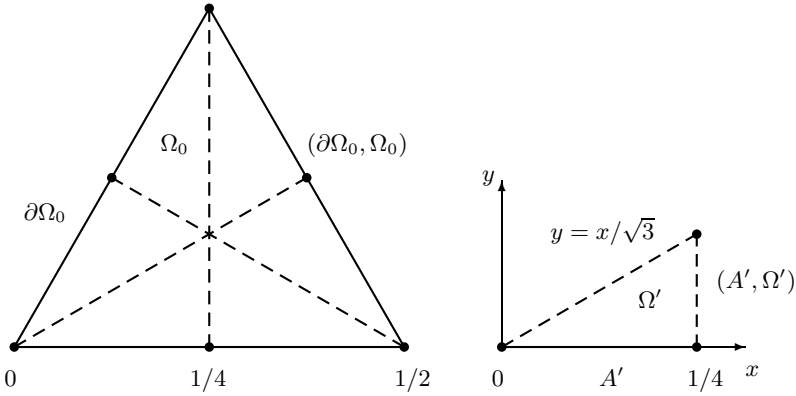

\begin{example} ({\em Relative Sierpi\'nski gasket}).\label{6.15}\index{relative Sierpi\'nski!gasket|textbf}
Let $A$ be the Sierpi\'nski gasket in $\eR^2$, the outer boundary of which is an equilateral triangle with unit sides. 
Consider the countable family of all open triangles in the standard construction of the gasket. (Namely, these are the open triangles which are deleted at each stage of the construction.)
If $\O$ is the largest open triangle (with unit sides),
then the {\em relative Sierpi\'nski gasket}\index{relative Sierpi\'nski gasket|textbf} is defined as the ordered pair $(A,\O)$.
The distance zeta function $\zeta_{A,\O}$ of the relative Sierpi\'nski gasket $(A,\O)$ can be
computed as the distance zeta function of the following relative fractal spray (see the end of Definition \ref{spray}):
$$
\operatorname{Spray}(\O_0, \g=1/2,b=3),
$$
where $\O_0$ is the first deleted open triangle with sides $1/2$. It suffices to apply Equation \eqref{10151/4c} from Theorem \ref{sprayzc}.
Decomposing $\O_0$ into the union of six congruent right triangles (determined by the heights of the triangle $\O_0$, see Figure \ref{sierpinski_triangle}) with disjoint interiors, we have that 
\begin{equation}\label{sierpinski6}
\begin{aligned}
\zeta_{\pa\O_0,\O_0}(s)&=6\,\zeta_{A',\O'}(s)=6\,\int_{\O'}d((x,y),A')^{s-2}\D x\,\D y\\
&=6\int_0^{1/4}\D x\int_0^{x/\sqrt3} y^{s-2}\D y=6\frac{(\sqrt3)^{1-s}2^{-s}}{s(s-1)},
\end{aligned}
\end{equation}
for all $s\in\Ce$ with $\re s>1$.
Using Equation \eqref{10151/4c} and appealing to Lemma \ref{congruentl}, we deduce that the distance zeta function of the relative Sierpi\'nski gasket $(A,\O)$ satisfies
\begin{equation}\label{gaszeta}
\zeta_{A,\O}(s)=\frac{6(\sqrt3)^{1-s}2^{-s}}{s(s-1)(1-3\cdot2^{-s})}\sim\frac1{1-3\cdot2^{-s}},
\end{equation}
where the equality holds for all $s\in\Ce$ with $\re s>\log_23$ and the equivalence $\sim$ holds in the sense of Definition \ref{equ}.
Therefore, by the principle of analytic continuation, it follows that $\zeta_{A,\O}$ 
has a meromorphic extension to the entire complex plane, given by the same closed form as in Equation \eqref{gaszeta}.
More specifically,
\begin{equation}\label{gaszeta2}
\zeta_{A,\O}(s)=\frac{6(\sqrt3)^{1-s}2^{-s}}{s(s-1)(1-3\cdot2^{-s})},\q\mbox{for all}\q s\in\Ce.
\end{equation}
Hence, the set of all of the complex dimensions (in $\Ce$) of the relative Sierpi\'nski gasket is given by
\begin{equation}\label{gaszeta1}
\po(\zeta_{A,\Omega})=\Big(\log_23+\frac{2\pi}{\log2}{\I}\Ze\Big)\cup\{0,1\}.
\end{equation}
Each of these complex dimensions in \eqref{gaszeta1} is simple (i.e., is a simple pole of $\zeta_{A,\O}$). Note that here, $\{0,1\}$ can be interpreted as the set of {\em integer dimensions} of $A$, in the sense of [{LapPe2--3}] and [{LapPeWi1}].
In particular, we deduce from \eqref{gaszeta1} that $D(\zeta_{A,\O})=\log_23$, and we thus recover a well-known result. Namely, the set 
$\dim_{PC} (A,\O):=\po_c(\zeta_{A,\Omega})$ of principal complex dimensions of the relative Sierpi\'nski gasket $(A,\O)$ is given by 
\begin{equation}\label{gaszetapc}
\dim_{PC} (A,\O)=\log_23+\mathbf{p}{\I}\Ze,
\end{equation} 
where $\mathbf{p}={2\pi}/{\log 2}$ is the oscillatory period\index{oscillatory period!of the Sierpi\'nski gasket} of the Sierpi\'nski gasket; see \cite[\S6.6.1]{lapidusfrank12}. 
\medskip

Note, however, that in 
[{Lap-vFr1--3}], 
the complex dimensions are obtained in a completely different manner (via an associated spectral zeta function) and not geometrically. 
In addition, all of the complex dimensions of the Sierpi\'nski gasket $A$ are shown to be principal (that is, to be located on the vertical line 
$\re s=\log_23=\dim_BA$), a conclusion which is slightly different from (\ref{gaszeta1}) above.\footnote{Analogously, for a fractal string RFD $(A,\O)$, we have $\po(\zeta_{A,\O})=\po(\zeta_{\mathcal{L}})\cup\{0\}$; here, of course, $N=1$ instead of $N=2$.}
We also refer to \cite{ChrIvLa} and [{LapSar}], as well as to
[{LapPe2--3}] 
and
[{LapPeWi1--2}], for different approaches (via a spectral zeta function associated to a suitable geometric Dirac operator and via a 
self-similar tiling associated with $A$, respectively) leading to the same conclusion.

In light of \eqref{gaszeta2}, the residue of the distance zeta function $\zeta_{A,\O}$ of the relative Sierpi\'nski gasket computed at any principal pole $s_k:=\log_23+\mathbf{p}k\I$, $k\in\Ze$, is given by
$$
\res(\zeta_{A,\O},s_k)=\frac{6(\sqrt3)^{1-s_k}}{2^{s_k}(\log2)s_k(s_k-1)}.
$$
In particular,
$$
|\res(\zeta_{A,\O},s_k)|\sim \frac{6(\sqrt3)^{1-D}}{D\,2^D\log2}\,k^{-2}\q\mbox{as\q$k\to\pm\ty$,}
$$
where $D:=\log_23$.
\end{example}\label{6.15end}

The following proposition shows that the relative Sierpi\'nski gasket can be viewed as the relative fractal spray generated by the relative fractal drum $(A',\O')$ appearing on the right-hand side of Figure \ref{sierpinski_triangle}.

\begin{prop}[{Relative Sierpi\'nski gasket}]\label{s_gasket} Let $(A',\O')$ be the relative fractal drum defined on the right part of  Figure \ref{sierpinski_triangle}. Let $(A,\O)$ be the relative fractal spray generated by the base relative fractal drum $(A',\O')$, with scaling ratio $\g=1/2$ and with multiplicities $m_k=6\cdot 3^{k-1}$, for any positive integer~$k$:
\begin{equation}
(A,\O)=\operatorname{Spray}((A',\O'),\,\,\g=1/2,\,\,m_k=6\cdot 3^{k-1}\,\,\mbox{\rm for $k\in\eN$}).
\end{equation}
 $($Note that we assume here that the base relative fractal drum $(A',\O')$ has a multiplicity equal to $8$.$)$ Then, the relative distance zeta function of the relative fractal spray $(A,\O)$ coincides with the relative distance zeta function of the relative Sierpi\'nski gasket; see Equation \eqref{gaszeta2}.
\end{prop}

\begin{example} ({\em Inhomogeneous Sierpi\'nski $N$-gasket RFD}).\label{Ngasket}\index{inhomogeneous!Sierpi\'nski $N$-gasket RFD}\index{Sierpi\'nski $N$-gasket RFD!inhomogeneous}
The usual Sierpi\'nski gasket is contained in the unit triangle in the plane. Its analog in $\eR^3$, which we call the {\em inhomogeneous Sierpi\'nski $3$-gasket} or {\em inhomogeneous tetrahedral gasket},\index{tetrahedral inhomogeneous gasket, $A_3$|textbf} and denote by $A_3$, is obtained by deleting the middle open octahedron (from the initial compact, convex unit tetrahedron), defined as the interior of the convex hull of the midpoints of each of the six edges of the initial tetrahedron (thus, four sub-tetrahedrons are left after the first step), and so on.
Such sets, along with their higher-dimensional counterparts, are analogous to, but not identical with, the (homogeneous) self-similar $N$-gaskets discussed, for example, in \cite{KiLa}.

\begin{figure}[t]
\begin{center}
\includegraphics[width=9cm]{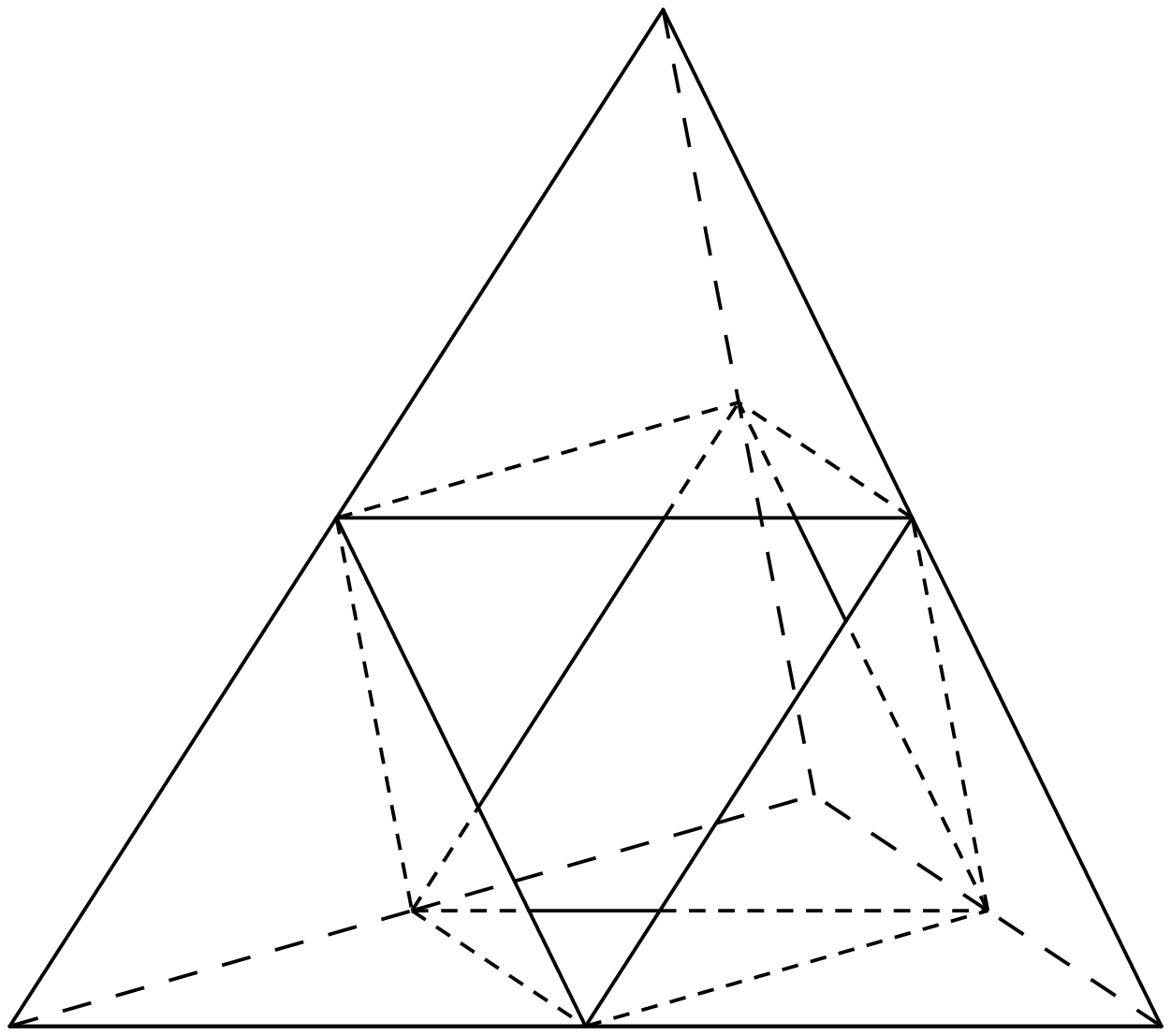}
\caption{\small The open octahedron $\O_{3,0}$ inscribed into the largest (compact) tetrahedron $\O_3$, surrounded with $4$ smaller (compact) tetrahedra scaled by the factor $1/2$. Each of them contains analogous scaled open octahedra, etc. The countable family of all open octahedra (viewed jointly with their boundaries) constitutes the tetrahedral gasket RFD or the Sierpi\'nski $3$-gasket RFD. The complement of the union of all open octahedra, with respect to the initial tetrahedron $\O_3$, is called the {\em inhomogeneous Sierpi\'nski $3$-gasket}.\newline
\hbox{}\q Unlike the classic Sierpi\'nski $3$-gasket (also known as the Sierpi\'nski pyramid or tetrahedron) $S_3$, which is a (homogeneous or standard) self-similar set in $\eR^3$ and satisfies the usual fixed point equation, $S=\cup_{j=1}^4\Phi_j(S)$, where $\{\Phi_j\}_{j=1}^4$ are suitable contractive similitudes of $\eR^3$ with respective fixed points $\{P_j\}_{j=1}^4$ and scaling ratios $\{r_j\}_{j=1}^4$ of common value $1/2$, the inhomogeneous Sierpi\'nski $3$-gasket $A_3$ is {\em not} a self-similar set.
Instead, it is an {\em inhomogeneous self-similar set} (in the sense of \cite{bd}, see also Equation \eqref{KB} below and the discussion surrounding it).
More specifically, $A:=A_3$ satisfies the following {\em inhomogeneous}\index{inhomogeneous!fixed point equation} 
fixed point equation (of which it is the unique solution in the class of all nonempty compact subsets of $\eR^3$),
$
A=\bigcup_{j=1}^4\Phi_j(A)\cup B,
$
where $B$ is the boundary of the first octahedron $\Omega_{3,0}$ (in fact, $B$ can simply be taken as the union of four middle triangles on the boundary of the outer tetrahedron~$\O_3$).}\label{octahedron}
\end{center}\index{inhomogeneous!Sierpi\'nski $3$-gasket}\index{inhomogeneous!self-similar set|textbf}\index{fixed point equation}
\end{figure}

More generally, for any integer $N\ge2$, the {\em inhomogeneous Sierpi\'nski $N$-gasket} $A_N$,\label{inh-gasket} contained in $\eR^N$, can be defined as follows. 
Let $V_N:=\{P_1,P_2,\dots,P_{N+1}\}$ be a set of $N$ points in $\eR^N$ such that the mutual distance of any two points from the set is equal to $1$.

The set $V_N$, where $N\ge2$, with the indicated property, can be constructed inductively as follows. For $N=2$, we take $V_2$ to be the set of vertices of any unit triangle in $\eR^2$.
We then reason by induction. Given $N\ge2$, we assume that
 the set $V_N$ of $N+1$ points in $R^N$ has been constructed. Note that the set $V_N$ is contained in a sphere, whose center is denoted by $O$. Let us consider the line of $\eR^{N+1}=\eR^N\times\eR$ through the point $O$ and perpendicular to the hyperplane $\eR^N=\eR^N\times\{0\}$ in $\eR^{N+1}$. There exists a unique point $P_{N+2}$ in the half-plane $\{x_{N+1}>0\}$ of $\eR^{N+1}$, which is a unit distance from all of the $N$ vertices of $V_N$. (Here, we identify $V_N$ with $V_N\times\{0\}\st\eR^{N+1}$.) We then define $V_{N+1}$ by $V_{N+1}:=V_N\cup\{P_{N+2}\}$.

Let us define $\O_N$ as the convex hull of the set $V_N$. As usual, we call it the {\em $N$-simplex}.\index{simplex@$N$-simplex|textbf} Let $\O_{N,0}$, called the {\em $N$-plex},\index{plex@$N$-plex|textbf} be the open set defined as the interior of the convex hull of the set of midpoints of all of the $\binom {N+1}2$ edges of the $N$-simplex $\O_N$. [For example, for $N=2$, the set $\O_{2,0}$ (that is, the $2$-plex) is an open equilateral triangle in $\eR^2$ of side lengths equal to $1/2$, while for $N=3$, the set $\O_{3,0}$ (that is, the $3$-plex) is an open octahedron in $\eR^3$ of side lengths equal to $1/2$.] The set $\ov\O_N\setminus\O_{N,0}$ is equal to the union of $N+1$ congruent $N$-simplices with disjoint interiors, having all of their side lengths equal to $1/2$. This is the first step of the construction. We proceed analogously with each of the $N+1$ compact $N$-simplices. The compact set $A_N$ obtained in this way is called the {\em inhomogeneous Sierpi\'nski $N$-gasket}.\index{inhomogeneous!Sierpi\'nski $N$-gasket|textbf} It can be identified with the relative fractal spray $(A_N,\O_N)$ in $\eR^N$, called the {\em inhomogeneous Sierpi\'nski $N$-gasket RFD}\index{Sierpi\'nski $N$-gasket RFD!inhomogeneous|textbf} (and, in short, the {\em inhomogeneous $N$-gasket\label{inh-gasketRFD} RFD}),\index{inhomogeneous!$N$-gasket RFD|textbf} defined by
\begin{equation}
(A_N,\O_N)=\operatorname{Spray}\,((\pa\O_{N,0},\O_{N,0}),\g=1/2,b=N+1).
\end{equation}
(See the end of Definition~\ref{spray}.)
According to Theorem \ref{sprayzc}, we have
\begin{equation}\label{Ngasketzf}
\zeta_{A_N,\O_N}(s)=\zeta_{\mathfrak S}(s)\cdot\zeta_{\pa\O_{N,0},\O_{N,0}}(s),
\end{equation}
where the {\em scaling zeta function}\index{scaling zeta function, $\zeta_{\mathfrak S}$}\index{fractal zeta function!scaling zeta function, $\zeta_{\mathfrak S}$} $\zeta_{\mathfrak S}(s)$ of the $N$-gasket RFD is given
for all $s\in\Ce$ such that $\re s>\log_2(N+1)$ by
\begin{equation}\label{Ngasketzf2}
\zeta_{\mathfrak S}(s)=\sum_{k=0}^\ty(N+1)^{k}(2^{-k})^s=\frac1{1-(N+1)2^{-s}}.
\end{equation}
Upon analytic continuation, it follows that $\zeta_{\mathfrak S}$ can be meromorphically continued to the whole of $\Ce$ and is given by
\begin{equation}\label{Ngasketzf3}
\zeta_{\mathfrak S}(s)=\frac1{1-(N+1)2^{-s}},\mbox{\q for all $s\in\Ce$.}
\end{equation}
Since (by \eqref{Ngasketzf3}) the set of poles of $\zeta_{\mathfrak S}$ is given by
\begin{equation}
\po(\zeta_{\mathfrak S})=\log_2(N+1)+\frac{2\pi}{\log2}\,\I\Ze
\end{equation}
and the set of poles of the distance zeta function of the {\em relative $N$-plex} $(\pa\O_{N,0},\O_{N,0})$ is given by
\begin{equation}\label{AN0p}
\po(\zeta_{\pa\O_{N,0},\O_{N,0}})=\{0,1,\dots,N-1\},
\end{equation}
and $\zeta_{\pa\O_{N,0},\O_{N,0}}(s)\ne0$ for all $s\in\Ce\setminus \po(\zeta_{\pa\O_{N,0},\O_{N,0}})$,
we conclude that the set of poles (complex dimensions) of the {\em relative Sierpi\'nski $N$-gasket}\index{relative Sierpi\'nski!gasket@$N$-gasket|textbf}  $(A_N,\O_N)$ is given by
\begin{equation}\label{poNg}
\po(\zeta_{A_N,\O_N})=\{0,1,\dots,N-1\}\cup\Big\{\log_2(N+1)+\frac{2\pi}{\log2}\,\I\Ze\Big\},
\end{equation}
where each complex dimension is simple.
In particular, the set of principal complex dimensions of the RFD $(A_N,\O_N)$ is given by\footnote{Recall that, by definition, $\dim_{PC}(A_N,\O_N)=\po_c(\zeta_{A_N,\O_N})$.}
\begin{equation}\label{ANON}
\dim_{PC}(A_N,\O_N)=
\begin{cases}
\log_23+\frac{2\pi}{\log2}\,\I\Ze&\mbox{ for $N=2$,}\\
2+\frac{2\pi}{\log2}\,\I\Ze&\mbox{ for $N=3$,}\\
\{N-1\}&\mbox{ for $N\ge4$,}
\end{cases}
\end{equation}
and
\begin{equation}\label{givvv}
\ov\dim_B(A_N,\O_N)=
\begin{cases}
\log_23&\mbox{ for $N=2$,}\\
N-1&\mbox{ for $N\ge3$,}
\end{cases}
\end{equation}
which extends the well-known results for $N=2$ and $3$, corresponding to the usual Sierpi\'nski gasket in $\eR^2$ and the tetrahedral gasket in $\eR^3$, respectively. (Namely, their respective relative box dimensions are equal to $\log_23$ and $2$).

It can be shown that in this case, $\dim_B(A_N,\O_N)$ and $\dim_BA_N$ exist and
\begin{equation}\label{BH}
\dim_B(A_N,\O_N)=\dim_BA_N=\dim_HA_N,
\end{equation}
as given by the right-hand side of \eqref{givvv}, where (as before) $\dim_H(\,\cdot\,)$ denotes the Hausdorff dimension.\index{fractal dimension!Hausdorff dimension} Furthermore, it is easy to see that $\dim_{PC}(A_N,\O_N)=\dim_{PC}A_N$.

 The relative distance zeta function $\zeta_{\pa\O_{N,0},\O_{N,0}}$ of the $N$-plex RFD can be explicitly computed as follows, in the case when $N=3$. It is easy to see that the octahedral RFD $(\pa\O_{3,0},\O_{3,0})$ can be identified with sixteen copies of disjoint RFDs, each of which is congruent to the pyramidal RFD $(T,\O')$ in $\eR^3$, where $\O'$  is the open (irregular) pyramid with vertices at $O(0,0,0)$, $A(1/4,0,0)$, $B(1/4,1/4,0)$ and $C(0,0,1/(2\sqrt2))$, while the triangle $T=\operatorname{conv}\,(A,B,C)$ is a face of the pyramid. Since for any $(x,y,z)\in\O'$ we have
\begin{equation}
d((x,y,z),T)=\frac1{\sqrt3}\left(\frac{1}{2\sqrt2}-\sqrt2 x-z\right),
\end{equation}
we deduce that
\begin{equation}\label{iiint}
\begin{aligned}
\zeta_{\pa\O_{3,0},\O_{3,0}}(s)&=16\zeta_{T,\O'}(s)\\
&=16\,\iiint_{\O'}d((x,y,z),T)^{s-3}\D x\,\D y\,\D z\\
&=16\int_{0}^{1/4}\D x\int_{0}^x\D y\int_{0}^{\frac1{2\sqrt2}-\sqrt2\,x} \left(\frac{\frac{1}{2\sqrt2}-\sqrt2 x-z}{\sqrt3}\right)^{s-3} \D z.
\end{aligned}
\end{equation}
Evaluating the last integral in \eqref{iiint}, we obtain by a direct computation that
\begin{equation}\label{30}
\begin{aligned}
\zeta_{\pa\O_{3,0},\O_{3,0}}(s)&=16\frac{(\sqrt3)^{3-s}}{s-2}\int_0^{1/4}\Big(\frac1{2\sqrt2}-\sqrt2\,x\Big)^{s-2}x\,\D x\\
&=8\frac{(\sqrt3)^{3-s}}{s-2}\int_0^{1/(2\sqrt2)}u^{s-2}\Big(\frac1{2\sqrt2}-u\Big)\,\D u\\
&=\frac{8(\sqrt3)^{3-s}(2\sqrt2)^{-s}}{s(s-1)(s-2)},
\end{aligned}
\end{equation}
for any complex number $s$ such that $\re s>2$. Therefore, we deduce from \eqref{Ngasketzf} that the distance zeta function of the thetrahedral RFD in $\eR^3$ can be meromorphically extended to the whole complex plane and is given for all $s\in\Ce$ by
\begin{equation}\label{zeta3sg}
\zeta_{A_3,\O_3}(s)=\frac{8(\sqrt3)^{3-s}(2\sqrt2)^{-s}}{s(s-1)(s-2)(1-4\cdot2^{-s})}.
\end{equation}
It is worth noting that $s=2$ is the only pole of order $2$, since $s=2$ is the simple pole of both $(s-2)^{-1}$ and $(1-4\cdot2^{-s})^{-1}$. More specifically, since the derivative of $1-4\cdot2^{-s}$ computed at $s=2$ is nonzero (and, in fact, is equal to $4\log2$), then $s=2$ is a simple zero of $1-4\cdot2^{-s}$; that is, it is a simple pole of $1/(1-4\cdot2^{-s})$. 

Moreover, it immediately follows from Equation \eqref{zeta3sg} that  
\begin{equation}
\zeta_{A_3,\O_3}(s)\sim\frac1{(s-2)(1-4\cdot2^{-s})}.
\end{equation}
In particular, as we have already seen in Equation \eqref{ANON} (recall that $N:=3$ here), we have
\begin{equation}
\dim_{PC}(A_3,\O_3)=2+\frac{2\pi}{\log2}\,\I\Ze.
\end{equation}
Since $D=2$ is a simple pole of both $1/(s-2)$ and $1/(1-4\cdot2^{-s})$, we conclude that $D=2$ is the only complex dimension of order two of the RFD $(A_3,\O_3)$.
Consequently, the case of the relative Sierpi\'nski $3$-gasket $(A_3,\O_3)$ reveals a new phenomenon: its relative box dimension $D=2$ is a complex dimension of order (i.e., multiplicity) two, while all the other complex dimensions of the relative Sierpi\'nski $3$-gasket (including the double sequence of nonreal complex dimensions on the critical line of convergence $\{\re s=2\}$) are simple.
\smallskip

By using arguments similar to those used when $N=3$, one can show that for any $N\ge3$, the distance zeta function of the relative $N$-plex $(\pa\O_{N,0},\O_{N,0})$ is of the form
\begin{equation}\label{zetaNsg}
\zeta_{\pa\O_{N,0},\O_{N,0}}(s)=\frac{g(s)}{s(s-1)\dots(s-(N-1))},
\end{equation}
where $g(s)$ is a nonvanishing entire function.
In the special case when $N=3$, this is in accordance with Equation \eqref{30} above.
Therefore, from Equations \eqref{Ngasketzf} and \eqref{Ngasketzf3} above, we conclude that
\begin{equation}\label{zetaNsg2}
\zeta_{A_N,\O_N}(s)=\frac{g(s)}{s(s-1)\dots(s-(N-1))(1-(N+1)2^{-s})}.
\end{equation}
This extends Equation \eqref{zeta3sg} to any $N\ge3$.

In the case when $N\ge4$, $D=N-1$ is the only principal complex dimension of the relative Sierpi\'nski $N$-gasket. (Indeed, for $N\ge4$ we have that $\log_2(N+1)<N-1$ (i.e., $N+1<2^{N-1}$), which can be easily proved, for example, by using mathematical induction on~$N$.) Also, all the other complex dimensions are simple. Furthermore,
we immediately deduce from Equation \eqref{zetaNsg2} that
\begin{equation}
\zeta_{A_N,\O_N}(s)\sim\frac1{s-(N-1)}.
\end{equation}

Moreover, if $N\ge 4$ is of the form $N=2^k-1$ for some integer $k\ge3$, then $s=k$ (note that it is smaller than $D=N-1$) is the only complex dimension of order two (since it is a simple pole of both $(s-k)^{-1}$ and $(1-(N+1)2^{-s})^{-1}$), while all the other complex dimensions are simple.

On the other hand, if $N\ge 4$ is not of the form $N=2^k-1$ for any integer $k\ge3$, then all of the complex dimensions of the relative Sierpi\'nski $N$-gasket are simple.

Roughly speaking, in the case when $N=3$, the fact that $s=2$ has multiplicity two can be explained geometrically as follows: firstly, $s=2$ is a simple pole arising from the self-similarity of the RFD $(A_3,\O_3)$,\footnote{Indeed, the similarity dimension of the $3$-gasket $A_3$ is equal to $2$.} while at the same time, $s=2$ is a simple pole arising from the geometry of the boundary of the first (deleted) octahedron, which is also 2-dimensional.

In the case of the ordinary Sierpi\'nski gasket, i.e, of the relative Sierpi\'nski $2$-gasket, the value of $s=\log_2 3$ (which is the simple pole arising from the self-similarity of $(A_2,\O_2)$) is strictly larger than the dimension $s=1$ of the boundary of the deleted triangle (i.e., of the $2$-plex $\O_{2,0}$). Moreover, the relative Sierpi\'nski $2$-gasket is Minkowski nondegenerate and Minkowski nonmeasurable, while the relative Sierpi\'nski $3$-gasket is Minkowski degenerate, with its $2$-dimensional Minkowski content being equal to $+\ty$.

On the other hand, when $N\ge4$, the dimension $N-1$ of the boundary of the $N$-plex $\O_{N,0}$ is larger than the similarity dimension $\log_2(N+1)$ arising from ``fractality''. Since $D=\log_2(N+1)$ is the only complex dimension on the critical line (and it is simple), we conclude that for $N\ge4$, the RFD $(A_N,\O_N)$ is Minkowski measurable (see \cite{cras2}). Thus, the case when $N=3$ is indeed very special among all relative Sierpi\'nski $N$-gaskets.
\end{example}\label{Ngasketend}

\medskip

We refer the interested reader to \cite{cras2} and \cite{mm} (as well as to the relevant part of \cite[\S5.5]{fzf}) for a detailed discussion of the property of Minkowski measurability (or of Minkowski nonmeasurability) of the $N$-gasket RFD $(A_N,\O_N)$, for any $N\ge2$ and for the corresponding fractal tube formulas.\index{fractal tube formula} Let us simply mention here that for $N=3$, a suitable gauge function can be found with respect to which $A_3$ is not only Minkowski nondegenerate but is also Minkowki measurable. (Note that for $N\ne3$, $A_N$ is Minkowski nondegenerate in the usual sense, that is, relative to the trivial gauge function obeying the standard power law.)

\medskip

Let $\s_0$ be the common {\em similarity dimension}\index{similarity dimension} of the inhomogneous Sierpi\'nski $N$-gasket $A_N$,
 the relative Sierpi\'nski $N$-gasket $(A_N,\O_N)$ (where the latter is viewed as a self-similar fractal spray or RFD)
and the classic Sierpi\'nski $N$-gasket $S_N$ (to be discussed below). 
Since the corresponding scaling ratios $\{r_j\}_{j=1}^{N+1}$ satisfy $r_1=\dots=r_{N+1}=1/2$, 
the similarity dimension $\s_0$, defined as being the unique real solution $s$ of the Moran equation $\sum_{j=1}^{N+1}r_j^s=1$ (i.e., here, $2^s=N+1$, $s\in\eR$), is given by
\begin{equation}
\s_0=\log_2(N+1).
\end{equation}
In light of Equation \eqref{dirichletd} and since $\dim_B(A_{N,0},\O_{N,0})=N-1$ (by Equation \eqref{AN0p}),
we see that $\dim_B(A_N,\O_N)=\s_0$ for $N=2$ or $N=3$, that
\begin{equation}\label{s01}
\s_0=\dim_B(A_N,\O_N)>\dim_B(A_{N,0},\O_{N,0})
\end{equation}
for $N=2$,
\begin{equation}\label{s02}
\s_0=\dim_B(A_N,\O_N)=\dim_B(A_{N,0},\O_{N,0})\,\,\,(=2)
\end{equation}
for $N=3$, whereas for every $N\ge4$, we have that 
\begin{equation}\label{s03}
\s_0=\log_2(N+1)<\dim_B(A_N,\O_N)=\dim_B(A_{N,0},\O_{N,0})=N-1.
\end{equation}
(Recall from \eqref{BH} that $\dim_BA_N=\dim_B(A_N,\O_N)$.)
On the other hand, if $S_N$ denotes the {\em classic Sierpi\'nski $N$-gasket}\label{classicNgasket}\index{classic Sierpi\'nski $N$-gasket, $S_N$}\index{Sierpi\'nski $N$-gasket!classic, $S_N$} in $\eR^N$
(to be further discussed below), then
for every $N\ge2$, we have that
\begin{equation}\label{s04}
\dim_B S_N\,\,(=\dim_HS_N)\,\,=\s_0=\log_2(N+1).
\end{equation}
The latter statement follows from a classic result of Hutchinson in \cite{hutchinson}
for self-similar sets satisfying the open set condition (which is the case of $S_N$ for every $N\ge2$)
and extending to higher dimensions the basic result of Moran \cite{mora} for one-dimensional self-similar sets.\footnote{Note that $S_1\st\eR$ is just the unit interval, viewed as a self-similar set with scaling ratios $r_1=r_2=1/2$. However, in the present discussion, we consider the more interesting case when $N\ge2$.} (See \cite[Theorem 9.3]{falc} for the statement and a detailed proof of this theorem.)

\medskip

We close this discussion of the $N$-gasket RFD $(A_N,\O_N)$ by explaining the discrepancy between the results obtained in \eqref{s01}, \eqref{s02} and, especially, \eqref{s03} for the self-similar spray $(A_N,\O_N)$ and the usual result \eqref{s04} for the self-similar set $S_N$, the classic Sierpi\'nski $N$-gasket.

First of all, note that in light of \eqref{Ngasketzf} and \eqref{Ngasketzf3} (see also Theorem \ref{sprayz}), we must have
\begin{equation}\label{s05}
\begin{aligned}
\ov\dim_B(A_N,\O_N)&\,\,(=\ov\dim_BA_N)\\
&=\max\{\s_0,\ov\dim_B(A_{N,0},\O_{N,0})\}\\
&=\max\{\log_2(N+1),N-1\},
\end{aligned}
\end{equation}
and that, in the present case, the upper Minkowski dimensions can be replaced by the  Minkowski dimensions in Equation \eqref{s05}.

Indeed, by \eqref{Ngasketzf}, we have
\begin{equation}
D(\zeta_{A_N,\O_N})=\max\{D(\zeta_{\mathfrak S}),D(\zeta_{A_{N,0},\O_{N,}})\}
\end{equation}
and by part ($b$) of Theorem \ref{an_rel}, we have
$$
D(\zeta_{A_N,\O_N})=\ov\dim_BD(A_N,\O_N)
$$
and
$$
D(\zeta_{A_{N,0},\O_{N,0}})=\ov\dim_BD(A_{N,0},\O_{N,0}),
$$
from which \eqref{s05} follows since $\s_0=\log_2(N+1)$. (See Theorem \ref{sprayz}.)

Identity \eqref{s05} explains why \eqref{s01}, \eqref{s02} and \eqref{s03} hold.
Indeed, if we let $D_G:=\dim_B(A_{N,0},\O_{N,0})$ (the Minkowski dimension of the base RFD $(A_{N,0},\O_{N,0})$ generating the self-similar RFD $(A_N,\O_N)$) and $D:=\dim_B(A_N,\O_N)$, we deduce from \eqref{s05} and an elementary computation that $D=\s_0$ if $N=2$, $D=\s_0=D_G$ if $N=3$, whereas $D=D_G$ if $N\ge4$, in agreement with \eqref{s01}, \eqref{s02} and \eqref{s03}, respectively.

From the geometric point of view, the difference between $A_N$ and $S_N$ can be explained as follows. As is well known (see, e.g., [{KiLap}] and the relevant references therein), the (homogeneous) Sierpi\'nski $N$-gasket $S_N$ is a self-similar set (satisfying the open set condition), associated with the iterated function system (IFS) $\{\Phi_j\}_{j=1}^{N+1}$, where (for $j=1,\dots,N+1$) each $\Phi_j$ is a contractive similitude of $\eR^N$ with fixed point $P_j$ and scaling ratio $r_j=1/2$; i.e., the associated scaling ratio list $\{r_j\}_{j=1}^{N+1}$ of $\{\Phi_j\}_{j=1}^{N+1}$ is given by $r_1=\cdots=r_{N+1}=1/2$. More specifically, $S_N$ is the unique (nonempty) compact subset $K$ of $\eR^N$ which is the solution of the (homogeneous) fixed point equation
\begin{equation}\label{PhiN}
K={\mathbf\Phi}(K):=\bigcup_{j=1}^{N+1}\Phi_j(K).
\end{equation}

On the other hand, unless $N=2$, the inhomogeneous Sierpi\'nski $N$-gasket $A_N$ is {\em not} a self-similar set in the classic sense of \cite{hutchinson} (see also \cite[Chapter 9]{falc}). (For $N=2$, $A_2$ coincides with the usual Sierpi\'nski gasket $S_2$.) However, interestingly, it is an {\em inhomogeneous} self-similar set, in the sense of Barnsley and Demko \cite{bd} (see also \cite{bd,fraser} and the relevant references therein for further results about such sets).
More specifically, $A_N$ is the unique (nonempty) solution $K$ of the {\em inhomogeneous} fixed point equation\index{fixed point equation|textbf}
\begin{equation}\label{KB}
K={\mathbf\Phi}(K)\cup B,
\end{equation}
where ${\mathbf\Phi}$ is defined as in Equation~\eqref{PhiN} above and $B$ is a suitable compact subset of $\eR^N$.
For $N=2$, the set $A_2=S_2$ is both homogeneous and inhomogeneous, since it satisfies Equation \eqref{KB}
both for $B=\emptyset$ and $B=\pa A_{2,0}$ (the boundary of the unit triangle).
By contrast, when $N\ge 3$, the compact set $B$ is {\em nonempty} and hence, $A_N$ is not self-similar for this IFS
$\{\Phi_j\}_{j=1}^{N+1}$.
 For $N=3$, a description of several possible choices for $B$ can be found in the caption of Figure \ref{octahedron}. When $N\ge3$, let us simply state that we can choose $B$ to be the boundary of $\O_{N,0}$: $B=\pa\O_{N,0}$. (Other choices are possible, however.)

\medskip

\begin{example}\label{sierpinski_carpetr} ({\em Relative Sierpi\'nski carpet}).\index{relative Sierpi\'nski!carpet|textbf}
Let $A$ be the Sierpi\'nski carpet contained in the unit square $\O$. Let $(A,\O)$ be the corresponding {\em relative Sierpi\'nski carpet}, with $\Omega$ being the unit square. Its distance zeta function $\zeta_{A,\Omega}$ coincides with the distance zeta function of the following relative fractal spray
 (see the end of Definition \ref{spray}):
$$
\operatorname{Spray}(\O_0, \g=1/3,b=8),
$$
where $\O_0$ is the first deleted open square with sides $1/3$.
Similarly as in Example \ref{6.15}, using Theorem \ref{sprayzc} and Lemma \ref{congruentl}, we obtain that $\zeta_{A,\Omega}$, the relative distance zeta functions of $(A,\O)$, has a meromorphic continuation to the entire complex plane given for all $s\in\Ce$ by
\begin{equation}\label{10.111/4}
\zeta_{A,\Omega}(s)=\frac{8\cdot 6^{-s}}{s(s-1)(1-8\cdot3^{-s})}.
\end{equation}

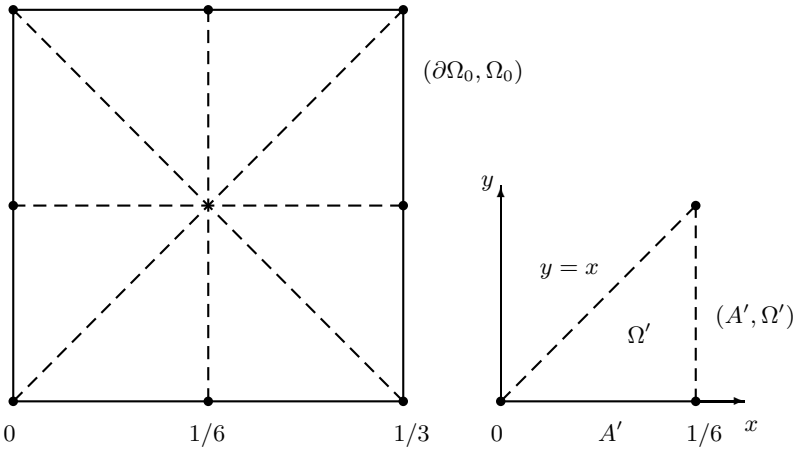
\begin{figure}[t]
\psset{unit=1.3}
\begin{pspicture}(-2.6,-1)(6,4.3)

\pstGeonode[PointName=none](-2,0){A}(2,0){B}(2,4){C}(-2,4){D}
\pstMiddleAB[PointName=none]{A}{B}{E}
\pstMiddleAB[PointName=none]{B}{C}{F}
\pstMiddleAB[PointName=none]{C}{D}{G}
\pstMiddleAB[PointName=none]{D}{A}{H}
\pstLineAB{A}{B}
\pstLineAB{B}{C}
\pstLineAB{C}{D}
\pstLineAB{D}{A}
\put(-2.1,-0.4){\small$0$}
\put(-0.2,-0.4){\small$1/6$}
\put(1.9,-0.4){\small$1/3$}
\put(2.2,3.3){\small$(\pa\O_0,\O_0)$}

\pstLineAB[linestyle=dashed]{A}{C}
\pstLineAB[linestyle=dashed]{B}{D}
\pstLineAB[linestyle=dashed]{E}{G}
\pstLineAB[linestyle=dashed]{F}{H}

\pstGeonode[PointName=none](3,0){I}(5,0){J}(5,2){K}
\pstLineAB{I}{J}
\pstLineAB[linestyle=dashed]{J}{K}
\pstLineAB[linestyle=dashed]{I}{K}
\put(2.9,-0.4){\small$0$}
\put(4.9,-0.4){\small$1/6$}
\put(3.4,1.3){\small$y=x$}
\put(5,0){\vector(1,0){0.5}}
\put(3,0){\vector(0,1){2.2}}
\put(5.5,-0.3){\small$x$}
\put(2.8,2.2){\small$y$}
\put(4,-0.4){\small$A'$}
\put(4.3,0.6){\small$\O'$}
\put(5.2,0.8){\small$(A',\O')$}

\end{pspicture}
\caption{\small On the left is the base relative fractal drum $(\pa\O_0,\O_0)$ of the relative Sierpi\'nski carpet $(A,\O)$ described in Example \ref{sierpinski_carpetr}, where $\O_0$ is the associated (open) square with sides $1/3$. The base relative fractal drum $(\pa\O_0,\O_0)$ can be viewed as the
(disjoint) union of eight RFDs, all of which are congruent to the relative fractal drum $(A',\O')$ on the right.
This figure explains Equation \eqref{sierpinski8}; see Lemma \ref{congruentl}.}
\label{sierpinski_carpetr2}
\end{figure}

Indeed, clearly, the base relative fractal drum $(\pa\O_0,\O_0)$ is the (disjoint) union of eight relative fractal drums, each of which is congruent to a relative fractal drum $(A',\O')$, where $\O'$ is an appropriate isosceles right triangle; see Figure \ref{sierpinski_carpetr2}. We then deduce from Lemma \ref{congruentl} that 
\begin{equation}\label{sierpinski8}
\begin{aligned}
\zeta_{\pa\O_0,\O_0}(s)&=8\,\zeta_{A',\O'}(s)=8\,\int_{\O'}d((x,y),A')^{s-2}\D x\,\D y\\
&=8\int_0^{1/6}\D x\int_0^x y^{s-2}\D y=\frac{8\cdot 6^{-s}}{s(s-1)}
\end{aligned}
\end{equation}
for all $s\in \Ce$ with $\re s>1$, and hence, in light of Theorem \ref{sprayzc}, that $\zeta_{A,\O}(s)$ is given by \eqref{10.111/4}.
Note that, after analytic continuation, we also have
\begin{equation}\label{4.2.39.1/2}
\zeta_{\pa\O_0,\O_0}(s)=\frac{8\cdot 6^{-s}}{s(s-1)},\q\mbox{for all $s\in\Ce$.}
\end{equation}

Since by \eqref{4.2.39.1/2},
$$
\zeta_{A,\Omega}(s)\sim\frac1{1-8\cdot3^{-s}},
$$
one deduces from this equivalence that the abscissa of convergence of $\zeta_{A,\Omega}$ is given by $D=\log_38=\dim_B(A,\Omega)$, where the equality follows from Theorem \ref{an_rel}($b$).

Here, the relative box dimension of $A$ coincides with its usual box dimension, namely, $\log_38$. Moreover, the set $\po_c(\zeta_{A,\Omega})$ of relative principal complex dimensions of the Sierpi\'nski carpet $A$ with respect to the unit square $\O$ is given by 
\begin{equation}\label{sierpinski_rfd}
\dim_{PC}(A,\O)=\log_38+\mathbf{p}{\I}\Ze,
\end{equation}
where 
$\mathbf{p}:=2\pi/\log 3$ 
is the oscillatory period\index{oscillatory period!of the Sierpi\'nski carpet} of the Sierpi\'nski carpet $A$.

Observe that it follows immediately from $(\ref{10.111/4})$ that the set $\po(\zeta_{A,\Omega})$ of all relative complex dimensions of the Sierpi\'nski carpet $A$ (with respect to the unit square $\Omega$) is given by 
$$
\po(\zeta_{A,\Omega})=\dim_{PC}A\cup\{0,1\}=(\log_38+\mathbf{p}{\I}\Ze)\cup\{0,1\},
$$
where $\{0,1\}$ can be viewed as the set of `integer dimensions' of $A$ (in the sense of 
[{LapPe2--3}] 
and \cite{lappewi1}, see also \cite[\S13.1]{lapidusfrank12}). Furthermore, each of these relative complex dimensions is simple (i.e., is a simple pole of $\zeta_{A,\O}$). Interestingly, these are exactly the complex dimensions which one would expect to be associated with $A$, according to the theory developed in 
[{LapPe2--3}] 
 and
[{LapPeWi1--2}] (as described in \cite[\S13.1]{lapidusfrank12})
via self-similar tilings (or sprays) and associated tubular zeta functions. 

In light of \eqref{10.111/4},
the residue of the distance zeta function of the relative Sierpi\'nski carpet $(A,\O)$ computed at any principal pole $s_k:=\log_38+\mathbf{p}{\I}k$, $k\in\Ze$ is given by
$$
\res(\zeta_{A,\O},s_k)=\frac{2^{-s_k}}{(\log3)s_k(s_k-1)}.
$$
In particular,
$$
|\res(\zeta_{A,\O},s_k)|\sim \frac{2^{-D}}{D\log3}\,k^{-2}\q\mbox{as\q$k\to\pm\ty$,}
$$
where $D:=\log_38$.
\end{example}\label{sierp_end}

\medskip

Similarly as in the case of the relative Sierpi\'nski gasket (see Proposition \ref{s_gasket}), the relative Sierpi\'nski carpet can be viewed as a fractal spray generated by the base RFD appearing in Figure \ref{sierpinski_carpetr2} on the right.

\begin{prop}[{Relative Sierpi\'nski carpet}] Let $(A',\O')$ be the RFD defined on the right-hand side of Figure \ref{sierpinski_carpetr2}. Let $(A,\O)$ be the relative fractal spray generated by the base relative fractal drum $(A',\O')$, with scaling ratio $\g=1/3$ and with multiplicities $m_k=8^k$ for any positive integer $k$:
\begin{equation}
(A,\O)=\operatorname{Spray}((A',\O'),\,\,\g=1/3,\,m_k=8^k\,\,\mbox{for $k\in\eN$}).
\end{equation}
$($Note that we assume here that the base relative fractal drum $(A',\O')$ has a multiplicity equal to $8$.$)$ Then, the relative distance zeta function of the relative fractal spray $(A,\O)$ coincides with the relative distance zeta function of the relative Sierpi\'nski carpet. $($See Equation \eqref{10.111/4}$.)$
\end{prop}

\begin{example} ({\em Sierpi\'nski $N$-carpet}).\label{carpetN}\index{Sierpi\'nski $N$-carpet|textbf}
It is easy to generalize the example of the standard Sierpi\'nski carpet (which is a compact subset of the unit square $[0,1]^2\st\eR^2$, see Example \ref{sierpinski_carpetr} above), to the {\em Sierpi\'nski $N$-carpet} (or {\em$N$-carpet}, for short), defined analogously as a compact subset $A$ of the unit $N$-dimensional cube $[0,1]^N\st\eR^N$. More precisely, we divide $[0,1]^N$ into the union of $3^N$ congruent $N$-dimensional subcubes of length $1/3$ and with disjoint interiors and then remove the middle open subcube. The remaining compact set is denoted by $F_1$. We then remove the middle open $N$-dimensional cubes of length $1/3^2$ from the remaining $3^N-1$ subcubes. The resulting compact subset is denoted by $F_2$. Proceeding analogously ad infinitum, we obtain a decreasing sequence of compact subsets $F_k$ of $[0,1]^N$, for $k\ge1$. 
The Sierpi\'nski $N$-carpet $A$ is then defined by
\begin{equation}
A:=\bigcap_{k=1}^\ty F_k.
\end{equation}
(Note that the Sierpi\'nski $1$-carpet coincides with the usual ternary Cantor set, while the Sierpi\'n\-ski $2$-carpet coincides with the usual Sierpi\'nski carpet discussed in Example \ref{sierpinski_carpetr}; furthermore, the Sierpi\'nski $3$-carpet is discussed in \cite[Example 2]{brezish}.)

It is clear that the {\em Sierpi\'nski $N$-carpet RFD}\index{Sierpi\'nski $N$-carpet!relative fractal drum} $(A,\O)$, where $A$ is the Sierpi\'nski $N$-carpet and $\O:=(0,1)^N$ is the open unit cube of $\eR^N$, can be viewed as the following relative fractal spray; see the end of Definition~\ref{spray}:
\begin{equation}
(A,\O)=\operatorname{Spray}\,((\pa\O_0,\O_0),\g=1/3,b=3^N-1).
\end{equation}
(Here, the cube $\O_0=(0,1/3)^N$ is obtained by a suitable translation of the middle open subcube from the first step of the construction of the set $A$.)
According to Theorem \ref{sprayzc}, we then have that
\begin{equation}\label{Ncarpetz}
\begin{aligned}
\zeta_{A,\O}(s)&=f(s)\cdot\zeta_{\pa\O_0,\O_0}(s)\\
&=\frac{\zeta_{\pa\O_0,\O_0}(s)}{1-(3^N-1)3^{-s}} \sim
\frac1{1-(3^N-1)3^{-s}}.
\end{aligned}
\end{equation}
Since $\O_0$ has a Lipschitz boundary and $\log_{1/\g}b=\log_3(3^N-1)\in(N-1,N)$, we deduce from \eqref{rfs_pcd} in Theorem \ref{sprayzc} that the set of principal complex dimensions of the relative Sierpi\'nski $N$-carpet spray is given by
\begin{equation}\label{Ncarpet_pcd}
\dim_{PC}(A,\O)=\log_3(3^N-1)+\frac{2\pi}{\log 3}\,\I\Ze
\end{equation}
and hence,
$$
\dim_{PC}(A,\O)\st\{\re s=\log_3(3^N-1)\}\st \{N-1<\re s<N\}.
$$
In particular, according to Theorem \ref{an_rel}$(b)$, we have that
\begin{equation}
\ov\dim_B(A,\O)=\log_3(3^N-1).
\end{equation}
Furthermore, it can be shown that in the present case of the Sierpi\'nski $N$-carpet RFD, we have that
$\dim_BA$ and $\dim_B(A,\O)$ exist and
\begin{equation}\label{crfd}
\ov\dim_B(A,\O)=\dim_B(A,\O)=\dim_BA=\log_3(3^N-1).
\end{equation}

It is easy to see that the set of principal complex dimensions $\dim_{PC}A$ of the Sierpi\'nski $N$-carpet $A$ in $\eR^N$ coincides with the set $\dim_{PC}(A,\O)$ appearing in Equation \eqref{Ncarpet_pcd}. As simple special cases, we obtain the set of principal complex dimensions of the ternary Cantor set or of the usual Sierpi\'nski carpet appearing in Equation \eqref{sierpinski_rfd}, for $N=1$ or $N=2$, respectively.

Since the set of all complex dimensions of the RFD $(\pa\O_0,\O_0)$ is equal to $\{0,1,\dots, N-1\}$,\footnote{Note that the relative zeta function $\zeta_{A,\O}$ appearing in Equation \eqref{Ncarpetz} can be meromophically extended in a unique way to the whole complex plane $\Ce$ since the same can be done with $\zeta_{\pa\O_0,,\O_0}$. See, for example, Equation \eqref{4.2.39.1/2} dealing with the case when $N=2$.} it follows from Equation \eqref{Ncarpetz} that the set of all complex dimensions of the Sierpi\'nski $N$-carpet relative fractal spray $(A,\O)$ is given by
\begin{equation}
\begin{aligned}
\po(\zeta_{A,\Omega})&=\dim_{PC}(A,\O)\cup\{0,1,\dots, N-1\}\\
&=\Big(\log_3(3^N-1)+\frac{2\pi}{\log 3}\,{\I}\Ze\Big)\cup\{0,1,\dots,N-1\}.
\end{aligned}
\end{equation}

This concludes our study of the relative fractal drum $(A,\Omega)$ naturally associated with
the $N$-dimensional Sierpi\'nski carpet.
\end{example}

\section{Self-similar sprays and RFDs}\label{goran}

Let us now recall the definition of a self-similar spray or tiling (see [{LapPe2--3}], 
[{LapPeWi1--2}], \cite[\S13.1]{lapidusfrank12}).
More precisely, let us state this definition slightly more generally and in the context of relative fractal drums.

\begin{defn}({\em Self-similar spray or tiling}).\label{ss_spray}\index{self-similar!spray (or tiling)|textbf}
Let $G$ be a given open subset ({\em base set} or {\em generator})\index{base set (or generator)|textbf}\index{generator (or base set)|textbf} of $\eR^N$ of finite $N$-dimensional Lebesgue measure and let $\{r_1,r_2,\ldots,r_J\}$ be a finite multiset (also called a {\em ratio list})\index{ratio list of a self-similar spray (or tiling)|textbf} of positive real numbers such that $J\in\eN$, $J\geq 2$ and
\begin{equation}\label{r_J^N}
\sum_{j=1}^{J}r_j^N<1.
\end{equation}
Furthermore, let $\Lambda$ be the multiset consisting of all the possible `words' of multiples of the scaling factors $r_1,\ldots,r_J$; that is, let
\begin{equation}\label{LAMBDA}
\begin{aligned}
\Lambda&:=\{1,r_1,\dots,r_J,r_1r_1,\ldots,r_1r_J,r_2r_1,\ldots,r_2r_J,\ldots,r_Jr_1,\ldots,r_Jr_J,\\
&\phantom{:=\{}\,r_1r_1r_1,\ldots,r_1r_1r_J,\ldots\}
\end{aligned}
\end{equation}
and arrange all of the elements of the multiset $\Lambda$ into a {\em scaling sequence}\index{scaling sequence|textbf} $(\lambda_i)_{i\geq0}$, where $\lambda_0:=1$.
(Note that $0<\lambda_i<1$, for every $i\geq 1$.)

A {\em self-similar spray}\index{self-similar!spray (or tiling)|textbf} (or {\em tiling}), generated by the base set $G$ and the ratio list $\{r_1,r_2,\ldots,r_J\}$ is an RFD $(\pa\O,\O)$ in $\eR^N$, where $\O$ is a disjoint union of open sets $G_{i}$; i.e.,
\begin{equation}\label{Galfa}
\O:=\bigsqcup_{i=0}^{\ty} G_{i},
\end{equation}
such that each $G_{i}$ is congruent to $\lambda_i G$, for every $i\geq 0$.
Here, the disjoint union $\sqcup$ can be understood as the disjoint union of RFDs given in Definition \ref{union}, with $(A_i,\O_i):=(\pa G_i,G_i)$ for each $i\geq 0$, in the notation of that definition.
In the sequel, $(\pa G, G)$ is also referred to as a {\em self-similar RFD}.\index{self-similar!RFD|textbf}
%
%
\end{defn}

\medskip

\begin{remark}
Note that in the above definition, the scaling sequence $(\lambda_{i})_{i\geq 0}$  consists of all the products of ratios $r_1,\ldots,r_J$ appearing in the infinite sum
\begin{equation}
\sum_{n=0}^{\ty}\bigg(\sum_{j=1}^{J}r_j\bigg)^n,
\end{equation}
after expanding the powers and counted with their multiplicities.
More precisely, we have that for every multi-index $\a=(\a_1,\ldots,\a_J)\in\eN_0^J$, the multiplicity of $r_1^{\a_1}r_2^{\a_2}\dots r_J^{\a_J}$ in the multiset $\Lambda$ is equal to the multinomial coefficient
\begin{equation}\label{multinom_l}
\binom{|\a|}{\a_1,\a_2,\ldots,\a_J}=\frac{|\a|!}{\a_1!\,\a_2!\cdots\a_J!},
\end{equation}
where $|\a|:=\sum_{j=1}^J\a_j$.
Of course, depending on the specific values of the ratios $r_1,\ldots,r_J$, some of the numbers $r_1^{\a_1}r_2^{\a_2}\dots r_J^{\a_J}$ may be equal for different multi-indices $\a\in\eN_0^J$.

Furthermore, the condition \eqref{r_J^N} ensures that the set $\O=\sqcup_{i\geq 0}G_{i}$ has finite $N$-dimensional Lebesgue measure.
Indeed, we have
\begin{equation}
\begin{aligned}
|\O|&=\sum_{i=0}^{\ty}|G_{i}|=\sum_{i=0}^{\ty}|\g_i G|=|G|\sum_{i=0}^{\ty}\lambda_{i}^N\\
&=|G|\sum_{n=0}^{\ty}\bigg(\sum_{j=1}^{J}r_j^N\bigg)^n=\frac{|G|}{1-\Big(\sum_{j=1}^{J}r_j^N\Big)^n},
\end{aligned}
\end{equation}
since \eqref{r_J^N} is satisfied.
Note that the second to last equality above follows from the construction of the scaling sequence $(\lambda_{i})_{i\geq 0}$. 
\end{remark}

Consider now a self-similar spray as a relative fractal drum $(A,\O)$; that is, let $A:=\partial\O$ and $\O:=\sqcup_{i\geq 0}G_{i}$ (see Definition~\ref{ss_spray}). 
The `self-similarity' of $(A,\O)$ is nicely exhibited by the scaling relation \eqref{ssss} given in the following lemma.

\begin{lemma}\label{self_lemma}
Let  $(A,\O)$ be a self-similar spray in $\eR^N$, as in Definition \ref{ss_spray}.
Then, the relative fractal drum $(A,\O)$ satisfies  the following  {\rm `self-similar identity'}$:$\index{self-similar!identity of an RFD $(A,\O)$|textbf}
\begin{equation}\label{ssss}
(A,\O)=(\partial G,G)\sqcup \bigsqcup_{j=1}^J r_j(A,\O),
\end{equation}
where $($with the exception of the first term on the right-hand side of \eqref{ssss}$)$ the symbol $\sqcup_{j=1}^J$ indicates that this represents a disjoint union of copies of $(A,\O)$ scaled by factors $r_1,\ldots,r_J$ and displaced by isometries of $\eR^N$.
\end{lemma}

\begin{proof}
Let us re-index the scaling sequence $(\lambda_{i})_{i\geq 0}$ in a way that keeps track of the actual construction of the numbers $\lambda_i$ out of the scaling ratios $r_1,\ldots,r_J$; see Equation~\eqref{LAMBDA} above.
We let
\begin{equation}
I:=\{\emptyset\}\cup\bigcup_{m=1}^{\ty}\{1,\ldots,J\}^{m}
\end{equation}
be the set of all finite sequences consisting of numbers $1,\ldots,J$ (or, equivalently, of all finite words based on the alphabet $\{1,\ldots,J\}$).
Furthermore, for every $\a\in I$, define
\begin{equation}
\lambda_{\a}:=\begin{cases}
1,&\a=\emptyset\\
r_{\a_1}r_{\a_2}\cdots r_{\a_m},&\a\neq\emptyset.
\end{cases}
\end{equation}
We then deduce from the construction of $(A,\O)$ that
$$
\begin{aligned}
(A,\O)&=\bigsqcup_{i=0}^{\ty}(\pa G_{i},G_{i})=\bigsqcup_{i=0}^{\ty}\lambda_{i}(\pa G,G)\\
&=\bigsqcup_{\a\in I}\lambda_{\a}(\pa G,G)=(\pa G,G)\sqcup\bigsqcup_{\a\in I\setminus\{\emptyset\}}\lambda_{\a}(\pa G,G).\\
\end{aligned}
$$
Observe now that in the last disjoint union above, every $\a\in\{1,\ldots,J\}^{m}$ can be written as $\{j\}\times\{1,\ldots,J\}^{m-1}$, for some $j\in\{1,\ldots,J\}$, if we identify $\{j\}$ with $\{j\}\times\{\emptyset\}$ when $m=1$.
Note that this identification is consistent with the definition of $\lambda_\a$, in the sense that $\lambda_{\{j\}\times\beta}=r_j\lambda_{\beta}$ for all $j\in\{1.\ldots,J\}$ and $\b\in I$.
In light of this, we can next partition the last union above with respect to which number $j\in\{1,\ldots,J\}$ the sequence $\a$ begins with:
$$
\begin{aligned}
(A,\O)&=(\pa G,G)\sqcup\bigsqcup_{j=1}^{J}\bigsqcup_{\a\in \{j\}\times I}\lambda_{\a}(\pa G,G)=(\pa G,G)\sqcup\bigsqcup_{j=1}^{J}\bigsqcup_{\b\in I}r_j\lambda_{\b}(\pa G,G)\\
&=(\pa G,G)\sqcup\bigsqcup_{j=1}^{J}r_j\bigg(\bigsqcup_{\b\in I}\lambda_{\b}(\pa G,G)\bigg)=(\pa G,G)\sqcup\bigsqcup_{j=1}^{J}r_j(A,\O).
\end{aligned}
$$
This completes the proof of the lemma.
%
\end{proof}

In light of the identity \eqref{ssss} and a very special case of Theorem~\ref{unionz}, it is now clear that the distance zeta function of $(A,\O)$ satisfies the following functional equation, which itself can be considered as a {\em self-similar identity}:
\begin{equation}\label{self-sim_eq}
\zeta_{A,\O}(s)=\zeta_{\partial G,G}(s)+\sum_{j=1}^J\zeta_{r_j(A,\O)}(s),
\end{equation}
for all $s\in\Ce$ with $\re s$ sufficiently large.\footnote{For instance, it suffices to assume that $\re s>N$ since, by Theorem~\ref{an_rel}, all of the zeta functions appearing in \eqref{self-sim_eq} are holomorphic on the right half-plane $\{\re s>N\}$.}
Furthermore, for such $s$, by using the scaling property of the relative distance zeta function (Theorem \ref{scaling}), we deduce that the above equation then becomes
\begin{equation}
\zeta_{A,\O}(s)=\zeta_{\partial G,G}(s)+\sum_{j=1}^Jr_j^s\zeta_{A,\O}(s).
\end{equation}
Finally, this last identity together with an application of the principle of analytic continuation now yields the following theorem.

\begin{theorem}\label{ss_spray_zeta}
Let $G$ be the generator of a self-similar spray in $\eR^N$, and let $\{r_1,r_2,\ldots,r_J\}$, with $r_j>0$ $($for $j=1,\ldots,J$, $J\geq 2$$)$ and such that $\sum_{j=1}^{J}r_j^N<1$, be its scaling ratios.
Furthermore, let $(A,\O):=(\pa\O,\O)$ be the self-similar spray generated by $G$, as in Definition \ref{ss_spray}.
Then, the distance zeta function of $(A,\O)$ is given by
\begin{equation}\label{ss_spray_form}
\zeta_{A,\O}(s)=\frac{\zeta_{\partial G,G}(s)}{1-\sum_{j=1}^{J}r_j^s},
\end{equation}
for all $s\in\Ce$ with $\re s$ sufficiently large.
In addition,
\begin{equation} 
D(\zeta_{A,\O})=\max\{\ov{\dim}_B(\pa G,G),D\},
\end{equation}
where $D>0$ is the unique real solution of $\sum_{j=1}^{J}r_j^D=1$ $($i.e., $D$ is the {\rm similarity dimension}\index{similarity dimension|textbf} of the self-similar spray $(\pa\O,\O)$$)$.

More specifically, given a connected open neighborhood $U$ of the critical line $\{\re s=D\}$, $\zeta_{A,\O}$ has a meromorphic continuation to $U$ if and only if $\zeta_{\pa G,G}$ does, and in that case, $\zeta_{A,\O}(s)$ is given by \eqref{ss_spray_form} for all $s\in U$.
Consequently, the visible complex dimensions of $(A,\O)$ satisfy
\begin{equation}\label{ss_spray_po}
\po(\zeta_{A,\O},U)\subseteq(\mathfrak{D}\cap U)\cup\po(\zeta_{\partial G,G},U),
\end{equation}
where $\mathfrak{D}$ is the set of all the complex solutions of the Moran equation $\sum_{j=1}^{J}r_j^s=1$ $($i.e., the {\rm scaling complex dimensions} of the fractal spray$)$; see Remark \ref{fractality0} for detailed information about $\mathfrak{D}$.
Finally, if there are no zero-pole cancellations in \eqref{ss_spray_form}, then we have an equality in \eqref{ss_spray_po}.
\end{theorem}


\begin{remark}\label{5.151/2} ({\em Complex dimensions and the definition of fractality}).\index{fractality and complex dimensions}
In [{Lap-vFr1--3}], 
a geometric object is said to be ``fractal'' if the associated zeta function has at least one nonreal complex dimension (with positive real part).
(See \cite[\S12.1 and \S12.2]{lapidusfrank12} for a detailed discussion.)
In \cite{lapidusfrank06,lapidusfrank12}, in order, in particular, to take into account some possible situations pertaining to random fractals (see \cite{HamLa}, partly described in \cite[\S13.4]{lapidusfrank12}), the definition of fractality (within the context of the theory of complex dimensions) was extended so as to allow for the case described in part $(i)$ of Definition~\ref{hyperfractal} just above, namely, the existence of a (meromorphic) natural boundary along a screen.
(See \cite[\S13.4.3]{lapidusfrank12}.)

We note that in \cite{lapidusfrank12} (and the other aforementioned references), the term ``hyperfractal'' was not used to refer to case $(i)$ (or to any other situation).
More important, except for fractal strings and in very special higher-dimensional situations (such as suitable fractal sprays), one did not have to our disposal (as we now do), 
a general definition of ``fractal zeta function'' associated with an arbitrary bounded subset of $\eR^N$, for every $N\geq 1$.
Therefore, we can now define the ``fractality'' of any bounded subset of $\eR^N$ (including Julia sets and the Mandelbrot set) and, more generally, of any relative fractal drum, by the presence of a nonreal complex dimension or else by the ``hyperfractality'' (in the sense of part $(i)$ of Definition~\ref{hyperfractal}) of the geometric object under consideration.
Here, ``complex dimension'' is understood as a (visible) pole of the associated fractal zeta function (the distance or tube zeta function of a bounded subset or a relative fractal drum of $\eR^N$, or else, as was the case in most of \cite{lapidusfrank12}, the geometric zeta function of a fractal string).

Much as in [{Lap-vFr1--3}] and [{Lap3--8}], this terminology (concerning fractality, hyperfractality, and complex dimensions), can be extended to `virtual geometries', as well as to (absolute or) relative fractal drums, noncommutative geometries, dynamical systems, and arithmetic geometries, via suitably associated `fractal zeta functions', be they absolute or relative distance or tube zeta functions, spectral zeta functions, dynamical zeta functions, or arithmetic zeta functions (or their logarithmic derivatives thereof).
\end{remark}

We will return to the discussion of the notion of fractality in the closing chapter of this paper, namely,~Chapter \ref{fractality}; see also the next remark.
\medskip

\begin{remark}\label{fractality0}
The multiset $\mathfrak{D}$ of scaling complex dimensions of the self-similar spray $(A,\O)$ is analyzed in detail in [Lap-vFr3, Chapter 3; esp., Theorem 3.6].\footnote{See also [Lap-vFr3, Chapter 2; esp., Theorem 2.16] for the one-dimensional case, corresponding to self-similar strings.} Accordingly, there is a natural lattice\,/\,nonlattice dichotomy defined as follows: $(A,\O)$ is {\em lattice}\index{self-similar!spray (or tiling)!lattice|textbf} if the multiplicative subgroup $G$ of $(0,+\ty)$ generated by the distinct values of the scaling ratios $r_1,\dots,r_J$ is of rank~$1$ (i.e., is of the form $r^\Ze$ for some unique real number $r\in(0,1)$, called the {\em multiplicative generator} of the spray). It is {\em nonlattice},\index{self-similar!spray (or tiling)!nonlattice|textbf} otherwise (i.e., if the above group is of rank $>1$), and {\em generic nonlattice} if $G$ is of maximal rank $>1$ (i.e., of rank $J'$, the number of distinct elements in the ratio list $r_1,\dots, r_J$, and $J'>1$).

Then, according to [Lap-vFr3, Theorem~3.6], in the lattice case, all of the scaling complex dimensions are periodically distributed along finitely many vertical lines (the right most of which is the vertical line $\{\re s=D\}$) with the same period $T:=2\pi/\log(r^{-1})$, called the {\em oscillatory period} of the lattice self-similar spray.\footnote{On each of these vertical lines, the corresponding scaling complex dimensions all have the same multiplicity. In particular, along the vertical line $\{\re s=D\}$, they are all simple.} On the other hand, in the nonlatice case,  $\mathfrak{D}$ is simple and is the only principal scaling complex dimension (i.e., the only scaling complex dimension with real part $D$ located on the vertical line $\{\re s=D\}$). However, there is an infinite sequence of distinct scaling complex dimensions converging from the left to (but not touching) the vertical line $\{\re s=D\}$.

Moreover, it was conjectured in [Lap-vFr2,3, \S3.7] (and especially, in reference [Lap-vF7] of [Lap-vFr3]) that in the generic nonlattice case, the set of real parts of the scaling dimensions is dense in a compact interval $[\s_l,D]$, with $\s_l\in\eR$ and $\s_l<D$; i.e., the {\em set of ``fractality''} (that is, the closure of the above set of real parts, as defined in [Lap-vFr2,3]) is equal to $[\s_l,D]$, in striking contrast to the lattice case where it is a finite set. 
This conjecture has recently been proved in [MorSep], where it was also shown that in the nonlattice (but not necessarily, generic nonlattice) case, the set of ``fractality'' is equal to a finite (and nonempty) disjoint union of nonempty compact intervals.  

Finally, via Diophantine approximation techniques, the scaling complex dimensions of a nonlattice self-similar spray can be approximated by those of a sequence of lattice sprays with larger and larger periods. (See [Lap-vFr3, \S3.4, esp., Theorem 3.19].) Accordingly, in the nonlattice case, the scaling complex dimensions exhibit a quasiperiodic pattern (studied in detail both numerically and theoretically in [Lap-vFr3, Chapter~3]).
\end{remark}
\medskip

\begin{example}({\em The $1/2$-square fractal}).\label{kvadrat0.5}\index{square@$1/2$-square fractal|textbf}
In this planar example, we will further investigate and illustrate the new interesting phenomenon which occurs in the case of the Sierpi\'nski $3$-gasket RFD discussed in Example \ref{Ngasket}.
Namely, we start with the closed unit square $I=[0,1]^2$ in $\eR^2$ and subdivide it into $4$ smaller squares by taking the centerlines of its sides.
We then remove the two diagonal open smaller squares, denoted by $G_1$ and $G_2$ in Figure \ref{kv_0.5}, so that $G:=G_1\cup G_2$ is our generator in the sense of Definition \ref{ss_spray}.
Next, we repeat this step with the remaining two closed smaller squares and continue this process, ad infinitum.
The $1/2$-square fractal is then defined as the set $A$ which remains at the end of the process; see Figure \ref{kv_0.5}, where the first 6 iterations are shown.
More precisely, the set $A$ is the closure of the union of the boundaries of the disjoint family of open squares appearing in Figure \ref{kv_0.5} and packed in the unit square $I$.
If we now let $\O:=(0,1)^2$, we have that $(A,\O)$ is an example of a self-similar spray (or tiling), in the sense of Definition \ref{ss_spray}, with generator $G=G_1\cup G_2$ and scaling ratios $r_1=r_2=1/2$.
Note, however, that $A$ is not a ({\em homogeneous}) self-similar set in the usual sense (see, e.g., \cite{falc,hutchinson}), defined via iterated function systems (or, in short, IFS), but it is an {\em inhomogeneous}\index{homogeneous self-similar set} self-similar set.\label{IFS_label}

More specifically, the set $A$ is the unique nonempty compact subset of $\eR^2$ which is the solution of the inhomogeneous fixed point equation\index{fixed point equation!inhomogeneous} 
\begin{equation}\label{4.2.102.1/2EE}
A=\bigcup_{j=1}^2\Phi_j(A)\cup B,
\end{equation}
where $\Phi_1$ and $\Phi_2$ are contractive similitudes of $\eR^2$ with fixed points located at the lower left vertex and the upper right vertex of the unit square, respectively, and with a common scaling ratio equal to $1/2$ (i.e., $r_1=r_2=1/2$, where $\{r_j\}_{j=1}^2$ are the scaling ratios of the self-similar RFD $(A,\O)$).
Furthermore, the nonempty compact set $B$ in Equation \eqref{4.2.102.1/2EE} is the union of the left and upper sides of the square $G_1$ and the right and lower sides of the square $G_2$; see Figure \ref{kv_0.5}.
We note that here, the corresponding (classic or homogeneous) self-similar set (i.e., the unique nonempty compact subset $C$ of $\eR^2$ which is the solution of the homogeneous fixed point equation, $C=\cup_{j=1}^2\Phi_j(C)$), is the diagonal $C$ of the unit square connecting the lower left and the upper right vertices of the unit square.

\begin{figure}[ht]
\begin{center}
\includegraphics[width=9cm]{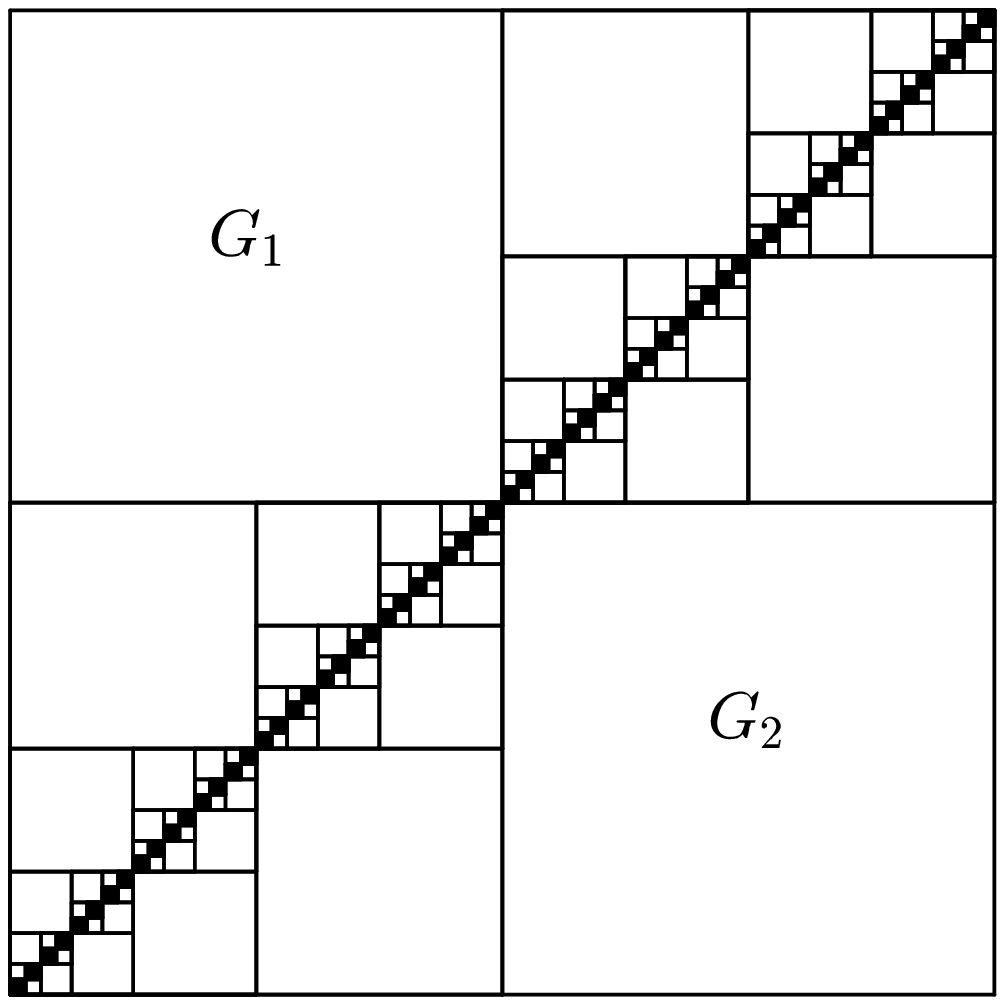}
\caption{The $1/2$-square fractal $A$ from Example \ref{kvadrat0.5}. The first 6 iterations are depicted. Here, $G:=G_1\cup G_2$ is the single generator of the corresponding self-similar spray or RFD $(A,\O)$, in the sense of Definition \ref{ss_spray}. The set $A$ is equal to the complement of the union of the disjoint family of all open squares, with respect to $\O=(0,1)^2$. Equivalently, the set $A$ coincides with the closure of the union of the boundaries of all the open squares.}\label{kv_0.5}
\end{center}
\end{figure}

Let us now compute the distance zeta function $\zeta_A$ of the $1/2$-square fractal.
Without loss of generality, we may assume that $\delta>1/4$; so that we have
\begin{equation}\label{eqq1}
\zeta_{A}(s)=\zeta_{A,\O}(s)+\zeta_{I}(s),
\end{equation}
where, intuitively, $\zeta_{I}$ denotes the distance zeta function corresponding to the `outer' $\d$-neigh\-bor\-hood of $A$.
Clearly, $\zeta_I$ is equal to the distance zeta function of the unit square $I:=[0,1]^2$; it is straightforward to compute it and show that it has a meromorphic extension to all of $\Ce$ given by
\begin{equation}\label{eqq2}
\zeta_{I}(s)=\frac{4\delta^{s-1}}{s-1}+\frac{2\pi\delta^s}{s},
\end{equation}
for all $s\in\Ce$.

Furthermore, by using Theorem \ref{ss_spray_zeta}, we obtain that
\begin{equation}\label{eqq3}
\zeta_{A,\O}(s)=\frac{\zeta_{\pa G,G}(s)}{1-2\cdot 2^{-s}}=\frac{2^s\zeta_{\pa G,G}(s)}{2^s-2},
\end{equation}
for all $s\in\Ce$ with $\re s$ sufficiently large.
Next, we compute the distance zeta function of $(\pa G,G)$ by subdividing $G=G_1\cup G_2$ into 16 congruent triangles (see also Figure \ref{sierpinski_carpetr2}, which describes the way we subdivide both $G_1$ and $G_2$) and by using local Cartesian coordinates $(x,y)\in\eR^2$ to deduce that
\begin{equation}\nonumber
\zeta_{\pa G,G}(s)=16\int_0^{1/4}\di x\int_0^{x}y^{s-2}\di y=\frac{4^{-s}}{s(s-1)},
\end{equation}
for all $s\in\Ce$ with $\re s>1$.
Hence,
\begin{equation}\label{eqq4}
\zeta_{\pa G,G}(s)=\frac{4^{-s}}{s(s-1)},
\end{equation}
an identity valid initially for all $s\in\Ce$ such that $\re s>1$, and then, after meromorphic continuation, for all $s\in\Ce$.
Finally, by combining Equations \eqref{eqq1}--\eqref{eqq4}, we conclude that the distance zeta function $\zeta_A$ is meromorphic on all of $\Ce$ and is given by
\begin{equation}\label{1/2-square_dist}
\zeta_A(s)=\frac{2^{-s}}{s(s-1)(2^s-2)}+\frac{4\delta^{s-1}}{s-1}+\frac{2\pi\delta^s}{s},
\end{equation}
for all $s\in\Ce$.

Consequently, we have that $\dim_BA$ exists,\footnote{The existence of $\dim_BA$ in Example \ref{kvadrat0.5} (as well as in Examples \ref{kvadrat_0.33} and \ref{ss_fractal_nest} below) follows from \cite[Theorem 5.4.30]{fzf} (see also \cite[Theorem 4.2]{mm}.)
}
\begin{equation}\label{4.2.108E}
\begin{gathered}
D(\zeta_A)=\dim_BA=1,\\
\po(\zeta_A):=\po(\zeta_A,\Ce)=\{0\}\cup\left(1+\mathbf{p}\I\Ze\right)
\end{gathered}
\end{equation}
and
\begin{equation}\label{PCAsq}
\dim_{PC}A:=\po_c(\zeta_A)=1+\mathbf{p}\I\Ze,
\end{equation}
where the oscillatory period\index{oscillatory period!of the $1/2$-square fractal} $\mathbf{p}$ of $A$ is given by $\mathbf{p}:=\frac{2\pi}{\log 2}$.
All of the complex dimensions in $\po(\zeta_A)$ are simple except for $\omega=1$, which is a double pole of $\zeta_A$.
Finally, we note that in light of Equation \eqref{4.2.108E} (and hence, in light of the presence of nonreal complex dimensions), the set $A$ is indeed {\em fractal} according to our proposed definition of fractality given in Remark \ref{5.151/2} and further discussed in Chapter \ref{fractality} below. In fact, according to Equation \eqref{PCAsq}, it is {\em critically fractal}\index{critical fractality|textbf} (i.e., fractal in dimension $d:=1=\dim_BA$, in the sense of \S\ref{sub_rfd}).
\end{example}

\begin{example}({\em The $1/3$-square fractal}).\label{kvadrat_0.33}\index{squaree@$1/3$-square fractal|textbf}
In the present planar example, we illustrate a situation which is similar to that of the inhomogeneous Sierpi\'nski $N$-gasket RFD discussed in Example \ref{Ngasket} for $N\geq 4$.
Again, we start with the closed unit square $I=[0,1]^2$ in $\eR^2$ and subdivide it into $9$ smaller congruent squares (similarly as in the case of the Sierpi\'nski carpet).
Next, we remove 7 of those smaller squares; that is, we only leave the lower left and the upper right squares (see Figure \ref{kv_0.33}).
In other words, our generator $G$ (in the sense of Definition \ref{ss_spray}) is the (nonconvex) open polygon depicted in Figure \ref{kv_0.33}.

As usual, we proceed by iterating this procedure with the two remaining closed squares and then continue this process ad infinitum.
(The first 4 iterations are depicted in Figure \ref{kv_0.33}.)
The $1/3$-square fractal is then defined as the set $A$ which remains at the end of the process.
We now let $\O:=(0,1)^2$, which makes the RFD $(A,\O)$ a self-similar spray (or tiling), in the sense of Definition \ref{ss_spray}, with generator $G$ and scaling ratios $\{r_j\}_{j=1}^2$ such that $r_1=r_2=1/3$.
Again, the set $A$ is not a homogeneous self-similar set, but is instead an inhomogeneous self-similar set.

More specifically, the set $A$ is the unique nonempty compact subset of $\eR^2$ which is the solution of the inhomogeneous equation
\begin{equation}\label{4.2.110.1/2EE}
A=\bigcup_{j=1}^2\Phi_j(A)\cup B,
\end{equation}
where $\Phi_1$ and $\Phi_2$ are contractive similitudes of $\eR^2$ with fixed points located at the lower left vertex and the upper right vertex of the unit square, respectively, and with a common scaling ratio equal to $1/3$.
Furthermore, the set $B$ in Equation \eqref{4.2.110.1/2EE} is equal to the boundary of $G$ without the part belonging to the boundary of the two smaller squares which are left behind in the first iteration; see Figure \ref{kv_0.33}.
We also observe that here, the corresponding (classic or homogeneous) self-similar set generated by the IFS consisting of $\Phi_1$ and $\Phi_2$, is the ternary Cantor set located along the diagonal of the unit square.

\begin{figure}[ht]
\begin{center}
\includegraphics[width=9cm]{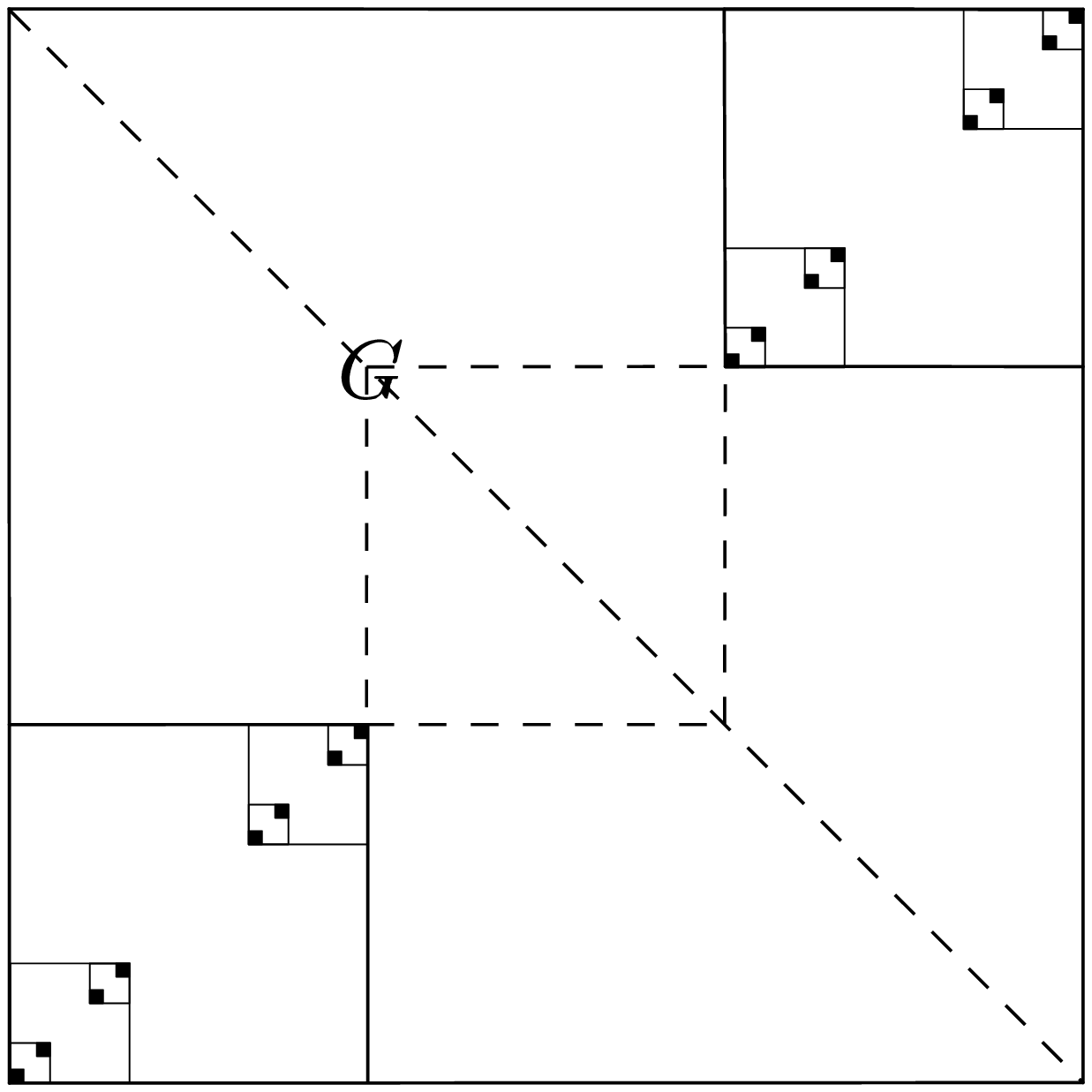}
\caption{The $1/3$-square fractal $A$ from Example \ref{kvadrat_0.33}. The first 4 iterations are depicted. Here, $G$ is the single generator of the corresponding self-similar spray or RFD $(A,\O)$, in the sense of Definition \ref{ss_spray}. The set $A$ is equal to the complement of the union of the disjoint family of all the open $8$-gons, with respect to the open square $\O=(0,1)^2$. The largest $8$-gon is equal to the union of two open squares indicated by dashed sides of length $2/3$, while each of the next two smaller $8$-gons is obtained by scaling the first one by the factor $1/3$.  Any of the $2^k$ $8$-gons of the $k$-th generation is obtained by scaling the first one by the factor $1/3^{k-1}$, for any $k\in\eN$. Equivalently, $A$ coincides with the closure of the union of the boundaries of all the $8$-gons.}\label{kv_0.33}
\end{center}
\end{figure}

We now proceed by computing the distance zeta function $\zeta_A$ of the $1/3$-square fractal.
Without loss of generality, we may assume that $\delta>1/4$; so that we have
\begin{equation}\label{eqq13}
\zeta_{A}(s)=\zeta_{A,\O}(s)+\zeta_{I}(s),
\end{equation}
where, as before in Example \ref{kvadrat0.5}, $\zeta_{I}$ denotes the distance zeta function corresponding to the `outer' $\d$-neighborhood of $A$ and coincides with the distance zeta function of the unit square $I:=[0,1]^2$.
Recall that $\zeta_{I}$ was computed in Example \ref{kvadrat0.5} and is given by Equation \eqref{eqq2}.

Furthermore, by using Theorem \ref{ss_spray_zeta}, we obtain that
\begin{equation}\label{eqq33}
\zeta_{A,\O}(s)=\frac{\zeta_{\pa G,G}(s)}{1-2\cdot 3^{-s}}=\frac{3^s\zeta_{\pa G,G}(s)}{3^s-2},
\end{equation}
for all $s\in\Ce$ with $\re s$ sufficiently large.

Next, we compute the distance zeta function of $(\pa G,G)$ by subdividing $G$ into 14 congruent triangles denoted by $G_i$, for $i=1,\ldots,14$ (see Figure \ref{kv_0.33}).
Therefore, by symmetry, we obtain the following functional equation:
\begin{equation}\label{eqq43}
\zeta_{\pa G,G}(s)=12\zeta_{\pa G,G_1}(s)+2\zeta_{\pa G,G_{13}},
\end{equation}
valid initially for all $s\in\Ce$ such that $\re s$ is sufficiently large.

We use local Cartesian coordinates $(x,y)\in\eR^2$ in order to compute $\zeta_{\pa G,G_1}$ and obtain that
\begin{equation}\nonumber
\zeta_{\pa G,G_1}=\int_{0}^{1/3}\di x\int_0^{x}y^{s-2}\di y=\frac{3^{-s}}{s(s-1)}.
\end{equation}
Hence,
\begin{equation}\label{eqq53}
\zeta_{\pa G,G_1}=\frac{3^{-s}}{s(s-1)},
\end{equation}
an identity valid initially for all $s\in\Ce$ such that $\re s>1$, and then, after meromorphic continuation, for all $s\in\Ce$.
In order to compute $\zeta_{\pa G,G_{13}}$, we use local polar coordinates $(r,\theta)$ and deduce that
\begin{equation}\label{eqq63}
\begin{aligned}
\zeta_{\pa G,G_{13}}(s)&=\int_0^{\pi/2}\di\theta\int_0^{3^{-1}(\sin\theta+\cos\theta)^{-1}}r^{s-1}\di r\\
&=\frac{3^{-s}}{s}\int_0^{\pi/2}(\cos\theta+\sin\theta)^{-s}\di\theta,
\end{aligned}
\end{equation}
valid, initially, for all $s\in\Ce$ such that $\re s>0$ and then, after meromorphic continuation, for all $s\in\Ce$.
It is easy to check that
\begin{equation}\label{eqq73}
Z(s):=\int_0^{\pi/2}(\cos\theta+\sin\theta)^{-s}\di\theta
\end{equation}
is an entire function, since it is a generalized DTI $f(s):=\int_E \f(\theta)^s\di\mu(\theta)$, where $E:=[0,\pi/2]$,  $\f(\theta):=(\cos\theta+\sin\theta)^{-1}$ for all $\theta\in E$ is uniformly bounded by positive constants both from above and below, and $\di\mu(\theta):=\di\theta$.

Finally, by combining Equation \eqref{eqq2} and Equations \eqref{eqq13}--\eqref{eqq73}, we obtain that $\zeta_A$ is given by
\begin{equation}\label{zeta_1/3_square}
\zeta_{A}(s)=\frac{2}{s(3^s-2)}\left(\frac{6}{s-1}+Z(s)\right)+\frac{4\delta^{s-1}}{s-1}+\frac{2\pi\delta^s}{s},
\end{equation}
an identity valid initially for all $s\in\Ce$ with $\re s>1$ and then, after meromorphic continuation, for all $s\in\Ce$.

Consequently, we deduce that $\dim_BA$ exists,
\begin{equation}\label{c_dim_A1/3}
\begin{gathered}
D(\zeta_A)=\dim_BA=1,\\
\po(\zeta_A):=\po(\zeta_A,\Ce)\subseteq\{0\}\cup\left(\log_32+\mathbf{p}\I\Ze\right)\cup\{1\}
\end{gathered}
\end{equation}
and
\begin{equation}\label{PCA1}
\dim_{PC}A:=\po_c(\zeta_A)=\{1\},
\end{equation}
where the oscillatory period\index{oscillatory period!of the $1/3$-square fractal} $\mathbf{p}$ of $A$ is given by $\mathbf{p}:=\frac{2\pi}{\log 3}$.
In Equation \eqref{c_dim_A1/3}, we only have an inclusion since, in principle, some of the complex dimensions with real part $\log_32$ may be canceled by the zeros of $6/(s-1)+Z(s)$.
However, it can be checked numerically that $\log_32\in\po(\zeta_A)$ and that there also exist nonreal complex dimensions with real part $\log_32$ in $\po(\zeta_A)$. 
All of the complex dimensions in $\po(\zeta_A)$ are simple.
We also note that $A$ is indeed {\em fractal},\index{fractality and complex dimensions} according to our proposed definition of fractality (see Remark \ref{5.151/2} above and Chapter \ref{fractality} below).
More precisely, in light of Equations \eqref{c_dim_A1/3} and \eqref{PCA1}, it is {\em strictly subcritically fractal}\index{strictly subcritical fractal|textbf} and {\em fractal in dimension} $d=\log_32$, in the sense of~\S\ref{sub_rfd}.
\end{example}

\begin{example}({\em A self-similar fractal nest}).\label{ss_fractal_nest}\index{fractal nest|textbf}\index{self-similar!fractal nest|textbf}
In the final planar example of this section, we investigate the case of a self-similar fractal nest.\footnote{As we shall see, throughout this example, the use of the adjective ``self-similar'' is somewhat abusive since only one similarity transformation is involved.}
The set $A$ which we now define is an inhomogeneous self-similar set.
Similarly as in Example \ref{kvadrat_0.33}, the set $A$ will be {\em fractal}\index{fractality and complex dimensions} in the sense of our proposed definition of fractality given in Remark \ref{5.151/2} and, moreover, will be {\em strictly subcritically fractal}\index{strictly subcritically fractal|textbf} in the sense of~\S\ref{sub_rfd}.

Let $a\in(0,1)$ be a real parameter.
We define the set $A$ as the union of concentric circles with center at the origin and of radius $a^k$ for $k\in\eN_0$ (see Figure \ref{nest_0.8}).
Furthermore, let $G$ be the open annulus such that $\pa G$ consists of the circles of radius $1$ and $a$, as depicted in Figure \ref{nest_0.8}, and let $\O:=B_1(0)$.
We can now consider the RFD $(A,\O)$ as a self-similar spray with generator $G$, in the sense of Definition \ref{ss_spray}.

We note that even though $(A,\O)$ is a fractal spray, with a single generator $G$, it is not (strictly speaking) self-similar in the traditional sense because it only has one scaling ratio $r=a$ (associated with a single contractive similitude).
However, we will continue using this abuse of language throughout this example.
Also, a moment's reflection reveals that this fact does not affect any of the conclusions relevant to the distance zeta function of such an RFD.
\begin{figure}[ht]
\begin{center}
\includegraphics[width=9cm]{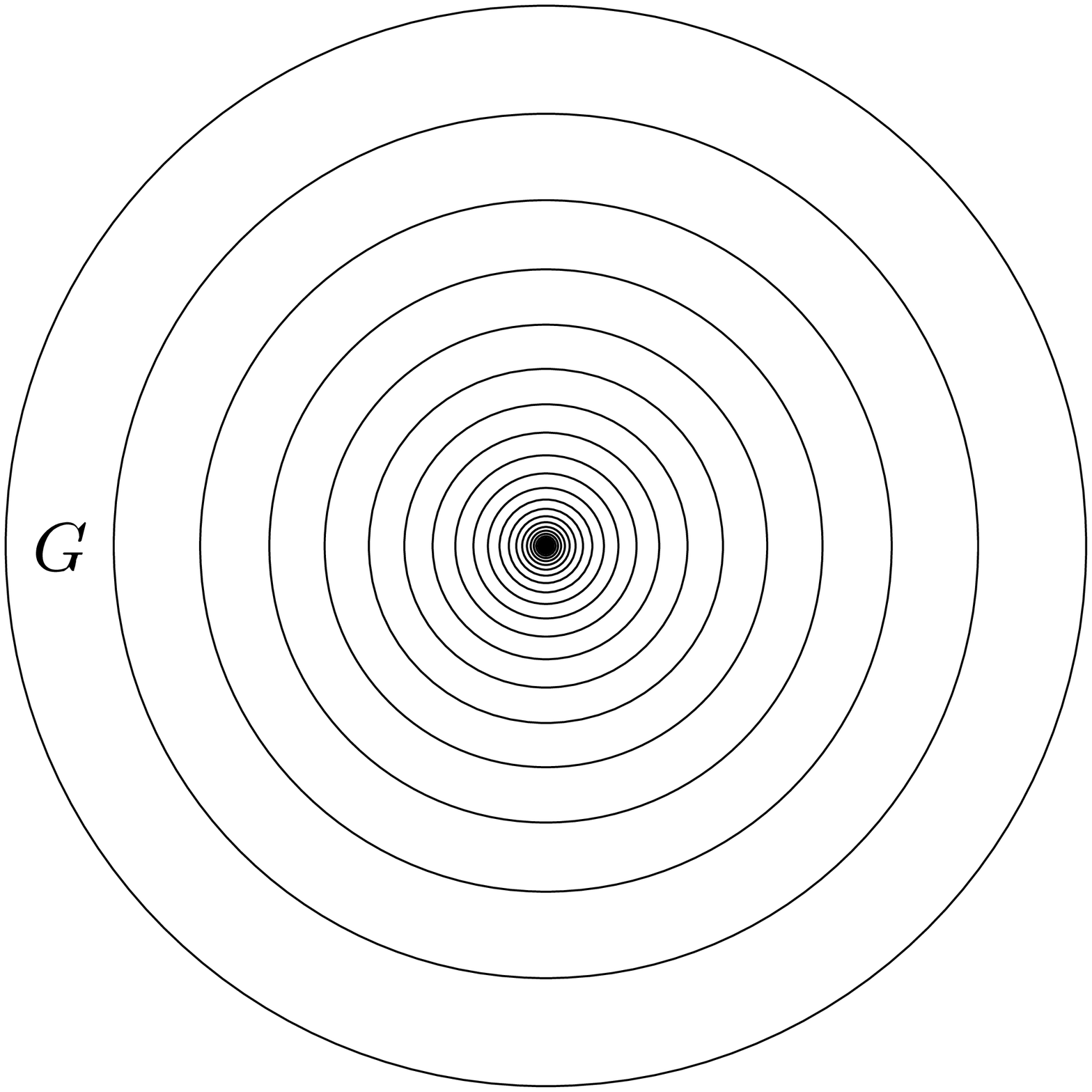}
\caption{\small The self-similar fractal nest from Example \ref{ss_fractal_nest}.}\label{nest_0.8}
\end{center}
\end{figure}
Namely, we obviously have that
\begin{equation}
(A,\O)=(\pa G,G)\sqcup a(A,\O);
\end{equation}
so that
\begin{equation}
\zeta_{A,\O}(s)=\zeta_{\pa G,G}(s)+\zeta_{a(A,\O)}(s),
\end{equation}
for all $s\in\Ce$ such that $\re s$ is sufficiently large.
Furthermore, by using the scaling property of the relative distance zeta function (see Theorem \ref{scaling}), we conclude that
\begin{equation}\label{eqA1}
\zeta_{A,\O}(s)=\frac{\zeta_{\pa G,G}(s)}{1-a^s},
\end{equation}
again, for all $s\in\Ce$ such that $\re s$ is sufficiently large.

Next, we compute the distance zeta function of the generator by using polar coordinates $(r,\theta)$:
\begin{equation}\label{eqA2}
\begin{aligned}
\zeta_{\pa G,G}(s)&=\int_0^{2\pi}\di\theta\int_a^{(1+a)/2}(r-a)^{s-2}r\,\di r\\
&\phantom{=}+\int_0^{2\pi}\di\theta\int_{(1+a)/2}^{1}(1-r)^{s-2}r\,\di r\\
&=\frac{2^{2-s}\pi(1+a)(1-a)^{s-1}}{s-1},
\end{aligned}
\end{equation}
an identity valid, after meromorphic continuation, for all $s\in\Ce$.

Equation \eqref{eqA2} combined with Equation \eqref{eqA1} now yields that $\zeta_{A,\O}$ is meromorphic on all of $\Ce$ and is given for all $s\in\Ce$ by
\begin{equation}
\zeta_{A,\O}(s)=\frac{2^{2-s}\pi(1+a)(1-a)^{s-1}}{(s-1)(1-a^s)}.
\end{equation}

Finally, we fix an arbitrary $\d>(1-a)/2$ and observe that for the distance zeta function of $A$, we have
\begin{equation}\label{eqA3}
\zeta_{A}(s)=\zeta_{A,\O}(s)+\zeta_{A,B_{1+\d}(0)\setminus\O}(s),
\end{equation}
for all $s\in\Ce$ with $\re s$ sufficiently large.
Furthermore, we have that
\begin{equation}\label{eqA4}
\zeta_{A,B_{1+\d}(0)\setminus\O}(s)=\int_0^{2\pi}\di\theta\int_1^{1+\d}(r-1)^{s-2}r\,\di r=\frac{2\pi\d^{s-1}}{s-1}+\frac{2\pi\d^s}{s},
\end{equation}
where the last equality is valid, initially, for all $s\in\Ce$ such that $\re s>1$, and then, after meromorphic continuation, for all $s\in\Ce$.

Combining now the above equation with \eqref{eqA3}, we finally obtain that $\zeta_A$ is meromorphic on all of $\Ce$ and is given by
\begin{equation}\label{ss_nest_zeta}
\zeta_{A}(s)=\frac{2^{2-s}\pi(1+a)(1-a)^{s-1}}{(s-1)(1-a^s)}+\frac{2\pi\d^{s-1}}{s-1}+\frac{2\pi\d^s}{s},
\end{equation}
for all $s\in\Ce$.

Consequently, we have that $\dim_BA$ exists,
\begin{equation}\label{pizA}
\begin{gathered}
D(\zeta_A)=\dim_BA=1\\
\po(\zeta_A):=\po(\zeta_A,\Ce)=\mathbf{p}\I\Ze\cup\{1\}
\end{gathered}
\end{equation}
and
\begin{equation}
\dim_{PC}(A):=\po_c(\zeta_A)=\{1\},
\end{equation}
where the oscillatory period\index{oscillatory period!of a self-similar fractal nest} $\mathbf{p}$ of $A$ is given by $\mathbf{p}:=\frac{2\pi}{\log a^{-1}}$ and all of the complex dimensions in $\po(\zeta_A)$ are simple.

In closing, we mention that $A$ is indeed {\em fractal}\index{fractality and complex dimensions} according to our proposed definition of fractality (see Remark \ref{5.151/2} and Chapter \ref{fractality}).
More specifically, in light of Equation \eqref{pizA}, $A$ is {\em strictly subcritically fractal}\index{strictly subcritically fractal|textbf} and {\em fractal in dimension} $d:=0$, in the sense of~\S\ref{sub_rfd} below.
\end{example}


\section{Generating complex dimensions of RFDs of any multiplicity}
\label{multiplicity}

A key tool in generating (principal) complex dimensions of higher multiplicities is the tensor product of bounded fractal strings, which we now briefly define; see \cite{mezf} for more details. If $\mathcal L_1:=(\ell_{1j})_{j\ge1}$ and $\mathcal L_2:=(\ell_{2k})_{k\ge1}$ are two given bounded fractal strings, then the {\em tensor product}\index{tensor product!of fractal strings, ${\mathcal L}_1\otimes{\mathcal L}_2$|textbf} ${\mathcal L_1}\otimes{\mathcal L_2}$ is defined as the multiset consisting of all possible products of the form $\ell_{1j}\ell_{1k}$  for all ordered pairs $(j,k)\in\eN^2$; hence, we take into account the multiplicities. It is easy to see that the tensor product ${\mathcal L_1}\otimes{\mathcal L_2}$ is also a bounded fractal string. Furthermore, we have that the geometric zeta function of the tensor product is equal to the product of each of the component geometric zeta functions. More specifically, we have that
\begin{equation}\label{ztens}
\zeta_{{\mathcal L_1}\otimes{\mathcal L_2}}(s)=\zeta_{\mathcal L_1}(s)\cdot\zeta_{\mathcal L_2}(s)
\end{equation}
for all $s\in\Ce$ with $\re s>\max\{D(\zeta_{\mathcal L_1}),D(\zeta_{\mathcal L_2})\}$ and
$$
D(\zeta_{{\mathcal L_1}\otimes{\mathcal L_2}})=\max\{D(\zeta_{\mathcal L_1}),D(\zeta_{\mathcal L_2})\};
$$
see \cite[Lemma 4.13]{mezf}. 

The following example provides a class of bounded relative fractal drums, generated by an $a$-string, which illustrates Theorem \ref{rel_measurable} above. Note that here, we have a unique, nonsimple, principal complex dimension, $D$, on the critical line, and that its multiplicity is equal to an arbitrarily prescribed positive integer $m$. We shall need the notion of the {\em disjoint union}\index{disjoint union!of bounded fractal strings, $\sqcup_{m=1}^\ty\mathcal L_m$} (of an at most countable family) {\em of bounded fractal strings} $\mathcal L_m=(\ell_{mj})_{j\ge1}$ for $m\in\eN$:
\begin{equation}\label{union_s}
\mathcal L:=\bigsqcup_{m=1}^\ty\mathcal L_m,
\end{equation}
defined as the multiset $\mathcal L$ consisting of elements of the union of fractal strings, counting their multiplicities.
 Assuming, additionally, that $\ell_{m1}\to0^+$ as $m\to\ty$, it is easy to see that the mulitplicity of each of the elements of the multiset $\mathcal L:=\sqcup_{m=1}^\ty\mathcal L_m$ is finite, so that $\mathcal L$ is indeed a bounded fractal string.

\begin{example} ({\em $m$-th order $a$-string})\label{Lmloga}\index{a@$a$-string of higher order|textbf} 
Let $\mathcal L(a):=\{\ell_k:=k^{-a}-(k+1)^{-a}\}_{k=1}^{\ty}$ be the $a$-string, where $a>0$ (see \cite[Example 5.1]{Lap1} and \cite[\S6.5.1]{lapidusfrank12}), and let $m$ be a positive integer. Let 
$\mathcal L_m(a)$ be defined by $\mathcal L_1(a):=\mathcal L(a)$ for $m=1$ and as the $(m-1)$-fold tensor product for $m\ge2$; that is,
\begin{equation}\label{Lma}
\mathcal L_m(a):=
\begin{cases}
\mathcal L(a)&\mbox{for $m=1$},\\
\mathcal L(a)\otimes\cdots\otimes\mathcal L(a)&\mbox{for $m\ge2$},
\end{cases}
\end{equation}
which we call the {\em $m$-th order $a$-string}. 
Then $D=1/(1+a)$ is the only principal complex dimension of $\mathcal L_m(a)$, and it is of multiplicity $m$, since,
in light of Equation \eqref{ztens}, we have that
\begin{equation}\label{zLma}
\zeta_{\mathcal L_m(a)}(s)=[\zeta_{\mathcal L(a)}(s)]^m=\left(\sum_{k=1}^\ty\ell_k^s\right)^m,
\end{equation}
for all $s\in\Ce$ with $\re s>1/(1+a)$.
Defining $h(t):=(\log t^{-1})^{m-1}$ for all $t\in(0,1)$, and using \cite[Theorem 5.4 and Example 3.7]{mm}, we deduce that the fractal string $\mathcal L_m(a)$ is $h$-Minkowski measurable. Moreover, also according to \cite[Theorem 5.4]{mm} 
(see also \cite[Theorem 5.4.27 and Example 5.5.10]{fzf}), we have that $\mathcal L_m(a)$ has the following tube asymptotics:
\begin{equation}
|A_t|=t^{1-D}h(t)\,(\mathcal M+o(1))\q\mbox{as $t\to0^+$,}
\end{equation}
where $\mathcal M\in(0,+\ty)$ is the {\em $h$-Minkowski content} of $\mathcal L_m(a)$ and can be explicitly computed in terms of the $-m$-th coefficient $c_{-m}$ of the Laurent expansion\index{Laurent expansion} of the tube zeta function $\tilde\zeta_{A}$ around $s=D$, as follows: ${\mathcal M}=c_{-m}/(m-1)!$; see \cite[Theorem 5.4]{mm}. In particular, 
$\mathcal L_m(a)$ is {\em $h$-Minkowski measurable}. 
\end{example}
\medskip

In \cite[\S4.4]{mezf}, we have constructed a (Cantor-type) bounded fractal string $\mathcal L_m$ which has infinitely many principal complex dimensions of arbitrary prescribed multiplicity $m\ge2$. 
The bounded fractal string was obtained by taking $m-1$ consecutive tensor products of the usual Cantor string $\mathcal L_{CS}$; i.e., the $(m-1)$-fold tensor product:
\begin{equation}\label{Lm}
{\mathcal L}_m:=
\begin{cases}
{\mathcal L}_{CS}&\mbox{for $m=1$},\\
{\mathcal L}_{CS}\otimes\dots\otimes{\mathcal L}_{CS}&\mbox{for $m\ge2$},
\end{cases}
\end{equation}
which we call {\em $m$-th order Cantor string}\index{Cantor string!of higher order ($m$-Cantor string)|textbf} or the {\em $m$-Cantor string}, in short.
The corresponding multiset of principal complex dimensions is
\begin{equation}\label{CSm}
\dim_{PC}{\mathcal L}_m=\log_32+\frac{2\pi}{\log 3}\I\Ze,
\end{equation}
and each of its elements has multiplicity $m$. Note that by \cite[Theorem 3.1]{mm} (see also
\cite[Theorem 5.4.20]{fzf}),
the $m$-th order Cantor string ${\mathcal L}_m$ is not Minkowski measurable for $m\ge2$.


Furthermore, by letting
\begin{equation}\label{Lty}
{\mathcal L}_\ty:=\bigsqcup_{m=1}^{\ty} \frac{3^{-m}}{m!}{\mathcal L}_m,
\end{equation}
we obtain a bounded fractal string ${\mathcal L}_\ty$, called the {\em Cantor string of infinite order}\index{Cantor string!of infinite order ($\ty$-Cantor string)|textbf} or the {\em $\ty$-Cantor string}, such that its geometric zeta function $\zeta_{{\mathcal L}_\ty}$
has an infinite sequence of {\em essential}\index{essential singularities on the critical line} singularities along the critical line $\{\re s=D\}$, 
located at each of the points $D+\I{\mathbf p}k$ (with $k\in\Ze$, $D:=\log_32$ and ${\mathbf p}:=2\pi/\log3$)
of the periodic set
defined by the right-hand side of \eqref{CSm}. 
\medskip

\begin{example}
Let $m$ be a fixed positive integer and let $a$ be a positive real number chosen small enough, so that $D:=1/(1+a)>\log_32$. Consider the following bounded fractal string $\mathcal L$ defined by
\begin{equation}\label{Lmixed}
{\mathcal L}:=\mathcal L_m(a)\sqcup\mathcal L_\ty,
\end{equation}
where the bounded fractal strings ${\mathcal L}_m(a)$ ($m$-th order $a$-string) and $\mathcal L_\ty$ ($\ty$-Cantor string) are defined by Equations \eqref{Lma} and \eqref{Lty}, respectively, and generated by tensor products of $a$-strings and Cantor strings, respectively. Here, we have that $D_{\rm mer}(\zeta_{\mathcal L})=\log_32$, since the geometric zeta function
\begin{equation}
\zeta_{\mathcal L}(s)=\zeta_{\mathcal L_m(a)}(s)+\zeta_{\mathcal L_\ty}(s)
\end{equation}
 is holomorphic on the connected open set $\{\re s>0\}\setminus\big(\{D\}\cup(\log_32+\frac{2\pi}{\log3}\I\Ze)\big)$, where $D=\dim_B\mathcal L$ and is the (unique) pole of $\zeta_{\mathcal L}$ of order $m$ in the open right half-plane $\{\re s>0\}$, while $\log_32+\frac{2\pi}{\log3}\I\Ze$ is the set of essential singularities\index{essential singularities on the critical line} of $\zeta_{\mathcal L}$ in $\{\re s>0\}$. Denoting by $A_{\mathcal L}:=\{a_k:=\sum_{j=k}^\ty\ell_j:k\in\eN\}$ the canonical representation of the fractal string $\mathcal L:=\{\ell_j\}_{j=1}^\ty$, and applying \cite[Theorem 5.4.27]{fzf} to the RFD $(A_{\mathcal L},(A_{\mathcal L})_\d)$ (for any fixed positive real number $\d$),\footnote{The open right half-plane $\{\re s>D_{\rm mer}(\zeta_{\mathcal L})\}$ does not contain any other poles of $\zeta_{\mathcal L}$, except for $s=D$.} we obtain the following asymptotic expansion of the tube function of the set $A_{\mathcal L}$:
\begin{equation}
|(A_{\mathcal L})_t|=t^{1-D}h(t)\big(\mathcal M+O(t^{D-D_{\rm mer}(\zeta_{\mathcal L})-\e}\big)\q\mbox{as $t\to0^+$,}
\end{equation} 
for any $\e>0$, where $h(t):=(\log t^{-1})^{m-1}$ for all $t\in(0,1)$; i.e.,
\begin{equation}\label{tubeLmixed}
|(A_{\mathcal L})_t|=t^{a/(1-a)}h(t)\big(\mathcal M+O(t^{\frac1{1+a}-\log_32-\e}\big)\q\mbox{as $t\to0^+$,}
\end{equation}
 where $\mathcal M$ is a positive real number (the $h$-Minkowski content) and can be computed (see \cite[Theorem 4.5]{mm} or \cite[Theorem 5.4.27]{fzf}). According to \cite[Theorem 5.6]{mm} (or \cite[Theorem 5.4.29]{fzf}), the exponent $\frac1{1+a}-\log_32$ appearing on the right-hand side of Equation \eqref{tubeLmixed}, is optimal; i.e., it cannot be replaced by a larger exponent.
\end{example}
\medskip

In Examples \ref{tensorr} and \ref{tensorm} below, we construct {\em Minkowski measurable}\index{Minkowski measurable RFD} RFDs which possess infinitely many complex dimensions of arbitrary multiplicity $m$, with $m\ge1$, or even essential singularities. 

\begin{example}\label{tensorr}
Let us first define
the unit square RFD $(A_0,\O_0)$ by $\O_0:=[0,1]^2$ and $A_0:=\pa\O_0$. We introduce the RFD
\begin{equation}
(A_m',\O_m'):=(A_0,\O_0)\sqcup\mathcal L_m,
\end{equation}
where we embed $\mathcal L_m$ via its {\em canonical geometric representation}\index{canonical geometric representation $A_{\mathcal L}$ of a bounded fractal string $\mathcal L$} $A_{\mathcal L_m}$ into the $x$-axis of the $2$-dimensional plane $\eR^2$.
Since $\zeta_{A_m',\O_m'}(s)=\zeta_{A_0,\O_0}(s)+\zeta_{\mathcal L_m}(s)$, we have that
\begin{equation}
\po(\tilde\zeta_{A_m',\O_m'})=\{0,1\}\cup{\mathcal P}',
\end{equation}
where $\mathcal P':=\log_32+\frac{2\pi}{\log3}\I\Ze$ and each of the complex dimensions
$D+\I{\mathbf p}k$ (with $k\in\Ze$, $D:=\log_32$) of $\mathcal P'$
is of multiplicity $m$.
On the other hand, the only principal complex dimension of $(A_m',\O_m')$ is $1$, and it is simple (i.e.,  of multiplicity $1$). 
Therefore, according to \cite[Theorem 5.2]{cras2}, the RFD $(A_m',\O_m')$ is Minkowski measurable. 
\end{example}

\begin{example}\label{tensorm} Let us again define
the unit square RFD $(A_0,\O_0)$ by $\O_0:=[0,1]^2$ and $A:=\pa\O$. We introduce the RFD 
\begin{equation}
(A_m,\O_m):=(A_0,\O_0)\otimes\mathcal L_m,
\end{equation}
where $\mathcal L_m$ is the $m$-th order Cantor string defined by \eqref{Lm}, and the tensor product
of the RFD $(A_0,\O_0)$ and $\mathcal L_m$  is defined analogously as in Example \ref{tensorE} above.
Using Equation \eqref{dirichlete} from Theorem \ref{sprayz} above (see also Equation \eqref{10151/4c} from Theorem \ref{sprayzc}), we obtain that
\begin{equation}
\zeta_{A_m,\O_m}(s)=\zeta_{A_0,\O_0}(s)\cdot\zeta_{\mathcal L_m}(s)=\frac{g(s)}{s(s-1)(3^s-2)^m}.
\end{equation}
Here, $g(s)$ is an entire function without zeros at $0$, $1$ or at any point of the arithmetic set 
\begin{equation}\label{P'}
{\mathcal P}':=\log_32+\frac{2\pi}{\log3}\I\Ze.
\end{equation}
 In other words, 
\begin{equation}
\dim_{PC}(A_m,\O_m)=\{1\}, \q\po(A_m,\O_m)=\{0,1\}\cup{\mathcal P}',
\end{equation}
and each complex dimension of $(A,\O)$ lying in the arithmetic set ${\mathcal P}'$ has multiplicity $m$.
The value of $D:=\ov\dim_B(A_m,\O_m)=1$ is the only complex dimension located on the critical line $\{\re s=1\}$, while the infinite set ${\mathcal P}'$ is contained in the vertical line $\{\re s=\log_32\}$ located strictly to the left of the critical line.
It follows from \cite[Theorem 5.4]{mm} (or \cite[Theorem 5.4.29]{fzf}) that the RFD $(A,\O)$ is $h$-Minkowski measurable with respect to the gauge function $h(t):=(\log t^{-1})^{m-1}$, for all $t\in(0,1)$.

We can further define the RFD
\begin{equation}
(A_\ty,\O_\ty):=\bigsqcup_{m=2}^\ty\frac{3^{-m}}{m!}\cdot(A_m,\O_m).
\end{equation}
Similarly as above, we have that (with ${\mathcal P}'$ given by \eqref{P'})
\begin{equation}
\dim_{PC}(A_\ty,\O_\ty)=\{1\}, \q\po(A_\ty,\O_\ty)=\{0,1\}\cup{\mathcal P}',
\end{equation}
and each complex dimension of $(A_\ty,\O_\ty)$ lying in the arithmetic set ${\mathcal P}'=\log_32+\frac{2\pi}{\log3}\I\Ze$ is an essential singularity\index{essential singularities on the critical line} of $\zeta_{A_\ty,\O_\ty}$.
\end{example}



\chapter{Fractality, complex dimensions and singularities}\label{fractality}

We close this article by specifying, within the general theory of fractal zeta functions developed
here and in [{LapRa\v Zu1--8}], the elusive notion of ``fractality''. Much as in [{Lap-vFr1--3}]
(see, especially, \cite[\S12.1 and \S13.4.3]{lapidusfrank12}), but now using the general higher-dimensional notion of fractal zeta function and associated notion of complex dimensions, we say that a bounded set $A$ (or, more generally, an RFD $(A,\O)$) in 
$\eR^N$ is {\em fractal}\index{fractality and complex dimensions|textbf} if it has at least one nonreal (visible) complex dimension
(i.e., a nonreal pole for its associated fractal zeta function),\footnote{Provided $D:=\ov\dim_BA$ (resp., $\ov\dim_B(A,\O)$) $<N$, it does not matter whether we use $\zeta_A$ or $\tilde\zeta_A$
(resp., $\zeta_{A,\O}$ or $\tilde\zeta_{A,\O}$) throughout this definition.} relative to some screen $\bm S$, or else if there exists a screen $\bm S$ which is a (meromorphic) natural boundary for its fractal zeta function (i.e., such that the fractal zeta function cannot be meromorphically extended to the left of $\bm S$). In the latter situation, $A$ (or, more generally, $(A,\O)$) is said to be {\em hyperfractal}.\index{hyperfractal} In particular, it is said to be {\em strictly hyperfractal}\index{strict hyperfractal|textbf}\index{hyperfractal!strict|textbf} if we may choose $\bm{S}=\{\re s=D\}$, and {\em maximally hyperfractal}\index{hyperfractal!maximal|textbf} if the critical line $\bm{S}=\{\re s=D\}$ consists entirely of nonremovable singularities of the fractal zeta function; see Definition \ref{hyperfractal}. Here, as before, we let $D:=\ov\dim_BA$ (or $D:=\ov\dim_B(A,\O)$). Recall that in Theorem \ref{qp}, we have constructed a family of maximally hyperfractal RFDs.

\section{Fractal and subcritically fractal RFDs}\label{sub_rfd}

In this work, we have seen many examples of fractals (that are not hyperfractal), for instance,
the Cantor string or set, the relative Sierpi\'nski gasket and carpet (Examples \ref{6.15} and \ref{sierpinski_carpetr}) or, more generally, the relative $N$-gasket RFD and the $N$-carpet RFD (Examples \ref{Ngasket} and \ref{carpetN}), along with the examples discussed in \S\ref{goran}.
Among these examples, some have nonreal complex dimensions located on the critical line (such as, for instance, the Cantor string,
the inhomogeneous Sierpi\'nski gasket and carpet RFDs, the $N$-carpet RFD for any $N\ge2$, as well as the inhomogeneous $N$-gasket $(A_N,\O_N)$  when $N=2$ or $3$). These are called {\em critically fractal}.
Yet others only have nonreal complex dimensions with real parts strictly less than $D$. The latter are called {\em subcritically fractal}.
In addition, {\em strictly subcritical fractals} are subcritical fractals which do not have any nonreal principal complex dimensions (i.e., complex dimensions with real part $D$).
Examples of strictly subcritical fractals include the inhomogeneous Sierpi\'nski $N$-gasket RFD when $N\ge4$ (Example \ref{6.15} above) , the $1/3$-square fractal (Example \ref{kvadrat_0.33}), a self-similar fractal nest (Example~\ref{ss_fractal_nest}), as well as the modified devil's staircase (or Cantor graph) RFD to be discussed in Example \ref{stair} below.

Finally, we complete this list of definitions
by stating that, given $d\in\eR$, the bounded set $A$ (or, more generally, the RFD $(A,\O)$)
is {\em fractal in dimension} $d$ if it has nonreal complex dimensions of real parts $d$. (In light of Theorem \ref{an_rel}, we must then necessarily have $d\le N$.)
Hence, a critical fractal is such that $d:=D$, while a strictly subcritical fractal is such that $d<D$.
For instance, with the notation of Example \ref{Ngasket}, for $N\ge4$, the Sierpi\'nski $N$-gasket RFD $(A_N,\O_N)$ is fractal in dimension 
\begin{equation}
d=\s_0=\log_2(N+1)<D=N-1=\dim_B(A_{N,0},\O_{N,0})
\end{equation}
but not in dimension $D$, and therefore, it is strictly subcritically fractal. By contrast, when $N=2$ or $3$, it is critically fractal 
(indeed, in those cases, it is fractal in dimension $d:=D=\s_0$, the similarity dimension).


We point out that, much as as was the case in the one-dimensional situation in \cite[Chapter 12]{lapidusfrank12}, based on the general explicit formulas and fractal tube formulas\index{fractal tube formula} obtained in [{Lap-Fr1--3}] (see, especially, \cite[Chapters 5 and 8]{lapidusfrank12}), the definitions of fractality, critical fractality and (strict) subcritical fractality are justified in part by the general fractal tube formulas obtained in \cite{cras2} (see also [LapRa\v Zu4] and \cite[Chapter 5]{fzf}).\footnote{These fractal tube formulas generalize to any $N\ge1$ and to arbitrary bounded sets $A$ (or, more generally, RFDs) in $\eR^N$ the ones obtained for fractal strings (i.e., when $N=1$) in [{Lap-vFr1--3}]
(see, especially, \cite[Chapter 8]{lapidusfrank12}), as well as for the very special but important higher-dimensional case of fractal sprays, in [{LapPe2--3}] and, more generally, in [{LapPeWi1--2}] (see \cite[\S13.1]{lapidusfrank12} for an exposition).} 
Indeed, the latter tube formulas show that, under mild assumptions, the presence of nonreal complex dimensions of real part $d\in\eR$ corresponds to oscillations of order $d$ in the geometry of $A$ (or of $(A,\O)$). Similarly, roughly speaking, critical fractality (along with the simplicity of $D$) corresponds to the Minkowski nonmeasurability of $A$ (or $(A,\O)$), while strict subcritical fractality  (still assuming the simplicity of $D$) not only corresponds to (critical) Minkowski measurability but also to (strictly subcritical) Minkowski nonmeasurability in dimension $d<D$. (See also [LapRa\v Zu6].)
This is the case, for instance, for the inhomogeneous Sierpi\'nski $N$-gasket RFD\index{Sierpi\'nski $N$-gasket RFD} (see Example \ref{6.15}) whenever $N\ge4$ (and avoiding nongeneric values of $N$), for the RFDs of Examples \ref{kvadrat_0.33} and \ref{ss_fractal_nest}, as well as for the (modified) devil's staircase RFD, which we discuss in Example \ref{stair} just below.



Finally, we note that it follows from the discussion in Remark \ref{fractality0} (and from Theorem~\ref{ss_spray_zeta}) above that, under mild assumptions on their generators,\footnote{It suffices to assume that the base RFD $(\pa G,G)$ is ``nonfractal'' (so that it does not have any nonreal complex dimensions) and ``sufficiently nice'' (so that $\zeta_{\pa G,G}$ has a meromorphic continuation to all of $\Ce$). Both conditions are satisfied, for instance, if $G$ is the interior of a convex polytope, which is the case for essentially all of the classical examples.} self-similar sprays (in the sense of Definition~\ref{ss_spray}) are fractal in dimension $d$ for only a finite (but nonempty) set of values of $d$ in the lattice case, whereas they are fractal in dimension $d$ for an infinite countable and dense set of values of $d$ in the nonlattice case. More specifically, the set of $d$'s for which nonlattice (respectively, generic nonlattice) self-similar RFDs are fractal in dimension $d$ is dense in finitely many nonempty compact intervals (respectively, in a single compact interval of the form $[D_l,D]$, where $D_l\in\eR$ and $D_l<D$).\footnote{We refer to Remark \ref{fractality0} for the definitions of the terms ``lattice'', ``nonlattice'' and ``generic nonlattice'', as well as for the appropriate references.}
More generally, we conjecture that under suitable mild hypotheses, self-similar RFDs and sets satisfying the open set condition enjoy the same properties.

\medskip

\section{The Cantor graph relative fractal drum}\label{cg_rfd}

Recall that in the introduction (i.e., in \S1), we have discussed the Cantor graph (or devil's staircase)\index{Cantor graph (full)|(} in relation with the notion of fractality. We end this chapter by considering a closely related example, namely, the Cantor graph RFD.

\begin{example}\label{stair}({\em The Cantor graph RFD}).\index{Cantor graph RFD|textbf}
In this example, we compute the distance zeta function of the RFD $(A,\O)$ in $\eR^{2}$, where $A$ is the graph of the Cantor function and $\O$ is the union of triangles $\triangle_k$ that lie above and the triangles $\tilde{\triangle}_k$ that lie below each of the horizontal parts of the graph denoted by $B_k$.
(At each step of the construction there are $2^{k-1}$ mutually congruent triangles $\triangle_k$ and $\tilde{\triangle}_k$.)
Each of these triangles is isosceles, has for one of its sides a horizontal part of the Cantor function graph, and has a right angle at the left end of $B_k$, in the case of $\triangle_k$, or at the right end of $B_k$, in the case of $\tilde{\triangle}_k$.
(See Figure~\ref{devil_sl}.)

\begin{figure}[ht]
\begin{center}
\includegraphics[trim=6cm 3cm 9cm 2cm,clip=true,width=8cm]{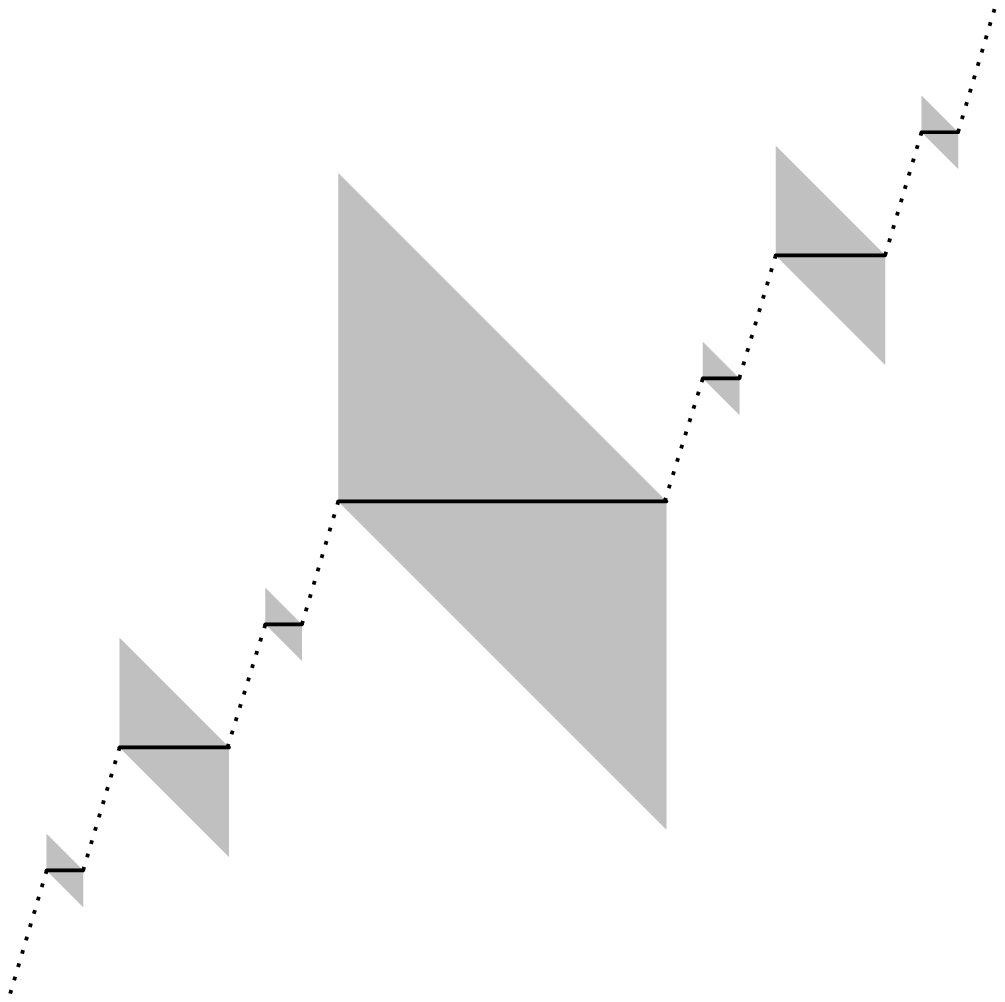}
\end{center}
\caption{\small The third step in the construction of the Cantor graph relative fractal drum $(A,\O)$ from Example~\ref{stair}.
One can see, in particular, the sets $B_k$, $\triangle_k$ and $\tilde{\triangle}_k$ for $k=1,2,3$.}
\label{devil_sl}
\end{figure}

For obvious geometric reasons and by using the scaling property of the relative distance zeta function of the resulting RFD $(A,\O)$ (see Theorem~\ref{scaling}), we then have the following identity:
\begin{equation}\label{eq5.5.52}
\begin{aligned}
\zeta_{A,\O}(s)&=\sum_{k=1}^{\ty}2^{k}\zeta_{B_k,\triangle_k}(s)
=\sum_{k=1}^{\ty}2^{k}\zeta_{3^{-k}B_1,3^{-k}\triangle_1}(s)\\
&=\zeta_{B_1,\triangle_1}(s)\sum_{k=1}^{\ty}\frac{2^{k}}{3^{ks}}=\frac{2\zeta_{B_1,\triangle_1}(s)}{3^s-2},
\end{aligned}
\end{equation}
valid for all $s\in\Ce$ with $\re s$ sufficiently large. 
Here, $(B_1,\triangle_1)$ is the relative fractal drum described above with two perpendicular sides of length equal to 1.
It is straightforward to compute its relative distance zeta function:
\begin{equation}
\zeta_{B_1,\triangle_1}(s)=\int_0^1\di x\int_0^xy^{s-2}\di y=\frac{1}{s(s-1)},
\end{equation}
valid, initially, for all $s\in\Ce$ such that $\re s>1$ and then, upon meromorphic continuation, for all $s\in\Ce$.
This fact, combined with (the last equality of) Equation \eqref{eq5.5.52}, yields the distance zeta function of $(A,\O)$, which is clearly meromorphic on all of $\Ce$:
\begin{equation}\label{zeta_devil_stair}
\zeta_{A,\O}(s)=\frac{2}{s(3^s-2)(s-1)},\quad\textrm{for all}\ s\in\Ce.
\end{equation}
We therefore deduce that the set of complex dimensions of the RFD $(A,\O)$ is given by
\begin{equation}\label{devil_dim}
\po(\zeta_{A,\O}):=\po({\zeta}_{A,\O},\Ce)=\{0,1\}\cup\left(\log_32+\frac{2\pi}{\log3}\I\Ze\right),
\end{equation}
with each complex dimension being simple.
Hence, its set of principal complex dimensions is given by
\begin{equation}\label{dimAO1}
\dim_{PC}(A,\O):=\po_c(\zeta_{A,\O})=\{1\}.
\end{equation}

We conclude from part ($b$) of Theorem~\ref{an_rel} that $\dim_B(A,\O)=1$ and that the RFD $(A,\O)$ is Minkowski measurable.
Moreover, one also deduces from \cite[Theorem~4.2]{mm} 
that the (one-dimensional) Minkowski content of $(A,\O)$ is given by 
\begin{equation}\label{devil_mink}
\mathcal{M}^{1}(A,\O)=\frac{\res(\zeta_{A,\O},1)}{2-1}=2,
\end{equation}
which coincides with the length of the Cantor graph (i.e., the graph of the Cantor function, also called the devil's staircase in \cite{Man}).

In the sequel, we associate the RFD $(A,A_{1/3})$ in $\eR^2$ to the classic Cantor graph.
We do not know if the right-hand side of~\eqref{devil_dim} coincides with the set of complex dimensions of the `full' graph of the Cantor function (i.e., the original devil's staircase), or equivalently, the RFD $(A,A_{1/3})$, but we expect that this is indeed the case since $(A,\O)$ is a `relative fractal subdrum' of $(A,A_{1/3})$.
Moreover, it clearly follows from the construction of $(A,\O)$ that for the distance zeta function of the RFD $(A,A_{1/3})$ associated with the graph of the Cantor function, we have
\begin{equation}\label{pole-pole}
{\zeta}_{A,A_{1/3}}(s)=\zeta_{A,\O}(s)+\zeta_{A,A_{1/3}\setminus\O}(s).
\end{equation}
In order to prove that $\po(\zeta_{A,\O})$, given by~\eqref{devil_dim}, is a subset of the set of complex dimensions of the `full' Cantor graph, it would therefore remain to show that $\zeta_{A,A_{1/3}\setminus\O}$ has a meromorphic continuation to some connected open neighborhood $U$ of the critical line $\{\re s=1\}$ such that $U$ contains the set of complex dimensions of $(A,\O)$, as given by \eqref{devil_dim}, and that there are no pole-pole cancellation in the right-hand side of~\eqref{pole-pole}.
\end{example}

We now return to the RFD $(A,\O)$ (that is, the Cantor graph relative fractal drum).
It follows from \eqref{devil_dim} that $(A,\O)$ is fractal, in our sense.
More specifically, in light of \eqref{dimAO1}, {\em it is not critically fractal} (because its only complex dimension of real part $D_{CG}\ (=\ov{D}=\dim_{B}(A,\O))=1$ is $1$ itself, the Minkowski dimension of the Cantor graph RFD, and it is simple) {\em but it is strictly subcritically fractal}.
In fact, it is subcritically fractal in a single dimension, namely, in dimension $d:=D_{CS}=\log_32$, the Minkowski dimension of the Cantor set.

We expect the exact same statements to be true for the devil's staircase itself (i.e., the `full' graph of the Cantor function), represented by the RFD $(A,A_{1/3})$ and of which $(A,\O)$ is a `relative fractal subdrum', as was explained above.
Clearly, in light of \eqref{pole-pole} and \eqref{devil_dim}, we have the following inclusions (between multisets):
\begin{equation}\label{5.5.34.4/5}
\begin{aligned}
\po(\zeta_{A,A_{1/3}})&\subseteq\po(\zeta_{A,\Omega})\cup\po(\zeta_{A,A_{1/3}\setminus\O})\\
&\subseteq\{0,1\}\cup\left\{D_{CS}+\frac{2\pi}{\log 3}\I\Ze\right\}.
\end{aligned}
\end{equation}

Also, we know for a fact that $\dim_B(A,A_{1/3})$ exists and 
\begin{equation}\label{DA1/3}
D(\zeta_{A,A_{1/3}})=\dim_B(A,A_{1/3})=1,
\end{equation}
so that
\begin{equation}\label{DA1/3PC}
\dim_{PC}(A,A_{1/3}):=\po_c(\zeta_{A,A_{1/3}})=\{1\}.
\end{equation}
(Thus, we have that $\{1\}\stq \po(\zeta_{A,A_{1/3}})$ in \eqref{5.5.34.4/5}.) Note that \eqref{DA1/3} (and hence, \eqref{DA1/3PC}) follows from the rectifiability of the devil's staircase, combined with a well-known result in \cite{federer} and with part $(b)$ of Theorem \ref{an_rel}.

As was conjectured in \cite[\S12.1.2 and \S12.3.2]{lapidusfrank12}, 
based on an `approximate tube formula', we expect that $\po(\zeta_{A,A_{1/3}})=\po(\zeta_{A,\O})$, as given by \eqref{devil_dim}, and hence, that we actually have equalities instead of inclusions in \eqref{5.5.34.4/5}, even equalities between multisets.
If so, then the `full' Cantor graph $(A,A_{1/3})$ is fractal, not critically fractal, but (strictly) subcritically fractal in the single dimension $d:=D_{CS}=\log_32$.

In his celebrated book, {\em The Fractal Geometry of Nature} [Man], Mandelbrot reluctantly defined ``fractality'' by the property that a geometric object has Hausdorff dimension strictly greater than (i.e., different from) its topological dimension; see [Man, p.\ 15]. However, he was aware of an obvious counterexample to his definition; namely, the Cantor graph (or devil's staircase, depicted in [Man, plate 83, p.\ 83]), for which all the notions of fractal dimensions (Hausdorff, Minkowski, etc.) coincide with the topological dimension (i.e., one). In this regard, he stated in [Man, p.\ 82] about the devil's staircase that ``{\em one would love to call the present curve a fractal, but to achieve this goal, we would have to define fractals less stringently, on the basis of notions other than {\rm [the Hausdorff dimension]} alone}.''

The above paradox has puzzled the first author from the very beginning and was one of the key motivations for the development of the mathematical theory of complex dimensions, eventually in [Lap-vFr1--3], for fractal strings (i.e., when $N=1$), and now (in higher dimensions) in [LapRa\v Zu1--8]. If we replace the (full) Cantor graph $(A,A_{1/3})$ by the Cantor graph RFD $(A,\O)$, then the paradox is completely resolved since, as was discussed above, $(A,\O)$ is ``fractal'', in the sense of the theory of complex dimensions. Nevertheless, the fact that $(A,\O)$ is strictly subcritically fractal (i.e., does not have nonreal principal complex dimensions, that is, with real part $1$, but has nonreal complex dimensions with real part $<1$, namely, on the vertical line $\{\re s=\log_32\}$) shows that the issue at hand is rather subtle. 

Since, according to the above discussion, the (full) Cantor graph (or devil's staircase) is also expected to be fractal (as well as strictly subcritically fractal), the original paradox should itself be completely resolved in the future within the present theory of complex dimensions of relative fractal drums. Naturally, we expect that many other apparent paradoxes can be similarly resolved within the theory developed in this paper and in [LapRa\v Zu1--8].\index{Cantor graph (full)|)}

\backmatter

\bibliographystyle{amsalpha}


\renewcommand{\indexname}{Subject Index}


\printindex



\end{document}